\documentclass[12pt]{article}
\usepackage[margin=1in]{geometry}

\usepackage[utf8]{inputenc} 
\usepackage[T1]{fontenc}    
\usepackage{url,bbm,mathrsfs,mathtools,bm,amsthm,amsmath,algorithm,algorithmic,afterpage,amssymb,setspace,amsfonts,array,verbatim,newfloat,caption,afterpage,float,rotating,subfiles,csquotes,indentfirst,cancel,subfigure,enumitem,xcolor} 
\allowdisplaybreaks
\usepackage[singlelinecheck=false]{caption}
\usepackage[backend=biber,style=numeric,citestyle=numeric-comp,sorting=none,giveninits=true,maxcitenames=1,uniquelist=false,maxbibnames=10,natbib=true]{biblatex}
\bibliography{refs}

\usepackage{units}
\DeclareMathOperator{\Tr}{Tr}
\DeclareMathOperator{\naive}{(naive)}
\DeclareMathOperator{\known}{(known)}
\DeclareMathOperator{\fs}{(FS)}
\DeclareMathOperator{\ipw}{(IPW)}
\DeclareMathOperator{\im}{Im}
\DeclareMathOperator{\diag}{diag}
\DeclareMathOperator*{\argmin}{argmin}
\DeclareMathOperator*{\argmax}{argmax}

\DeclareMathOperator{\E}{\mathbb{E}}
\DeclareMathOperator{\Prob}{Pr}

\DeclareMathOperator{\V}{\mathbb{V}}

\DeclareMathOperator{\myvec}{vec}
\DeclarePairedDelimiter\abs{\lvert}{\rvert}%
\DeclarePairedDelimiter\norm{\lVert}{\rVert}%
\DeclareMathOperator{\edist}{\stackrel{d}{=}}
\DeclareMathOperator{\Bprob}{Pr}

\DeclareMathOperator{\tdist}{\stackrel{\text{d}}{\to}}
\DeclareMathOperator{\IPW}{(IPW)}

\makeatletter
\let\oldabs\abs
\def\abs{\@ifstar{\oldabs}{\oldabs*}}
\let\oldnorm\norm
\def\norm{\@ifstar{\oldnorm}{\oldnorm*}}
\makeatother

\makeatletter
\newcommand*\bigcdot{\mathpalette\bigcdot@{1.0}}
\newcommand*\bigcdot@[2]{\mathbin{\vcenter{\hbox{\scalebox{#2}{$\m@th#1*$}}}}}
\makeatother

\numberwithin{equation}{section}

\newcommand{\LtabMetab}{
 \begin{minipage}{\linewidth}
   \centering
   \footnotesize
   \vspace{0.25cm}
   \captionof{table}[Parameters used to simulate latent confound effects.]{The $\pi_k$ and $\tau_k$ values used to simulate $\bm{\ell}_{1},\ldots,\bm{\ell}_{p}$ ($k=1,\ldots,10$).}\label{supp:Table:LMetab}
	\begin{tabular}{c | c | c | c | c | c | c | c | c | c | c}
 Factor number ($k$) & $1$ & $2$ & $3$ & $4$ & $5$ & $6$ & $7$ & $8$ & $9$ & $10$\\\hline
  $\pi_k$ & 0 & 0 & 0.80 & 0.60 & 0.50 & 0.35 & 0.30 & 0.20 & 0.20 & 0.20\\\hline
  $\tau_k$ & 0.80 & 0.60 & 0.5 & 0.5 & 0.5 & 0.5 & 0.5 & 0.5 & 0.5 & 0.5
  \end{tabular}
  \vspace{0.4cm}
  \end{minipage}
}

\newtheorem{Assumption}{\textit{Assumption}}

\newtheorem{remark}{\textit{Remark}}

\newtheorem{Proposition}{\textit{Proposition}}
\newtheorem{corollary}{\textit{Corollary}}
\newtheorem{theorem}{\textit{Theorem}}
\newtheorem{lemma}{\textit{Lemma}}
\renewcommand{\theequation}{\thesection.\arabic{equation}}
\renewcommand{\thetheorem}{\thesection.\arabic{theorem}}
\renewcommand{\theremark}{\thesection.\arabic{remark}}

\renewcommand{\thelemma}{\thesection.\arabic{lemma}}

\begin{document}

\title{From differential abundance to mtGWAS: accurate and scalable methodology for metabolomics data with non-ignorable missing observations and latent factors}

\author{Shangshu Zhao$^{1}$, Kedir Turi$^{2}$, Tina Hartert$^{2}$, Carole Ober$^{3}$,\\Klaus B{\o}nnelykke$^{4}$, Bo Chawes$^{4}$, Hans Bisgaard$^{4}$, Chris McKennan$^{1,\ast}$\\[4pt]
\textit{$^1$Department of Statistics,
University of Pittsburgh}
\\
\textit{$^2$Department of Medicine, Vanderbilt University Medical Center}\\
\textit{$^3$Department of Human Genetics, University of Chicago}\\
\textit{$^4$COPSAC, Copenhagen Prospective Studies on Asthma in Childhood,}\\\textit{Herlev and Gentofte Hospital, University of Copenhagen}\\[2pt]
{\footnote{To whom correspondence should be addressed.} chm195@pitt.edu}}


\maketitle


\begin{abstract}
Metabolomics is the high-throughput study of small molecule metabolites. Besides offering novel biological insights, these data contain unique statistical challenges, the most glaring of which is the many non-ignorable missing metabolite observations. To address this issue, nearly all analysis pipelines first impute missing observations, and subsequently perform analyses with methods designed for complete data. While clearly erroneous, these pipelines provide key practical advantages not present in existing statistically rigorous methods, including using both observed and missing data to increase power, fast computation to support phenome- and genome-wide analyses, and streamlined estimates for factor models. To bridge this gap between statistical fidelity and practical utility, we developed MS-NIMBLE, a statistically rigorous and powerful suite of methods that offers all the practical benefits of imputation pipelines to perform phenome-wide differential abundance analyses, metabolite genome-wide association studies (mtGWAS), and factor analysis with non-ignorable missing data. Critically, we tailor MS-NIMBLE to perform differential abundance and mtGWAS in the presence of latent factors, which reduces biases and improves power. In addition to proving its statistical and computational efficiency, we demonstrate its superior performance using three real metabolomic datasets.
\end{abstract}

\allowdisplaybreaks

\section{Introduction}
\label{section:introdution}
Metabolomics is the high-throughput study of small molecule metabolites, and can help understand human variation and the etiology of disease \citep{Bilirubin}. Metabolite abundances are typically measured via mass spectrometry, which, while sensitive, produces a large amount of non-ignorable missing data in which low abundance metabolites are less likely to be observed \citep{MetabMiss}. This precludes the use of the many complete data methods able to perform the three core metabolomic analyses: differential abundance, metabolome genome-wide association studies (mtGWAS), and factor analysis \citep{CMS_metab,MetabMiss}. Factor analysis, while important in its own right, is required in differential abundance analyses and, as we show in Section~\ref{subsection:mtGWAS}, mtGWAS, as it helps recover latent factors that plague metabolomic data and confound relationships of interest \citep{MetabMiss}.

Consequently, nearly all existing analysis pipelines first impute missing data, which acts as a crude solution to issues of method incompatibility \citep{Impute2} and offers the following important practical advantages: (i) ensuing analyses use both observed and missing data to improve power, (ii) downstream computation is fast enough to perform metabolite phenome- and genome-wide studies, and (iii) factor models can be estimated. Despite its expedience, it is well known that imputing non-ignorable missing data can beget biased estimators and spurious inference \citep{ImputationFlaws}. However, to our knowledge, \citet{MetabMiss} is the only work to provide a rigorous alternative to imputation while also considering latent confounding factors. Although a step in the right direction, their work does not offer the aforementioned advantages of imputation, as it discards missing data and does not provide methodology to perform an mtGWAS. And while it does provide a method to perform factor analysis, its theoretical properties are completely unknown. Therefore, it is questionable whether the statistical rigor offered by \citet{MetabMiss} is sufficient to offset the expediency of imputation.

To bridge the gap between statistical fidelity and practical utility, we developed MS-NIMBLE (\underline{M}ethod\underline{s} for \underline{N}on-\underline{I}gnorable \underline{M}issing \underline{M}eta\underline{b}o\underline{l}omic Obs\underline{e}rvations), a suite of statistically rigorous methods to perform differential abundance, mtGWAS, and factor analysis in metabolomic data that offers all of the practical advantages of imputation. Like \citet{MetabMiss}, we estimate each metabolite's missingness mechanism once per dataset and store it to facilitate efficient downstream computation. However, unlike \citet{MetabMiss}, subsequent estimators use both observed and missing data by leveraging the approximate conditional normality of metabolite levels. Our method for mtGWAS is able to partition low rank and idiosyncratic genetic variation, and we prove the statistical and computational efficiency of our factor analysis-related and other estimators. We lastly use simulated and three real metabolomic datasets to show that MS-NIMBLE significantly outperforms the method proposed in \citet{MetabMiss} and existing imputation pipelines. An R package and code to reproduce our simulations are available from https://github.com/chrismckennan/MSNIMBLE.

%







\section{Notation, problem setup, and statistical models}
\label{section:Problem}

Let $[m] = \{1,\ldots,m\}$ for $m>0$ and $y_{gi}$ be the possibly missing log-abundance of metabolite $g \in [p]$ in sample $i \in [n]$. For observed covariates $\bm{x}_i \in \mathbb{R}^d$ and latent factors $\bm{c}_i \in \mathbb{R}^K$, assume
\begin{align}
\label{equation:MainModel}
    y_{gi} = \bm{\beta}_g^{\top}\bm{x}_i + \bm{\ell}_g^{\top}\bm{c}_i + e_{gi}, \quad (e_{g1},\ldots,e_{gn})^{\top} \sim N(0,\sigma_g^2 I_n), \quad g \in [p]; i \in [n]
\end{align}
for some unknown and non-random $\bm{\beta}_g \in \mathbb{R}^d$ and $\bm{\ell}_g \in \mathbb{R}^K$. We will assume the number of latent factors $K$ is known, although we estimate $K$ with parallel analysis \citep{BujaFA} in practice. In differential abundance, $\bm{\beta}_g$ is of interest and $\bm{c}_i$ confounds the relationship between $\bm{x}_i$ and $y_{gi}$. In factor analysis and mtGWAS, $\bm{\beta}_g$ is a nuisance parameter and $\bm{c}_i$ and $\bm{\ell}_g$ are of interest. Other than assuming the design matrix with rows $(\bm{x}_i^{\top}, \bm{c}_i^{\top})$, $i\in [n]$, is full rank, we assume nothing about the relationship between $\bm{x}_i$ and $\bm{c}_i$, which facilitates the analysis of data with arbitrarily complex latent confounding. While our theoretical results require assumptions on the moments of $\bm{c}_i$, our methodology is agnostic to these assumptions, and therefore postpone their discussion to Section~\ref{section:Theory}. The normality of $e_{gi}$, which we leverage to design efficient estimators, is a common assumption in mass spectrometry data \citep{Impute2,Impute1}. However, we do not require $y_{gi}$ be normal, as the elements of $\bm{c}_i$ are often highly skewed (Figure~\ref{supp:Figure:Normal}).

It is well known metabolite levels depend on genotype \citep{CMS_metab}. However, since genotype does not appear in \eqref{equation:MainModel}, it is possible that its effect is mediated by $\bm{c}_i$ or appears in the idiosyncratic error terms $e_{gi}$, which belies the canonical factor analysis assumption that $\bm{c}_i$ is independent of $e_{gi}$ \citep{BCconf}. The genetic effects in $e_{gi}$ also imply the normality of $e_{gi}$ may only be an approximation, and that $e_{gi}$, $e_{hi}$ may be dependent for $g \neq h$. Our theoretical work in Section~\ref{section:Theory} accommodate all of these observations.


To describe the missing data model, let $r_{gi} = I(y_{gi}$ is observed$)$. We follow \citet{MetabMiss} and assume that for some known cumulative distribution function $\Psi$ and unknown, metabolite-specific scale and location parameters $\alpha_g\geq 0$ and $\delta_g \in \mathbb{R}$,
\begin{align}
\label{equation:MissingDataModel}
    \Prob(r_{gi}=1 \mid y_{gi},\bm{x}_i,\bm{c}_i) = \Prob(r_{gi}=1 \mid y_{gi}) = \Psi\{ \alpha_g(y_{gi} - \delta_{g}) \}, \quad g\in [p]; i\in [n],
\end{align}
where $\{r_{gi}\}_{g \in [p]; i \in [n]}$ are independent conditional on $\{y_{gi}\}_{g \in [p]; i \in [n]}$. This, along with the assumptions that $\alpha_g\geq 0$ and the distribution of $r_{gi}$ only depends on $y_{gi}$, is justified because nearly all missing data are due to an artifact of the mass spectrometer, where analytes with low abundances are less likely to be observed \citep{MetabMiss}. \citet{MetabMiss} contains additional justifications of \eqref{equation:MissingDataModel}.

We assume $\Psi$ in \eqref{equation:MissingDataModel} is known, which is ostensibly allowed to be any cumulative distribution function (CDF). While typical choices for $\Psi$ include the CDFs of the logistic and normal distributions \citep{Impute1}, our theoretical work in Section~\ref{section:Theory} requires the left hand tail of $\Psi$ go to zero no faster than a polynomial rate. Our default choice for $\Psi$ is therefore the CDF of the t-distribution with four degrees of freedom, which we show gives excellent results in real data.

\section{When do the missing data matter?}
\label{section:Imputation}
Ignoring or incorrectly modeling non-ignorable missing data can bias estimators \citep{ImputationFlaws}. Despite this, differential abundance simulations routinely suggest that errant imputation techniques have a trivial effect on type I error \citep{Impute1}. This begs the question when, or if, we have to account for the non-ignorable missing data in metabolomic analyses. We study this in Proposition~\ref{proposition:Impute}, which analyzes estimates from errantly imputed data. 

\begin{Proposition}
\label{proposition:Impute}
Let $x_i\in \mathbb{R}$. Assume \eqref{equation:MainModel} satisfies $y_{gi}=\mu_g + x_i\beta_g + \bm{c}_i^{\top} \bm{\ell}_g + e_{gi}$, $(e_{g1},\ldots,e_{gn})^{\top} \sim N(0,\sigma_g^2 I_n)$, and the regularity conditions in Section~\ref{supp:section:MinImp} hold. Suppose \eqref{equation:MissingDataModel} holds and we impute missing $y_{gi}$'s as $a* \min( \{ y_{gi} \}_{\{i: r_{gi}=1\}} )$ for any constant $a \in \mathbb{R}$. Then for $\hat{\beta}_g,\hat{s}_g$ the resulting ordinary least squares estimate and standard error for $\beta_g$ when $\bm{c}_i$ is known, $(\hat{\beta}_g-\beta_g)/\hat{s}_g \to N(0,1)$ as $n \to \infty$ if (i) the null hypothesis $H_{0,g}:\beta_g=0$ holds and (ii) $\bm{\ell}_g=0$ or $x_i$ is independent of $\bm{c}_i$.
\end{Proposition}

\begin{remark}
\label{remark:MinImputation}
Minimum imputation from Proposition~\ref{proposition:Impute} is one of the most common ways to handle missing metabolomic data \citep{Impute2}. Note $\bm{c}_i$ is observed in Proposition~\ref{proposition:Impute}.
\end{remark}

\noindent Proposition~\ref{proposition:Impute} shows errant imputation can beget valid type I error rates provided (ii) holds, i.e. $\bm{c}_i$ does not confound the relationship between $x_i$ and $y_{gi}$. This result explains the abovementioned befuddling observations that incorrectly modelling simulated non-ignorable missing metabolomic data has a trivial effect on type I error rates, since their simulations did not consider confounders.

The proof of the asymptotic normality in Proposition~\ref{proposition:Impute} relies on $x_i$ being independent of $y_{gi}$, which is only true if (i) and (ii) hold. This suggests properly handling missing $y_{gi}$'s is critical when estimating intervals for non-zero effects $\beta_g$, and when controlling type I error in the presence of confounding factors $\bm{c}_i$, even when $\bm{c}_i$ is observed. We show this using simulated and real data. 





\section{Estimation and inference with MS-NIMBLE}
\label{section:Estimation}
We must overcome several challenging features of \eqref{equation:MainModel}, \eqref{equation:MissingDataModel}, and metabolomic experiments in general. First, \eqref{equation:MainModel} is not congruent with existing maximum likelihood estimators designed for normally distributed data \citep{Impute1}, since $\bm{c}_i$'s distribution may be highly non-normal (Figure~\ref{supp:section:Normal}). Second, leveraging the approximate normality of the errors $e_{gi}$ to improve estimates requires integrating over missing $y_{gi}$, which can be prohibitively slow for theoretically valid choices of $\Psi$ discussed in Section~\ref{section:Theory}. Lastly, our estimators must scale to facilitate phenome- and genome-wide analyses. Figure~\ref{Figure:Overview} gives an overview of the steps in our method. For simplicity of presentation, we assume in Sections~\ref{section:Estimation} and \ref{section:Theory} that all metabolites have missing data, but provide extensions in supplemental Section~\ref{supp:section:FullyObs} to allow fully observed metabolites. Section~\ref{subsection:MissMech} gives a brief description of the estimators for $\alpha_g,\delta_g$, as they mirror those from \citet{MetabMiss}. Sections~\ref{subsection:LF}-\ref{subsection:mtGWAS} contain detailed descriptions of our novel methodological components.

\begin{figure}
\centering
\includegraphics[width=0.75\textwidth]{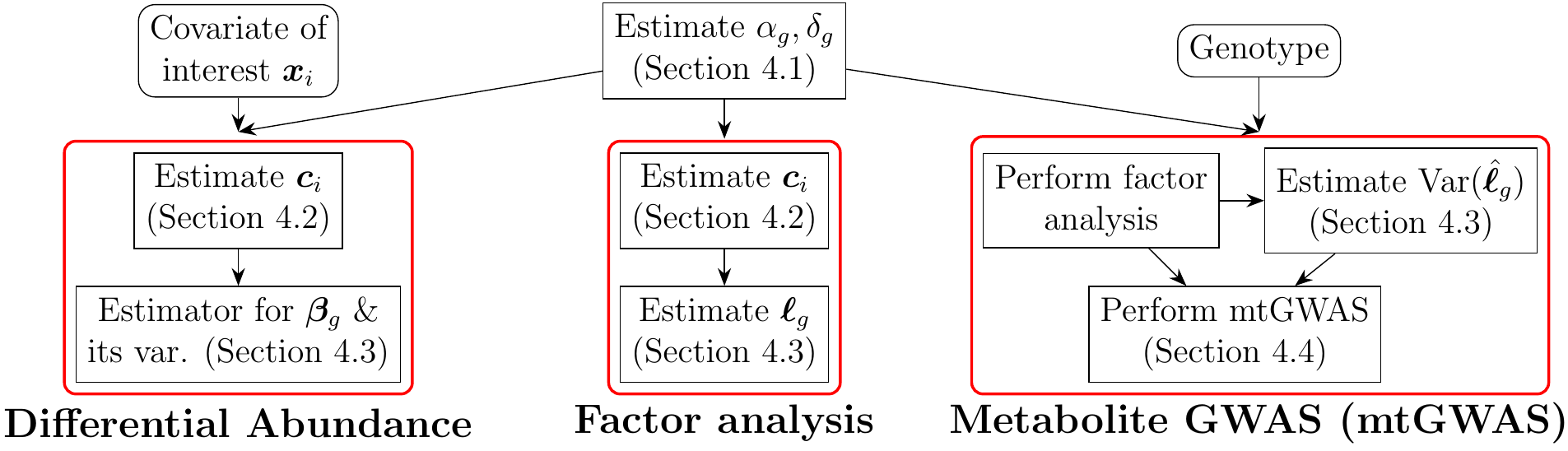}
\caption{Method overview and how estimators are used to solve different problems in metabolomics.}\label{Figure:Overview}
\end{figure}

\subsection{Estimating the missingness mechanisms}
\label{subsection:MissMech}
We follow \citet{MetabMiss} and estimate $\alpha_g,\delta_g$ from \eqref{equation:MissingDataModel} using a Bayesian generalized method of moments estimator. Briefly, for some observed $\bm{u}_{gi} \in \mathbb{R}^r$, we consider the observable sample moment $\bm{m}_g(\tilde{\alpha},\tilde{\delta})=n^{-1/2}\sum_{i=1}^n \bm{u}_{gi}[1 - r_{gi}/\Psi\{\tilde{\alpha}(y_{gi}-\tilde{\delta})\}]$, which is mean 0 and asymptotically normal when $(\tilde{\alpha},\tilde{\delta})=(\alpha_g,\delta_g)$ and $\bm{u}_{gi}$ is independent of $r_{gi}$ conditional on $y_{gi}$. Treating $\bm{m}_g$ as our ``data'', we estimate $\alpha_g,\delta_g$ as $(\hat{\alpha}_g,\hat{\delta}_g)=\E\{ (\alpha_g,\delta_g) \mid \bm{m}_g(\alpha_g,\delta_g) \}$, where we approximate the posterior $\Bprob\{ \alpha_g,\delta_g \mid \bm{m}_g(\alpha_g,\delta_g) \} \propto \Bprob\{ \bm{m}_g(\alpha_g,\delta_g) \mid \alpha_g ,\allowbreak\delta_g\}\Bprob(\alpha_g ,\delta_g)$ assuming $\bm{m}_g(\alpha_g,\delta_g)$ is normally distributed. Since $y_{gi}$ must be dependent on $\bm{u}_{gi}$, we let $\bm{u}_{gi} \in \mathbb{R}^r$ be $r$ of the first few principal components of the data matrix of fully observed metabolites. Sections 3 and 4 of \citet{MetabMiss} contain additional details.

Critically, the estimators $\hat{\alpha}_g,\hat{\delta}_g$ only depend on the dataset $\{r_{gi}y_{gi}\}_{g \in [p]; i \in [n]}$, and are invariant to the covariate of interest $\bm{x}_i$ and genotype. We therefore only compute $\hat{\alpha}_g,\hat{\delta}_g$ once per dataset and store the results, which helps make downstream analyses computationally tractable.



\subsection{Estimating latent factors}
\label{subsection:LF}
We describe estimates for latent factors $\bm{c}_i$ in differential abundance problems, and show how these can be used to derive estimates in factor analysis and mtGWAS applications as well in Sections~\ref{subsection:CoeffInt} and \ref{subsection:mtGWAS}. Let $\bm{X}=(\bm{x}_1\cdots \bm{x}_n)^{\top}$ and $P_{\bm{X}}^{\perp} \in \mathbb{R}^{n \times n}$ be the orthogonal projection matrix that projects vectors onto the kernel of $\bm{X}^{\top}$. We can express $\bm{C}=(\bm{c}_1 \cdots \bm{c}_n)^{\top} \in \mathbb{R}^{n \times K}$ as $\bm{C} = P_{\bm{X}}^{\perp}\bm{C} + \bm{X}\bm{\Omega}$, where $\bm{\Omega} = (\bm{X}^{\top}\bm{X})^{-1}\bm{X}^{\top}\bm{C}$. Model \eqref{equation:MainModel} can then be re-written as
\begin{align}
\label{equation:ModelFactor}
    y_{gi} = \bm{b}_g^{\top} \bm{x}_i + \bm{\ell}_g^{\top} [P_{\bm{X}}^{\perp}\bm{C}]_{i *} + e_{gi}, \quad \bm{b}_g=\bm{\beta}_g + \bm{\Omega}\bm{\ell}_g, \quad (e_{g1},\ldots,e_{gn}) \sim N(0,\sigma_g^2 I_n),
\end{align}
where $[P_{\bm{X}}^{\perp}\bm{C}]_{i *}\in\mathbb{R}^K$ is the $i$th row of $P_{\bm{X}}^{\perp}\bm{C}$. We utilize the paradigm from \citet{BCconf} and sequentially estimate $P_{\bm{X}}^{\perp}\bm{C}$ and $\bm{\Omega}$, where the latter estimate adjusts for confounding. It seems natural to use the normality of $e_{gi}$ in \eqref{equation:ModelFactor} to obtain optimal maximum likelihood estimates for $P_{\bm{X}}^{\perp}\bm{C}$ and $\bm{\Omega}$. However, this would beget computationally expensive iterative algorithms that require numerically integrating over all missing $y_{gi}$'s at each iteration. Instead, we use inverse probability weighting (IPW) to derive computationally efficient estimators. Remarkably, we prove in Section~\ref{section:Theory} that the loss of statistical efficiency that accompanies IPW has an asymptotically negligible effect on downstream inference.

If $\bm{Y}=[y_{gi}] \in \mathbb{R}^{p \times n}$ were observed, a natural estimate for $P_{\bm{X}}^{\perp}\bm{C}$ is the first $K$ right singular vectors of $\bm{Y}P_{\bm{X}}^{\perp}$ \citep{BCconf}, which is equivalent to minimizing $\sum_{g,i} \allowbreak \{ \allowbreak y_{gi} \allowbreak - \allowbreak ( \bm{b}_g^{\top}\bm{x}_i \allowbreak + \allowbreak \bm{\ell}_g^{\top} \bm{C}^{\perp}_{i*} ) \}^2$ over $\bm{C}^{\perp}$, as well as $\bm{b}_g$ and $\bm{\ell}_g$, such that $\bm{X}^{\top}\bm{C}^{\perp}=0$. This motivates estimating $P_{\bm{X}}^{\perp}\bm{C}$ when $y_{gi}$'s may be missing by solving the following IPW-inspired optimization problem:
\begin{align}
\label{equation:OptC}
\begin{aligned}
   &\{P_{\bm{X}}^{\perp}\hat{\bm{C}}, \{\tilde{\bm{b}}_g,\tilde{\bm{\ell}}_g\}_{g \in [p]}\} \in \argmin_{\substack{\bm{C}^{\perp} \in \mathbb{R}^{n \times K},\, \bm{b}_g \in \mathbb{R}^d, \,  \bm{\ell}_g \in \mathbb{R}^K\\ \text{such that $\bm{X}^{\top} \bm{C}^{\perp}=0$}}} \sum_{g=1}^p \sum_{i=1}^n \hat{w}_{gi} f_{gi}(\bm{C}^{\perp},\bm{b}_g,\bm{\ell}_g)\\
   &f_{gi}(\bm{C}^{\perp},\bm{b}_g,\bm{\ell}_g) = \{ y_{gi}-( \bm{b}_g^{\top}\bm{x}_i + \bm{\ell}_g^{\top} \bm{C}^{\perp}_{i*} ) \}^2, \quad \hat{w}_{gi}=r_{gi}/\Psi\{ \hat{\alpha}_g(y_{gi}-\hat{\delta}_g) \}.
\end{aligned}
\end{align}
Here, $\hat{\alpha}_g,\hat{\delta}_g$ are given in Section~\ref{subsection:MissMech} and the objective on the first line is observable because $\hat{w}_{gi}=0$ if $y_{gi}$ is missing. If $\hat{\alpha}_g=\alpha_g$, $\hat{\delta}_g=\delta_g$, and we replace $\hat{w}_{gi}$ with its expectation $\E(\hat{w}_{gi}\mid y_{gi})=1$, the above discussion implies \eqref{equation:OptC} is equivalent to singular value decomposition. Unlike maximum likelihood estimators that use the normality of $e_{gi}$ in \eqref{equation:ModelFactor}, iterative updates in \eqref{equation:OptC} have a closed form and beget fast computation. Since $P_{\bm{X}}^{\perp}\bm{C}$ is not identifiable in \eqref{equation:ModelFactor}, $P_{\bm{X}}^{\perp}\hat{\bm{C}}$ is not unique. While we address this in factor analysis applications by requiring $P_{\bm{X}}^{\perp}\hat{\bm{C}}$ have orthogonal columns, such an identification criterion is unnecessary in differential abundance and mtGWAS.

To recover $\bm{\Omega}$, we see \eqref{equation:OptC} provides estimates $\tilde{\bm{b}}_g,\tilde{\bm{\ell}}_g$ for $\bm{b}_g,\bm{\ell}_g$ in \eqref{equation:ModelFactor}. Since $P_{\bm{X}}^{\perp}\bm{C}$ is orthogonal to $\bm{X}$, we should be able to separate variation due to $P_{\bm{X}}^{\perp}\bm{C}$ and $\bm{X}$, which suggests $\tilde{\bm{b}}_g,\tilde{\bm{\ell}}_g$ are reasonably accurate. If $\bm{\beta}_g=0$ for all $g\in[p]$, then the expression for $\bm{b}_g$ in \eqref{equation:ModelFactor} indicates we can estimate $\bm{\Omega}$ by regressing $(\tilde{\bm{b}}_1\cdots \tilde{\bm{b}}_p)$ onto $(\tilde{\bm{\ell}}_1\cdots \tilde{\bm{\ell}}_p)$. While not all $\bm{\beta}_g$'s will be 0, covariates of interest encoded by $\bm{X}$ typically correlate with only a few metabolites \citep{Bilirubin}. We use this to justify estimating $\bm{\Omega}$ with the aforementioned regression:
\begin{align}
\label{equation:Omegahat}
    \hat{\bm{\Omega}} = \textstyle \argmin_{\bm{\Omega} \in \mathbb{R}^{d \times K}} \textstyle\sum_{g=1}^p \norm*{ \tilde{\bm{b}}_g - \bm{\Omega}\tilde{\bm{\ell}}_g }_2^2 = (\textstyle \sum_{g=1}^p \tilde{\bm{b}}_g \tilde{\bm{\ell}}_g^{\top}) ( \sum_{g=1}^p \tilde{\bm{\ell}}_g \tilde{\bm{\ell}}_g^{\top} )^{-1}.
\end{align}
In addition to adjusting for latent confounds, we show $\hat{\bm{\Omega}}$ can be used to test if latent factors depend on $\bm{X}$ in Sections~\ref{section:Theory} and \ref{section:RealData}. We show in supplemental Section~\ref{supp:section:RefineOmega} that \eqref{equation:Omegahat} can be further refined by iteratively removing ``outlying'' covariate-dependent metabolites from the regression.

We estimate $\bm{C}$ as $\hat{\bm{C}} = P_{\bm{X}}^{\perp}\hat{\bm{C}} + \bm{X} \hat{\bm{\Omega}}$ in differential abundance problems. Since $\bm{X}$ is a nuisance covariate in factor analysis and mtGWAS, we let $\hat{\bm{C}}$ be the solution to \eqref{equation:OptC} in those applications.

\subsection{Estimation and inference on coefficients of interest}
\label{subsection:CoeffInt}
Here we consider $\bm{\theta}_g = (\bm{\beta}_g^{\top},\bm{\ell}_g^{\top})^{\top}$, where $\bm{\beta}_g$ is the inferential target in differential abundance and $\bm{\ell}_g$ is important in factor analysis and mtGWAS. Our goal is to develop statistically efficient estimators that can be computed quickly. Throughout Section~\ref{subsection:CoeffInt}, we let $\hat{\bm{z}}_i=(\bm{x}_i^{\top},\allowbreak\hat{\bm{c}}_i^{\top})^{\top}$ for $\hat{\bm{c}}_i \in \mathbb{R}^{K}$ the $i$th row of $\hat{\bm{C}}$ defined in Section~\ref{subsection:LF} (our estimate for $\bm{c}_i$ in \eqref{equation:MainModel}).

Having estimated $\{\alpha_g,\delta_g,\bm{z}_i=( \bm{x}_i^{\top},\bm{c}_i^{\top} )^{\top}\}$ as $\{\hat{\alpha}_g,\hat{\delta}_g,\hat{\bm{z}}_i\}$ in Sections~\ref{subsection:MissMech} and \ref{subsection:LF}, we consider estimating $\bm{\theta}_g$ and $\sigma_g$ via the log-likelihood $h_{g}(\bm{\theta},\sigma)$ of the observed data $\{r_{gi}y_{gi}\}_{i \in [n]}$ implied by \eqref{equation:MainModel} and \eqref{equation:MissingDataModel} using the plug-in estimators $\{\hat{\alpha}_g,\hat{\delta}_g,\hat{\bm{z}}_i\}$:
\begin{align}
\label{equation:thetagLikelihood}
\begin{aligned}
h_{g}(\bm{\theta},\sigma) =& \textstyle\sum_{i=1}^n -r_{gi}\{\log(\sigma) +( y_{gi}-\bm{\theta}^{\top}\hat{\bm{z}}_i )^2/(2\sigma^2)\}\\& + \textstyle\sum_{i=1}^n(1-r_{gi})\log[1- \smallint \Psi\{ \hat{\alpha}_g( \bm{\theta}^{\top}\hat{\bm{z}}_i + \sigma e - \hat{\delta}_g ) \}\phi(e)\text{d}e ],
\end{aligned}
\end{align}
\noindent where $\phi(e)$ is the standard normal density. While Section~\ref{supp:subsection:DiffAbund} of the Supplement shows that directly maximizing $h_{g}$ will accurately estimate $\bm{\theta}_g$, this fails to consider the computational cost of numerically integrating the second line of \eqref{equation:thetagLikelihood}. To address this, we design an appropriately initialized algorithm that only requires a small number of iterations, and therefore numerical integrations, to accurately approximate the maximizer of \eqref{equation:thetagLikelihood}. Briefly, for $\hat{w}_{gi}$ given in \eqref{equation:OptC}, let
\begin{align}
\label{equation:IPW}
\begin{aligned}
    \hat{\bm{\theta}}_g^{\IPW} &= (\textstyle\sum_{i=1}^n \hat{w}_{gi} \hat{\bm{z}}_i\hat{\bm{z}}_i^{\top})^{-1}(\textstyle\sum_{i=1}^n \hat{w}_{gi} \hat{\bm{z}}_i y_{gi})\\
    \hat{\sigma}_g^{\IPW} &= [ (\textstyle\sum_{i=1}^n \hat{w}_{gi})^{-1}\textstyle\sum_{i=1}^n \hat{w}_{gi}\{ y_{gi} - \hat{\bm{z}}_i^{\top}\hat{\bm{\theta}}_g^{\IPW} \}^2 ]^{1/2}
\end{aligned}
\end{align}
be the inverse probability weighted (IPW) estimators of $\bm{\theta}_g$ and $\sigma_g$. Since $\hat{w}_{gi}= 0$ if $y_{gi}$ is missing, the estimators in \eqref{equation:IPW} only use observed data, and are therefore sub-optimal. However, they are easy to compute and, as we show in supplemental Section~\ref{supp:subsection:DiffAbund}, consistent, which make them appropriate starting points. We then iteratively update our estimates for $\bm{\theta}_g$ and $\sigma_g$ with Fisher scoring using the information matrix $\mathcal{I}_g(\bm{\theta},\sigma)=\E_{\{\bm{\theta},\sigma\}}[ \nabla^2 h_g(\bm{\theta},\sigma) \mid \{\hat{\bm{z}}_i\}_{i\in[n]} ]$, where the expectation ignores the uncertainty in $\hat{\bm{z}}_i$, $\hat{\alpha}_g$, and $\hat{\delta}_g$. While running this algorithm to completion is potentially computationally expensive, we prove in Section~\ref{subsection:theory:DA} that we only require one Fisher scoring step to achieve asymptotically optimal estimates. In practice, our software default is $\leq 10$ iterations. Letting $\hat{\bm{\theta}}_g=( \hat{\bm{\beta}}_g^{\top},\hat{\bm{\ell}}_g^{\top} )^{\top}$ and $\hat{\sigma}_g$ be the resulting estimates, we perform inference on $\bm{\beta}_g$ assuming $\hat{\bm{\beta}}_g \approx N(\bm{\beta}_{g},\hat{\V}( \hat{\bm{\beta}}_{g}))$ for $\hat{\V}( \hat{\bm{\beta}}_{g})$ the first $d \times d$ block of $\{-\mathcal{I}_g(\hat{\bm{\theta}}_g,\hat{\sigma}_g)\}^{-1}$. 

Two features of this procedure cast doubt on its fidelity. The first is the assumption in \eqref{equation:MainModel} that $e_{gi}$ is normally distributed, as the existence of genetic and possibly other non-normal variation in $e_{gi}$ suggest the likelihood in \eqref{equation:thetagLikelihood} is incorrect. While this is not a concern in fully observed data, estimates from missing data may be sensitive to distributional assumptions \citep{ImputationFlaws}. The second is $\hat{\bm{\beta}}_g$ depends on the estimated latent factors $\hat{\bm{c}}_1,\ldots,\hat{\bm{c}}_n$ whose theoretical properties are unknown. We address these concerns in Section~\ref{subsection:theory:DA}, where we prove inference with $\hat{\bm{\beta}}_g$ is asymptotically equivalent to knowing both the non-normal genetic effects and latent factors. While the uncertainty in $\hat{\alpha}_g,\hat{\delta}_g$ ostensibly poses a third issue, the strong theoretical and simulation results in \citet{MetabMiss} proving their accuracy suggest this is trivial.



\subsection{Metabolite genome-wide association study}
\label{subsection:mtGWAS}
We lastly consider performing an mtGWAS. We set $\bm{x}_i$ in \eqref{equation:MainModel} to be 0 for simplicity, but show in supplemental Section~\ref{supp:section:TheorymtGWAS} how to extend our method to allow $\bm{x}_i \neq 0$. Let $G_{si} \in \{0,1,2\}$ be the genotype at single nucleotide polymorphism (SNP) $s$ in sample $i$. Given \eqref{equation:MainModel}, the effect of $G_{si}$ on $y_{gi}$ can either appear in the idiosyncratic error $e_{gi}$, or be mediated by $\bm{c}_i$. We therefore assume $e_{gi} = \gamma^{(e)}_{gs} G_{si} + \Delta_{gi}^{(e)}$ and $\bm{c}_i = \bm{\gamma}_s^{(c)} G_{si} + \bm{\Delta}_i^{(c)}$, where $\gamma^{(e)}_{gs} \in \mathbb{R}$, $\bm{\gamma}_s^{(c)} \in \mathbb{R}^K$ quantify the effect of $G_{si}$ on $e_{gi}$ and $\bm{c}_i$, respectively, and $\Delta_{gi}^{(e)} \in \mathbb{R}$, $\bm{\Delta}_i^{(c)} \in \mathbb{R}^K$ are mean 0 errors. This implies
\begin{align}
\label{equation:ygenetic}
    y_{gi} = \{ \bm{\ell}_g^{\top}\bm{\gamma}_s^{(c)} + \gamma^{(e)}_{gs}\}G_{si} + \{\bm{\ell}_g^{\top}\bm{\Delta}_i^{(c)} +  \Delta_{gi}^{(e)} \},
\end{align}
where $\gamma^{(e)}_{gs}$ and $\bm{\ell}_g^{\top}\bm{\gamma}_s^{(c)}$ are interpretable as the idiosyncratic and low rank genetic effects. We develop methodology below to perform inference on $\gamma^{(e)}_{gs}$, $\bm{\ell}_g^{\top}\bm{\gamma}_s^{(c)}$, and the total effect $\bm{\ell}_g^{\top}\bm{\gamma}_s^{(c)} + \gamma^{(e)}_{gs}$. 


Consider testing $H_{0,gs}^{(e)}: \gamma_{gs}^{(e)}=0$. Classic Wald tests would require optimizing \eqref{equation:thetagLikelihood} for all \#metabolites $\times$ \#SNPs pairs $g$ and $s$. While this is reasonable for \#SNPs $\lesssim 10^2$ (i.e. on the order of a phenome-wide association study), it is infeasible in genome-wide studies, where \#SNPs $\gtrsim 10^6$. To circumvent this, we propose a novel and tractable score test. Briefly, consider the log-likelihood $h_{gs}( \gamma,\bm{\ell},\sigma )$ for $\{r_{gi}y_{gi}\}_{i \in [n]}$ under \eqref{equation:MainModel} and \eqref{equation:MissingDataModel} assuming $e_{gi} \sim N(\gamma G_{si},\sigma^2)$:
\begin{align*}
    h_{gs}( \gamma,\bm{\ell},\sigma ) =& \textstyle\sum_{i=1}^n -r_{gi}[\log(\sigma) + \{ y_{gi}-(\bm{\ell}^{\top}\hat{\bm{c}}_i + \gamma G_{si})\}^2/(2\sigma^2)]\\
    &+ \textstyle\sum_{i=1}^n(1-r_{gi})\log[1- \smallint \Psi\{ \hat{\alpha}_g( \bm{\ell}^{\top}\hat{\bm{c}}_i + \gamma G_{si} + \sigma e - \hat{\delta}_g ) \}\phi(e)\text{d}e ].
\end{align*}
If $H_{0,gs}^{(e)}: \gamma_{gs}^{(e)}=0$ is true, $h_{gs}\{ \gamma_{gs}^{(e)},\bm{\ell},\sigma \} = h_{g}( \bm{\ell},\sigma )$ for $h_g$ as defined in \eqref{equation:thetagLikelihood}. Then for $\hat{\bm{\ell}}_g,\hat{\sigma}_g$ the approximate maximizers of $h_g$ described in Section~\ref{subsection:CoeffInt}, we define the score statistic $\eta_{gs}^{(e)}$ to be
\begin{align}
\label{equation:ScoreTest}
\eta_{gs}^{(e)} = \{ \textstyle \tfrac{\partial}{\partial \gamma} h_{gs}( \gamma,\hat{\bm{\ell}}_g,\hat{\sigma}_g )\mid_{\gamma=0} \}^2 [ \{ -\mathcal{I}_{gs}(0,\hat{\bm{\ell}}_g,\hat{\sigma}_g) \}^{-1} ]_{11}, 
\end{align}
where $\mathcal{I}_{gs}(\gamma,\bm{\ell},\sigma)$ is the Fisher information matrix assuming $h_{gs}(\gamma,\bm{\ell},\sigma)$ is the log-likelihood for $\{r_{gi}y_{gi}\}_{i \in [n]}$. A p-value for $H_{0,gs}^{(e)}$ is computed by comparing $\eta_{gs}^{(e)}$ to the upper quantiles of a $\chi^2_1$.

Several features of \eqref{equation:ScoreTest} make our test computationally and statistically efficient. First, since $\hat{\bm{\ell}}_g,\hat{\sigma}_g$ are the approximate maximizers of $h_g$ in \eqref{equation:thetagLikelihood}, they do not depend on genotype, and consequently only need to be computed once per metabolite $g$. Therefore, as we show in supplemental Section~\ref{supp:section:TheorymtGWAS}, \eqref{equation:ScoreTest} is a simple function of genotype and metabolite-specific terms that can be pre-computed. Second, \eqref{equation:ScoreTest} uses all available data and does not errantly impute missing data, which is the prevailing practice in mtGWAS studies. Lastly, and most importantly, inference with \eqref{equation:ScoreTest} is done conditional on the estimated latent factors $\hat{\bm{c}}_i$, which de-noises the data to substantially improve power by reducing residual variances. For example, we show that the variance reduction in our data example is equivalent to increasing the sample size by 67\%. 

We next consider $\bm{\ell}_g^{\top}\bm{\gamma}_s^{(c)}$ from \eqref{equation:ygenetic}, which is interpretable as the effect of SNP $s$ on metabolite $g$ that is mediated through the latent factors $\bm{c}_i$. Let $\hat{\bm{\ell}}_g$ as defined above, and let $\hat{\V}(\hat{\bm{\ell}}_g)$ be its its estimated variance obtained using the inferential procedure outlined in Section~\ref{subsection:CoeffInt}. Since $\bm{\gamma}_s^{(c)}$ satisfies $\E(\bm{c}_i \mid G_{si}) = \bm{\gamma}_s^{(c)} G_{si}$, we define $\hat{\bm{\gamma}}_s^{(c)}$ and $\hat{\V}\{ \hat{\bm{\gamma}}_s^{(c)} \}$ to be the ordinary least squares estimate and its corresponding estimated variance from the regression of $[\hat{\bm{c}}_1 \cdots \hat{\bm{c}}_n]^{\top}$ onto $(G_{s1}\cdots G_{sn})^{\top}$, which can be efficiently computed at the genome-wide scale. If $\hat{\bm{c}}_i=\bm{c}_i$ and there were no genetic effects on $e_{gi}$, standard arguments can be used to show $\hat{\bm{\ell}}_g$ is  asymptotically independent of $\hat{\bm{\gamma}}_s^{(c)}$. We therefore test $H_{0,gs}^{(c)}: \bm{\ell}_g^{\top}\bm{\gamma}_s^{(c)}=0 $ by comparing the following to the upper quantiles of a $\chi^2_1$:
\begin{align}
\label{equation:ScoreTest:C}
    \eta_{gs}^{(c)} = \{ \hat{\bm{\ell}}_g^{\top} \hat{\bm{\gamma}}_s^{(c)} \}^2/[ \hat{\bm{\ell}}_g^{\top}\widehat{\V}\{\hat{\bm{\gamma}}_s^{(c)}\}\hat{\bm{\ell}}_g + \{\hat{\bm{\gamma}}_s^{(c)}\}^{\top}\widehat{\V}(\hat{\bm{\ell}}_g) \hat{\bm{\gamma}}_s^{(c)}].
\end{align}

We lastly test whether SNP $s$ has any effect on metabolite $g$'s abundance. Given \eqref{equation:ygenetic}, the classic approach would test the null that $\bm{\ell}_g^{\top}\bm{\gamma}_s^{(c)} + \gamma_{gs}^{(e)} = 0$. However, as discussed above, this is not practical because it would require estimating $\gamma_{gs}^{(e)}$. Instead, since $\bm{c}_i$ and $e_{gi}$ are typically assumed to be independent in metabolomic data \citep{MetabMiss}, we assume their corresponding genetic effects reflect unrelated variation. This suggests a metabolite's abundance is genetically regulated if $\bm{\ell}_g^{\top}\bm{\gamma}_s^{(c)}$ or $\gamma_{gs}^{(e)}$ is 0. We therefore propose testing $H_{0,gs}^{(c,e)}: \bm{\ell}_g^{\top}\bm{\gamma}_s^{(c)} = \gamma_{gs}^{(e)}=0$ using $\eta_{gs}^{(c,e)} = \eta_{gs}^{(c)} + \eta_{gs}^{(e)}$, which we show in Section~\ref{subsection:theory:mtGWAS} is approximately $\chi^2_2$ under $H_{0,gs}^{(c,e)}$.


\section{Theoretical guarantees}
\label{section:Theory}
Here we justify estimators and inference from Section~\ref{section:Estimation}. Since \citet{MetabMiss} detailed the theoretical properties of $\hat{\alpha}_g,\hat{\delta}_g$ defined in Section~\ref{subsection:MissMech}, we focus on the properties and impact of the latent factor estimates $\hat{\bm{c}}_i$ from Section~\ref{subsection:LF}, as their theoretical properties are unknown but critical to the fidelity of estimators proposed in Sections~\ref{subsection:LF}-\ref{subsection:mtGWAS}. Given the accuracy of $\hat{\alpha}_g,\hat{\delta}_g$ \citep{MetabMiss} and the negligible impact their uncertainty has in real and simulated data (see Sections~\ref{section:RealData} and \ref{supp:section:simulations}), we assume $\hat{\alpha}_g=\alpha_g, \hat{\delta}_g=\delta_g$ to make proofs tractable, which is common in the non-random missing data literature \citep{MissKnownTheory}.

Section~\ref{subsection:theory:Assumptions} details our assumptions, and Sections~\ref{subsection:theory:LF}-\ref{subsection:theory:mtGWAS} contain our theoretical results. In addition to providing the theoretical foundation for estimators in Section~\ref{section:Estimation}, these results help us specify a software default choice for $\Psi$ defined in \eqref{equation:MissingDataModel}. All proofs are in the supplement.


\subsection{Assumptions}
\label{subsection:theory:Assumptions}
Let $\bm{X}=[\bm{x}_1 \cdots \bm{x}_n]^{\top} \in \mathbb{R}^{n \times d}$, $\bm{C}=[\bm{c}_1 \cdots \bm{c}_n]^{\top} \in \mathbb{R}^{n \times K}$, and $\bm{1}_n=(1,\ldots,1)^{\top}\in\mathbb{R}^n$. For $\bm{M}\in\mathbb{R}^{n\times m}$, let $P_{\bm{M}}^{\perp}\in\mathbb{R}^{n\times n}$ be the orthogonal projection onto the kernel of $\bm{M}^{\top}$. We first place assumptions on $y_{gi}$.


\begin{Assumption}
\label{assumption:y}
For $g \in [p]$, $i \in [n]$, and $s \in [S]$, let $y_{gi}=\bm{\beta}_g^{\top}\bm{x}_i + \bm{\ell}_g^{\top}\bm{c}_i + e_{gi}$, $G_{si} \in \{0,1,2\}$, and $\mathcal{G}=\{G_{si}\}_{s \in S; i \in [n]}$. Then the following hold for constants $a_1>0$ and $\epsilon \in (0,1/2\wedge a_1)$.
\begin{enumerate}[label=(\alph*)]
\item $\bm{X}=[\tilde{\bm{X}},\bm{1}_n]$ is non-random, $n^{-1}\tilde{\bm{X}}^{\top}P_{\bm{1}_n}^{\perp}\tilde{\bm{X}} \succeq \epsilon I_{d-1}$, $\norm*{ \tilde{\bm{X}} }_{\infty},\norm*{\bm{\beta}_g}_2\leq a_1$, $\mathcal{G}$'s elements are independent, $\{G_{si}\}_{i \in [n]}$ are identically distributed for each $s\in[S]$, and $\epsilon n \leq p \leq a_1 n$. \label{assump:y:XG}
\item The eigenvalues $\lambda_1,\ldots,\lambda_K > 0$ of $p^{-1}\sum_{g=1}^p \bm{\ell}_g \bm{\ell}_g^{\top}$ satisfy $n^{-1/2+\epsilon}\allowbreak \lesssim \allowbreak \lambda_K \allowbreak \leq \allowbreak \cdots \leq \allowbreak \lambda_1 \allowbreak \lesssim 1$, $\lambda_1/\lambda_K \leq a_1$, and $\norm{ \bm{\ell}_g }_2 \leq a_1 \lambda_1^{1/2}$. Further, $\bm{c}_i = \bm{f}(\bm{x}_i)+ \sum_{s=1}^S \bm{\gamma}_{s}^{(c)} G_{si} + \bm{\Delta}_i^{(c)} \in \mathbb{R}^K$, where:
\begin{enumerate}[label=(\roman*)]
    \item $\bm{f}:\mathbb{R}^{d} \to \mathbb{R}^K$ is a continuous function and $\{\bm{\gamma}_{s}^{(c)}\}_{s \in [S]}$ are non-random and satisfy $\sum_{s=1}^S \norm*{\bm{\gamma}_{s}^{(c)}}_2 \leq a_1$, $\sum_{s=1}^S 1\{\bm{\gamma}_{s}^{(c)} \neq 0\} \leq a_1 p^{1/2}$, and $\max_{s \in [S]}\norm*{\bm{\gamma}_{s}^{(c)}}_2 = o(n^{-1/4})$.\label{assump:y:Lambda:gamma}
    \item $\{\bm{\Delta}_i^{(c)}\}_{i \in [n]}$ are independent, identically distributed, independent of $\mathcal{G}$, $\V\{\bm{\Delta}_i^{(c)}\} \succeq \epsilon I_K$, and $\E\{ \abs*{ \bm{\Delta}_{i_k}^{(c)} }^m \} \leq b_m$ for $k \in [K]$, all $m>0$, and constants $b_m>0$.\label{assump:y:Lambda:Delta}
\end{enumerate}\label{assump:y:Lambda}
\item For non-random parameters $\{\gamma_{gs}^{(e)}\}_{g\in[p];s\in[S]}$, $e_{gi} = \sum_{s=1}^S \gamma_{gs}^{(e)}G_{si} + \Delta_{gi}^{(e)}$ such that:
\begin{enumerate}[label=(\roman*)]
    \item $\sum_{s=1}^S 1\{\gamma_{gs}^{(e)} \neq 0\} \leq a_1$, $\max_{g\in[p]; s \in [S]}\abs*{\gamma_{gs}^{(e)}} = o(n^{-1/4})$, $\Delta_{gi}^{(e)} \sim N(0,\sigma_g^2)$, $\sigma_g^2 \leq a_1$, and $\{\Delta_{gi}^{(e)}\}_{g \in [p];i \in [n]}$ are independent and are independent of $\{\mathcal{G},\bm{C}\}$.
    \item Each connected component of the metabolite graph created by placing an edge between metabolites $g,h\in[p]$ if $\gamma_{gs}^{(e)} \gamma_{hs}^{(e)}\neq 0$ has $\leq a_1$ metabolite vertices.\label{assump:y:e:Network}
\end{enumerate}\label{assump:y:e}
\end{enumerate}
\end{Assumption}
\noindent We require $\bm{X}$ contain an intercept in \ref{assump:y:XG}. The assumptions on genotype $G_{si}$ in \ref{assump:y:XG} are akin to assuming each linkage disequilibrium block contains at most one causal SNP. The eigenvalues in \ref{assump:y:Lambda} quantify the average magnitude of $\bm{\ell}_1,\allowbreak\ldots,\allowbreak\bm{\ell}_p$, where we let eigenvalues be moderate ($\asymp n^{-1/2+\epsilon}$) or large ($\asymp 1$). While some datasets may have eigenvalues even smaller than $n^{-1/2+\epsilon}$, they likely make a trivial contribution to metabolite variation and are therefore not considered here. 

Since metabolites may be genetically regulated, we allow latent factors $\bm{c}_i$ and errors $e_{gi}$ to be dependent on genotype. This implies $\bm{c}_i$ and $e_{gi}$ may be dependent, which violates the assumptions of most factor analysis methods \citep{BCconf}. To our knowledge, our theoretical work is the first to consider genetic dependence between latent factors and errors.

We assume genetic effects $\bm{\gamma}^{(c)}_{s}$ and $\gamma^{(e)}_{gs}$ decay with sample size, which is a common assumption in GWAS \citep{SNPAsymp}. However, we will have asymptotically perfect power if the genetic effect is $\gtrsim n^{-1/2+\eta}$ in magnitude for any $\eta>0$ and the number of tested SNPs is polynomial in $n$. Assumption \ref{assump:y:e}\ref{assump:y:e:Network} assumes metabolites can be partitioned into pathways where, conditional on latent factors, metabolites in different pathways are independent, which is a common assumption \citep{Bilirubin}. We next place assumptions on the missing data.

\begin{Assumption}
\label{assumption:Miss}
Model \eqref{equation:MissingDataModel} and the following hold for some constants $a_2>1$, $m>0$:
\begin{enumerate}[label=(\alph*)]
\item $\{r_{gi}\}_{g \in [p]; i \in [n]}$ are independent conditional on $\{y_{gi}\}_{g \in [p]; i \in [n]}$ and $\alpha_g\in (0,a_2), \abs*{\delta_g} \leq a_2$.\label{assump:MissData:Indep}
\item $\Psi$ is a six times continuously differentiable CDF that satisfies (i) $\Psi(-x)=1-\Psi(x)$, (ii) $\abs*{x}^m\Psi(x) \geq a_2^{-1}$ for all $x<-a_2$, and (iii) $\abs*{x}^m\abs*{\frac{\text{d}^{(j)}}{\text{d}x^{(j)}} \Psi(x)} \leq a_2$ for all $j\in [6]$ and $\abs*{x}>a_2$.\label{assump:MissData:Psi} 
\end{enumerate}
\end{Assumption}
\begin{remark}
\label{remark:Assumption2}
Assumption~\ref{assump:MissData:Psi} is satisfied when $\Psi$ is the CDF of a t-distribution.
\end{remark}
Section~\ref{section:Problem} discusses the conditional independence assumption in \ref{assump:MissData:Indep}. Assumption~\ref{assump:MissData:Psi}(ii) requires the left hand tail of $\Psi$ to go to 0 at a polynomial rate, which ensures the inverse probability weighted estimator in \eqref{equation:OptC} is well-behaved. Remark~\ref{remark:Assumption2} inspires our software-default choice for $\Psi$ to be the CDF of a t-distribution with four degrees of freedom, which also reduces the impact of outlying observations on our estimates for $\bm{\beta}_g$ (see supplemental Remark~\ref{supp:remark:PsiBound}). Note \ref{assump:MissData:Psi}(ii) excludes the usual assumption that $\Psi$ is the CDF of a logistic or normal random variable \citep{Impute1}, as their left hand tails go to 0 at exponential and super-exponential rates. 

\subsection{Accuracy of and inference with latent factor estimates}
\label{subsection:theory:LF}
We first consider the accuracy of $P_{\bm{X}}^{\perp}\hat{\bm{C}}$ defined in \eqref{equation:OptC}, which is critical to the estimate for $\hat{\bm{C}}$ in differential abundance and is exactly $\hat{\bm{C}}$ in factor analysis applications.
\begin{theorem}
\label{theorem:PxC}
Suppose Assumptions~\ref{assumption:y} and \ref{assumption:Miss} hold, and let $\hat{\mathcal{P}},\mathcal{P} \in \mathbb{R}^{n \times n}$ be the orthogonal projections that project vectors onto $\im(P_{\bm{X}}^{\perp}\hat{\bm{C}})$ and $\im(P_{\bm{X}}^{\perp}\bm{C})$. Then there exists a constant $\eta>0$ such that if we require $\norm*{ \hat{\mathcal{P}} - \mathcal{P} }_F \leq \eta$, then $\norm*{ \hat{\mathcal{P}} - \mathcal{P} }_F^2 = o_P(n^{-1/2})$.
\end{theorem} 
\begin{remark}
\label{remark:PxC}
The objective in \eqref{equation:OptC}, which is expressed as a function of the matrix parameter $\bm{C}^{\perp}$, only depends on $\bm{C}^{\perp}$ through $\im(\bm{C}^{\perp})$, and is therefore actually a function of orthogonal projection matrices. The requirement that $\norm*{ \hat{\mathcal{P}} - \mathcal{P} }_F \leq \eta$ implies the desired minimizer of \eqref{equation:OptC} may only be a local minima, and we implicitly assume $\norm*{ \hat{\mathcal{P}} - \mathcal{P} }_F \leq \eta$ in all future theoretical statements. We show in supplemental Section~\ref{supp:section:FullyObs} that, under minor conditions, we can guarantee $\norm*{ \hat{\mathcal{P}} - \mathcal{P} }_F \leq \eta$ by initializing \eqref{equation:OptC} using metabolites with fully observed data.
\end{remark}

\noindent Theorem~\ref{theorem:PxC} is, to our knowledge, the first result proving the fidelity of factor analysis in data with non-random missing observations. Remarkably, this result mirrors the best known factor analysis results when data are observed \citep{BCconf}, and accounts for possible genetic-related dependencies between $\bm{c}_i$ and $e_{gi}$, which are not allowed to exist in most factor analysis-related theoretical results \citep{CATE,BCconf}.

We next consider our estimate for $\bm{\Omega}$ from \eqref{equation:Omegahat}, which helps ensure our estimates for $\bm{\beta}_g$ are not biased by latent factors $\bm{c}_i$. While its theoretical properties derived in supplemental Section~\ref{supp:subsection:Omega} are critical for Sections~\ref{subsection:theory:DA} and \ref{subsection:theory:mtGWAS}, we show in Theorem~\ref{theorem:Omega} below that it can also be used to formally test whether $\bm{c}_i$ confounds the relationship between $\bm{x}_i$ and $y_{gi}$. 



\begin{theorem}
\label{theorem:Omega}
Fix a $j \in [d-1]$. In addition to Assumptions~\ref{assumption:y} and \ref{assumption:Miss}, suppose (i) $p^{-1}\sum_{g=1}^p \allowbreak 1\{\bm{\beta}_{g_j} \neq 0\}=o(\lambda_1^{1/2}n^{-1/2})$ and (ii) $\E(\bm{c}_i) = \bm{A}^{\top}\bm{x}_{i}$ for some non-random $\bm{A} \in \mathbb{R}^{d \times K}$. Then if the null hypothesis $H_{0,j}: \bm{A}_{j*}=0$ is true, $\hat{\bm{\Omega}}_{j*}^{\top} \hat{\bm{\Omega}}_{j*}/\tilde{x}_j^2 \tdist \chi^2_K$, where $\bm{A}_{j*},\hat{\bm{\Omega}}_{j*} \in\mathbb{R}^K$ are the $j$th rows of $\bm{A},\hat{\bm{\Omega}}$ and $\tilde{x}_j^2$ is the $j$th diagonal of $(\bm{X}^{\top}\bm{X})^{-1}$.
\end{theorem}

\begin{remark}
\label{remark:sparsity}
The sparsity assumption in (i) is weaker than the usual assumption $p^{-1}\sum_{g=1}^p \allowbreak 1\{\bm{\beta}_{g_j} \neq 0\}=o(\lambda_1 n^{-1/2})$ made by methods that require fully observed data \citep{BCconf}, since $\lambda_1 < \lambda_1^{1/2}$ if $\lambda_1<1$. Note (i) is only required for the $j$th coefficient.
\end{remark}




\subsection{The statistical and computational efficiency of differential abundance estimates}
\label{subsection:theory:DA}
We next consider our estimate for $\bm{\beta}_g$ from Section~\ref{subsection:CoeffInt}. While we want to ensure its statistical fidelity, we are also interested studying its computational efficiency, since maximizing the likelihood in \eqref{equation:thetagLikelihood} requires expensive numerical integrations. We first state a proposition.

\begin{Proposition}
\label{proposition:Cknown}
Suppose Assumptions~\ref{assumption:y} and \ref{assumption:Miss} hold, let $h^{\known}_g(\bm{\beta}_g,\bm{\ell}_g,\sigma_g)$ be the log-likelihood for $\{r_{gi}y_{gi}\}_{i\in[n]}$ when $\bm{C}$ and $\{\E(e_{gi}\mid \mathcal{G})\}_{i\in[n]}$ are known, and let $\hat{\bm{\beta}}_g^{\known}$ be $\bm{\beta}_g$'s corresponding consistent maximum likelihood estimate. Then $\{ \bm{V}_g^{\known} \}^{-1/2} \allowbreak\{\hat{\bm{\beta}}_g^{\known}\allowbreak - \allowbreak \bm{\beta}_g\} \tdist N(0,I_d)$ for $\bm{V}_g^{\known}$ the first $d \times d$ block of $[-\E\{\nabla^2 h^{\known}_g(\bm{\beta}_g,\bm{\ell}_g,\sigma_g^2) \mid \bm{C}, \mathcal{G}\}]^{-1}$.
\end{Proposition}
\noindent Unsurprisingly, estimates are asymptotically normal when we observe the full covariate matrix $[\bm{X},\bm{C}]$ and know the genetic effects $\{\E(e_{gi}\mid \mathcal{G})\}_{i\in[n]}$. The latter is important, since the missing data likelihood is incorrect when the non-normal genetic effects are unknown, which risks biasing estimates. Remarkably, Theorem~\ref{theorem:betag} shows that our estimator for $\bm{\beta}_g$, which replaces $\bm{C}$ with its estimate from Section~\ref{subsection:LF} and ignores genetic effects, is asymptotically equivalent to $\hat{\bm{\beta}}_g^{\known}$.

\begin{theorem}
\label{theorem:betag}
Let $d_1 \leq d-1$. In addition to Assumptions~\ref{assumption:y} and \ref{assumption:Miss}, assume (i) in the statement of Theorem~\ref{theorem:Omega} holds for all $j \in [d_1]$. Suppose we initialize the optimization to maximize \eqref{equation:thetagLikelihood} at the IPW estimates defined in \eqref{equation:IPW}, and let $\hat{\bm{\beta}}_g$ be the estimate for $\bm{\beta}_g$ after updating the IPW estimates with one Fisher scoring step. Then for $\hat{\bm{\beta}}_g^{\known}$ and $\bm{V}_g^{\known}$ defined in Proposition~\ref{proposition:Cknown},
\begin{align}
\label{equation:AsymEquivalent}
    n^{1/2}\norm*{ \hat{\bm{\beta}}_{g (1:d_1)} - \hat{\bm{\beta}}_{g (1:d_1)}^{\known} }_2 = o_P(1), \quad n\norm*{ \hat{\V}(\hat{\bm{\beta}}_g)_{(1:d1)} - \bm{V}_{g (1:d_1)}^{\known} }_2 = o_P(1),
\end{align}
where $\hat{\bm{\beta}}_{g (1:d_1)},\hat{\bm{\beta}}_{g (1:d_1)}^{\known}\in\mathbb{R}^{d_1}$ are the first $d_1$ elements of $\hat{\bm{\beta}}_{g},\hat{\bm{\beta}}_{g}^{\known}$. The matrices $\bm{V}_{g (1:d_1)}^{\known}$ and the observable $\hat{\V}(\hat{\bm{\beta}}_g)_{(1:d1)}$ are the first $d_1 \times d_1$ blocks of $\bm{V}_g^{\known}$ and the minus inverse Fisher information for the likelihood $h_g$ in \eqref{equation:thetagLikelihood} evaluated at the first Fisher scoring step, respectively.
\end{theorem}
\noindent Result \eqref{equation:AsymEquivalent} indicates both the estimate $\hat{\bm{\beta}}_{g (1:d_1)}$ and corresponding inference using $\hat{\V}(\hat{\bm{\beta}}_g)_{ (1:d1)}$ is asymptotically equivalent to that when both $\bm{C}$ and genetic effects are known. Together with Proposition~\ref{proposition:Cknown}, this justifies using standard Wald intervals and tests to perform inference.

Two features of Theorem~\ref{theorem:betag} imply our estimates are computationally efficient. First, we need only apply a single iteration of Fisher scoring per metabolite. While we allow more than one iteration in practice, convergence is fast (see supplemental Section~\ref{supp:section:RealData}). Second, Theorem~\ref{theorem:betag} indicates differential abundance inference incurs no cost when using the computationally efficient, but statistically sub-optimal, IPW-based estimate for $\bm{C}$ in Section~\ref{subsection:LF}. This is critical, since the likelihood-based estimate is prohibitively slow to compute due to repeated numerical integration.

\subsection{Fidelity of latent factor-corrected mtGWAS}
\label{subsection:theory:mtGWAS}
Here we justify our mtGWAS method from Section~\ref{subsection:mtGWAS}. Recall $\eta_{gs}^{(e)}$, $\eta_{gs}^{(c)}$, and $\eta_{gs}^{(c,e)}$ are the test statistics that test whether the genotype at SNP $s$ affects metabolite $g$'s idiosyncratic variation $e_{gi}$, low-dimensional variation $\bm{\ell}_g^{\top}\bm{c}_i$, and total variation $\bm{\ell}_g^{\top}\bm{c}_i + e_{gi}$. As we did in Section~\ref{subsection:mtGWAS}, we assume $\bm{X}=0$ for simplicity, but show in Section~\ref{supp:section:TheorymtGWAS} that the extension to general $\bm{X}$ is simple.

\begin{theorem}
\label{theorem:mtGWAS}
Fix a $g \in [p]$, suppose $\bm{X}=0$ and Assumptions~\ref{assumption:y} and \ref{assumption:Miss} hold, and let $\gamma_{gs}^{(e)},\bm{\gamma}_{s}^{(c)}$ be as defined in Assumption~\ref{assumption:y}. Then $\eta_{gs}^{(e)} \tdist \chi^2_1$ if $H_{0,gs}^{(e)}: \gamma_{gs}^{(e)}=0$ is true. If $n^{1/2}\norm{\bm{\ell}_g}_2 \to \infty$, then $\eta_{gs}^{(c)} \tdist \chi^2_1$ if $H_{0,gs}^{(c)}: \bm{\ell}_g^{\top}\bm{\gamma}_{s}^{(c)}=0$ is true and $\eta_{gs}^{(c,e)} \tdist \chi^2_2$ if $H_{0,gs}^{(c,e)}: \gamma_{gs}^{(e)}=\bm{\ell}_g^{\top}\bm{\gamma}_{s}^{(c)}=0$ is true.
\end{theorem}

\begin{remark}
The non-trivial effect of latent factors suggests $n^{1/2}\norm{\bm{\ell}_g}_2$ is large for most $g$. 
\end{remark}


\section{Real data analysis}
\label{section:RealData}
We used three metabolomic datasets to evaluate our method MS-NIMBLE. Table~\ref{Table:Cohorts} describes the data, which were collected from the plasma of children that were part of the Copenhagen Prospective Study on Asthma in Childhood (COPSAC) \citep{COPSAC} or Infant Susceptibility to Pulmonary Infections and Asthma Following RSV Exposure Study (INSPIRE) \citep{INSPIRE} cohorts. We partitioned metabolites into ``observed'' metabolites ($<5\%$ missing data) and metabolites with missing data ($\geq 5\%$ but $\leq 50\%$ missing data), and discarded metabolites with $>50\%$ missing data. We were primarily interested in metabolites with missing data. Supplemental Section~\ref{supp:section:simulations} provides simulations further demonstrating MS-NIMBLE's superior performance. 


\begin{table}
\centering
\includegraphics[width=0.85\textwidth]{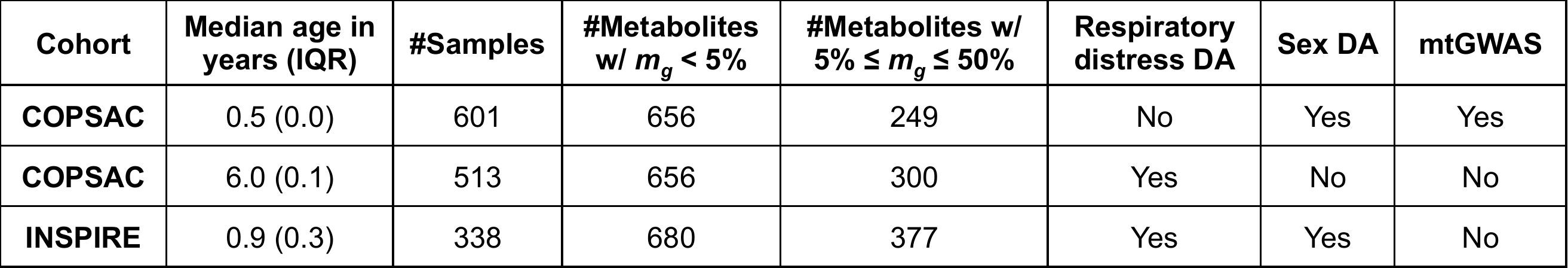}
\caption{An overview of the real data analyzed in Section~\ref{section:RealData}, where $m_g$ is the fraction of metabolite $g$'s observations that are missing. The sixth and seventh columns indicate whether a differential abundance (DA) analysis was performed using respiratory-related traits and sex. The last column indicates if the dataset was used to perform the mtGWAS.}\label{Table:Cohorts}
\end{table}

\subsection{Real data differential abundance analyses}
\label{subsection:realData:DA}
Since the COPSAC and INSPIRE studies were designed to investigate respiratory illness through childhood, we first used MS-NIMBLE to identify respiratory-related metabolites. Specifically, we considered the phenotypes specific airway resistance (sRAW), a measure of airway patency in the COPSAC cohort, and infant wheeze, defined as whether the infant wheezed during the first year of life in the INSPIRE cohort. Since there was no evidence of sRAW-related metabolites in infancy, we did not consider the 0.5 year COPSAC dataset in this analysis.

We compared MS-NIMBLE's estimators for and inference on $\bm{\beta}_g$ from Section~\ref{subsection:CoeffInt} to two competing approaches. The first, MetabMiss \citep{MetabMiss}, uses the estimates for the missingness mechanism parameters from Section~\ref{subsection:MissMech} and takes a similar approach as that in Section~\ref{subsection:LF} to recover latent factors. However, its estimates for $\bm{\beta}_g$ discard missing data, and are therefore expected to be substantially less powerful than MS-NIMBLE. The second imputes missing data using one of minimum imputation, singular value decomposition (SVD), K-nearest neighbors (KNN), or random forest (RF), the four most commonly used imputation techniques \citep{Impute2}, and subsequently estimates $\bm{\beta}_g$ using the latent factor-correction method CATE \citep{CATE}. While many methods can adjust for latent factors in imputed data, we found CATE gave the best simulation results in supplemental Section~\ref{supp:section:simulations}. To facilitate inter-method comparisons, the number of latent factors was set to be the same for each method, and, as done previously \citep{MetabMiss,Bilirubin}, was estimated via parallel analysis applied to metabolites with no missing data. Supplemental Section~\ref{supp:section:RealData} contains additional details, including method-specific software settings.

Figure~\ref{Figure:DAAsthma}(a) gives the number of respiratory-associated metabolites with missing data identified by each method at a q-value threshold of 0.2. As expected, MS-NIMBLE identifies over three times as many metabolites as MetabMiss, where the three piperine metabolites identified by MetabMiss, whose relationship with sRAW has previously been explored \citep{MetabMiss}, were also identified by MS-NIMBLE (Figure~\ref{Figure:DAAsthma}(b)). Figure~\ref{Figure:DAAsthma}(c) provides a biological explanation for the remaining metabolites in Figure~\ref{Figure:DAAsthma}(b) uniquely identified by MS-NIMBLE, which helps argue the veracity of MS-NIMBLE's identifications. The small p-values in Figure~\ref{Figure:DAAsthma}(a), which test the null hypothesis from Theorem~\ref{theorem:Omega}, suggest latent factors confound the relationship between the two respiratory traits and metabolite levels. As a consequence, Section~\ref{section:Imputation} and supplemental Section~\ref{supp:section:simulations} suggest imputation methods are inflating type I error rates, thereby casting doubt on their results.

\begin{figure}
\centering
\includegraphics[width=1\textwidth]{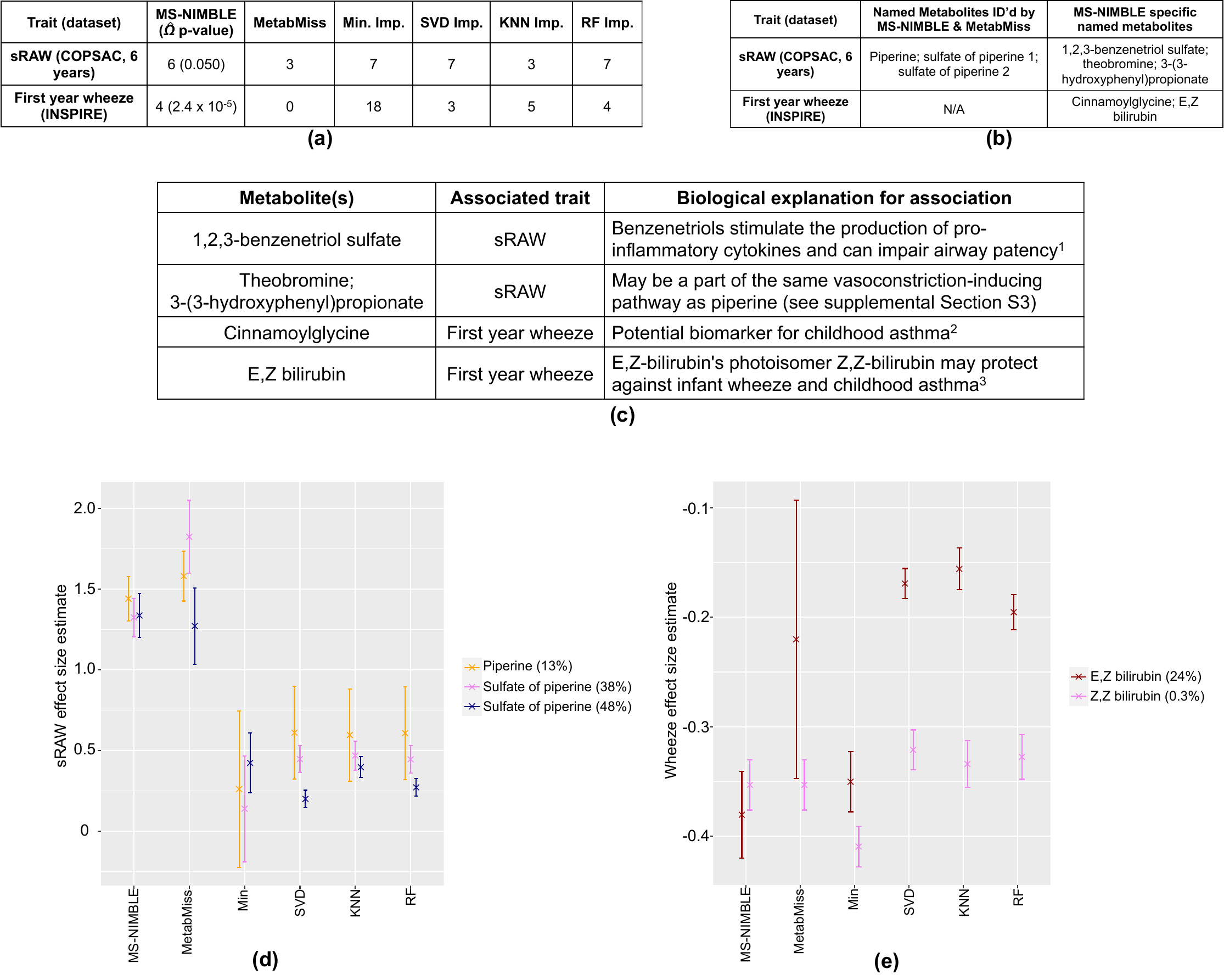}
\caption{Respiratory-related differential abundance results. \textbf{(a)}: The number of metabolites with missing data identified at a q-value threshold of 0.2. MS-NIMBLE's p-value is the p-value for the null hypothesis from Theorem~\ref{theorem:Omega} that the trait is not related to the latent factors. \textbf{(b)}: Named metabolites with missing data that were identified by MS-NIMBLE and MetabMiss (second column) and MS-NIMBLE but not MetabMiss (third column). MS-NIMBLE identified two unnamed wheeze-associated metabolites. \textbf{(c)}: Biological plausibility of metabolites from (b) uniquely identified by MS-NIMBLE. Superscripts are 1: \citet{Benzenetriol}; 2: \citet{Cinnamon}; 3: \citet{Bilirubin}. \textbf{(d)}-\textbf{(e)}: Effect estimates and 95\% confidence intervals for selected metabolites. Numbers in parentheses are the fractions of missing metabolite data.}\label{Figure:DAAsthma}
\end{figure}

Having argued hypothesis testing with MS-NIMBLE is sensitive and specific, we turn our attention to the reliability of MS-NIMBLE's coefficient estimates for respiratory-associated metabolites. Since the ground truth is unknown, we study similar metabolites, as they are likely to have similar effects. Given our results in Figure~\ref{Figure:DAAsthma}(b), we consider piperine- and bilirubin-related metabolites, where we chose the latter because E,Z-bilirubin's photoisomer Z,Z-bilirubin was fully observed and shown to be a replicable biomarker for infant wheeze \citep{Bilirubin}. Figures~\ref{Figure:DAAsthma}(d)-(e) provide the results, which illustrate the consistency of MS-NIMBLE's estimates. Figure~\ref{Figure:DAAsthma}(e) is particularly interesting, as it suggests MS-NIMBLE's estimates and standard errors for metabolites with missing data are as reliable as those for fully observed metabolites.

To further explore the fidelity of MS-NIMBLE's estimates, we compared estimators for the effect of sex, an important source of metabolite variation, on metabolite levels in the 0.5 year COPSAC and INSPIRE datasets. Let $\hat{\beta}_g^{(\text{C})},\hat{\beta}_g^{(\text{I})}$ be a method's sex effect estimates for metabolite $g$ in COPSAC and INSPIRE and $\hat{\V}(\cdot)$ their estimated variances. Since these data were collected from unrelated infants at similar ages, their sex effects should be the same, meaning the z-score $\{ \hat{\beta}_g^{(\text{C})}\allowbreak - \allowbreak \hat{\beta}_g^{(\text{I})} \}/[ \hat{\V}\{ \hat{\beta}_g^{(\text{C})} \} + \hat{\V}\{ \hat{\beta}_g^{(\text{I})} \} ]^{1/2}$ should be approximately $N(0,1)$. Interestingly, the metabolome-wide z-scores for imputation-based methods, but not MS-NIMBLE, were significantly inflated (Table~\ref{Table:DAGender}), indicating imputation-based estimates and their standard errors are unreliable. While several factors are likely responsible for this inflation, we hypothesized errant effect estimates for metabolites with missing data were partly responsible. Given Section~\ref{section:Imputation} and supplemental Figure~\ref{supp:figure:SimCI}'s simulation results showing estimates in trait-associated metabolites are most corrupted by missing data, we considered z-scores for the 64 sex-associated missing metabolites, defined as metabolites with missing data and sex q-values $\leq 0.2$ in at least one method, dataset pair. Consistent with our hypothesis, Table~\ref{Table:DAGender} shows these z-scores were inflated in imputation methods, whereas MS-NIMBLE showed no evidence of inflation. The conclusions were the same even when we separately examined each method's sex-associated missing metabolites (Figure~\ref{supp:figure:ZscoresSub}), implying differences between MS-NIMBLE and imputation methods could not be attributed to metabolite selection biases, and indicate MS-NIMBLE's estimates and standard errors are accurate.

\begin{table}
\centering
\includegraphics[width=0.8\textwidth]{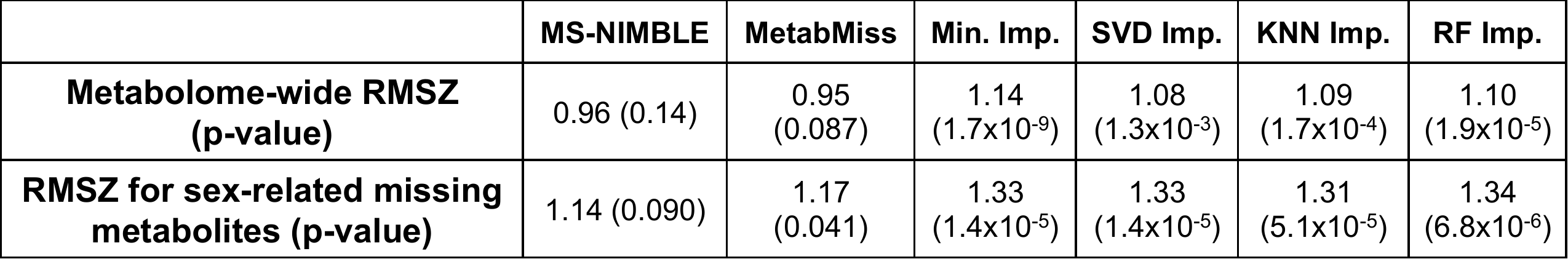}
\caption{Root mean squared z-score (RMSZ) for all analyzed metabolites (second row) and the 64 sex-associated missing metabolites (third row), where an RMSZ $>1$ suggests z-scores are inflated. The p-value is for the null hypothesis that z-scores are $N(0,1)$.}\label{Table:DAGender}
\end{table}






\subsection{Metabolite GWAS in the six month COPSAC data}
\label{subsection:realData:mtGWAS}
We examined the effect of genotype at 1.4 million SNPs on metabolite levels in the six month COPSAC data to evaluate the performance of our methodology proposed in Section~\ref{subsection:mtGWAS}. As far as we are aware, only \citet{CMS_metab} has considered controlling for latent sources of variation in mtGWAS studies. However, their method requires determining a set of latent covariates for each metabolite-SNP pair, which, as determined by their simulation results, would take 12 CPU Years if applied to our data. We therefore compared our results to those using the current state of the art, which involves first imputing missing metabolite levels and subsequently regressing them onto genotype without considering latent variation \citep{CMS_metab}. We present results for minimum imputation, but note imputation technique did not alter results.

Figure~\ref{Figure:mtGWAS}(a) contains the results, where the second and third rows imply nearly all of the genetic effect is idiosyncratic and appears in the error terms $e_{gi}$, whereas there is no evidence indicating latent factors $\bm{c}_i$ mediate genetic effects. This suggests mtGWAS analyses should be performed conditional on estimated latent factors, as in the second row of Figure~\ref{Figure:mtGWAS}(a), which is equivalent to data de-noising. This is recapitulated by Figure~\ref{Figure:mtGWAS}(b), which shows such de-noising reduces the residual variance by $\approx 40\%$, thereby effectively increasing the sample size by 67\%.

\begin{figure}[t!]
\centering
\includegraphics[width=0.9\textwidth]{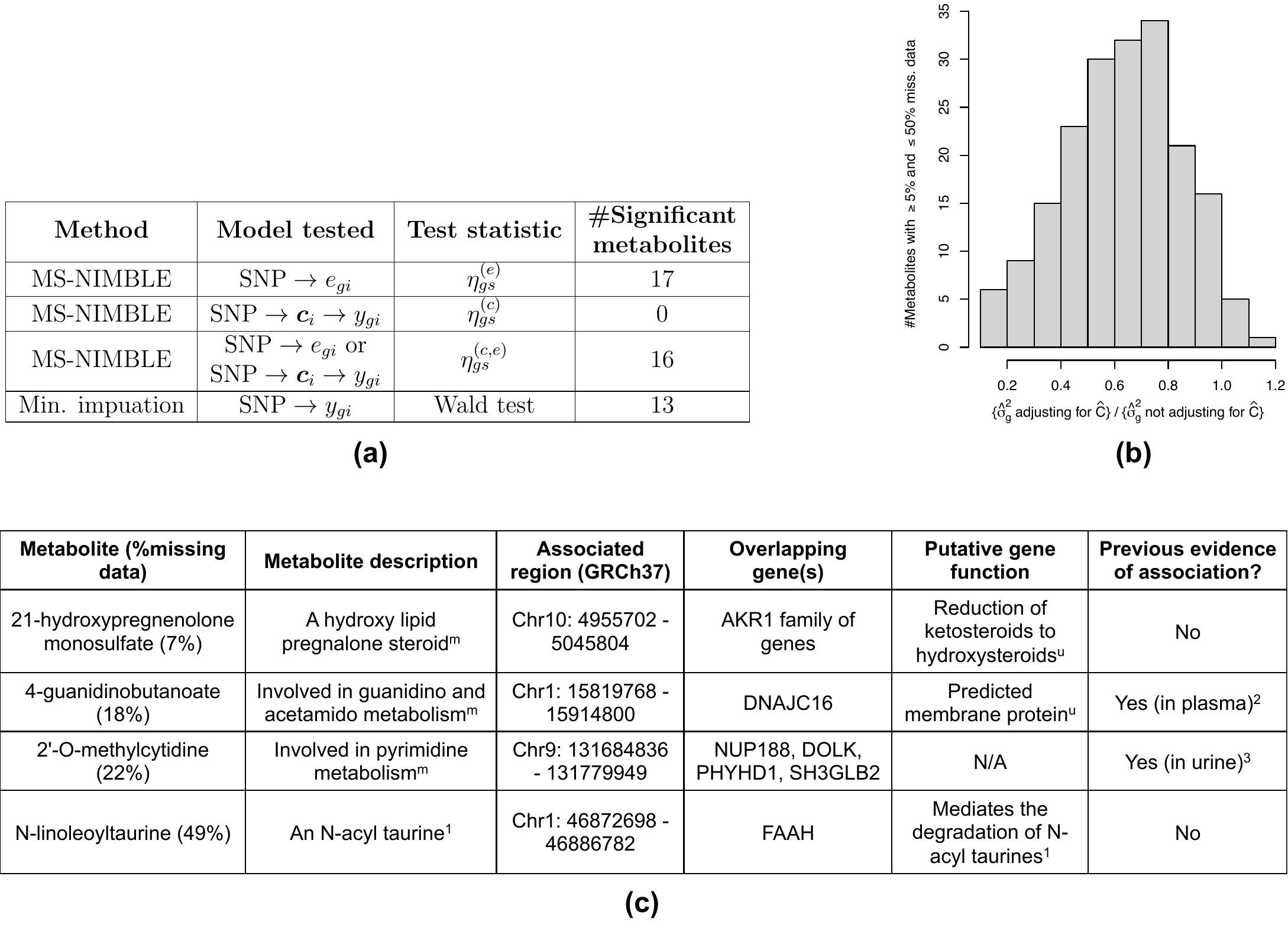}
\caption{mtGWAS results for metabolites with missing data. \textbf{(a)}: A metabolite was ``significant'' if it was associated with at least one SNP at the Bonferroni p-value threshold $(5 \times 10^{-8})/(656+249)$. \textbf{(b)}: Reduction in residual variance after adjusting for latent factors. \textbf{(c)}: Named metabolites that were identified by MS-NIMBLE in the second row of (a), but not minimum imputation. A metabolite-genomic region association had previous evidence if the region contained SNPs previously shown to be associated with the metabolite. Superscripts are m: derived from Metabolon; u: obtained from Uniprot; 1: \citet{mtGWAS_NAT}; 2: \citet{mtGWAS_4guan}; 3: \citet{mtGWAS_methyl}.}\label{Figure:mtGWAS}
\end{figure}

The last row of Figure~\ref{Figure:mtGWAS}(a) indicates existing approaches are underpowered, where 11 out of the 13 metabolites identified by minimum imputation were among the 17 metabolites identified by our proposed method in row two of Figure~\ref{Figure:mtGWAS}(a). To explore the veracity our method's results, we sought to evaluate the biological significance of the four named metabolites uniquely identified by our method. We did not consider the other two metabolites, since they were unnamed. Figure~\ref{Figure:mtGWAS}(c) shows that two out of the four associations have previously been observed, whereas, to the best of our knowledge, the results involving 21-hydroxypregnenolone monosulfate and N-linoleoyltaurine are novel. Critically, their metabolite descriptions and associated gene functions are congruent, suggesting our method improves power to identify genuine mtGWAS associations.

\section{Conclusion}
\label{section:Conclusion}
We developed MS-NIMBLE, a rigorous suite of methods to analyze metabolomics data with non-ignorable missing observations and latent factors that offers all the practical advantages of missing data imputation. We derived its theoretical properties and demonstrated its superior performance in differential abundance and mtGWAS using three real datasets. We believe this work offers a critical step towards reliable estimation and inference in metabolomic studies.


\section*{Acknowledgments}
\noindent This work was supported by the NIH [UL1 TR001857 (C.M.), K01 HL149989 (K.T.), U19 AI 095227 (T.H.), UG3/UH3 OD023282 (T.H.), UL1 TR002243 (T.H.), R01 HL129735 (C.O.)]. All funding received by COPSAC is listed on www.copsac.com.

\printbibliography

\newpage

\begin{center}
    {\Large \textbf{Supplemental material for ``From differential abundance to mtGWAS: accurate and scalable methodology for metabolomics data with non-ignorable missing observations and latent factors''}}
\end{center}
\allowdisplaybreaks

\setcounter{equation}{0}
\setcounter{theorem}{0}
\setcounter{figure}{0}
\setcounter{table}{0}
\setcounter{section}{0}
\renewcommand{\thefigure}{S\arabic{figure}}
\renewcommand{\thetable}{S\arabic{table}}
\renewcommand{\thesubfigure}{(\alph{subfigure})}
\renewcommand{\theequation}{\thesection.\arabic{equation}}
\renewcommand{\thesection}{S\arabic{section}}
\renewcommand{\thetheorem}{\thesection.\arabic{theorem}}
\renewcommand{\theremark}{\thesection.\arabic{remark}}
\renewcommand{\thelemma}{\thesection.\arabic{lemma}}
\renewcommand{\thecorollary}{\thesection.\arabic{corollary}}

\section{The normality assumption}
\label{supp:section:Normal}
We rely on the assumption that $e_{gi}$ in \eqref{equation:MainModel} is approximately normally distributed to develop statistically efficient estimators. However, while we assume $e_{gi}$ is approximately normal, we do not assume $y_{gi}$ is normal. This is a critical distinction, as Figure~\ref{supp:Figure:Normal} indicates the latent factors $\bm{c}_i$ may be highly skewed.

\begin{figure}[b!]
\centering
\includegraphics[width=0.8\textwidth]{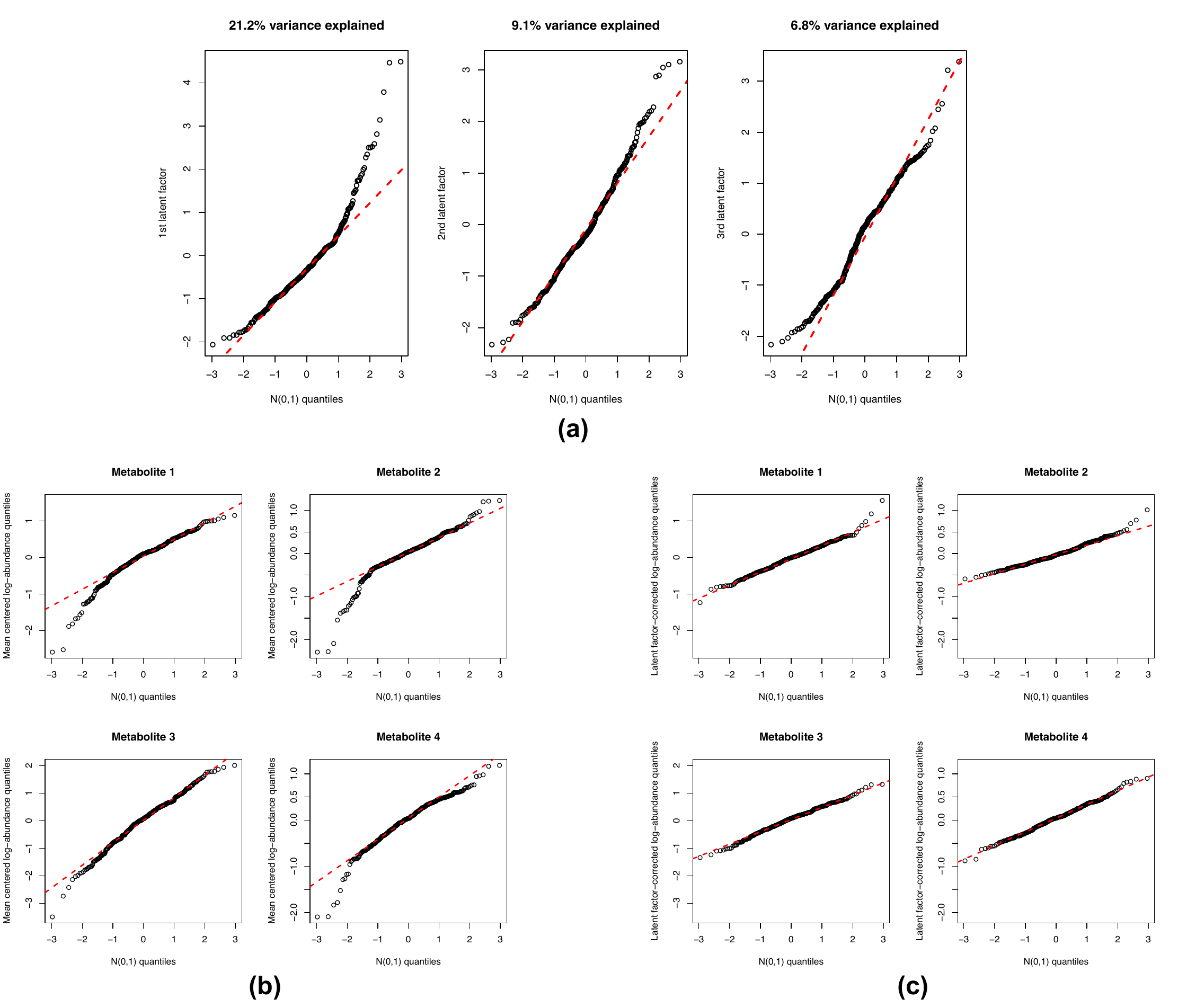}
\caption{Normality of plasma metabolite levels from \citet{Bilirubin}. (a) Normal Q-Q plots for the first three estimated latent factors. (b) Normal Q-Q plots for four randomly chosen metabolites. (c) Q-Q plots for the same four metabolites, except after regressing out the $K=19$ latent factors.}\label{supp:Figure:Normal}
\end{figure}

\section{Simulations}
\label{supp:section:simulations}

\subsection{Simulation setup}
\label{supp:subsection:SimulationSetup}
We simulated 50 datasets containing $p=1200$ metabolites measured in $n=600$ samples with missing observations and $K=10$ latent factors to best mirror our real data from Section~\ref{section:RealData}. We partitioned individuals into equal sized treatment and control groups, where the covariate of interest $\bm{X} \in \{0,1\}^n$ denotes treatment status. For some constant $a \in \mathbb{R}$ controlling the dependence of latent factors on $\bm{X}$, metabolite levels $y_{gi}$ and missingness indicators $r_{gi}$ were then simulated according to \eqref{supp:equation:MetabMiss:Simulation} below.

\begin{subequations}
\label{supp:equation:MetabMiss:Simulation}
\begin{align}
    \label{supp:equation:Sim:adelta}
    &\log\left( \alpha_{g}\right) \sim N_1\left( \mu_{\alpha},0.4^2\right), \quad \delta_{g} \sim N_1\left( 16, 1.2^2\right), \quad g \in [p]\\
    \label{supp:equation:Sim:C}
    &\bm{C}=\left(\bm{c}_1 \cdots \bm{c}_n\right)^{\top} \sim MN_{n \times K}\left( \left( a\bm{X}, \, a\bm{X},\, \bm{0}_n \cdots \bm{0}_n\right), I_n, I_K \right)\\
    \label{supp:equation:Sim:ell}
    &\bm{\ell}_{g_k} \sim \pi_k \delta_0 + \left( 1-\pi_k\right)N_1\left( 0,\tau_k^2\right), \quad g\in[p]; k\in[K]\\
    \label{supp:equation:Sim:meanvar}
    &\mu_{g} \sim N_1\left(18, 5^2 \right), \quad \sigma_{g}^2 \sim \text{Gamma}\left( 0.2^{-2}, 0.2^{-2}\right), \quad g\in[p]\\
    \label{supp:equation:Sim:effect}
    &\beta_{g} \sim 0.8\delta_0 + 0.2N_1\left( 0, 0.4^2\right), \quad g\in[p]\\
    \label{supp:equation:Sim:y}
    &y_{gi} \sim N_1\left( \mu_{g} + \bm{X}_{i}\beta_{g} + \bm{c}_i^{\top} \bm{\ell}_{g}, \sigma_{g}^2 \right), \quad g\in[p]; i\in[n]\\
    \label{supp:equation:Sim:r}
    &r_{gi} \sim \text{Bernoulli}\left[ \tilde{\Psi}\left\lbrace \alpha_{g}\left(y_{gi}-\delta_{g} \right) \right\rbrace \right], \quad g\in[p]; i\in[n]
\end{align}
\end{subequations}
where $\delta_0$ is the point mass at 0 and $\mu_{\alpha}$ in \eqref{supp:equation:Sim:adelta} was set so that if $Z$ has cumulative distribution function $\tilde{\Psi}\left\lbrace \exp\left( \mu_{\alpha}\right)x \right\rbrace$, $\V(Z)=1$. To study scenarios where we incorrectly specify $\Psi$ in \eqref{equation:MissingDataModel}, we let $\tilde{\Psi}$ in \eqref{supp:equation:Sim:r} be the cumulative distribution function (CDF) of a logistic random variable, but analyzed the data assuming $\Psi$ was the CDF of a t-distribution with four degrees of freedom. The normal means and variances in \eqref{supp:equation:Sim:adelta} and \eqref{supp:equation:Sim:meanvar} were chosen to match those estimated in the three datasets from Section~\ref{section:RealData}, and the parameters used to simulate the loadings $\bm{\ell}_g$ in \eqref{supp:equation:Sim:ell} are given in Table~\ref{supp:Table:LMetab}. The loadings were chosen so that the eigenvalues $\lambda_1,\ldots,\lambda_K$ from Assumption~\ref{assumption:y} ranged from $n^{-0.47}=0.05$ to 0.80 on average, which mirrored the eigenvalues estimated from the six year COPSAC data (see Table~\ref{Table:Cohorts}). The constant $a$ in \eqref{supp:equation:Sim:C} was chosen so that $\bm{C}$ explained 60\% of the variance in $\bm{X}$ on average, and was chosen to match the substantial correlation between latent factors and infant wheeze in INSPIRE (see Figure~\ref{Figure:DAAsthma}(a)). Lastly, we simulated the effects of interest $\beta_g$ in \eqref{supp:equation:Sim:effect} to violate the sparsity assumption in (i) of Theorem~\ref{theorem:Omega}, which is also used to prove Theorem~\ref{theorem:betag}.

\LtabMetab

\subsection{Simulation results}
\label{supp:subsection:SimulationResults}
We compared MS-NIMBLE's estimates for $\beta_g$ to those from MetabMiss \citep{MetabMiss} and imputation-based methods, the latter of which first impute missing data with one of minimum imputation, singular value decomposition (SVD), K-nearest neighbors (KNN), or random forest (RF), and subsequently use CATE \citep{CATE} to estimate latent factors. While other methods are capable of estimating latent factors in complete data, we found that CATE gave the best results. Imputation hyperparameters were $K=10$ factors for SVD and the software defaults recommended in \citet{Impute2} for KNN and RF. The estimates for $\alpha_g$ and $\delta_g$, which were used by both MS-NIMBLE and MetabMiss, were obtained using the method proposed in \citet{MetabMiss} and outlined in Section~\ref{subsection:MissMech} with 5 potential instruments. We do not include results when $\bm{C}$ is known or when it is ignored, as they both performed similarly to and uniformly worse than KNN imputation, respectively.

On the average, 485 metabolites were fully observed (i.e. missing in $< 5\%$ of samples) and 300 were missing (i.e. $\geq 5\%$ but $\leq 50\%$ missing data). Metabolites with $>50\%$ missing data were discarded. We first consider each method's ability to identify missing metabolites with non-zero $\beta_g$. Figure~\ref{supp:figure:SimFDP} gives the results, where Figure~\ref{supp:figure:SimFDP}(a) indicates MS-NIMBLE and, to a lesser extent, MetabMiss are able to control false discovery rates at their nominal levels. However, Figure~\ref{supp:figure:SimFDP}(b) indicates MS-NIMBLE has 50\% greater power than MetabMiss to identify treatment-related metabolites with missing data. These results are consistent with the fact that while MetabMiss does use inverse probability weighting to account for the non-ignorable missing data, their estimates for $\beta_g$ discard missing data, and are therefore less powerful. On the other hand, imputation-based methods inflate error rates and have poor power. The former is consistent with our discussion from Section~\ref{section:Imputation}, as their false discovery proportions resembled nominal levels when we simulated data with latent factors $\bm{C}$ that did not depend on $\bm{X}$.

\begin{figure}
\centering
\includegraphics[width=0.8\textwidth]{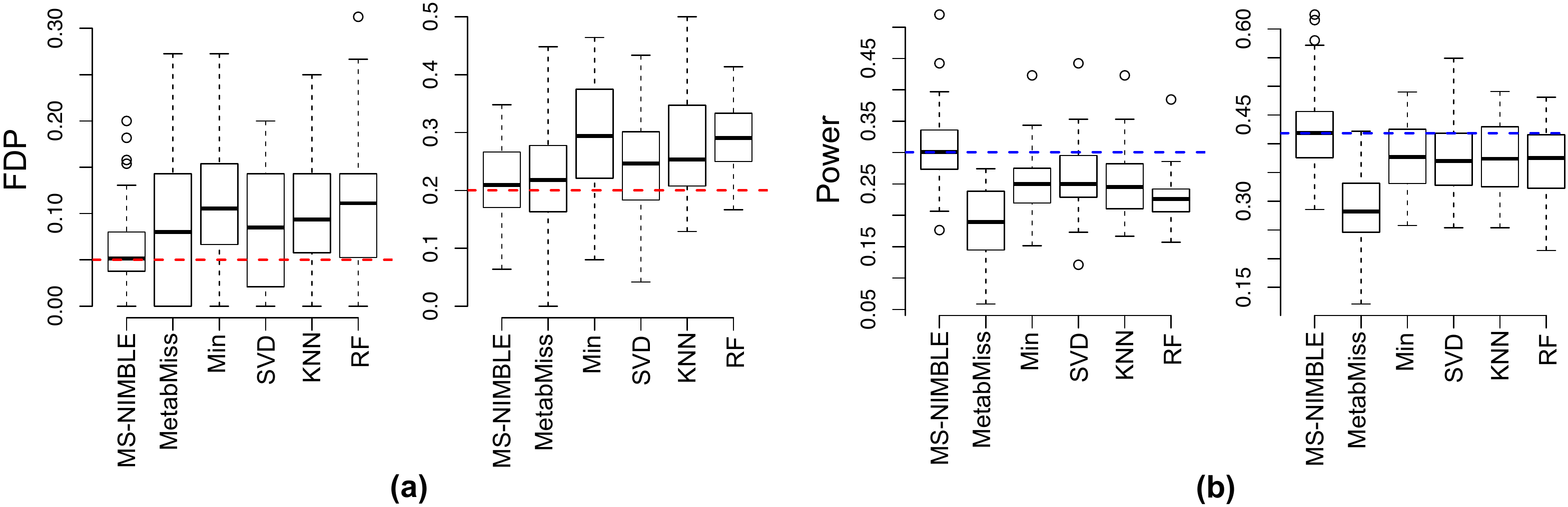}
\caption{False discovery proportion (a) and power (b) for metabolites with missing data at q-value thresholds of 0.05 (left) and 0.2 (right). The dashed red and blue lines indicate the q-value thresholds and MS-NIMBLE's median power, respectively.}\label{supp:figure:SimFDP}
\end{figure}

We lastly considered each method's estimates and 95\% confidence intervals for $\beta_g$ for metabolites $g$ with missing data, where confidence intervals were standard Wald intervals assuming estimates for $\beta_g$ were approximately normal. Figure~\ref{supp:figure:SimCI}(a) contains the results, where only MS-NIMBLE and MetabMiss return accurate intervals. However, consistent with the above discussion and results from Section~\ref{subsection:realData:DA}, MetabMiss's intervals are on average over 25\% wider than MS-NIMBLE's (Figure~\ref{supp:figure:SimCI}(b)). We also see that imputation-based intervals become less accurate as $\abs*{\beta_g}$ increases, which corroborates our discussion in Section~\ref{section:Imputation}.

\begin{figure}
\centering
\includegraphics[width=0.8\textwidth]{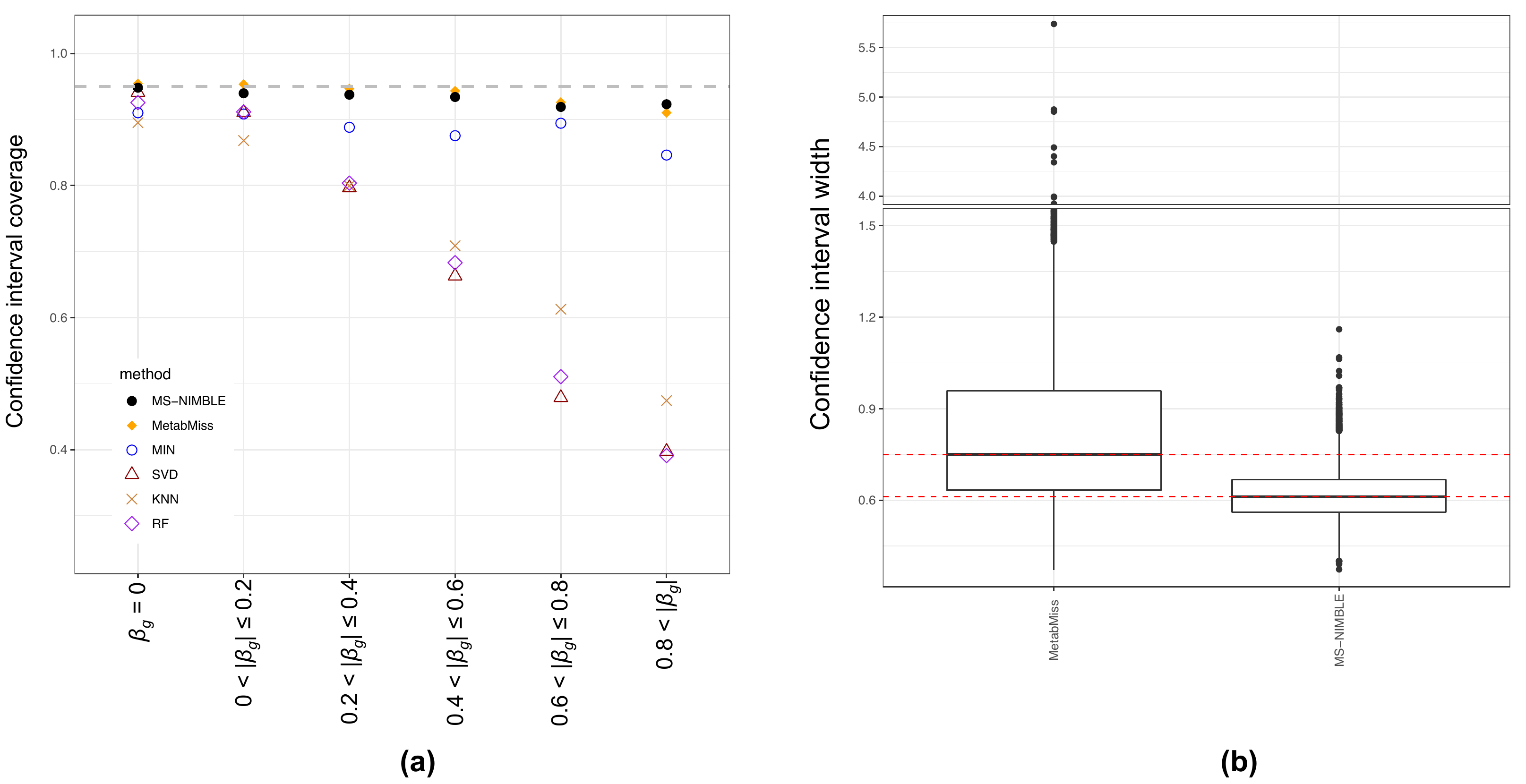}
\caption{\textbf{(a)}: 95\% confidence interval coverage for metabolites with missing data. The dashed grey line indicates 95\% coverage. \textbf{(b)}: 95\% confidence interval widths for metabolites with missing data. Each point represents a simulated metabolite with missing data. MS-NIMBLE's and MetabMiss's confidence interval widths did not depend on $\beta_g$.}\label{supp:figure:SimCI}
\end{figure}

\section{Additional real data details and results from Section~\ref{section:RealData}}
\label{supp:section:RealData}

\subsection{Additional real data and analysis details}
\label{supp:subsection:RealData:details}
Raw metabolite intensities were log base 2-transformed. There were no additional quality control or pre-processing steps.

Missingness mechanism parameters $\alpha_g,\delta_g$, which are used by both MS-NIMBLE and MetabMiss, were estimated using the procedure outlined in Section~\ref{subsection:MissMech} and described in detail in \citet{MetabMiss} with 10 potential instruments. Missing data were imputed exactly as described in Section~\ref{supp:subsection:SimulationResults}.

Differential abundance regressions in INSPIRE were performed by controlling for the observed covariates daycare status (yes/no), breast-feeding status (exclusively breast-fed or not in the first six months of life), age in months, and sex in the first year wheeze analysis. The sRAW analysis in the six year COPSAC dataset was done conditional on sex, and we did not include any observed nuisance covariates in the 0.5 year COPSAC sex regression.

\subsection{Additional real data results}
\label{supp:subsection:RealData:results}
We first justify the observed relationship between infant wheeze and theobromine and 3-(3-hydroxyphenyl)propionate levels in INSPIRE (see Figure~\ref{Figure:DAAsthma}(c)), where wheezers tended to have higher plasma concentrations of both metabolites. Theobromine is an alkaloid commonly found in the cacao plant, and is a notable adenosine receptor antagonist \citep{Theobromine1}. Higher theobromine concentrations tend to increase plasma adenosine levels \citep{Theobromine2}, thereby potentially exacerbating adenosine's bronchoconstricting properties \citep{Theobromine2,Adenosine}. The metabolite 3-(3-hydroxyphenyl)propionate is a phenolic degradation product of proanthocyanidins, the most abundant polyphenols present in chocolate \citep{Prop}, and therefore may simply correlate with infant wheeze because it correlates with theobromine levels.

We next consider the sex-related z-scores defined in Section~\ref{subsection:realData:DA}. To argue that the inter-method differences in root mean squared z-scores for the 64-sex related metabolites was not due to metabolite selection bias (i.e. winner's curse), we investigated each method's sex-associated metabolites with missing data. The results are given in Figure~\ref{supp:figure:ZscoresSub}, and show that only MS-NIMBLE's z-scores show no evidence of inflation. This suggests that differences between MS-NIMBLE and imputation methods in Table~\ref{Table:DAGender} cannot be attributed to metabolite selection bias.

We lastly consider MS-NIMBLE's computation time. The most computationally demanding component in differential abundance analyses is estimating each the missingness mechanism parameters $\alpha_g,\delta_g$ (see Section~\ref{subsection:MissMech}), which took 40 minutes for the 0.5 year COSAC dataset (the dataset with the largest sample size). However, this only needed to be computed once, and was stored for use in all downstream analyses. The subsequent sex analysis in the 0.5 year COSAC dataset took 3.4 minutes.

\begin{figure}
\centering
\includegraphics{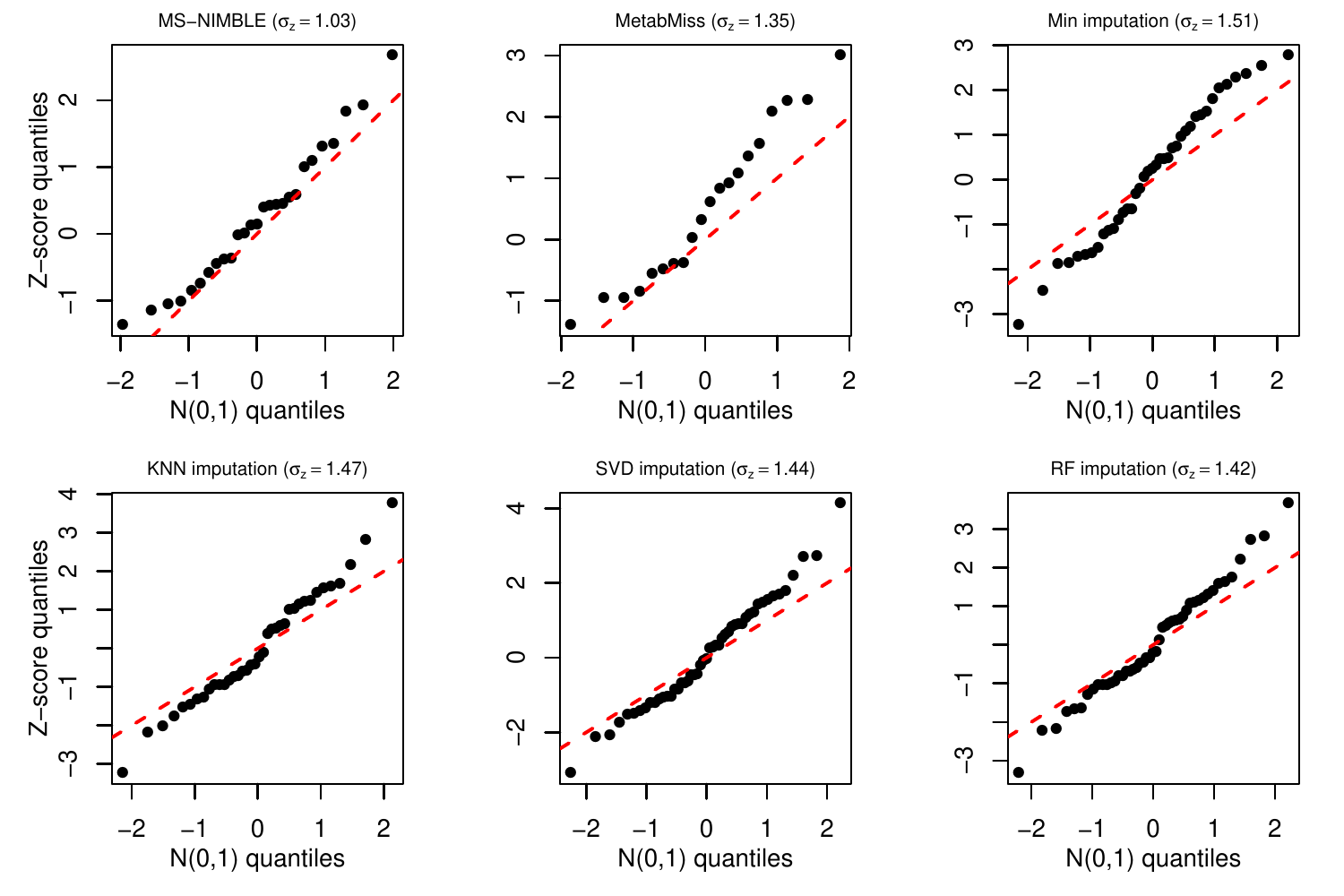}
\caption{Q-Q plot for each method's sex-associated missing metabolites. The z-score is as defined in Section~\ref{subsection:realData:DA} and a metabolite was sex-associated in a method if it (i) was analyzed in both the six month COPSAC and INSPIRE datasets, (ii) contained missing data in at least one dataset, and (iii) had a q-value $\leq 0.2$ using that method. The statistic $\sigma_z$ is the method's root mean squared z-score for their sex-associated missing metabolites.}\label{supp:figure:ZscoresSub}
\end{figure}

\section{Refining our estimator for $\bm{\Omega}$}
\label{supp:section:RefineOmega}
Here we provide a way to refine our estimator for $\bm{\Omega}$ in \eqref{equation:Omegahat} that iteratively removes ``outlying'' metabolites that likely depend on the covariate(s) of interest. It should be noted that this is our software default estimator for $\bm{\Omega}$.

Briefly, assume $\bm{X}$ can be written as $\bm{X}=[\bm{X}_I,\bm{X}_N]$, where $\bm{X}_I \in \mathbb{R}^{n \times d_I}$ contains the $d_I$ covariates of interest and $\bm{X}_N$ contains the remaining nuisance covariates. Let $\bm{\beta}_{g,I},\allowbreak \hat{\bm{\beta}}_{g,I}^{\ipw} \in \mathbb{R}^{d_I}$ be the first $d_I$ elements of $\bm{\beta}_g$ and the inverse probability weighted estimator $\hat{\bm{\theta}}_g^{\ipw}$ defined in \eqref{equation:IPW}, respectively. Then Lemma~\ref{supp:lemma:IPWest} shows that under the same assumptions used to prove Theorem~\ref{theorem:betag}, $\norm*{ \hat{\bm{\beta}}_{g,I}^{\ipw} - \bm{\beta}_{g,I} }_2=O_P(n^{-1/2})$ if $d_I\leq d_1$. Further, it is straightforward to extend Lemma~\ref{supp:lemma:IPWest} to show that for
\begin{align*}
    \hat{\V}\{ \hat{\bm{\beta}}_{g,I}^{\ipw} \} = \left( (\textstyle\sum_{i=1}^n \hat{w}_{gi} \hat{\bm{z}}_i\hat{\bm{z}}_i^{\top})^{-1} [ \textstyle\sum_{i=1}^n \hat{w}_{gi}^2\{y_{gi} - \hat{\bm{z}}_i^{\top}\bm{\theta}_g^{\ipw} \}^2 \hat{\bm{z}}_i\hat{\bm{z}}_i^{\top} ] (\textstyle\sum_{i=1}^n \hat{w}_{gi} \hat{\bm{z}}_i\hat{\bm{z}}_i^{\top})^{-1} \right)_{1:d_I,1:d_I}
\end{align*}
the sandwich estimator for $\V\{ \hat{\bm{\beta}}_{g,I}^{\ipw} \}$, $x_g^2 =\{ \hat{\bm{\beta}}_{g,I}^{\ipw} \}^{\top} [ \hat{\V}\{ \hat{\bm{\beta}}_{g,I}^{\ipw} \} ]^{-1}\hat{\bm{\beta}}_{g,I}^{\ipw}\allowbreak \tdist \allowbreak \chi_{d_I}^2$ under the null hypothesis $H_{0,g}: \bm{\beta}_{g,I}=0$. To refine our estimate for $\bm{\Omega}$, we compute p-values for $H_{0,g}$ by comparing $x_g^2$ to the upper quantiles of a $\chi_{d_I}^2$, use \citet{qvalue} to subsequently determine q-values, and re-estimate $\bm{\Omega}$ using the regression in \eqref{equation:Omegahat} after removing metabolites from said regression whose q-values fall below a user-specified threshold $q$. Our software default is to let $q=0.1$ and iterate this procedure 3 times.

\section{Extensions when some metabolites have fully observed data}
\label{supp:section:FullyObs}

\subsection{Methodological extensions}
\label{supp:subsection:FullObs:Method}
The factor analysis- and mtGWAS-related estimators are the only estimators that need to be updated to allow fully observed metabolites. For the former, we simply let $\hat{w}_{gi}$ in \eqref{equation:OptC} be 1 if metabolite $g$ has no missing data. For the mtGWAS estimators described in Section~\ref{subsection:mtGWAS}, we regress $y_{gi}$ onto genotype $G_{si}$ and estimated latent factors $\hat{\bm{c}}_i$ to estimate $\gamma_{sg}^{(e)}$ and the estimator's variance. We then use standard Wald-based inference to test $H_{0,sg}^{(e)}$. Testing $H_{0,sg}^{(c)}$ remains unchanged. Since the test statistics used to test $H_{0,sg}^{(e)}$ and $H_{0,sg}^{(c)}$ are asymptotically $\chi_1^2$ and independent under Assumptions~\ref{assumption:y} and \ref{assumption:Miss}, we simply add the test statistics and compare it to the upper quantiles of a $\chi_2^2$ to test $H_{0,sg}^{(c,e)}$.

\subsection{Theoretical extensions}
\label{supp:subsection:FullObs:Theory}
The only theoretical extension we must consider is choosing an appropriate starting point for the estimator $\hat{\mathcal{P}}$ from Theorem~\ref{theorem:PxC}, which is discussed in Remark~\ref{remark:PxC}. Let $\mathcal{O} \subset [p]$ be the set of metabolites with fully observed data, $\lambda^{(\mathcal{O})}_1\geq \cdots \geq \lambda^{(\mathcal{O})}_K$ be the eigenvalues of $\sum_{g \in \mathcal{O}} \bm{\ell}_g \bm{\ell}_g^{\top}$, $\hat{\bm{V}} \in \mathbb{R}^{n \times K}$ be the first $K$ right singular vectors of $[y_{gi}]_{g \in \mathcal{O}; i \in [n]}P_{\bm{X}}^{\perp} \in \mathbb{R}^{\abs*{\mathcal{O}} \times n}$, and define $\hat{\mathcal{P}}^{(\mathcal{O})} = \hat{\bm{V}}\hat{\bm{V}}^{\top}$. Then under Assumptions~\ref{assumption:y} and \ref{assumption:Miss}, the proof of Theorem~4 in \citet{FALCO} can easily be used to show that $\norm*{ \hat{\mathcal{P}}^{(\mathcal{O})} - \mathcal{P} }_F^2 = O_P[\{\lambda^{(\mathcal{O})}_K n\}^{-1}]$. Therefore, if $\lambda^{(\mathcal{O})}_K \gtrsim n^{-1+\epsilon}$ for any $\epsilon>0$, Corollary~\ref{supp:corollary:Htilde} in Section~\ref{supp:subsection:FARate} implies $\hat{\mathcal{P}}$ will satisfy the condition $\norm*{ \hat{\mathcal{P}} - \mathcal{P} }_F \leq \eta$ in Theorem~\ref{theorem:PxC} when we solve the optimization in \eqref{equation:OptC} by initializing $\bm{C}_{\perp} = \hat{\bm{V}}$ and iteratively updating $\{\bm{b}_g,\bm{\ell}_g\}_{g \in [p]}$ and $\bm{C}_{\perp}$.

\section{Outline and notation for the rest of the supplement}
\label{supp:section:Notation}

\subsection{Outline for the remaining supplement}
The rest of the supplement is devoted to proving the theoretical statements made in Sections~\ref{section:Imputation} and \ref{section:Theory}. Due to its length, we give a compendious outline below.
\begin{itemize}
\item Section~\ref{supp:section:MinImp}: we provide the regularity conditions for and prove Proposition~\ref{proposition:Impute} stated in Section~\ref{section:Imputation}.
\item Section~\ref{supp:section:FATheory}: we prove Theorem~\ref{theorem:PxC}, Theorem~\ref{theorem:Omega}, Proposition~\ref{proposition:Cknown}, and Theorem~\ref{theorem:betag}. The proofs can be found in:
\begin{itemize}
    \item Theorem~\ref{theorem:PxC}: Corollary~\ref{supp:corollary:Pchat} in Section~\ref{supp:subsection:FARate}.
    \item Theorem~\ref{theorem:Omega}: Corollary~\ref{supp:corollary:OmegaInference} in Section~\ref{supp:subsection:Omega}.
    \item Proposition~\ref{proposition:Cknown}: A direct consequence of Lemma~\ref{supp:lemma:Cknown} in Section~\ref{supp:subsection:DiffAbund}. See Remark~\ref{supp:remark:Cknown}.
    \item Theorem~\ref{theorem:betag}: proven in Theorem~\ref{supp:theorem:InferenceBeta} in Section~\ref{supp:subsection:DiffAbund}.
\end{itemize}
\item Section~\ref{supp:section:TheorymtGWAS}: we extend our mtGWAS test statistics to allow $\bm{x}_i \neq 0$, prove an extension of Theorem~\ref{theorem:mtGWAS} that allows $\bm{x}_i \neq 0$, and illustrate the computational efficiency of our mtGWAS test statistics.
\end{itemize}

\subsection{Notation}
\label{supp:subsection:Notation}
For any matrix $\bm{M} \in \mathbb{R}^{m \times n}$, we define $\bm{M}_{i \bigcdot} \in \mathbb{R}^{n}$, $\bm{M}_{\bigcdot j} \in \mathbb{R}^m$, and $\bm{M}_{ij} \in \mathbb{R}$ to be the $i$th row, $j$th column, and $(i,j)$th element of $\bm{M}$, respectively. We also define $P_{\bm{M}}, P_{\bm{M}}^{\perp}\in \mathbb{R}^{n \times n}$ to be the orthogonal projections matrices that project vectors onto the image of $\bm{M}$ and kernel of $\bm{M}^{\top}$. Let $\{\bm{X}_n\}_{n \geq 1}$ be a sequence of random vectors or matrices. Unless otherwise specified, $\bm{X}_n = O_p(a_n)$ if $\norm*{\bm{X}_n}_2/a_n = O_P(1)$ and $\bm{X}_n = o_p(a_n)$ if $\norm*{\bm{X}_n}_2/a_n = o_P(1)$ as $n \to \infty$. Lastly, for random vector $\bm{e}$, we use the notation $\bm{e} \sim (\bm{\mu},\bm{V})$ if $\E(\bm{e})=\bm{\mu}$ and $\V(\bm{e})=\bm{V}$.

\section{Proof of Proposition~\ref{proposition:Impute}}
\label{supp:section:MinImp}
We first state the complete set of sufficient conditions needed to prove Proposition~\ref{proposition:Impute}.
\begin{Assumption}[Proposition~\ref{proposition:Impute}]
In addition to the assumptions in the statement of Proposition~\ref{proposition:Impute}, assume the following hold:
\begin{enumerate}[label=(\alph*)]
\item The elements of $\{x_i,\bm{c}_i\}_{i \in [n]}$ are independent and identically distributed and are independent of $\{e_{gi}\}_{i \in [n]}$.
\item $\E(x_i^4) \leq c$ for some constant $c>0$ and $\E(\abs*{ \bm{c}_{i_k} }^m) \leq c_m$ for all $k \in [K]$, $m>0$, and some constants $c_m>0$.
\item $\alpha_g > 0$.
\end{enumerate}
\end{Assumption}
\begin{remark}
\label{supp:remark:ImputeC}
The moment assumption on $\bm{c}_i$ is the same as that in Assumption~\ref{assumption:y}.
\end{remark}

\begin{proof}[Proof of Proposition~\ref{proposition:Impute}]
We drop the subscript $g$ to simplify notation. Since $\{x_i,\bm{c}_i\}_{i \in [n]}$ are identically distributed and our design matrix includes the intercept, it suffices to assume $\E(x_i)$ and $\E(\bm{c}_i)$ are 0. Let $m = a \min_{i: r_{i}=0}y_{i} = a\mu + a \min_{i: r_{i}=0}(\bm{\ell}^{\top}\bm{c}_i + e_i)$. Since $e_i$ is sub-Gaussian and by the moment assumptions on $\bm{c}_i$, $\abs{m} = O_P(n^{\epsilon})$ for any $\epsilon > 0$. For any $M>0$, the Gaussian assumption on $e_i$ and the moment assumptions on $\bm{c}_i$ also imply $\Prob\{y_i \in (-2M,-M)\} = \delta_{1,M} > 0$. Since $\Prob\{r_i=1 \mid y_i \in (-2M,-M)\} \geq \Prob\{r_i=1 \mid y_i =-2M\}=\delta_{2,M}>0$, this implies the event $\{y_i \in (-2M,-M), r_i=1\}$ occurs infinitely often as $n \to \infty$, meaning $m \to -\infty$ as $n\to \infty$.

We consider the cases $\bm{\ell}=0$ and $x_i$ is independent of $\bm{c}_i$ separately. Suppose first that $\bm{\ell}=0$. Then $\bm{z}_i=(x_i,\bm{c}_i^{\top})^{\top}$ is independent of $y_{i}$ and the elements of $\{\bm{z}_i,y_{i}\}_{i\in[n]}$ are independent and identically distributed. Let $\bm{y}=(y_1,\ldots,y_n)$, $\bm{V}=\V(\bm{z}_i)$, $\bm{R}=\diag(r_1,\ldots,r_n)$, $\bm{Z} = (\bm{z}_1 \cdots \bm{z}_n)^{\top}$, and $\bm{y}_I = \bm{R}\bm{y} + m(I_n - \bm{R})\bm{1}$ be the imputed data. Then for $\bm{e}_1 \in \{0,1\}^{K+1}$ the first standard basis vector,
\begin{align*}
    n^{1/2}\hat{\beta} &= \bm{e}_1^{\top} (\bm{V}^{-1}\hat{\bm{V}})^{-1} \bm{V}^{-1}(n^{-1/2}\bm{Z}^{\top}P_{\bm{1}}^{\perp} \bm{y}_I), \quad \hat{\bm{V}} = n^{-1}\bm{Z}^{\top}P_{\bm{1}}^{\perp}\bm{Z}\\
    n\hat{s}^2 &= \hat{\sigma}^2 \bm{e}_1^{\top} \hat{\bm{V}}^{-1} \bm{e}_1 , \quad \hat{\sigma}^2 = (n-K-2)^{-1}\bm{y}_I^{\top}P_{[\bm{1}, \bm{Z}]}^{\perp} \bm{y}_I.
\end{align*}
Since $\norm*{ \bm{V}^{-1}\hat{\bm{V}})^{-1} - I_{K+1} }_2 = o_P(1)$, we need only show that for $\bm{v} = \bm{Z}\bm{V}^{-1}\bm{e}_1 (\bm{e}_1^{\top}\bm{V}^{-1}\bm{e}_1)^{-1/2}$,
\begin{align*}
    (\hat{\sigma}^2)^{-1/2} (n^{-1/2}\bm{v}^{\top}P_{\bm{1}}^{\perp} \bm{y}_I) \to N(0,1).
\end{align*}
We start by studying $\hat{\sigma}^2$. First, it is easy to see that because $\abs*{m} \to \infty$, $(n-K-2)^{-1}\bm{y}_I^{\top} P_{\bm{1}}^{\perp} \bm{y}_I = m^2\{c+o_P(1)\}$ for some $c>0$. Next, for $\tilde{n}=n-K-2$,
\begin{align*}
    \hat{\sigma}^2 = \tilde{n}^{-1} \bm{y}_I^{\top} P_{\bm{1}}^{\perp} \bm{y}_I - \{\tilde{n}^{-1} \bm{y}_I^{\top}(\bm{Z} - \bm{1}\bar{\bm{z}}^{\top})\} \hat{\bm{V}}^{-1} \{n^{-1}(\bm{Z} - \bm{1}\bar{\bm{z}}^{\top})^{\top} \bm{y}_I\}, \quad \bar{\bm{z}}= n^{-1}\bm{Z}^{\top}\bm{1},
\end{align*}
where because $\bm{Z}$ is independent of $\bm{y}_I$ and $\bar{\bm{z}} = O_P(n^{-1/2})$,
\begin{align*}
    \{\tilde{n}^{-1} \bm{y}_I^{\top}(\bm{Z} - \bm{1}\bar{\bm{z}}^{\top})\} \hat{\bm{V}}^{-1} \{n^{-1}(\bm{Z} - \bm{1}\bar{\bm{z}}^{\top})^{\top} \bm{y}_I\} =& O_P\{ \norm*{ n^{-1} \bm{y}_I^{\top}\bm{Z} }_2^2 + \norm*{ n^{-1} \bm{y}_I^{\top}\bm{1}\bar{\bm{z}}^{\top} }_2 \}\\ =& O_P(\abs*{m}n^{-1/2}) = o_P(1).
\end{align*}
Since the entries of $\bm{v}$ are mean 0, variance 1, independent, and independent of $\bm{y}_I$, $\norm*{ n^{-1/2}\bm{v}^{\top}P_{\bm{1}}^{\top}\bm{y}_I }_2 = O_P(\abs*{m})$, meaning
\begin{align*}
    (\hat{\sigma}^2)^{-1/2} (n^{-1/2}\bm{v}^{\top}P_{\bm{1}}^{\perp} \bm{y}_I) = \tilde{\sigma}^{-1} (n^{-1/2}\bm{v}^{\top}\tilde{\bm{y}}_I) + o_P(1), \quad \tilde{\bm{y}}_I = P_{\bm{1}}^{\perp} \bm{y}_I, \quad \tilde{\sigma}^2 = \tilde{n}^{-1}\bm{y}_I^{\top} P_{\bm{1}}^{\perp} \bm{y}_I.
\end{align*}
We therefore need only show $\tilde{\sigma}^{-1} (n^{-1/2}\bm{v}^{\top}\tilde{\bm{y}}_I) \to N(0,1)$. To prove this, we note that $\E\{ \tilde{\sigma}^{-1} (n^{-1/2}\bm{v}^{\top}\tilde{\bm{y}}_I) \mid \bm{y} \} = 0$, $\V\{ \tilde{\sigma}^{-1} (n^{-1/2}\bm{v}^{\top}\tilde{\bm{y}}_I) \mid \bm{y} \} = 1$, the elements of $\bm{v}$ are independent conditional on $\bm{y}$, and
\begin{align*}
    n^{-2}\sum_{i=1}^n \E(\tilde{\sigma}^{-4} \bm{v}_i^4 \tilde{\bm{y}}_{I_i}^4 \mid \bm{y} ) \leq c \tilde{\sigma}^{-4} n^{-1} (\max_{i \in [n]} y_i^4) = o_P(1)
\end{align*}
for some constant $c>0$. The asymptotic normality of $\tilde{\sigma}^{-1} (n^{-1/2}\bm{v}^{\top}\tilde{\bm{y}}_I)$ follows by the Lindeberg central limit theorem.

We lastly consider the case when $x_i$ is independent of $\bm{c}_i$. Here, $\bm{\ell}$ may not be 0, so $y_i$ and $\bm{c}_i$ may be dependent. Let $\bm{x}=(x_1,\ldots,x_n)^{\top}$. Let $v = \V(x_i)$. We can express $\hat{\beta}/\hat{s}$ as
\begin{align*}
    \hat{\beta}/\hat{s} = (n^{-1}\bm{x}^{\top} P_{[\bm{1},\bm{C}]}^{\perp} \bm{x}v^{-1})^{-1/2} \hat{\sigma}^{-1} (n^{-1/2} \tilde{\bm{x}}^{\top} P_{[\bm{1},\bm{C}]}^{\perp} \bm{y}_I), \quad \tilde{\bm{x}} = v^{-1/2}\bm{x}.
\end{align*}
Since $n^{-1}\bm{x}^{\top} P_{[\bm{1},\bm{C}]}^{\perp} \bm{x}v^{-1} = 1+o_P(1)$, it suffices to show $\hat{\sigma}^{-1} (n^{-1/2} \tilde{\bm{x}}^{\top} P_{[\bm{1},\bm{C}]}^{\perp} \bm{y}_I) \tdist N(0,1)$ to complete the proof. The above proof of the asymptotic normality when $\bm{\ell}=0$ implies this will be true if $\hat{\sigma}^2 = m^2\{c+o_P(1)\}$ for some constant $c>0$ and if $\hat{\sigma}^2 = \tilde{n}^{-1} \bm{y}_I^{\top} P_{[\bm{1},\bm{C}]}^{\perp} \bm{y}_I + o_P(1)$. For the former, we see that for $\tilde{\bm{z}}_i=\{ 1 \oplus v^{-1/2} \oplus \V(\bm{c}_i)^{-1/2}\} (1,\bm{x}_i,\bm{c}_i^{\top})^{\top}$,
\begin{align*}
    \hat{\sigma}^2 = O_P(m) + m^2 [ \E(1-r_1) - \E\{(1-r_1)\tilde{\bm{z}}_1\}^{\top} \E\{(1-r_1)\tilde{\bm{z}}_1\} ].
\end{align*}
For any non-random unit vector $\bm{u} \in \mathbb{R}^{K+2}$, Holder's inequality implies 
\begin{align*}
    \bm{u}^{\top}[\E\{(1-r_1)\tilde{\bm{z}}_1\} \E\{(1-r_1)\tilde{\bm{z}}_1\}^{\top}] \bm{u}&=[\E\{(1-r_1)(\tilde{\bm{z}}_1^{\top}\bm{u})\}]^2 \leq \E(1-r_1) \E\{(\tilde{\bm{z}}_1^{\top}\bm{u})^2\}\\& = \bm{u}^{\top}\{\E(1-r_1) \E(\tilde{\bm{z}}_1 \tilde{\bm{z}}_1^{\top})\} \bm{u},
\end{align*}
where the inequality holds with equality if and only if $(1-r_1) \propto \bm{u}^{\top}\tilde{\bm{z}}_1$ a.s. Since this does not hold for any non-random $\bm{u}$, we must have
\begin{align*}
    \E\{(1-r_1)\tilde{\bm{z}}_1\} \E\{(1-r_1)\tilde{\bm{z}}_1\}^{\top} &\prec \E(1-r_1) \underbrace{\E(\tilde{\bm{z}}_1 \tilde{\bm{z}}_1^{\top})}_{=I_{K+2}}\\
    \Rightarrow \E\{(1-r_1)\tilde{\bm{z}}_1\}^{\top} \E\{(1-r_1)\tilde{\bm{z}}_1\} &= \norm*{ \E\{(1-r_1)\tilde{\bm{z}}_1\} \E\{(1-r_1)\tilde{\bm{z}}_1\}^{\top} }_2 < \E(1-r_1) \norm*{ \E(\tilde{\bm{z}}_1 \tilde{\bm{z}}_1^{\top}) }_2\\& = \E(1-r_1),
\end{align*}
which implies $\hat{\sigma}^2 = m^2\{c+o_P(1)\}$ for some constant $c>0$. Lastly,
\begin{align*}
    \hat{\sigma}^2 =& \tilde{n}^{-1}\bm{y}_I^{\top} P_{[\bm{1},\bm{C}]}^{\perp} \bm{y}_I - \{\tilde{n}^{-1}\bm{y}_I^{\top} \bm{x} - \tilde{n}^{-1}\bm{y}_I^{\top}[\bm{1},\bm{C}]\hat{\bm{M}}^{-1} (n^{-1}[\bm{1},\bm{C}]^{\top} \bm{x}) \}\hat{v}^{-1} \\
    &\times \{n^{-1}\bm{y}_I^{\top} \bm{x} - n^{-1}\bm{y}_I^{\top}[\bm{1},\bm{C}]\hat{\bm{M}}^{-1} (n^{-1}[\bm{1},\bm{C}]^{\top} \bm{x}) \}\\
    \hat{\bm{M}} =& n^{-1}[\bm{1},\bm{C}]^{\top} [\bm{1},\bm{C}], \quad \hat{v} = n^{-1}\bm{x}^{\top} P_{[\bm{1},\bm{C}]}^{\perp} \bm{x}.
\end{align*}
Since $\bm{x}$ is mean 0 and independent of $\{\bm{y},\bm{C}\}$, it is easy to see that 
\begin{align*}
    n^{-1}\bm{y}_I^{\top} \bm{x} - n^{-1}\bm{y}_I^{\top}[\bm{1},\bm{C}]\hat{\bm{M}}^{-1} (n^{-1}[\bm{1},\bm{C}]^{\top} \bm{x}) = o_P(1),
\end{align*}
which completes the proof.
\end{proof}

\section{Theoretical guarantees for factor analysis and differential abundance}
\label{supp:section:FATheory}
Section~\ref{supp:section:FATheory} proves theoretical statements from Sections~\ref{subsection:theory:LF} and \ref{subsection:theory:DA} in the main text. We prove Theorem~\ref{theorem:mtGWAS} from Section~\ref{subsection:theory:mtGWAS} in Section~\ref{supp:section:TheorymtGWAS}.

\subsection{Problem statement and preliminaries}
\label{supp:subsection:FASetup}
We consider the model $\bm{y}_g = \bm{X}\bm{\beta}_g + \bm{C}\bm{\ell}_g + \bm{e}_g$, where $\bm{e}_g \sim (0, \sigma_g^2 I_n)$ is sub-Gaussian with independent entries. We also define the diagonal matrix of weights $\bm{W}_g = \diag(w_{g1},\ldots,w_{gn})$ to be $r_{gi}\{\pi_{g}(y_{gi})\}^{-1}$ for $\pi_{g}(y_{gi}) = \Psi\{ \alpha_g(y_{gi} - \delta_g) \}$. Note that $\E(\bm{W}_{g} \mid \bm{y}_g,\bm{X},\bm{C}) = I_n$. Let $\bm{L} = [\bm{\ell}_1 \cdots \bm{\ell}_p]^{\top} \in \mathbb{R}^{p \times K}$, $\lambda = np^{-1}\Tr(\bm{L}^{\top} \bm{L})$, and $\bm{B} = [\bm{\beta}_1 \cdots \bm{\beta}_p]^{\top} \in \mathbb{R}^{p \times d}$. If $\hat{w}_{gi}=w_{gi}$, the optimization problem in \eqref{equation:OptC} is equivalent to
\begin{align*}
    (P_{\bm{X}}^{\perp}\hat{\bm{C}},\hat{\bm{L}},\hat{\bm{B}}) \in \argmin_{\substack{\bm{C}_{\perp},\bm{L},\bm{B}\\ \bm{X}^{\top}\bm{C}_{\perp}=\bm{0}\\ \bm{C}_{\perp}^{\top}\bm{C}_{\perp} = I_K}} \frac{1}{2\lambda p}\sum_{g=1}^p (\bm{y}_g - \bm{C}_{\perp}\bm{\ell}_g - \bm{X}\bm{\beta}_g)^{\top}\bm{W}_g (\bm{y}_g - \bm{C}_{\perp}\bm{\ell}_g - \bm{X}\bm{\beta}_g).
\end{align*}
Define $\bm{P}_g^{\perp} = \bm{W}_g - \bm{W}_g \bm{X}( \bm{X}^{\top}\bm{W}_g\bm{X} )^{-1}\bm{X}^{\top}\bm{W}_g$. Solving for $\bm{B}$ and using the fact that $\hat{\bm{\ell}}_g = (\bm{C}_{\perp}^{\top}\bm{P}_g^{\perp} \bm{C}_{\perp} )^{-1}\bm{C}_{\perp}^{\top} \bm{P}_g^{\perp} \bm{y}_g$, the profile likelihood for $\bm{C}_{\perp}$ can be expressed as
\begin{align}
\label{supp:equation:fMax}
    P_{\bm{X}}^{\perp}\hat{\bm{C}} \in \argmax_{\substack{\bm{U} \in \mathbb{R}^{n \times K}\\ \bm{X}^{\top}\bm{U}=\bm{0}\\\bm{U}^{\top}\bm{U} = I_K}} f(\bm{U}), \quad f(\bm{U}) = \frac{1}{2\lambda p} \sum_{g=1}^p \Tr\{ (\bm{U}^{\top}\bm{P}_g^{\perp} \bm{U} )^{-1} \bm{U}^{\top} \bm{P}_g^{\perp} \bm{y}_g \bm{y}_g^{\top}  \bm{P}_g^{\perp}\bm{U} \}
\end{align}
Expanding $\bm{y}_g \bm{y}_g^{\top}$, the objective function can be expressed as
\begin{align}
\label{supp:equation:Cobj}
\begin{aligned}
    f(\bm{U}) =&\frac{1}{2\lambda p}\sum_{g=1}^p \Tr\{ (\bm{U}^{\top}\bm{P}_g^{\perp} \bm{U} )^{-1} \bm{U}^{\top} \bm{P}_g^{\perp} \tilde{\bm{C}}\tilde{\bm{\ell}}_g \tilde{\bm{\ell}}_g^{\top} \tilde{\bm{C}}^{\top}  \bm{P}_g^{\perp}\bm{U} \}\\
    &+ \frac{1}{\lambda p}\Tr\{ (\bm{U}^{\top}\bm{P}_g^{\perp} \bm{U} )^{-1} \bm{U}^{\top} \bm{P}_g^{\perp} \tilde{\bm{C}}\tilde{\bm{\ell}}_g \bm{e}_g^{\top} \bm{P}_g^{\perp}\bm{U} \}\\
    &+ \frac{1}{2 \lambda p}\Tr\{ (\bm{U}^{\top}\bm{P}_g^{\perp} \bm{U} )^{-1} \bm{U}^{\top} \bm{P}_g^{\perp} \bm{e}_g \bm{e}_g^{\top} \bm{P}_g^{\perp}\bm{U} \}\\
    \tilde{\bm{C}} =& P_{X}^{\perp}\bm{C}(\bm{C}^{\top}P_{X}^{\perp}\bm{C})^{-1/2}, \quad \tilde{\bm{\ell}}_g = (\bm{C}^{\top}P_{X}^{\perp}\bm{C})^{1/2}\bm{\ell}_g.
\end{aligned}
\end{align}
We use this expression to prove the consistency of $P_{\bm{X}}^{\perp}\hat{\bm{C}}$ in Section~\ref{supp:subsection:FAConsistency} and derive its properties and rate of convergence in Section~\ref{supp:subsection:FARate}.

\subsection{Assumptions}
\label{supp:subsection:FAAssumptions}
We first re-state the assumptions on $y_{gi}$ and $\Psi$ from Section~\ref{subsection:theory:Assumptions} with a change in the scaling of the eigenvalues $\lambda_1,\ldots,\lambda_K$ defined in Assumption~\ref{assumption:y}. Note that the change is without loss of generality.

\begin{Assumption}
\label{supp:assumptions:FA}
For $g \in [p]$, $i \in [n]$, and $s \in [S]$, let $y_{gi}=\bm{\beta}_g^{\top}\bm{x}_i + \bm{\ell}_g^{\top}\bm{c}_i + e_{gi}$ and define $G_{si} \in \mathbb{R}$ to be a random variable. Then the following hold for constant $a>0$ and $\epsilon \in (0,1/2\wedge a)$:
\begin{enumerate}[label=(\alph*)]
\item $\bm{X}=[\bm{x}_1 \cdots \bm{x}_n]^{\top}$ and $\{\bm{\beta}_g\}_{g \in [p]}$ are non-random and satisfy $\abs*{ \bm{X}_{ij} },\norm*{\bm{\beta}_g}_2 \leq a$, $\bm{1}_n \in \im(\bm{X})$, and $n^{-1}\bm{X}^{\top}\bm{X} \succeq \epsilon I_d$.\label{supp:assumption:X}
\item The matrix $\bm{G}=[G_{si}] \in \mathbb{R}^{S \times n}$ is mean 0, has independent and uniformly bounded entries, and identically distributed columns.\label{supp:assumption:G}
\item The eigenvalues $\lambda_1,\ldots,\lambda_K > 0$ of $np^{-1}\sum_{g=1}^p \bm{\ell}_g \bm{\ell}_g^{\top}$ and $\bm{\ell}_g$ satisfy $n^{1/2+\epsilon}\allowbreak \lesssim \allowbreak \lambda_K \allowbreak \leq \allowbreak \cdots \leq \allowbreak \lambda_1 \allowbreak \lesssim n$, $\lambda_1/\lambda_K \leq a$, and $\norm{ \bm{\ell}_g }_2 \leq a (\lambda_1/n)^{1/2}$.\label{supp:assumption:lambda}
\item $\bm{C} = [\bm{c}_1 \cdots \bm{c}_n]^{\top} = \bm{G}^{\top} \bm{\gamma}^{(c)} + \bm{\Delta}^{(c)} \in \mathbb{R}^{n \times K}$ for $\bm{\gamma}^{(c)} \in \mathbb{R}^{S \times K}$ and $\bm{\Delta}^{(c)}_{i*} \in \mathbb{R}^K$ such that:
\begin{enumerate}[label=(\roman*)]
    \item $\bm{\gamma}^{(c)}$ is non-random and $\norm*{ \bm{\gamma}^{(c)} }_1 \leq a$, $\norm*{ \bm{\gamma}^{(c)} }_{\infty} = o(n^{-1/4})$, $\sum_{s=1}^S 1\{ \bm{\gamma}^{(c)}_{s*} \neq 0 \} \leq a p^{1/2}$.
    \item The rows of $\bm{\Delta}^{(c)} - \E\{ \bm{\Delta}^{(c)} \}$ are independent and identically distributed, $\V\{ \bm{\Delta}^{(c)}_{i *} \} \succeq \epsilon I_K$, and $\E\{\abs*{ \bm{\Delta}^{(c)}_{ik} }^m\} \leq a_m$ for all $i \in [n]$, $k \in [K]$, $m \geq 1$ and constant $a_m >0$ that may depend on $m$.\label{supp:assumption:c:Delta}
\end{enumerate}\label{supp:assumption:c}
\item $\bm{E} = [e_{gi}] = \{ \bm{\gamma}^{(e)} \}^{\top}\bm{G} + \bm{\Delta}^{(e)} \in \mathbb{R}^{p \times n}$, where $\bm{\gamma}^{(e)} \in \mathbb{R}^{S \times p}$, $\bm{\Delta}^{(e)} \in \mathbb{R}^{p \times n}$ satisfy:
\begin{enumerate}[label=(\roman*)]
    \item $\sup_{s \in [S];g \in [p]}\abs*{ \bm{\gamma}^{(e)}_{sg} }\leq a n^{-1/4}$ and $\sup_{g \in [p]}\sum_{s=1}^S 1\{ \bm{\gamma}^{(e)}_{sg} \neq 0 \} \leq a$.
    \item The columns of $\bm{\gamma}^{(e)}$ can be partitioned into disjoint sets containing $\leq a$ metabolites, where $\bm{\gamma}^{(e)}_{sg} \bm{\gamma}^{(e)}_{sh} = 0$ if columns $g$ and $h$ lie in different sets.
    \item $\bm{\Delta}^{(e)} \sim MN_{p \times n}\{0,\diag(\sigma_1^2,\ldots,\sigma_p^2),I_n\}$, where $\sigma_1,\ldots,\sigma_p^2 \in [\epsilon,a]$.
\end{enumerate}
\item $\bm{G}$, $\bm{\Delta}^{(c)}$, and $\bm{\Delta}^{(e)}$ are independent, and $p,n,S$ satisfy $p \in [\epsilon n, a n]$.
\item In differential abundance applications, $\bm{X}$ can be written as $\bm{X}=(\bm{X}_1, \bm{X}_2)$ for $\bm{X}_1 \in \mathbb{R}^{n \times d_1}$ and $\bm{X}_2 \in \mathbb{R}^{n \times d_2}$, where $\bm{X}_1$ are the $d_1$ covariates of interests and $p^{-1}\sum_{g=1}^p 1\{ \bm{\beta}_{g_j} \neq 0 \} = o(\lambda_1^{1/2}n^{-1})$ for all $j \in [d_1]$.\label{supp:assumption:SparseBeta}
\item Assumption~\ref{assumption:Miss} from the main text holds.\label{supp:assumption:Psi}
\end{enumerate}
\end{Assumption}
\noindent Assumption~\ref{supp:assumption:X} contains all the regularity conditions on the design matrix $\bm{X}$ mentioned in Section~\ref{subsection:theory:Assumptions}. The assumption that $\bm{1}_n \in \im(\bm{X})$ makes the assumption that $\E(\bm{G})=0$ in \ref{supp:assumption:G} without loss of generality. Note $\E\{ \bm{\Delta}^{(c)} \}$ in \ref{supp:assumption:c} may depend on $\bm{X}$. The assumptions in \ref{supp:assumption:c} are more general than (b) from Assumption~\ref{assumption:y} in the main text, since the latter assumes $\E(\bm{c}_i) - \{\bm{\gamma}^{(c)}\}^{\top}\bm{G}_{\bigcdot i}$ is continuous in $\bm{x}_i$, whereas the former only assumes $\norm*{ \E(\bm{c}_i) - \{\bm{\gamma}^{(c)}\}^{\top}\bm{G}_{\bigcdot i} }_2$ is bounded from above. The eigenvalues $\lambda_1,\ldots,\lambda_K$ in \ref{supp:assumption:lambda} have been scaled by a factor of $n$ to make notation in the below theoretical statements simpler, and to be consistent with \citet{BCconf,CorrConf,FALCO}. Assumption \ref{supp:assumption:SparseBeta} gives the sparsity assumption utilized in the statements of Theorem~\ref{theorem:Omega} and Theorem~\ref{theorem:betag} in the main text. Note this is only needed in differential abundance applications, and is not needed to prove Theorem~\ref{theorem:PxC}, Propoposition~\ref{proposition:Cknown}, or Theorem~\ref{theorem:mtGWAS}. We next prove two useful lemmas about $\bm{C}$ and its relationship with $\bm{E}$ that we will use in the theoretical results that follow.

\begin{lemma}
\label{supp:lemma:Cassumption}
Suppose Assumption~\ref{supp:assumptions:FA} holds. Then $\lim_{n \to \infty}\Prob(n^{-1}\bm{C}^{\top}P_{\bm{X}}^{\perp}\bm{C} \succeq \epsilon I_K)=1$ and $\E( \abs*{\bm{C}_{ik}}^m ) \leq \tilde{a}_m$ for all $m\geq 1$ and some constant $\tilde{a}_m>0$ that may depend on $m$.
\end{lemma}
\begin{proof}
Define $\bm{\mu} = \E\{\bm{\Delta}^{(c)}\}$. Then
\begin{align*}
    \bm{C}^{\top}P_{\bm{X}}^{\perp}\bm{C} \succeq & \{ \bm{\Delta}^{(c)} - \bm{\mu} \}^{\top} P_{\bm{X}}^{\perp} \{ \bm{\Delta}^{(c)} - \bm{\mu} \} + \bm{\mu}^{\top} P_{\bm{X}}^{\perp} \{ \bm{\Delta}^{(c)} - \bm{\mu} \} + [ \bm{\mu}^{\top} P_{\bm{X}}^{\perp} \{ \bm{\Delta}^{(c)} - \bm{\mu} \} ]^{\top}\\
    & + \{ \bm{\Delta}^{(c)} \}^{\top}P_{\bm{X}}^{\perp} \bm{G}^{\top} \{ \bm{\gamma}^{(c)} \} + [ \{ \bm{\Delta}^{(c)} \}^{\top}P_{\bm{X}}^{\perp} \bm{G}^{\top} \{ \bm{\gamma}^{(c)} \} ]^{\top}.
\end{align*}
First, $\bm{G}_i^{\top} \{ \bm{\gamma}^{(c)} \} \leq c$ for some constant $c>0$ by \ref{supp:assumption:G} and \ref{supp:assumption:c} in Assumption~\ref{supp:assumptions:FA}, meaning $\norm*{ P_{\bm{X}}^{\perp} \bm{G}^{\top} \{ \bm{\gamma}^{(c)} \} }_2^2 \leq n c^2$. This means $\norm*{ \{ \bm{\Delta}^{(c)} \}^{\top}P_{\bm{X}}^{\perp} \bm{G}^{\top} \{ \bm{\gamma}^{(c)} \} }_2 = O_P(n^{1/2})$. Since $\norm*{\bm{\mu}}_2^2 = O(n)$ by \ref{supp:assumption:c} in Assumption~\ref{supp:assumptions:FA}, we also have $\norm*{ \bm{\mu}^{\top} P_{\bm{X}}^{\perp} \{ \bm{\Delta}^{(c)} - \bm{\mu} \} }_2 = O_P(n^{1/2})$. Since the rows of $\bm{\Delta}^{(c)} - \bm{\mu}$ are independent and identically distributed, $n^{-1}\{ \bm{\Delta}^{(c)} - \bm{\mu} \}^{\top} P_{\bm{X}}^{\perp} \{ \bm{\Delta}^{(c)} - \bm{\mu} \} = \{ \bm{\Delta}^{(c)} - \bm{\mu} \}^{\top} \{ \bm{\Delta}^{(c)} - \bm{\mu} \} \succeq \{c+o_P(1)\} I_K$ for some constant $c>0$, which proves $\lim_{n \to \infty}\Prob(n^{-1}\bm{C}^{\top}P_{\bm{X}}^{\perp}\bm{C} \succeq \epsilon I_K)=1$.

Let $\norm*{z}_m = \{\E(\abs*{x}^m)\}^{1/m}$ for any random variable $z$. Then for some constant $c>0$ and $a_m$ as defined in Assumption~\ref{supp:assumptions:FA}, $\norm*{ \bm{C}_{ik} }_m \leq \norm*{ \bm{G}_{*i}^{\top} \bm{\gamma}_{*k}^{(c)} }_m + \norm*{ \bm{\Delta}_{ik}^{(c)} }_m \leq c + a_m$.
\end{proof}

\begin{lemma}
\label{supp:lemma:Cte}
Suppose Assumption~\ref{supp:assumptions:FA} holds, let $\epsilon>0$ be any constant and $m>0$ any integer, and let $h_i: \mathbb{R}^K \to \mathbb{R}$, $i \in [n]$, be uniformly bounded functions with uniformly bounded gradients. Then for $\bm{L} = [\bm{\ell}_1^{\top}\cdots \bm{\ell}_p^{\top}]^{\top}$,
\begin{subequations}
\label{supp:equation:CtE}
\begin{align}
\label{supp:equation:CtE:1}
    &\E\{ ( n^{-1/2}\bm{C}_{*k}^{\top}\bm{E}_{g*} )^{2m} \} \leq c_m, \quad k \in [K]; g \in [p]\\
\label{supp:equation:CtE:1b}
    &\textstyle\E[ \{ n^{-1/2}\sum_{i=1^n} h_i( \bm{C}_{i*}) \bm{E}_{gi}\}^{2m} ] \leq c_m, \quad g \in [p]\\
\label{supp:equation:CtE:2}
    &\norm*{ (\lambda_K p)^{-1/2}\bm{C}^{\top} \bm{E}^{\top}\bm{L} }_2 = O_P(1)\\
\label{supp:equation:CtE:2b}
    &\norm*{ (\lambda_K p)^{-1/2}\sum_{g=1}^p \sum_{i=1}^n h_i(\bm{C}_{i*}) \bm{E}_{gi} \bm{\ell}_g }_2 = O_P(1)
\end{align}
\end{subequations}
for some constant $c_m>0$ that may depend on $m$.
\end{lemma}

\begin{proof}
Under Assumption~\ref{supp:assumptions:FA},
\begin{align*}
    \bm{C}_{*k}^{\top}\bm{E}_{g*} = \{\bm{\gamma}^{(c)}_{*k}\}^{\top} \bm{G} \bm{G}^{\top} \bm{\gamma}^{(e)}_{*g} + \{\bm{\Delta}^{(c)}_{*k}\}^{\top}\bm{G}^{\top} \bm{\gamma}^{(e)}_{*g}  + \bm{C}_{*k}^{\top} \bm{\Delta}^{(e)}_{*g}
\end{align*}
Let $\epsilon>0$ be any constant. Since $\bm{C}$ is independent of $\bm{\Delta}^{(e)}_{*g}$, $\E\{ ( \bm{C}_{*k}^{\top} \bm{\Delta}^{(e)}_{*g} )^{2m} \} \leq c_{1,m}$ for some constant $c_{1,m}>0$ by Corollary~\ref{supp:corollary:MaximalIne}. Further, since at most finitely many entries of $\bm{\gamma}^{(e)}_{*g} \in \mathbb{R}^S$ are non-zero, $\bm{G}$ is independent of $\bm{\Delta}^{(c)}$, and the rows of $\bm{G}$ are mean 0, independent and sub-Gaussian, $\E[ \{ n^{-1/2}\{\bm{\Delta}_{*k}^{(c)}\}^{\top}\bm{G}^{\top} \bm{\gamma}^{(e)}_{*g} \}^{2m} ]\leq c_{2,m}$ for some constant $c_{2,m}>0$ by Corollary~\ref{supp:corollary:MaximalIne}. Let $\mathcal{I}_g = \{s \in [S]: \bm{\gamma}^{(e)}_{sg} \neq 0\}$ and $\mathcal{C} = \{s \in [S]: \bm{\gamma}^{(c)}_{sk} \neq 0\}$. Then
\begin{align*}
    &n^{-1/2} \{\bm{\gamma}^{(c)}_{*k}\}^{\top} \bm{G} \bm{G}^{\top} \bm{\gamma}^{(e)}_{*g} = \sum_{s \in \mathcal{I}_g} \tilde{\bm{\gamma}}^{(c)}_{sk} \tilde{\bm{\gamma}}^{(e)}_{sg} (n^{-1}\bm{G}_{s*}^{\top} \bm{G}_{s*}) + n^{-1/2} \sum_{r \in \mathcal{I}_g^{c} \cap \mathcal{C} } \bm{\gamma}^{(c)}_{rk} \bm{G}_{r*}^{\top}\sum_{s \in \mathcal{I}_g} \bm{\gamma}_{sg}^{(e)}\bm{G}_{s*}\\
    &\tilde{\bm{\gamma}}^{(c)} = n^{-1/4} \bm{\gamma}^{(c)}, \quad \tilde{\bm{\gamma}}^{(e)} = n^{-1/4} \bm{\gamma}^{(e)}.
\end{align*}
First, since $\E\{ ( n^{-1}\bm{G}_{s*}^{\top} \bm{G}_{s*} )^{2m} \} \leq c_{3,m}$ for all $s \in [S]$ and some constant $c_{3,m}>0$ that may depend on $m$, 
\begin{align*}
    \E[ \textstyle \{ \sum_{s \in \mathcal{I}_g} \tilde{\bm{\gamma}}^{(c)}_{sk} \tilde{\bm{\gamma}}^{(e)}_{sg} (n^{-1}\bm{G}_{s*}^{\top} \bm{G}_{s*}) \}^{2m} ] \leq a c_{3,m} 
\end{align*}
for some constant $a>0$ because $\abs*{ \tilde{\bm{\gamma}}^{(c)}_{sk} \tilde{\bm{\gamma}}^{(e)}_{sg} }$ and $\abs*{\mathcal{I}_g}$ are uniformly bounded from above. Let $\bm{Z}_g = \sum_{s \in \mathcal{I}_g} \bm{\gamma}^{(e)}_{sg} \bm{G}_{s*} \in \mathbb{R}^n$. Then $\bm{Z}_g$ is mean 0, has independent entries that are bounded above by $an^{-1/4}$ for some constant $a>0$, and is independent of $\bm{W}_g = \sum_{r \in \mathcal{I}_g^{c} \cap \mathcal{C} } \bm{\gamma}^{(c)}_{rk} \bm{G}_{r*} \in \mathbb{R}^n$. Here, $\bm{W}_g$ is also mean 0, has independent entries. Further, since $\{\bm{G}_{r*}\}_{r \in \mathcal{I}_g^{c} \cap \mathcal{C} }$ are independent and sub-Gaussian with uniformly bounded sub-Gaussian norm and $\sum_{r \in [S]}\{\bm{\gamma}^{(c)}_{rk}\}^2 \leq \sum_{r \in [S]}\abs*{ \bm{\gamma}^{(c)}_{rk} } \leq a$ for some constant $a>0$ by Assumption~\ref{supp:assumptions:FA}, $\bm{W}_g$ is also sub-Gaussian with uniformly bounded sub-Gaussian norm over $g \in [p]$. All this implies
\begin{align*}
    \E[ \textstyle\{ n^{-1/2} \sum_{r \in \mathcal{I}_g^{c} \cap \mathcal{C} } \bm{\gamma}^{(c)}_{rk} \bm{G}_{r*}^{\top}\sum_{s \in \mathcal{I}_g} \bm{\gamma}_{sg}^{(e)}\bm{G}_{s*} \}^{2m} ] = \E\{(n^{-1/2}\bm{W}_g^{\top} \bm{Z}_g)^{2m}\} \leq c_{4,m}, \quad g \in [p]
\end{align*}
for some constant $c_{4,m}>0$ that may depend on $m$. This completes the proof of \eqref{supp:equation:CtE:1}.

Since $h$ is bounded from above, the above work implies we need only show that $\E[ \{n^{-1/2}\allowbreak\sum_{i=1}^n \allowbreak h_i( \bm{C}_{i*})\bm{G}_{*i}^{\top}\bm{\gamma}_{*g}^{(e)}\}^{2m} ]\allowbreak \leq \allowbreak c_m$  to prove \eqref{supp:equation:CtE:1b}. For simplicity, we assume that $\mathcal{I}_g = \{s_g\}$, and note that the extension to general $\mathcal{I}_g$ under the Assumption~\ref{supp:assumptions:FA} is trivial. Then for $\bm{R}^{(g)}_{i*} = \bm{C}_{i*} - \bm{\gamma}_{s_g *}^{(c)}\bm{G}_{s_g i}$,
\begin{align*}
n^{-1/2}\sum_{i=1}^n h_i( \bm{C}_{i*})\bm{G}_{*i}^{\top}\bm{\gamma}_{*g}^{(e)}= & n^{-1/2}\bm{\gamma}_{s_g g}^{(e)}\sum_{i=1}^n h_i\{\bm{\gamma}_{s_g *}^{(c)}\bm{G}_{s_g i} + \bm{R}^{(g)}_{i*}\} \bm{G}_{s_gi}\\=& n^{-1/2}\bm{\gamma}_{s_g g}^{(e)}\sum_{i=1}^n h_i\{\bm{R}^{(g)}_{i*}\} \bm{G}_{s_gi}+ O(1),
\end{align*}
where the second equality follows because $\{h_i\}_{i \in [n]}$ have uniformly bounded gradients, $\bm{G}_{s_gi}$ is bounded, and $\abs*{\bm{\gamma}_{s_g g}^{(e)}} \norm*{\bm{\gamma}_{s_g *}^{(c)}}_2 = o(n^{-1/2})$. An application of Lemma~\ref{supp:lemma:Powera} then proves the result, which proves \eqref{supp:equation:CtE:1b}.

For \eqref{supp:equation:CtE:2},
\begin{align*}
    (\lambda_K p)^{-1/2}\bm{C}^{\top}\bm{E}^{\top}\bm{L} =& (\lambda_K p)^{-1/2}\bm{C}^{\top}\{\bm{\Delta}^{(e)}\}^{\top}\bm{L} + (\lambda_K p)^{-1/2}\{\bm{\Delta}^{(c)}\}^{\top}\bm{G}^{\top}\bm{\gamma}^{(e)}\bm{L}\\
    &+ (\lambda_K p)^{-1/2}\{\bm{\gamma}^{(c)}\}^{\top}\bm{G}\bm{G}^{\top}\bm{\gamma}^{(e)}\bm{L}.
\end{align*}
It is straightforward to show that
\begin{align*}
    \norm*{ (\lambda_K p)^{-1/2}\bm{C}^{\top}\{\bm{\Delta}^{(e)}\}^{\top}\bm{L} }_2, \, \norm*{ (\lambda_K p)^{-1/2}\{\bm{\Delta}^{(c)}\}^{\top}\bm{G}^{\top}\bm{\gamma}^{(e)}\bm{L} }_2 = O_P(1).
\end{align*}
For the remaining term, let $\tilde{\bm{L}} = n^{3/4}\lambda_K^{-1/2}\bm{\gamma}^{(e)}\bm{L} \in \mathbb{R}^{S \times K}$. Then at most $O(p)$ rows of $\tilde{\bm{L}}$ are non-zero and $\norm*{ \tilde{\bm{L}}_{s*} }_2 \leq c $ for some constant $c>0$ by Assumption~\ref{supp:assumptions:FA}. Then for $\mathcal{S} = \{s \in [S]: \bm{\gamma}_{s*},\tilde{\bm{L}}_{s*}\neq 0 \}$ and $\mathcal{R} = \{(r,s) \in [S] \times [S]: r \neq s, \bm{\gamma}_{s*}\odot\tilde{\bm{L}}_{s*}\neq 0\}$ (where $\odot$ is the Hadamard product),
\begin{align*}
    (\lambda_K p)^{-1/2}\{\bm{\gamma}^{(c)}\}^{\top}\bm{G}\bm{G}^{\top}\bm{\gamma}^{(e)}\bm{L} =& p^{-1/2} \sum_{s \in \mathcal{S}} \tilde{\bm{\gamma}}^{(c)}_{s*} \tilde{\bm{L}}_{s*}^{\top} (n^{-1}\bm{G}_{s*}^{\top} \bm{G}_{s*})\\
    &+ n^{-1}\sum_{(r,s) \in \mathcal{R}} \bm{\gamma}^{(c)}_{s*} \tilde{\bm{L}}_{s*}^{\top} (p^{-1/2}\bm{G}_{r *}^{\top}\bm{G}_{s *}).
\end{align*}
Since $\norm*{\tilde{\bm{\gamma}}^{(c)}_{s*}}_2, \norm*{ \tilde{\bm{L}}_{s*} }_2 \leq c$ and $\E(n^{-1}\bm{G}_{s*}^{\top} \bm{G}_{s*})\leq c$ for some constant $c>0$ and all $s \in [S]$,
\begin{align*}
    \E\{ \norm*{ p^{-1/2} \sum_{s \in \mathcal{S}} \tilde{\bm{\gamma}}^{(c)}_{s*} \tilde{\bm{L}}_{s*}^{\top} (n^{-1}\bm{G}_{s*}^{\top} \bm{G}_{s*}) }_2 \} \leq c^3p^{-1/2}\abs*{\mathcal{S}} = O(1)
\end{align*}
since $\abs*{\mathcal{S}} = O(p^{1/2})$ by Assumption~\ref{supp:assumptions:FA}. Lastly, since $\E( \bm{G}_{r *}^{\top}\bm{G}_{s *} \bm{G}_{r' *}^{\top}\bm{G}_{s' *} )=0$ for all $(r,s) \neq (r',s')$,
\begin{align*}
    \E[\{ n^{-1}\sum_{(r,s) \in \mathcal{R}} \bm{\gamma}^{(c)}_{sk_1} \tilde{\bm{L}}_{sk_2}^{\top} (p^{-1/2}\bm{G}_{r *}^{\top}\bm{G}_{s *}) \}^2] = O( n^{-2}\abs*{\mathcal{R}} ) = O(n^{-2}p^{3/2}) = O(1),
\end{align*}
which completes the proof of \eqref{supp:equation:CtE:2}.

Lastly, to prove \eqref{supp:equation:CtE:2b}, we again assume for simplicity that $\mathcal{I}_g=\{s_g\}$, but note that extending it to general $\mathcal{I}_g$ is trivial under Assumption~\ref{supp:assumptions:FA}. Since $h_i$ is bounded, the proof of \eqref{supp:equation:CtE:2} implies we need only consider the behavior of $\sum_{g=1}^p \bm{\gamma}_{s_g g}^{(e)}\sum_{i=1}^n h_i( \bm{C}_{i*}) \bm{G}_{s_g i} \bm{\ell}_g$. Fix a $k \in [K]$ and let $a_{gk}=n^{1/2}\lambda_K^{-1/2}\bm{\ell}_{g_k}$, where $\abs*{a_{gi}}$ is uniformly bounded from above by Assumption~\ref{supp:assumptions:FA}. Let $\mathcal{C} = \{ g \in [p]: \bm{\gamma}_{s_g*}^{(c)} \neq 0 \}$. Then $\abs*{\mathcal{C}} = O(p^{1/2})$ by Assumption~\ref{supp:assumptions:FA} and for any $k\in[K]$,
\begin{align*}
&(\lambda_K p)^{-1/2} \sum_{g=1}^p \bm{\gamma}_{s_g g}^{(e)}\sum_{i=1}^n h_i( \bm{C}_{i*}) \bm{G}_{s_g i} \bm{\ell}_{g_k} = (\lambda_K p)^{-1/2} \sum_{g\in \mathcal{C}} \bm{\gamma}_{s_g g}^{(e)}\sum_{i=1}^n h_i(\bm{C}_{i*}) \bm{G}_{s_g i} \bm{\ell}_{g_k} \\&+ (\lambda_K p)^{-1/2} \sum_{g\in \mathcal{C}^c} \bm{\gamma}_{s_g g}^{(e)}\sum_{i=1}^n h_i(\bm{C}_{i*}) \bm{G}_{s_g i} \bm{\ell}_{g_k}.
\end{align*}
Since the $\bm{G}_{s_g}$ is independent of $\bm{C}$ for $g \in \mathcal{C}^c$, the second term after the equality in the above expression is $O_P(1)$ because $h_i$ is uniformly bounded from above. For the first term, fix a $g \in \mathcal{C}$ and let $\bm{C}_{i*} = \bm{\gamma}_{s_g *}^{(e)}\bm{G}_{s_g i} + \bm{R}_{i*}^{(g)}$, where $\bm{R}_{i*}^{(g)}$ is independent of $\bm{G}_{s_g i}$. Then for $\tilde{\ell}_{gk} = n^{1/2}\lambda_K^{-1/2}\bm{\ell}_{gk}$ (which is uniformly bounded by Assumption~\ref{supp:assumptions:FA}),
\begin{align*}
    &(\lambda_K p)^{-1/2} \sum_{g\in \mathcal{C}} \bm{\gamma}_{s_g g}^{(e)}\sum_{i=1}^n h_i( \bm{C}_{i*}) \bm{G}_{s_g i} \bm{\ell}_{g_k} = p^{-1/2} \sum_{g\in \mathcal{C}} \tilde{\ell}_{gk}\bm{\gamma}_{s_g g}^{(e)} [n^{-1/2}\sum_{i=1}^n h_i\{\bm{R}^{(g)}_{i*}\} \bm{G}_{s_g i}] \\ &+ p^{-1/2}\sum_{g \in \mathcal{C}} [n^{1/2} \tilde{\ell}_{gk}\bm{\gamma}_{s_g g}^{(e)}\{ \bm{\gamma}_{s_g *}^{(c)} \}^{\top}] [n^{-1}\sum_{i=1}^n \bm{G}_{s_g i}^2 \smallint_{0}^1 \nabla h_i\{\bm{R}^{(g)}_{i*} + t \bm{\gamma}_{s_g *}^{(c)}\bm{G}_{s_g i}\}dt ],
\end{align*}
Since $h_i$ and $\nabla h_i$ are bounded,
\begin{align*}
    &\E( \textstyle[n^{-1/2}\sum_{i=1}^n h_i\{ \bm{R}^{(g)}_{i*}\} \bm{G}_{s_g i}]^2 ) \leq c\\
    &\E[ \norm*{ n^{-1}\sum_{i=1}^n \bm{G}_{s_g i}^2 \smallint_{0}^1 \nabla h_i\{\bm{R}^{(g)}_{i*} + t \bm{\gamma}_{s_g *}^{(c)}\bm{G}_{s_g i}\}dt }_2 ] \leq c
\end{align*}
for some constant $c>0$. Since $\abs*{\mathcal{C}}=O(p^{1/2})$, this completes the proof.
\end{proof}
\begin{remark}
\label{supp:remark:CtE}
The proof of Lemma~\ref{supp:lemma:Cte} can easily be extended to show that \eqref{supp:equation:CtE} holds when we replace $\bm{C}$ with $P_{\bm{X}}^{\perp}\bm{C}$. This will be useful in Lemma~\ref{supp:lemma:s11} below.
\end{remark}

\subsection{Consistency of $P_{\bm{X}}^{\perp}\hat{\bm{C}}$}
\label{supp:subsection:FAConsistency}
Here we use $f(\bm{U})$ from \eqref{supp:equation:Cobj} to prove the consistency of $P_{\bm{X}}^{\perp}\hat{\bm{C}}$. The main results are Lemmas~\ref{supp:lemma:f1}, \ref{supp:lemma:EtE}, \ref{supp:lemma:f2} and Corollary~\ref{supp:corollary:Consistency}. For ease of notation, we re-define $P_{\bm{X}}^{\perp}\hat{\bm{C}}$ to be $\hat{\bm{C}}$ in Sections~\ref{supp:subsection:FAConsistency} and \ref{supp:subsection:FARate}.

\begin{lemma}
\label{supp:lemma:Vz}
Let $\bm{U} \in \mathbb{R}^{n \times K}$ be a matrix with orthonormal columns that satisfies $\bm{U}^{\top}\bm{X}=\bm{0}$, and define $\delta = \norm{ P_{\bm{U}} - P_{\tilde{\bm{C}}} }_2$. Then there exist $\bm{v}_u \in \mathbb{R}^{K \times K}$, $\bm{z}_u \in \mathbb{R}^{(n-K-d) \times K}$, and some universal constant $c>1$ such that
\begin{align*}
    \bm{v}_u^{\top} \bm{v}_u + \bm{z}_u^{\top}\bm{z}_u = I_K, \quad \bm{U} = \tilde{\bm{C}}\bm{v}_u + \bm{Q}\bm{z}_u, \quad \norm{ \bm{v}_u - \bm{v} }_2 \in [c^{-1}\delta^2,c\delta^2], \quad \norm{ \bm{z}_u }_2 \in [c^{-1}\delta, c \delta],
\end{align*}
where $\bm{v} = \bm{A}_u\bm{B}_u^{\top} \in \mathbb{R}^{K \times K}$ for $\bm{A}_u\in \mathbb{R}^{K \times K}$ and $\bm{B}_u\in \mathbb{R}^{K \times K}$ the left and right singular vectors of $\bm{v}_u$.
\end{lemma}
\begin{proof}
We can express $P_{\bm{U}} = \bm{U}\bm{U}^{\top}$ and $P_{\tilde{\bm{C}}} = \tilde{\bm{C}}\tilde{\bm{C}}^{\top}$, and can always express $\bm{U} = \tilde{\bm{C}}\bm{v}_u + \bm{Q}\bm{z}_u$ where $\bm{v}_u^{\top} \bm{v}_u + \bm{z}_u^{\top}\bm{z}_u = I_K$ by the Fredholm alternative. Then
\begin{align*}
    &\norm{ P_{\bm{U}} - P_{\tilde{\bm{C}}} }_2^2 \leq \norm{ P_{\bm{U}} - P_{\tilde{\bm{C}}} }_F^2\leq K\norm{ P_{\bm{U}} - P_{\tilde{\bm{C}}} }_2^2\\
    &\norm{ P_{\bm{U}} - P_{\tilde{\bm{C}}} }_F^2 = 2\Tr(I_K - \bm{v}_u^{\top}\bm{v}_u) = 2\Tr(\bm{z}_u^{\top}\bm{z}_u) = 2\norm{\bm{z}_u}_F^2\\
    &K^{-1}\norm{\bm{z}_u}_F^2 \leq \norm{\bm{z}_u}_2^2\leq \norm{\bm{z}_u}_F^2.
\end{align*}
The first two lines imply $\norm{ \bm{z}_u }_F^2 \in [\delta^2/2,K\delta^2/2]$, which taken with the third line, implies $\norm{ \bm{z}_u }_2^2 \in [\delta^2/(2K),K\delta^2/2]$. Note that since $\bm{v}_u^{\top}\bm{v}_u \preceq I_K$, the singular values of $\bm{v}_u$ satisfy $0 \leq \sigma_K\leq \cdots \leq \sigma_1 \leq 1$ and
\begin{align*}
    &1-\sigma_K^2 = \norm{ I_K - \bm{v}_u^{\top}\bm{v}_u }_2 = \norm{ \bm{z}_u }_2^2 \in [\delta^2/(2K),K\delta^2/2]\\
    &\Rightarrow 1-\sigma_1,\ldots,1-\sigma_K \in [\delta^2/(4K),K\delta^2/2].
\end{align*}
If $\bm{v}_u = \bm{A}\diag(\sigma_1,\ldots,\sigma_K)\bm{B}^{\top}$ is the singular value decomposition of $\bm{v}_u$, this shows that $\norm{ \bm{A}\bm{B}^{\top} - \bm{v}_u }_2 \in [\delta^2/(4K),K\delta^2/2]$, which completes the proof.
\end{proof}

\begin{lemma}
\label{supp:lemma:Powera}
Suppose the random variables $z_1,\ldots,z_n$ satisfy the following for some integer $m\geq 1$ and constant $c>0$:
\begin{enumerate}
\item $\E(z_i^{2m}) < c$
\item[(ii)] There exists a $\sigma$-algebra $\mathcal{F}$ such that $\E(z_i \mid \mathcal{F}) = 0$ for all $i \in [n]$ and $z_1,\ldots,z_n$ are independent conditional on $\mathcal{F}$.
\end{enumerate}
Then $\E\{ (\sum_{i=1}^n z_i)^{2m} \} \leq c c_m n^m$, where $c_m$ is a constant that only depends on $m$.
\end{lemma}
\begin{proof}
\begin{align*}
    \E\left\{ \left(\sum_{i=1}^n z_i\right)^{2m} \right\} = &\sum_{i_1,\ldots,i_{2m} \in [n]} \E(z_{i_1} \cdots z_{i_{2m}}) = \sum_{\substack{ i_1,\ldots,i_{2m} \in [n]:\\\text{$\exists j \in [2m]$ such that}\\ i_j \notin \{i_s\}_{s \in [2m]\setminus \{j\}} }} \E(z_{i_1} \cdots z_{i_{2m}}) \\
    &+\sum_{\substack{i_1,\ldots,i_{2m} \in [n]:\\\text{ for all $j\in[2m]$, there exists }\\\text{$j' \neq j$ such that $i_j=i_{j'}$}}} \E(z_{i_1} \cdots z_{i_{2m}}),
\end{align*}
where
\begin{align*}
\sum_{\substack{ i_1,\ldots,i_{2m} \in [n]:\\\text{$\exists j \in [2m]$ such that}\\ i_j \notin \{i_s\}_{s \in [2m]\setminus \{j\}} }} \E(z_{i_1} \cdots z_{i_{2m}} \mid \mathcal{F}) = \sum_{\substack{ i_1,\ldots,i_{2m} \in [n]:\\\text{$\exists j \in [2m]$ such that}\\ i_j \notin \{i_s\}_{s \in [2m]\setminus \{j\}} }} \underbrace{\E(z_{i_j} \mid \mathcal{F})}_{=0} \E\left(\prod_{s\in [2m]\setminus \{j\}} z_{i_s} \mid \mathcal{F}\right) = 0
\end{align*}
and
\begin{align*}
    &\sum_{\substack{i_1,\ldots,i_{2m} \in [n]:\\\text{ for all $j\in[2m]$, there exists }\\\text{$j' \neq j$ such that $i_j=i_{j'}$}}} \E(z_{i_1} \cdots z_{i_{2m}})\\
    \leq & c \abs*{ \{ i_1,\ldots,i_{2m} \in [n]: \text{ for all $j\in[2m]$, there exists $j' \neq j$ such that $i_j=i_{j'}$} \} }\\
    \leq & c c_m n^m
\end{align*}
for some constant $c_m>0$ that only depends on $m$.
\end{proof}

\begin{corollary}
\label{supp:corollary:MaximalIne}
Let $r >0$ be an integer, and define the $r$ possibly dependent sets of random variables $\{z_{j1},\ldots,z_{jn}\}$ to be such that $z_{j1},\ldots,z_{jn}$ satisfy the conditions of Lemma~\ref{supp:lemma:Powera} for each $j \in [r]$. Then for $S_j = n^{-1}\sum_{i=1}^n z_{ji}$, $\tilde{S}_j = n^{-1/2}\sum_{i=1}^n z_{ji}$, and any $t > 0$, $\Prob(\max_{j \in [r]} \abs*{S_j} \geq t) \leq c c_m r/(nt^2)^m$ and $\Prob(\max_{j \in [r]} \abs*{\tilde{S}_j} \geq t) \leq c c_m r/t^{2m}$ for $c, c_m$ defined in the statement of Lemma~\ref{supp:lemma:Powera}.
\end{corollary}
\begin{proof}
Since $\max_{j \in [r]} \abs*{S_j} \leq \left( \sum_{j=1}^r S_j^{2m} \right)^{1/(2m)}$ and $\max_{j \in [r]} \abs*{\tilde{S}_j} \leq \left( \sum_{j=1}^r \tilde{S}_j^{2m} \right)^{1/(2m)}$, this follows immediately from Lemma~\ref{supp:lemma:Powera}.
\end{proof}
\begin{remark}
\label{supp:remark:MaximalIne}
If $r = n^a$ for some $a \in (0,m)$, then Corollary~\ref{supp:corollary:MaximalIne} implies $\max_{j \in [r]} \abs*{S_j} = O_P(c n^{-\delta})$ for $\delta = (1-a/m)/2 \in (0,1/2)$.
\end{remark}

\begin{lemma}
\label{supp:lemma:VershExt}
Let $c>1$ be a constant, and assume $\bm{e}_g \sim (\bm{0},\bm{V}_g)$, $g \in [p]$, are independent sub-Gaussian random vectors with sub-Gaussian norm $\norm{\bm{e}_g}_{\Psi_2} \leq c$. Then if $p \geq c^{-1} n$ and for $\bm{E} = (\bm{e}_1 \cdots \bm{e}_p)$, $\norm*{p^{-1/2} \bm{E}}_2 = O_P(1)$ as $n,p \to \infty$.
\end{lemma}
\begin{proof}
The proof is a simple extension of the proof of Theorem~5.39 in \citet{Vershynin}, and has been omitted. 
\end{proof}

\begin{lemma}
\label{supp:lemma:LatalaW}
Let $\bm{M} \in \mathbb{R}^{p \times n}$ such that $\bm{M}_{gi} = w_{gi}-1$, and suppose Assumption~\ref{supp:assumptions:FA} hold. Then for any constant $\epsilon > 0$, $\norm*{p^{-1/2}\bm{M}}_2 = O_P(n^{\epsilon})$ as $n,p \to \infty$.
\end{lemma}
\begin{proof}
Conditional of $\bm{e}_1,\ldots,\bm{e}_p$ and $\bm{C}$, the entries of $\bm{M}$ are mean 0, independent, and have finite fourth moments, where for any integer $m>0$ and some constant $c_m>0$ that only depends on $m$, 
\begin{align*}
\E(w_{gi}^{m} \mid \bm{C}_{i \bigcdot},\bm{e}_{gi}) =& \E( [ 1/\Psi\{\alpha_g(y_{gi} - \delta_{g})\} ]^{(m-1)} \mid \bm{C}_{i \bigcdot},\bm{e}_{gi}) \leq c_m\\
&+ c_m\{ \abs*{ \bm{e}_{gi} }^{a(m-1)} + \sum_{k=1}^K \abs*{ \bm{C}_{ik} }^{a(m-1)} \}    
\end{align*}
for some constant $a>0$ by \ref{supp:assumption:Psi} in Assumption~\ref{supp:assumptions:FA}. The result then follows by \citet{Latala} and the fact that, for any $\epsilon>0$, $\max_{g \in [p],i \in [n]}\abs*{\bm{e}_{gi}}, \max_{i \in [n],k \in [K]}\abs*{\bm{C}_{ik}} = O_P(n^{\epsilon})$.
\end{proof}

\begin{lemma}
\label{supp:lemma:LatalaExt}
Let $\bm{M} \in \mathbb{R}^{p \times n}$ such that $\bm{M}_{gi} = w_{gi}\bm{e}_{g_i}$, and suppose Assumption~\ref{supp:assumptions:FA} holds. Then for any fixed constant $\epsilon \in (0,1/2)$, $\norm{ p^{-1/2}\bm{M} }_2 = O_P(n^{\epsilon})$ as $n,p \to \infty$.
\end{lemma}
\begin{proof}
We can express $\bm{M}$ as $\bm{M} = \bm{M}^{(1)} + \bm{M}^{(2)}$ for $\bm{M}^{(1)}_{gi} = \bm{e}_{gi}$ and $\bm{M}^{(2)}_{gi} = (w_{gi}-1)\bm{e}_{g_i}$. By Lemma~\ref{supp:lemma:VershExt}, $\norm*{ \bm{M}^{(1)} }_2 = O_P(1)$, and a simple extension of the proof of Lemma~\ref{supp:lemma:LatalaW} can be used to show $\norm*{p^{-1/2}\bm{M}^{(2)}}_2 = O_P(n^{\epsilon})$.
\end{proof}

\begin{lemma}
\label{supp:lemma:f1}
Let $\bm{U} \in \mathbb{R}^{n \times K}$ be a matrix with orthonormal columns and define
\begin{align*}
    f_1(\bm{U}) = (\lambda p)^{-1} \sum_{g=1}^p \Tr\{ (\bm{U}^{\top}\bm{P}_g^{\perp} \bm{U} )^{-1} \bm{U}^{\top} \bm{P}_g^{\perp} \tilde{\bm{C}}\tilde{\bm{\ell}}_g \tilde{\bm{\ell}}_g^{\top} \tilde{\bm{C}}^{\top}  \bm{P}_g^{\perp}\bm{U} \}.
\end{align*}
Let $\delta_U = \norm*{ \tilde{\bm{C}}\tilde{\bm{C}}^{\top} - \bm{U}\bm{U}^{\top} }_2$, $\eta \in (0,1/2)$ be an arbitrarily small constant, and suppose Assumption~\ref{supp:assumptions:FA} holds. Then there exists a constant $c>1$ that does not depend on $n$ or $p$ such that for all $\bm{U}$ with $\delta_U \in (0,c^{-1})$ and any $\epsilon_1,\epsilon_2>0$, $f_1(\tilde{\bm{C}}) - f_1(\bm{U}) \geq c^{-1}\delta_U^2\{ 1 - c\delta_U(1+\epsilon_2 n^{-1/2+\eta}) \}$ with probability at least $1-\epsilon_1$ for all $n,p$ sufficiently large.
\end{lemma}

\begin{proof}
For notational simplicity, we set $\delta=\delta_U$. Let $\bm{U} = \bm{C}\bm{v}_u + \bm{Q}\bm{z}_u$ for $\bm{Q}$ as defined in Lemma~\ref{supp:lemma:Vz}, where $\norm{ \bm{z}_u }_2 \in [c^{-1}\delta,c\delta]$ and $\norm{\bm{v}_u - \bm{v}}_2 \leq c\delta^2$ for some constant $c>1$ and $K\times K$ unitary matrix $\bm{v}$. We let $\tilde{\bm{z}}_u = \bm{Q}\bm{z}_u$ for the remainder of the proof, and without loos of generality, assume $n^{-1}\bm{X}^{\top}\bm{X}=I_d$. Provided $\bm{U}^{\top}\bm{P}_g^{\perp} \bm{U}$ is invertible, define
\begin{align}
    f_{1g}(\bm{U}) =& \Tr\{ (\bm{U}^{\top}\bm{P}_g^{\perp} \bm{U} )^{-1} \bm{U}^{\top} \bm{P}_g^{\perp} \tilde{\bm{C}}\tilde{\bm{\ell}}_g \tilde{\bm{\ell}}_g^{\top} \tilde{\bm{C}}^{\top}  \bm{P}_g^{\perp}\bm{U} \}\nonumber\\
    \label{supp:equation:Tr}
    =& \Tr\{(\tilde{\bm{C}}^{\top}\bm{P}_g^{\perp} \tilde{\bm{C}})^{-1/2}\tilde{\bm{C}}^{\top}  \bm{P}_g^{\perp}\bm{U} (\bm{U}^{\top}\bm{P}_g^{\perp} \bm{U} )^{-1} \bm{U}^{\top} \bm{P}_g^{\perp} \tilde{\bm{C}}(\tilde{\bm{C}}^{\top}\bm{P}_g^{\perp} \tilde{\bm{C}})^{-1/2} \times\\
    &\times (\tilde{\bm{C}}^{\top}\bm{P}_g^{\perp} \tilde{\bm{C}})^{1/2} \tilde{\bm{\ell}}_g \tilde{\bm{\ell}}_g^{\top} (\tilde{\bm{C}}^{\top}\bm{P}_g^{\perp} \tilde{\bm{C}})^{1/2} \} \leq \Tr\{ (\tilde{\bm{C}}^{\top}\bm{P}_g^{\perp} \tilde{\bm{C}})^{1/2}\tilde{\bm{\ell}}_g \tilde{\bm{\ell}}_g^{\top}(\tilde{\bm{C}}^{\top}\bm{P}_g^{\perp} \tilde{\bm{C}})^{1/2} \} = f_{1g}(\bm{C}),\nonumber
\end{align}
where the inequality follows because the symmetric and positive semi-definite matrix in the second line has eigenvalues $\leq 1$. We first see that
\begin{align*}
    \tilde{\bm{C}}^{\top}  \bm{P}_g^{\perp}\bm{U} = \tilde{\bm{C}}^{\top} \bm{P}_g^{\perp} \tilde{\bm{C}} \bm{v}_u + \tilde{\bm{C}}^{\top} \bm{P}_g^{\perp} \tilde{\bm{z}}_u.
\end{align*}
Define the $K \times K$ matrix $\bm{A}_g = \tilde{\bm{C}}^{\top} \bm{P}_g^{\perp} \tilde{\bm{C}}$. Then the expression inside the $\Tr$ operator in \eqref{supp:equation:Tr} can be written as
\begin{align}
    &(\bm{A}_g^{1/2}\bm{v}_u + \bm{A}_g^{-1/2}\tilde{\bm{C}}^{\top} \bm{P}_g^{\perp} \tilde{\bm{z}}_u)( \bm{v}_u^{\top} \bm{A}_g \bm{v}_u + \bm{z}_u^{\top} \bm{Q}^{\top}\bm{P}_g^{\perp} \tilde{\bm{z}}_u + \bm{z}_u^{\top} \bm{Q}^{\top}\bm{P}_g^{\perp} \tilde{\bm{C}}\bm{v}_u + \bm{v}_u^{\top} \tilde{\bm{C}}^{\top} \bm{P}_g^{\perp} \tilde{\bm{z}}_u )^{-1}\nonumber\\
    &\times (\bm{A}_g^{1/2}\bm{v}_u + \bm{A}_g^{-1/2}\tilde{\bm{C}}^{\top} \bm{P}_g^{\perp} \tilde{\bm{z}}_u)^{\top} = \bm{B}_g \{ \bm{B}_g^{\top}\bm{B}_g + \bm{z}_u^{\top}\bm{Q}^{\top}( \bm{P}_g^{\perp} - \bm{P}_g^{\perp}\tilde{\bm{C}}\bm{A}_g^{-1}\tilde{\bm{C}}^{\top}\bm{P}_g^{\perp} )\tilde{\bm{z}}_u\}^{-1}\bm{B}_g^{\top}\nonumber\\
    =& (I_K + \bm{D}_g)^{-1}\\
    &\bm{B}_g = \bm{A}_g^{1/2}\bm{v}_u + \bm{A}_g^{-1/2}\tilde{\bm{C}}^{\top} \bm{P}_g^{\perp} \tilde{\bm{z}}_u\nonumber\\
    \label{supp:equation:Dg}
    &\bm{D}_g = \bm{B}_g^{-\top} \bm{z}_u^{\top}\bm{Q}^{\top}( \bm{P}_g^{\perp} - \bm{P}_g^{\perp}\tilde{\bm{C}}\bm{A}_g^{-1}\tilde{\bm{C}}^{\top}\bm{P}_g^{\perp} )\tilde{\bm{z}}_u \bm{B}_g^{-1}.
\end{align}
We first prove two lemmas that we will use throughout the proof.

\begin{lemma}
\label{supp:lemma:Bs}
Suppose Assumption~\ref{supp:assumptions:FA} holds and let $\tilde{\bm{B}}_g = \bm{A}_g^{-1/2}\bm{B}_g$. Then for all $\epsilon \in (0,1/2)$ and some constant $c>0$ that does not depend on $n$, $p$, or $\delta$,
\begin{align}
    \label{supp:equation:Abound}
    &\max_{g \in [p]}\norm*{ \bm{A}_g - I_K }_2 = O_P(n^{-1/2+\epsilon})\\
    \label{supp:equation:Bbound}
    &\max_{g \in [p]}\norm*{ \bm{B}_g^{\top} \bm{B}_g - I_K }_2, \, \max_{g \in [p]}\norm*{ \tilde{\bm{B}}_g^{\top} \tilde{\bm{B}}_g - I_K }_2 \leq c (\delta + \delta^2)\{1 + O_P(n^{-1/2+\epsilon})\} \text{ as $n,p \to \infty$.}
\end{align}
\end{lemma}
\begin{proof}
Define $\bm{R} = n^{-1}\bm{C}^{\top}P_X^{\perp} \bm{C}$ and let $\epsilon >0$ be an arbitrarily small constant. Then
\begin{align*}
    \bm{A}_g = \tilde{\bm{C}}^{\top}\bm{P}_g^{\perp} \tilde{\bm{C}} =& \bm{R}^{-1/2}\{n^{-1}\bm{C}^{\top}P_{X}^{\perp} \bm{W}_g P_{X}^{\perp} \bm{C}\\
    -& (n^{-1}\bm{C}^{\top}P_{X}^{\perp} \bm{W}_g\bm{X})(n^{-1} \bm{X}^{\top}\bm{W}_g\bm{X} )^{-1}(n^{-1}\bm{X}^{\top} \bm{W}_g P_{X}^{\perp} \bm{C}) \}\bm{R}^{-1/2},
\end{align*}
where $\E\{\bm{R}^{-1/2}( n^{-1}\bm{C}^{\top}P_{X}^{\perp} \bm{W}_g P_{X}^{\perp} \bm{C} )\bm{R}^{-1/2} \mid \bm{C}\} = I_K$ and $\E( \bm{C}^{\top}P_{X}^{\perp} \bm{W}_g\bm{X} \mid \bm{C} ) = \bm{0}$. First, Corollary~\ref{supp:corollary:MaximalIne} implies $\max_{g \in [p]} \norm*{ n^{-1} \bm{X}^{\top}\bm{W}_g\bm{X} - n^{-1}\bm{X}^{\top}\bm{X} }_2 = O_P(n^{-1/2 + \epsilon})$. Next,
\begin{align*}
    n^{-1}\bm{C}^{\top} P_{X}^{\perp} \bm{W}_g \bm{X} = n^{-1}\bm{C}^{\top}\bm{W}_g \bm{X} - n^{-1}\bm{C}^{\top}\bm{X}(n^{-1}\bm{X}^{\top}\bm{X})^{-1}(n^{-1}\bm{X}^{\top}\bm{W}_g\bm{X}),
\end{align*}
where a second application of Corollary~\ref{supp:corollary:MaximalIne} shows that $\max_{g \in [p]}\norm*{ n^{-1}\bm{C}^{\top} P_{X}^{\perp} \bm{W}_g \bm{X} }_2 = O_P(n^{-1/2 + \epsilon})$. Next,
\begin{align*}
    n^{-1}\bm{C}^{\top}P_{X}^{\perp} \bm{W}_g P_{X}^{\perp} \bm{C} =& n^{-1}\bm{C}^{\top}\bm{W}_g \bm{C}\\
    &+ (n^{-1}\bm{C}^{\top}\bm{X})(n^{-1}\bm{X}^{\top}\bm{X})^{-1} (n^{-1}\bm{X}^{\top}\bm{W}_g \bm{X}) (n^{-1}\bm{X}^{\top}\bm{X})^{-1} (n^{-1}\bm{X}^{\top}\bm{C})\\
    &- (n^{-1}\bm{C}^{\top} \bm{W}_g\bm{X})(n^{-1}\bm{X}^{\top}\bm{X})(n^{-1}\bm{X}^{\top}\bm{C})\\
    &- \{ (n^{-1}\bm{C}^{\top} \bm{W}_g\bm{X})(n^{-1}\bm{X}^{\top}\bm{X})(n^{-1}\bm{X}^{\top}\bm{C}) \}^{\top},
\end{align*}
where further applications of Corollary~\ref{supp:corollary:MaximalIne} to the terms in the above expression imply
\begin{align*}
    \max_{g \in [p]}\norm*{ n^{-1}\bm{C}^{\top}P_{X}^{\perp} \bm{W}_g P_{X}^{\perp} \bm{C} - \bm{R} }_2 = O_P(n^{-1/2 + \epsilon}).
\end{align*}
This proves \eqref{supp:equation:Abound}. Since $\tilde{\bm{B}}_g = \bm{A}_g^{-1/2}\bm{B}_g$, it suffices to only consider $\bm{B}_g^{\top}\bm{B}_g$ when proving \eqref{supp:equation:Bbound}. We have
\begin{align*}
    \bm{B}_g^{\top} \bm{B}_g = \bm{v}_u^{\top} \bm{A}_g \bm{v}_u + \bm{v}_u^{\top} \tilde{\bm{C}}^{\top} \bm{P}_g^{\perp} \tilde{\bm{z}}_u + (\bm{v}_u^{\top} \tilde{\bm{C}}^{\top} \bm{P}_g^{\perp} \tilde{\bm{z}}_u)^{\top} + \tilde{\bm{z}}_u^{\top} \bm{P}_g^{\perp} \tilde{\bm{C}} \bm{A}_g^{-1} \tilde{\bm{C}}^{\top} \bm{P}_g^{\perp} \tilde{\bm{z}}_u.
\end{align*}
By Lemma~\ref{supp:lemma:Vz} and \eqref{supp:equation:Abound}, $\norm*{\bm{v}_u^{\top} \bm{A}_g \bm{v}_u - I_K}_2 \leq c\delta^2\{ 1+O_P(n^{-1/2+\epsilon}) \}$ for some constant $c>0$. Since $\norm*{ \bm{v}_u^{\top} \tilde{\bm{C}}^{\top} \bm{P}_g^{\perp} \tilde{\bm{z}}_u }_2 \leq c\delta \norm*{ \tilde{\bm{C}}^{\top} (\bm{P}_g^{\perp})^2 \tilde{\bm{C}}}_2^{1/2}$ for some constant $c>0$, we need only show that $\norm*{ \tilde{\bm{C}}^{\top} (\bm{P}_g^{\perp})^2 \tilde{\bm{C}}}_2 \leq c\{1+O_P(n^{-1/2+\epsilon})\}$ for some constant $c>0$ to complete the proof. However, this follows from an identical analysis used to study the properties of $\bm{A}_g$, the details of which have been omitted.
\end{proof}

\begin{lemma}
\label{supp:lemma:IplusDg}
Suppose the assumptions of Lemma~\ref{supp:lemma:f1} hold and let $\tilde{\bm{M}} = [n^{-1/2}\bm{X}, \tilde{\bm{C}}]$. Then for any $a>0$, define $\tilde{\bm{W}}_{g,a} = \diag[w_{g1}1\{w_{g1}>a\},\ldots,w_{gn}1\{w_{gn}>a\}]$. Then there exists constants $c>0$ and $\eta_a>0$, the latter of which is a decreasing function of $a$, and a random variable $z = O_P(n^{-1/2+\epsilon})$ such that
\begin{align*}
    (I_K + \bm{D}_g)^{-1} \preceq & I_K + c \bm{B}_g^{-\top} \tilde{\bm{z}}_u^{\top} \bm{W}_g \tilde{\bm{M}} \tilde{\bm{M}}^{\top} \bm{W}_g\tilde{\bm{z}}_u\bm{B}_g^{-1} - \eta_a \bm{B}_g^{-\top} \tilde{\bm{z}}_u^{\top} \bm{W}_g \tilde{\bm{z}}_u \bm{B}_g^{-1}\\ &+ \eta_a \bm{B}_g^{-\top} \tilde{\bm{z}}_u^{\top} \tilde{\bm{W}}_{g,a} \tilde{\bm{z}}_u \bm{B}_g^{-1} + z I_K
\end{align*}
\end{lemma}

\begin{proof}
We assume $n^{-1}\bm{X}^{\top}\bm{X}=I_d$ without loss of generality. Then we can express $\bm{D}_g$ as
\begin{align*}
    \bm{D}_g =& \bm{B}_g^{-T}\tilde{\bm{z}}_u^{\top}\bm{W}_g^{1/2}\bm{N}_g\bm{W}_g^{1/2}\tilde{\bm{z}}_u\bm{B}_g^{-1}\\
    \bm{N}_g =& I_n - \bm{W}_g^{1/2}\bm{X}(\bm{X}^{\top}\bm{W}_g\bm{X})\bm{X}^{\top} - \tilde{\bm{P}}_g \bm{W}_g^{1/2}\tilde{\bm{C}}\bm{A}_g^{-1}\tilde{\bm{C}}^{\top}\bm{W}_g^{1/2}\tilde{\bm{P}}_g\\
    \tilde{\bm{P}}_g =& I_n - \bm{W}_g^{1/2}\bm{X}(\bm{X}^{\top}\bm{W}_g\bm{X})\bm{X}^{\top}.
\end{align*}
To simplify the expression for $\bm{N}_g$, we first see that
\begin{align*}
    \tilde{\bm{P}}_g \bm{W}_g^{1/2}\tilde{\bm{C}} = \bm{W}_g^{1/2}\tilde{\bm{C}} - (n^{-1/2}\bm{W}_g^{1/2}\bm{X})(n^{-1}\bm{X}^{\top}\bm{W}_g\bm{X})\{n^{-1}\bm{X}^{\top}(\bm{W}_g - I_n)\tilde{\bm{C}}\}.
\end{align*}
Corollary~\ref{supp:corollary:MaximalIne} can then be used to show that
\begin{align*}
    &\max_{g \in [p]}\norm*{n^{-1/2}\bm{W}_g^{1/2}\bm{X}}_2, \, \max_{g \in [p]}\norm*{ n^{-1}\bm{X}^{\top}\bm{W}_g\bm{X} }_2 \leq 1 + O_P(n^{-1/2+\epsilon})\\
    &\max_{g \in [p]} \norm*{ n^{-1}\bm{X}^{\top}(\bm{W}_g - I_n)\tilde{\bm{C}} }_2 = O_P(n^{-1/2+\epsilon}),
\end{align*}
which implies for $\tilde{\bm{X}} = n^{-1/2}\bm{X}$,
\begin{align*}
    &\max_{g \in [p]}\norm*{ \tilde{\bm{P}}_g \bm{W}_g^{1/2}\tilde{\bm{C}}\bm{A}_g^{-1}\tilde{\bm{C}}^{\top}\bm{W}_g^{1/2}\tilde{\bm{P}}_g - \bm{W}_g^{1/2}\tilde{\bm{C}}\bm{A}_g^{-1}\tilde{\bm{C}}^{\top}\bm{W}_g^{1/2} }_2 = O_P(n^{-1/2+\epsilon})\\
    &\max_{g \in [p]} \norm*{ \bm{W}_g^{1/2}\bm{X}(\bm{X}^{\top}\bm{W}_g\bm{X})\bm{X}^{\top} - \bm{W}_g^{1/2} \tilde{\bm{X}} \tilde{\bm{X}}^{\top} }_2 = O_P(n^{-1/2+\epsilon}).
\end{align*}
Lemma~\ref{supp:lemma:Bs} can then be used to simplify show
\begin{align*}
    \norm*{ \bm{W}_g^{1/2}\tilde{\bm{C}}\bm{A}_g^{-1}\tilde{\bm{C}}^{\top}\bm{W}_g^{1/2} - \bm{W}_g^{1/2}\tilde{\bm{C}} \tilde{\bm{C}}^{\top} \bm{W}_g^{1/2} }_2 = O_P(n^{-1/2+\epsilon}).
\end{align*}
Putting this all together implies for $\tilde{\bm{M}} = [\tilde{\bm{X}}, \tilde{\bm{C}}]$,
\begin{align*}
    \max_{g \in [p]} \norm*{\bm{N}_g - (I_n - \bm{W}_g^{1/2}\tilde{\bm{M}}\tilde{\bm{M}}^{\top}\bm{W}_g^{1/2}) }_2 = O_P(n^{-1/2+\epsilon}).
\end{align*}
Therefore, $\bm{D}_g$ satisfies
\begin{align*}
    \max_{g \in [p]}\norm*{ \bm{D}_g - \bm{B}_{g}^{-\top}\tilde{\bm{z}}_u^{\top}\bm{W}_g^{1/2}(I_n - \bm{W}_g^{1/2}\tilde{\bm{M}}\tilde{\bm{M}}^{\top}\bm{W}_g^{1/2}) \bm{W}_g^{1/2}\tilde{\bm{z}}_u\bm{B}_{g}^{-1} }_2 = O_P(n^{-1/2+\epsilon}),
\end{align*}
where for some constant $c>0$
\begin{align*}
    \max_{g \in [p]} \norm*{ \bm{B}_{g}^{-\top}\tilde{\bm{z}}_u^{\top}\bm{W}_g \tilde{\bm{M}}\tilde{\bm{M}}^{\top}\bm{W}_g\tilde{\bm{z}}_u\bm{B}_{g}^{-1} }_2 \leq \delta^2 c\{1 + O_P(n^{-1/2})\}.
\end{align*}
Therefore, there exists a constant $\eta_1 > 0$ and random variable $z = O_P(n^{-1/2+\epsilon})$ that does not depend on $g$ such that
\begin{align*}
    (I_K + \bm{D}_g)^{-1} \preceq ( I_K + \bm{B}_{g}^{-\top}\tilde{\bm{z}}_u^{\top}\bm{W}_g \tilde{\bm{z}}_u \bm{B}_{g}^{-1} )^{-1} + \eta_1 \bm{B}_{g}^{-\top}\tilde{\bm{z}}_u^{\top}\bm{W}_g \tilde{\bm{M}}\tilde{\bm{M}}^{\top}\bm{W}_g\tilde{\bm{z}}_u\bm{B}_{g}^{-1} + zI_K.
\end{align*}
Next, let $a>0$ be a constant and define $\bar{\bm{W}}_{g,a} = \diag[w_{g1}1\{w_{g1} \leq a\},\ldots,w_{gn}1\{w_{gn} \leq a\}]$. Then for some constant $\eta_{a}>0$ that is a decreasing function of $a$,
\begin{align*}
    ( I_K + \bm{B}_{g}^{-\top}\tilde{\bm{z}}_u^{\top}\bm{W}_g \tilde{\bm{z}}_u \bm{B}_{g}^{-1} )^{-1} \preceq ( I_K + \bm{B}_{g}^{-\top}\tilde{\bm{z}}_u^{\top}\bar{\bm{W}}_{g,a} \tilde{\bm{z}}_u \bm{B}_{g}^{-1} )^{-1} \preceq I_K - \eta_a \bm{B}_{g}^{-\top}\tilde{\bm{z}}_u^{\top}\bar{\bm{W}}_{g,a} \tilde{\bm{z}}_u \bm{B}_{g}^{-1},
\end{align*}
which completes the proof.
\end{proof}

Returning to the proof of Lemma~\ref{supp:lemma:f1}, let $a$, $\tilde{\bm{W}}_{g,a}$, and $\tilde{\bm{M}}$ be as given in the statement of Lemma~\ref{supp:lemma:IplusDg} and define
\begin{align*}
    &\tilde{\bm{B}}_g = \bm{A}_g^{-1/2}\bm{B}_g = \bm{v}_u + \bm{A}_g^{-1}\tilde{\bm{C}}^{\top} \bm{P}_g^{\perp} \tilde{\bm{z}}_u\\
    &\bm{S}_g = \bm{B}_g^{-\top}\tilde{\bm{z}}_u^{\top}\bm{W}_g\tilde{\bm{M}}\tilde{\bm{M}}^{\top} \bm{W}_g \tilde{\bm{z}}_u \bm{B}_g^{-1}\\
    &\bm{R}_g = \bm{B}_g^{-\top} \tilde{\bm{z}}_u^{\top}\bm{W}_g \tilde{\bm{z}}_u \bm{B}_g^{-1}, \quad \bm{R}_{g,a} = \bm{B}_g^{-\top} \tilde{\bm{z}}_u^{\top}\tilde{\bm{W}}_{g,a} \tilde{\bm{z}}_u \bm{B}_g^{-1}.
\end{align*}
Then for constants $c$ and $\eta_a$ as defined in the statement of Lemma~\ref{supp:lemma:IplusDg}, Lemma~\ref{supp:lemma:IplusDg} implies the objective can be lower bounded as
\begin{align}
\label{supp:equation:f1LowerBound}
\begin{aligned}
    f_{1g}(\bm{C}) - f_{1g}(\bm{U}) \geq & \eta_a\Tr\{ \bm{R}_g (\tilde{\bm{C}}^{\top}\bm{P}_g^{\perp} \tilde{\bm{C}})^{1/2}\tilde{\bm{\ell}}_g \tilde{\bm{\ell}}_g^{\top}(\tilde{\bm{C}}^{\top}\bm{P}_g^{\perp} \tilde{\bm{C}})^{1/2} \}\\&- \eta_a\Tr\{ \bm{R}_{g,a} (\tilde{\bm{C}}^{\top}\bm{P}_g^{\perp} \tilde{\bm{C}})^{1/2}\tilde{\bm{\ell}}_g \tilde{\bm{\ell}}_g^{\top}(\tilde{\bm{C}}^{\top}\bm{P}_g^{\perp} \tilde{\bm{C}})^{1/2} \}\\& - c\Tr\{ \bm{S}_g (\tilde{\bm{C}}^{\top}\bm{P}_g^{\perp} \tilde{\bm{C}})^{1/2}\tilde{\bm{\ell}}_g \tilde{\bm{\ell}}_g^{\top}(\tilde{\bm{C}}^{\top}\bm{P}_g^{\perp} \tilde{\bm{C}})^{1/2} \} + O_P(\lambda n^{-1/2+\epsilon}),
\end{aligned}
\end{align}
where the error term $O_P(\lambda n^{-1/2+\epsilon})$ is uniform over $g \in [p]$. For the third term in \eqref{supp:equation:f1LowerBound},
\begin{align*}
    \bm{M}^{(1)}_g&=\Tr\{\bm{S}_g (\tilde{\bm{C}}^{\top}\bm{P}_g^{\perp} \tilde{\bm{C}})^{1/2}\tilde{\bm{\ell}}_g \tilde{\bm{\ell}}_g^{\top}(\tilde{\bm{C}}^{\top}\bm{P}_g^{\perp} \tilde{\bm{C}})^{1/2} \} \leq \tilde{\bm{\ell}}_g^{\top}\tilde{\bm{\ell}}_g \Tr\{ \tilde{\bm{z}}_u^{\top} \bm{W}_g \tilde{\bm{M}}\tilde{\bm{M}}^{\top}\bm{W}_g\tilde{\bm{z}}_u(\tilde{\bm{B}}_g^{\top}\tilde{\bm{B}}_g)^{-1} \}\\
    & \leq c\{1+O_P(n^{-1/2+\epsilon})\}\tilde{\bm{\ell}}_g^{\top}\tilde{\bm{\ell}}_g \sum_{k=1}^K \tilde{\bm{z}}_{u_{\bigcdot k}}^{\top} \bm{W}_g \tilde{\bm{M}} \tilde{\bm{M}}^{\top}\bm{W}_g \tilde{\bm{z}}_{u_{\bigcdot k}}
\end{align*}
for some constant $c > 0$ that does not depend on $n$ or $p$, where the second inequality holds by Lemma~\ref{supp:lemma:Bs} for $\delta$ small enough. Note that the by Lemma~\ref{supp:lemma:Bs}, the $O_P(n^{-1/2+\epsilon})$ term is uniform over $g \in [p]$. Define $s_g^2 = \tilde{\bm{\ell}}_g^{\top}\tilde{\bm{\ell}}_g \leq c\lambda\{1+O_P(n^{-1/2})\}$, where the error is uniform over $g \in [p]$. Then
\begin{align*}
    (\lambda p)^{-1}\sum_{g=1}^p \bm{M}^{(1)}_g &\leq c\{1+O_P(n^{-1/2+\epsilon})\} \sum_{k=1}^K \tilde{\bm{z}}_{u_{\bigcdot k}}^{\top} \left\{(\lambda p)^{-1}\sum_{g=1}^p s_g^2\bm{W}_g \tilde{\bm{M}} \tilde{\bm{M}}^{\top}\bm{W}_g \right\} \tilde{\bm{z}}_{u_{\bigcdot k}}\\
    &= c\{1+O_P(n^{-1/2+\epsilon})\}\sum_{k=1}^K\sum_{r=1}^{d+K} (\lambda p)^{-1} \tilde{\bm{z}}_{u_{\bigcdot k}}^{\top} \sum_{g=1}^p s_g^2 \bm{W}_g \tilde{\bm{M}}_{\bigcdot r} \tilde{\bm{M}}_{\bigcdot r}^{\top}\bm{W}_g \tilde{\bm{z}}_{u_{\bigcdot k}}.
\end{align*}
We see that
\begin{align*}
     (\lambda p)^{-1} \tilde{\bm{z}}_{u_{\bigcdot k}}^{\top} \sum_{g=1}^p s_g^2 \bm{W}_g \tilde{\bm{M}}_{\bigcdot r} \tilde{\bm{M}}_{\bigcdot r}^{\top}\bm{W}_g \tilde{\bm{z}}_{u_{\bigcdot k}} =& p^{-1}\tilde{\bm{z}}_{\bigcdot k}^{\top}\bm{G}\bm{S}\bm{G}^{\top}\tilde{\bm{z}}_{\bigcdot k}\\
    \bm{G} =& [\bm{G}_1 \cdots \bm{G}_p] \in \mathbb{R}^{n \times p}, \quad \bm{G}_g = \bm{W}_g \tilde{\bm{M}}_{\bigcdot r}\\
    \bm{S} =& \diag(s_1^2/\lambda, \ldots, s_p^2/\lambda),
\end{align*}
where by Cauchy-Schwarz and the fact that $\norm{ \bm{S} }_2 \leq c\{1+O_P(n^{-1/2})\}$ for some constant $c>0$,
\begin{align*}
    0 \leq (\lambda p)^{-1} \tilde{\bm{z}}_{\bigcdot k}^{\top} \sum_{g=1}^p s_g^2 \bm{W}_g \tilde{\bm{M}}_{\bigcdot r} \tilde{\bm{M}}_{\bigcdot r}^{\top}\bm{W}_g \tilde{\bm{z}}_{u_{\bigcdot k}} \leq c\{1+O_P(n^{-1/2})\}( p^{-1}\tilde{\bm{z}}_{u_{\bigcdot k}}^{\top}\bm{G}\bm{G}^{\top}\tilde{\bm{z}}_{u_{\bigcdot k}})
\end{align*}
for some constant $c>0$. Let $\tilde{\bm{S}} = p^{-1} \bm{G}\bm{G}^{\top}$. Then
\begin{align}
\label{supp:equation:S}
\begin{aligned}
    \tilde{\bm{S}}_{ij} =& \tilde{\bm{M}}_{i r}\tilde{\bm{M}}_{j r} p^{-1} \sum_{g=1}^p (w_{gi} - 1)(w_{gj} - 1) + \tilde{\bm{M}}_{i r}\tilde{\bm{M}}_{j r} p^{-1} \sum_{g=1}^p w_{gi} + \tilde{\bm{M}}_{i r}\tilde{\bm{M}}_{j r} p^{-1} \sum_{g=1}^p w_{gj}\\
    &-  \tilde{\bm{M}}_{i r}\tilde{\bm{M}}_{j r}.
\end{aligned}
\end{align}
Since $\tilde{\bm{z}}_{u_{\bigcdot k}}^{\top}\tilde{\bm{M}}_{\bigcdot r}=0$ for all $k \in [K]$ and $r \in [d+K]$, the last three terms in \eqref{supp:equation:S} are nullified, implying we need only study the first term. We then have that for $\bm{M}=\in \mathbb{R}^{n \times p}$ such that $\bm{M}_{ig} = w_{gi} - 1$, $\tilde{\bm{S}} = p^{-1} \diag(\tilde{\bm{M}}_{\bigcdot r})\bm{M} \bm{M}^{\top}\diag(\tilde{\bm{M}}_{\bigcdot r})$. By Lemma~\ref{supp:lemma:LatalaW}, $\norm*{p^{-1/2}\bm{M}}_2 = O_P(n^{\epsilon})$ for an arbitrarily small constant $\epsilon>0$. Therefore, $\norm*{ \tilde{\bm{S}}}_2 = O_P(n^{\epsilon}\max_{i \in [n]} \tilde{\bm{M}}_{ir}^2) = O_P(n^{-1 + \epsilon})$, which implies $(\lambda p)^{-1}\sum_{g=1}^p\bm{M}_g^{(1)} = O_P(\delta^2n^{-1+\epsilon})$. 

We next consider the first term in \eqref{supp:equation:f1LowerBound}. Here,
\begin{align*}
    & \Tr\{\bm{R}_g (\tilde{\bm{C}}^{\top}\bm{P}_g^{\perp} \tilde{\bm{C}})^{1/2}\tilde{\bm{\ell}}_g \tilde{\bm{\ell}}_g^{\top}(\tilde{\bm{C}}^{\top}\bm{P}_g^{\perp} \tilde{\bm{C}})^{1/2}\}= \Tr( \tilde{\bm{z}}_u^{\top} \bm{W}_g \tilde{\bm{z}}_u \tilde{\bm{B}}_g^{-1}\tilde{\bm{\ell}}_g \tilde{\bm{\ell}}_g^{\top} \tilde{\bm{B}}_g^{-T} ).
\end{align*}
Since $f_{1g}(\bm{U})$ in \eqref{supp:equation:Tr} only depends on $\im(\bm{U})$, it suffices to assume $\norm*{ I_K-\bm{v}_u }_2 = O(\delta^2)$. Let $\bm{\Delta}_g = I_K - \tilde{\bm{B}}_g^{-1}$, where an identical analysis to that used to prove \eqref{supp:equation:Bbound} in Lemma~\ref{supp:lemma:Bs} can be used to show that $\max_{g \in [p]}\norm*{ \bm{\Delta}_g }_2 \leq c(\delta+\delta^2)\{1+O_P(n^{-1/2+\epsilon})\}$. Next,
\begin{align*}
    \Tr( \tilde{\bm{z}}_u^{\top} \bm{W}_g \tilde{\bm{z}}_u \tilde{\bm{B}}_g^{-1}\tilde{\bm{\ell}}_g \tilde{\bm{\ell}}_g^{\top} \tilde{\bm{B}}_g^{-T} ) =& \Tr( \tilde{\bm{z}}_u^{\top} \bm{W}_g \tilde{\bm{z}}_u \tilde{\bm{\ell}}_g \tilde{\bm{\ell}}_g^{\top} ) - \Tr\{ \tilde{\bm{z}}_u^{\top} \bm{W}_g \tilde{\bm{z}}_u ( \bm{\Delta}_g\tilde{\bm{\ell}}_g \tilde{\bm{\ell}}_g^{\top}\tilde{\bm{B}}_g^{-\top} + \tilde{\bm{B}}_g^{-1}\tilde{\bm{\ell}}_g \tilde{\bm{\ell}}_g^{\top}\bm{\Delta}_g^{\top} ) \}\\
    &+ \Tr( \tilde{\bm{z}}_u^{\top} \bm{W}_g \tilde{\bm{z}}_u \bm{\Delta}_g\tilde{\bm{\ell}}_g \tilde{\bm{\ell}}_g^{\top} \bm{\Delta}_g^{\top} )\\
    \geq & \Tr( \tilde{\bm{z}}_u^{\top} \bm{W}_g \tilde{\bm{z}}_u \tilde{\bm{\ell}}_g \tilde{\bm{\ell}}_g^{\top} ) - \Tr\{ \tilde{\bm{z}}_u^{\top} \bm{W}_g \tilde{\bm{z}}_u ( \bm{\Delta}_g\tilde{\bm{\ell}}_g \tilde{\bm{\ell}}_g^{\top}\tilde{\bm{B}}_g^{-\top} + \tilde{\bm{B}}_g^{-1}\tilde{\bm{\ell}}_g \tilde{\bm{\ell}}_g^{\top}\bm{\Delta}_g^{\top} ) \}.
\end{align*}
For $\bm{R} = n^{-1}\bm{C}^{\top}P_{X}^{\perp}\bm{C}$,
\begin{align*}
    (\lambda p)^{-1}\sum_{g=1}^p \Tr( \tilde{\bm{z}}_u^{\top} \bm{W}_g \tilde{\bm{z}}_u \tilde{\bm{\ell}}_g \tilde{\bm{\ell}}_g^{\top} ) =& (\lambda p)^{-1}\sum_{g=1}^p n\Tr( \tilde{\bm{z}}_u^{\top} \bm{W}_g \tilde{\bm{z}}_u \bm{R}^{1/2}\bm{\ell}_g \bm{\ell}_g^{\top}\bm{R}^{1/2} )\\
    =& (\lambda p)^{-1}\sum_{g=1}^p n\Tr( \tilde{\bm{z}}_u^{\top} \bm{W}_g \tilde{\bm{z}}_u \bm{\ell}_g \bm{\ell}_g^{\top} ) + O_P(\delta^2 n^{-1/2 + \epsilon})\\
    (\lambda p)^{-1}\sum_{g=1}^p n\Tr( \tilde{\bm{z}}_u^{\top} \bm{W}_g \tilde{\bm{z}}_u \bm{\ell}_g \bm{\ell}_g^{\top} ) \geq & \norm*{ \tilde{\bm{z}}_u }_F^2 \lambda_K/\lambda + \sum\limits_{r,s=1}^K \sum\limits_{i=1}^n \tilde{\bm{z}}_{u_{ir}} \tilde{\bm{z}}_{u_{is}} \underbrace{p^{-1}\sum\limits_{g=1}^p (n\bm{\ell}_{gr}\bm{\ell}_{gs}/\lambda)(w_{gi} - 1)}_{=x_{irs}}.
\end{align*}
Since $\{w_{gi} - 1\}_{g \in [p]}$ are mean 0 and independent conditional on $\bm{C}$ and $\max_{g \in [p]}\abs*{n\bm{\ell}_{gr}\bm{\ell}_{gs}/\lambda} \leq c$ for some constant $c>0$, Corollary~\ref{supp:corollary:MaximalIne} implies
\begin{align*}
    \max_{\substack{i \in [n], \, r,s \in [K]}}\abs{x_{irs}} = O_P(n^{-1/2+\epsilon}).
\end{align*}
Lastly, $\max_{g \in [p]}s_g^{-2}\norm*{ \bm{\Delta}_g\tilde{\bm{\ell}}_g \tilde{\bm{\ell}}_g^{\top}\tilde{\bm{B}}_g^{-\top} + \tilde{\bm{B}}_g^{-1}\tilde{\bm{\ell}}_g \tilde{\bm{\ell}}_g^{\top}\bm{\Delta}_g^{\top} }_2\leq c\delta \{1+O_P(n^{-1/2+\epsilon})\}$ for some constant $c>0$ and $\delta$ small enough, meaning
\begin{align*}
    \abs*{ \Tr\{ \tilde{\bm{z}}_u^{\top} \bm{W}_g \tilde{\bm{z}}_u ( \bm{\Delta}_g\tilde{\bm{\ell}}_g \tilde{\bm{\ell}}_g^{\top}\tilde{\bm{B}}_g^{-\top} + \tilde{\bm{B}}_g^{-1}\tilde{\bm{\ell}}_g \tilde{\bm{\ell}}_g^{\top}\bm{\Delta}_g^{\top} ) \} } \leq  c\delta\{1+O_P(n^{-1/2+\epsilon})\} s_g^2 \Tr( \tilde{\bm{z}}_u^{\top} \bm{W}_g \tilde{\bm{z}}_u )
\end{align*}
for some constant $c>0$. Therefore,
\begin{align*}
    \abs*{ (\lambda p)^{-1}\sum\limits_{g=1}^p \Tr\{ \tilde{\bm{z}}_u^{\top} \bm{W}_g \tilde{\bm{z}}_u ( \bm{\Delta}_g\tilde{\bm{\ell}}_g \tilde{\bm{\ell}}_g^{\top}\tilde{\bm{B}}_g^{-\top} + \tilde{\bm{B}}_g^{-1}\tilde{\bm{\ell}}_g \tilde{\bm{\ell}}_g^{\top}\bm{\Delta}_g^{\top} ) \} } \leq c \delta^3\{1+O_P(n^{-1/2+\epsilon})\}
\end{align*}
for some constant $c>0$.

We lastly consider the second term in \eqref{supp:equation:f1LowerBound}. For $a>0$ as defined in the statement of Lemma~\ref{supp:lemma:IplusDg}, $\E[w_{gi}1\{w_{gi}>a\}] \leq \epsilon_a$. for some constant $\epsilon_a \geq 0$ that is a non-increasing function of $a$, and can be made arbitrarily small. Then there exists a constant $c > 0$ and random variable $z=O_P(n^{-1/2+\epsilon})$ that does not depend on $g$ such that
\begin{align*}
    \bm{M}_g^{(2)} =& \Tr\{\bm{R}_{g,a}(\tilde{\bm{C}}^{\top}\bm{P}_g^{\perp} \tilde{\bm{C}})^{1/2}\tilde{\bm{\ell}}_g \tilde{\bm{\ell}}_g^{\top}(\tilde{\bm{C}}^{\top}\bm{P}_g^{\perp} \tilde{\bm{C}})^{1/2}\} \leq (c+z)\lambda \Tr( \tilde{\bm{z}}_u^{\top} \tilde{\bm{W}}_{g,a} \tilde{\bm{z}}_u ) \leq \delta^2 (c+z)\lambda \epsilon_a\\& + (c+z)\lambda \Tr[ \tilde{\bm{z}}_u^{\top}\{\tilde{\bm{W}}_{g,a} - \E(\tilde{\bm{W}}_{g,a})\}\tilde{\bm{z}}_u ].
\end{align*}
We see that
\begin{align*}
    p^{-1}\sum_{g=1}^p \Tr[ \tilde{\bm{z}}_u^{\top}\{\tilde{\bm{W}}_{g,a} - \E(\tilde{\bm{W}}_{g,a})\}\tilde{\bm{z}}_u ] = \sum_{k=1}^K \sum_{i=1}^n \tilde{\bm{z}}_{u_{ik}}^2 p^{-1}\sum_{g=1}^p ( w_{gi}1\{w_{gi}>a\} - \E[ w_{gi}1\{w_{gi}>a\} ] ),
\end{align*}
where
\begin{align*}
    \max_{i \in [n]} \abs*{ p^{-1}\sum_{g=1}^p ( w_{gi}1\{w_{gi}>a\} - \E[ w_{gi}1\{w_{gi}>a\} ] ) } = O_P(n^{-1/2 + \epsilon}).
\end{align*}
Choosing $a > 0$ large enough, and therefore $\epsilon_a \geq 0$ small enough, thus completes the proof.
\end{proof}

\begin{lemma}
\label{supp:lemma:UtPU}
Suppose Assumption~\ref{supp:assumptions:FA} holds and let $\Omega_{\delta} = \{\bm{U} \in \mathbb{R}^{n \times K}: \bm{U}^{\top}\bm{U}=I_K, \bm{U}^{\top}\bm{X}=\bm{0},\norm{P_{\bm{U}} - P_{\tilde{\bm{C}}}}_2\leq \delta\}$. Then for all $\delta>0$ sufficiently small, there exists a constant $c>0$ such that for all $\epsilon \in (0,1/2)$, $\sup_{\bm{U} \in \Omega_{\delta}}\max_{g \in [p]} \norm*{ (\bm{U}^{\top}\bm{P}_g^{\perp}\bm{U})^{-1} }_2 \leq c + O_P(n^{-1/2+\epsilon})$.
\end{lemma}
\begin{proof}
For $\bm{v}_u$ and $\bm{z}_u$ as defined in Lemma~\ref{supp:lemma:Vz}, Lemma~\ref{supp:lemma:Vz} implies that for some constant $c>0$,
\begin{align*}
    \bm{U}^{\top}\bm{P}_g^{\perp}\bm{U} \succeq (1-\delta^2 c)I_K + \bm{v}_u^{\top}(\tilde{\bm{C}}^{\top}\bm{P}_g^{\perp}\tilde{\bm{C}}-I_K) \bm{v}_u + \bm{v}_u^{\top}\tilde{\bm{C}}^{\top}\bm{P}_g^{\perp}\bm{Q}\bm{z}_u + (\bm{v}_u^{\top}\tilde{\bm{C}}^{\top}\bm{P}_g^{\perp}\bm{Q}\bm{z}_u)^{\top}.
\end{align*}
By Lemma~\ref{supp:lemma:Bs} and the proof of \eqref{supp:equation:Bbound} in Lemma~\ref{supp:lemma:Bs},
\begin{align*}
    &\max_{g \in [p]}\norm*{ \tilde{\bm{C}}^{\top}\bm{P}_g^{\perp}\tilde{\bm{C}}-I_K }_2 = O_P(n^{-1/2+\epsilon})\\
    &\sup_{\bm{U} \in \Omega_{\delta}}\max_{g \in [p]} \norm*{ \bm{v}_u^{\top}\tilde{\bm{C}}^{\top}\bm{P}_g^{\perp}\bm{Q}\bm{z}_u }_2 \leq \{c + O_P(n^{-1/2+\epsilon})\}\sup_{\bm{U} \in \Omega_{\delta}}\norm{ \bm{z}_u }_2
\end{align*}
for some constant $c>0$. Since $\sup_{\bm{U} \in \Omega_{\delta}}\norm{ \bm{z}_u }_2 = O(\delta)$ by Lemma~\ref{supp:lemma:Vz}, this completes the proof.
\end{proof}

\begin{lemma}
\label{supp:lemma:EtE}
Define $f_3(\bm{U}) = (\lambda p)^{-1}\sum_{g=1}^p \Tr\{ (\bm{U}^{\top} \bm{P}_g^{\perp}\bm{U} )^{-1} \bm{U}^{\top}\bm{P}_g^{\perp}\bm{e}_g \bm{e}_g^{\top}\bm{P}_g^{\perp}\bm{U} \}$ and suppose Assumption~\ref{supp:assumptions:FA} holds. Then for any constant $\epsilon \in (0,1/2)$,
\begin{align*}
    &\sup_{\bm{U}_{\delta} \in \Omega}f_3(\bm{U}) = O_P(\lambda^{-1+\epsilon}), \quad \Omega_{\delta} = \{\bm{U} \in \mathbb{R}^{n \times K}: \bm{U}^{\top}\bm{U}=I_K, \bm{U}^{\top}\bm{X}=\bm{0},\norm{P_{\bm{U}} - P_{\tilde{\bm{C}}}}_2\leq \delta\}
\end{align*}
for all $\delta>0$ sufficiently small.
\end{lemma}

\begin{proof}
Let $\epsilon \in (0,1/2)$ be an arbitrarily small constant. For any $\bm{U} \in \Omega_{\delta}$, there exists a constant $c>0$ such that
\begin{align*}
    &f_3(\bm{U}) \leq  c\{1+o_P(1)\}(\lambda p)^{-1}\sum_{g=1}^p \Tr\{ \bm{U}^{\top}\bm{P}_g^{\perp}\bm{e}_g \bm{e}_g^{\top}\bm{P}_g^{\perp}\bm{U} \}\\ =& c\{1+o_P(1)\}(\lambda p)^{-1}\sum_{g=1}^p \Tr\{ \bm{U}^{\top}\bm{W}_g\bm{e}_g \bm{e}_g^{\top}\bm{W}_g\bm{U} \}\\
    &-2c\{1+o_P(1)\}\underbrace{(\lambda p)^{-1}\sum_{g=1}^p \Tr\{ \bm{U}^{\top}\bm{W}_g\bm{X}(\bm{X}^{\top}\bm{W}_g\bm{X})^{-1}\bm{X}^{\top}\bm{W}_g\bm{e}_g \bm{e}_g^{\top}\bm{W}_g\bm{U} \}}_{=S(\bm{U})}\\
    &+c\{1+o_P(1)\}\underbrace{(\lambda p)^{-1}\sum_{g=1}^p \Tr\{ \bm{U}^{\top}\bm{W}_g\bm{X}(\bm{X}^{\top}\bm{W}_g\bm{X})^{-1}\bm{X}^{\top}\bm{W}_g\bm{e}_g \bm{e}_g^{\top}\bm{W}_g\bm{X}(\bm{X}^{\top}\bm{W}_g\bm{X})^{-1}\bm{X}^{\top}\bm{W}_g \bm{U} \}}_{=T(\bm{U})}
\end{align*}
by Lemma~\ref{supp:lemma:UtPU} for $\delta>0$ sufficiently small. By Lemma~\ref{supp:lemma:LatalaExt},
\begin{align*}
    \sup_{\bm{U} \in \Omega_{\delta}} (\lambda p)^{-1}\sum_{g=1}^p \Tr\{ \bm{U}^{\top}\bm{W}_g\bm{e}_g \bm{e}_g^{\top}\bm{W}_g\bm{U} \} = O_P(\lambda^{-1+\epsilon}).
\end{align*}
Since $f_3(\bm{U})$ only depends on $\bm{X}$ through $\im(\bm{X})$, it suffices to assume $n^{-1}\bm{X}^{\top}\bm{X}=I_d$, where by \ref{supp:assumption:X} of Assumption~\ref{supp:assumptions:FA}, the entries of $\bm{X}$ are uniformly bounded from above and below. Define $\bm{\Delta}_g = I_d - (n^{-1}\bm{X}^{\top} \bm{W}_g \bm{X})^{-1}$. By Corollary~\ref{supp:corollary:MaximalIne}, $\max_{g \in [p]}\norm*{ \bm{\Delta}_g }_2 = O_P(n^{-1/2 + \epsilon})$ and
\begin{align*}
    S(\bm{U}) =& (\lambda p)^{-1} \Tr\left( \bm{U}^{\top} \sum_{s=1}^d \bm{A}_s \bm{U} \right) + (\lambda p)^{-1} \Tr\left( \bm{U}^{\top} \sum_{r,s=1}^d \bm{B}_{rs} \bm{U} \right)\\
    \bm{A}_s =& \diag(\bm{X}_{\bigcdot s}) \bm{W}^{\top} \diag( n^{-1}\bm{X}_{\bigcdot s}^{\top}\bm{W}_1\bm{e}_1,\ldots, n^{-1}\bm{X}_{\bigcdot s}^{\top}\bm{W}_p\bm{e}_p)\tilde{\bm{E}}, \quad s\in [d]\\
    \bm{B}_{rs} =& \diag(\bm{X}_{\bigcdot r}) \bm{W}^{\top} \diag( n^{-1}\bm{\Delta}_{1_{rs}}\bm{X}_{\bigcdot s}^{\top}\bm{W}_1\bm{e}_1,\ldots, n^{-1}\bm{\Delta}_{p_{rs}}\bm{X}_{\bigcdot s}^{\top}\bm{W}_p\bm{e}_p)\tilde{\bm{E}}, \quad r,s \in [d]\\
    \bm{W}_{gi} =& w_{gi}-1, \quad \tilde{\bm{E}}_{gi} = \bm{e_{g_i}}w_{gi}, \quad g \in [p]; i \in [n].
\end{align*}
Since the entries of $\bm{X}$ are uniformly bounded, $\norm*{ \diag(\bm{X}_{\bigcdot s}) }_2 = O(1)$ for all $s \in [d]$. By Lemmas~\ref{supp:lemma:LatalaW} and \ref{supp:lemma:LatalaExt}, $\norm*{p^{-1/2} \bm{W}}_2, \norm*{p^{-1/2}\tilde{\bm{E}}}_2 = O_P(n^{\epsilon})$, and by Corollary~\ref{supp:corollary:MaximalIne},
\begin{align*}
    \norm*{ \diag( n^{-1}\bm{X}_{\bigcdot s}^{\top}\bm{W}_1\bm{e}_1,\ldots, n^{-1}\bm{X}_{\bigcdot s}^{\top}\bm{W}_p\bm{e}_p) }_2 = O_P(n^{-1/2+\epsilon}).
\end{align*}
Putting this all together implies $\sup_{\bm{U} \in \Omega}S(\bm{U}) = O_P(\lambda^{-1} n^{-1/2+\epsilon})$. An identical analysis can be used to show that $\sup_{\bm{U} \in \Omega}T(\bm{U}) = O_P(\lambda^{-1} n^{-1+\epsilon})$, which completes the proof.
\end{proof}

\begin{lemma}
\label{supp:lemma:f2}
Define $f_2(\bm{U}) = (\lambda p)^{-1}\sum_{g=1}^p \Tr\{ (\bm{U}^{\top}\bm{P}_g^{\perp} \bm{U} )^{-1} \bm{U}^{\top} \bm{P}_g^{\perp} \tilde{\bm{C}}\tilde{\bm{\ell}}_g \bm{e}_g^{\top} \bm{P}_g^{\perp}\bm{U} \}$. Then under the assumptions of Lemma~\ref{supp:lemma:EtE}, $\sup_{\bm{U} \in \Omega_{\delta}} \abs*{ f_2(\bm{U}) } = O_P(\lambda^{-1/2+\epsilon})$ for any constant $\epsilon \in (0,1/2)$ and $\delta>0$ sufficiently small.
\end{lemma}

\begin{proof}
Let $\bm{U} \in \Omega$ and define $a_{g_{rs}}$ to be the $r,s$ element of $( \bm{U}^{\top} \bm{P}_{g}^{\perp}\bm{U} )^{-1/2}$ for $(r,s) \in [K] \times [K]$. Note that $\max_{g\in[p]}\abs*{a_{g_{rs}}} \leq c\{1+o_P(1)\}$ for some constant $c>0$ by Lemma~\ref{supp:lemma:UtPU}. Then for $\bar{\bm{\ell}}_g = \lambda^{-1/2}\tilde{\bm{\ell}}$, where $\max_{g \in [p]} \bar{\bm{\ell}}_g \leq c\{1+o_P(1)\}$ for some constant $c>0$,
\begin{align*}
    f_2(\bm{U}) =& \lambda^{-1/2}\sum_{r,s=1}^K \bm{U}_{\bigcdot r}^{\top}\bm{A}^{\top} \bm{B}_{rs} \bm{U}_{\bigcdot s}, \quad \bm{A} = p^{-1/2} \begin{pmatrix}
    \bar{\bm{\ell}}_1^{\top} \tilde{\bm{C}}^{\top} \bm{P}_{1}^{\perp}\\
    \vdots\\
    \bar{\bm{\ell}}_p^{\top} \tilde{\bm{C}}^{\top} \bm{P}_{p}^{\perp}
    \end{pmatrix}, \quad \bm{B}_{rs} = p^{-1/2}\begin{pmatrix} 
    a_{1_{rs}}\bm{e}_1^{\top} \bm{P}_{1}^{\perp}\\
    \vdots\\
    a_{p_{rs}}\bm{e}_p^{\top} \bm{P}_{p}^{\perp}
    \end{pmatrix},
\end{align*}
where $\norm*{ \bm{A} }_2, \norm*{\bm{B}_{rs}}_2 = O_P(n^{\epsilon})$ by the proofs of Lemmas~\ref{supp:lemma:f1} and \ref{supp:lemma:EtE}.
\end{proof}

\begin{corollary}
\label{supp:corollary:Consistency}
Suppose the assumptions of Lemma~\ref{supp:lemma:EtE} hold and let $\Omega_{\delta}$ be as defined in the statement of Lemma~\ref{supp:lemma:EtE}. Then for $f$ defined in \eqref{supp:equation:fMax}, $\hat{\bm{C}} = \argmax_{\bm{U} \in \Omega_{\delta}} f(\bm{U})$, and $\delta>0$ sufficiently small, there exists a constant $\eta \in (0,1/4)$ such that $\norm*{ P_{\hat{\bm{C}}} - P_{\tilde{\bm{C}}} }_2 = O_P(n^{-\eta})$ as $n,p \to \infty$.
\end{corollary}

\begin{proof}
This is a direct consequence of Lemmas~\ref{supp:lemma:f1}, \ref{supp:lemma:EtE}, and \ref{supp:lemma:f2}.
\end{proof}

\subsection{Properties and rate of convergence of $\hat{\bm{C}}$}
\label{supp:subsection:FARate}
Here we study the properties and rate of convergence of $\hat{\bm{C}}$. To do so, we use the decomposition discussed in Lemma~\ref{supp:lemma:Vz}, where any $\bm{U} \in \mathbb{R}^{n \times K}$ such that $\bm{U}^{\top}\bm{U}=I_K$ can be expressed as $\bm{U}=\tilde{\bm{C}}\bm{v}_u + \bm{Q}\bm{z}_u$, where the columns of $\bm{Q} \in \mathbb{R}^{n \times (n-K)}$ form an orthonormal basis for $\ker(\tilde{\bm{C}}^{\top})$, $\bm{v}_u,\bm{z}_u$ depend on $\bm{U}$, and $\bm{v}_u^{\top} \bm{v}_u + \bm{z}_u^{\top}\bm{z}_u = I_K$. We can therefore write $\hat{\bm{C}}$ defined in the statement of Corollary~\ref{supp:corollary:Consistency} as $\hat{\bm{C}} = \tilde{\bm{C}}\hat{\bm{v}} + \bm{Q}\hat{\bm{z}}$, where to understand the properties of $\hat{\bm{C}}$, we need only determine $\hat{\bm{v}}$ and $\hat{\bm{z}}$.

Define $\tilde{f}\{(\bm{v}^{\top}\, \bm{z}^{\top})^{\top}\} = f(\tilde{\bm{C}}\bm{v} + \bm{Q}\bm{z})$, where $f$ is as defined in \eqref{supp:equation:fMax}. Then for $\bm{U} = \tilde{\bm{C}}\bm{v}_u + \bm{Q}\tilde{\bm{z}}_u$,
\begin{align}
\label{supp:equation:stilde}
\begin{aligned}
    &\tilde{\bm{s}}\{(\bm{v}_u^{\top}\, \bm{z}_u^{\top})^{\top}\} = \begin{pmatrix}\tilde{\bm{C}}^{\top}\\ \bm{Q}^{\top}\end{pmatrix}\nabla_{\bm{U}} f(\bm{U})\\
    =& (p\lambda)^{-1} \sum_{g=1}^p \begin{pmatrix}\tilde{\bm{C}}^{\top}\\ \bm{Q}^{\top}\end{pmatrix} \{ \bm{P}_{g}^{\perp} - \bm{P}_{g}^{\perp}\bm{U}(\bm{U}^{\top}\bm{P}_g^{\perp} \bm{U})^{-1} \bm{U}^{\top} \bm{P}_g^{\perp}\}\bm{y}_g\bm{y}_g^{\top} \bm{P}_{g}^{\perp}\bm{U}(\bm{U}^{\top}\bm{P}_g^{\perp} \bm{U})^{-1},
\end{aligned}
\end{align}
where for any unitary matrix $\bm{v} \in \mathbb{R}^{K \times K}$,
\begin{align}
\label{supp:equation:stildeC}
\begin{aligned}
    &\tilde{\bm{s}}\{(\bm{v}^{\top}\, \bm{0})^{\top}\} = \begin{pmatrix} \bm{0}_{K \times K}\\ (p\lambda)^{-1}\sum_{g=1}^p \bm{Q}^{\top}\{ \bm{P}_{g}^{\perp} - \bm{P}_{g}^{\perp}\tilde{\bm{C}}(\tilde{\bm{C}}^{\top}\bm{P}_g^{\perp} \tilde{\bm{C}})^{-1} \tilde{\bm{C}}^{\top} \bm{P}_g^{\perp}\}\bm{e}_g\tilde{\bm{\ell}}_g^{\top} \end{pmatrix}\bm{v}\\
    & + \begin{pmatrix} \bm{0}_{K \times K}\\ (p\lambda)^{-1}\sum_{g=1}^p \bm{Q}^{\top}\{ \bm{P}_{g}^{\perp} - \bm{P}_{g}^{\perp}\tilde{\bm{C}}(\tilde{\bm{C}}^{\top}\bm{P}_g^{\perp} \tilde{\bm{C}})^{-1} \tilde{\bm{C}}^{\top} \bm{P}_g^{\perp}\}\bm{e}_g \bm{e}_g^{\top} \bm{P}_{g}^{\perp}\tilde{\bm{C}}(\tilde{\bm{C}}^{\top}\bm{P}_g^{\perp} \tilde{\bm{C}})^{-1} \end{pmatrix}\bm{v}.
\end{aligned}
\end{align}
The Hessian can be expressed as
\begin{align}
\label{supp:equation:Htilde}
\begin{aligned}
    \tilde{\bm{H}}(\bm{U}) =& \{I_K\otimes (\tilde{\bm{C}}\, \bm{Q})^{\top}\}\nabla_{\bm{U}}^2f(\bm{U})\{I_K\otimes (\tilde{\bm{C}}\, \bm{Q})\}\\
    =& (\lambda p)^{-1} \sum_{g=1}^p \{I_K\otimes (\tilde{\bm{C}}\, \bm{Q})^{\top}\} (\bm{U}^{\top}\bm{P}_g^{\perp} \bm{U})^{-1} \otimes \{\bm{A}_g(\bm{U}) \bm{y}_g\bm{y}_g^{\top} \bm{A}_g(\bm{U})\}\{I_K\otimes (\tilde{\bm{C}}\, \bm{Q})\}\\
    &- (\lambda p)^{-1} \sum_{g=1}^p \{I_K\otimes (\tilde{\bm{C}}\, \bm{Q})^{\top}\} \bm{B}_g(\bm{U})^{\top} \otimes \{ \bm{A}_g(\bm{U})\bm{y}_g\bm{y}_g^{\top} \bm{B}_g(\bm{U}) \} \bm{\Pi} \{I_K\otimes (\tilde{\bm{C}}\, \bm{Q})\}\\
    &- (\lambda p)^{-1} \sum_{g=1}^p \{I_K\otimes (\tilde{\bm{C}}\, \bm{Q})^{\top}\}\{\bm{B}_g(\bm{U})^{\top}\bm{y}_g\bm{y}_g^{\top} \bm{B}_g(\bm{U}) \}\otimes \bm{A}_g(\bm{U})\{I_K\otimes (\tilde{\bm{C}}\, \bm{Q})\}\\
    &- (\lambda p)^{-1} \sum_{g=1}^p \{I_K\otimes (\tilde{\bm{C}}\, \bm{Q})^{\top}\}\{ \bm{B}_g(\bm{U})^{\top}\bm{y}_g\bm{y}_g^{\top} \bm{A}_g(\bm{U}) \} \otimes \bm{B}_g(\bm{U}) \bm{\Pi}\{I_K\otimes (\tilde{\bm{C}}\, \bm{Q})\}
\end{aligned}
\end{align}
\begin{align*}
    \bm{A}_g(\bm{U}) = \bm{P}_{g}^{\perp} - \bm{P}_{g}^{\perp}\bm{U}(\bm{U}^{\top}\bm{P}_g^{\perp} \bm{U})^{-1} \bm{U}^{\top} \bm{P}_g^{\perp}, \quad \bm{B}_g(\bm{U}) = \bm{P}_{g}^{\perp}\bm{U}(\bm{U}^{\top}\bm{P}_g^{\perp} \bm{U})^{-1},
\end{align*}
where $\bm{\Pi} \in \mathbb{R}^{nK \times nK}$ is a permutation matrix that satisfies $\bm{\Pi}\myvec(\bm{U}) = \myvec(\bm{U}^{\top})$ for $\bm{U} \in \mathbb{R}^{n \times K}$. We next prove a series of lemmas that facilitate understanding $\tilde{\bm{s}}$ and $\tilde{\bm{H}}$, and will lead to an exact expression for $\hat{\bm{C}} - \tilde{\bm{C}}$.

\begin{lemma}[First term in \eqref{supp:equation:Htilde}]
\label{supp:Lemma:Htilde1}
Define $\bm{H}_{g}^{(1)}(\bm{U}) = (\bm{U}^{\top}\bm{P}_g^{\perp} \bm{U})^{-1} \otimes \{\bm{A}_g(\bm{U}) \bm{y}_g\bm{y}_g^{\top} \bm{A}_g(\bm{U})\}$, suppose Assumption~\ref{supp:assumptions:FA} holds, let $\eta \in (0,1/2)$, and let $\hat{\bm{C}}$ be as defined in Corollary~\ref{supp:corollary:Consistency}. Then for $\delta = O_P(n^{-\eta})$ and any constant $\epsilon \in (0,\eta)$,
\begin{align*}
    \sup_{\bm{U} \in \Omega_{\delta}} \norm*{ (\lambda p)^{-1}\sum_{g=1}^p \bm{H}_{g}^{(1)}(\bm{U}) }_2 = O_P(n^{-2\eta+\epsilon} + \lambda^{-1+\epsilon}).
\end{align*}
\end{lemma}
\begin{proof}
By Lemma~\ref{supp:lemma:UtPU}, $(\bm{U}^{\top}\bm{P}_g^{\perp} \bm{U})^{-1} \preceq c\{1+o_P(1)\} I_K$ for some constant $c >0$. Therefore,
\begin{align*}
    (\lambda p)^{-1}\sum_{g=1}^p \bm{H}_{g}^{(1)}(\bm{U}) \preceq c\{1+o_P(1)\} I_K \otimes \left[(\lambda p)^{-1} \sum_{g=1}^p \{\bm{A}_g(\bm{U}) \bm{y}_g\bm{y}_g^{\top} \bm{A}_g(\bm{U})\} \right].
\end{align*}
First,
\begin{align}
\label{supp:equation:Agyg}
\begin{aligned}
    \bm{A}_g(\bm{U}) \bm{y}_g =& \bm{A}_g(\bm{U}) \tilde{\bm{C}}\tilde{\bm{\ell}}_g + \bm{A}_g(\bm{U}) \bm{e}_g\\
    \bm{A}_g(\bm{U}) \tilde{\bm{C}}\tilde{\bm{\ell}}_g =& \bm{P}_g^{\perp} \tilde{\bm{C}}\tilde{\bm{\ell}}_g - \bm{P}_g^{\perp}\bm{U}(\bm{U}^{\top} \bm{P}_g^{\perp} \bm{U})^{-1} \bm{U}^{\top} \bm{P}_g^{\perp} \tilde{\bm{C}}\tilde{\bm{\ell}}_g\\
    \bm{A}_g(\bm{U}) \bm{e}_g =& \bm{P}_g^{\perp} \bm{e}_g - \bm{P}_g^{\perp}\bm{U}(\bm{U}^{\top} \bm{P}_g^{\perp} \bm{U})^{-1} \bm{U}^{\top} \bm{P}_g^{\perp} \bm{e}_g.
\end{aligned}
\end{align}
By the proof of Lemma~\ref{supp:lemma:EtE},
\begin{align*}
    \norm*{ (\lambda p)^{-1}\sum_{g=1}^p \bm{P}_g^{\perp} \bm{e}_g \bm{e}_g^{\top} \bm{P}_g^{\perp} }_2 = O_P(\lambda^{-1+\epsilon}).
\end{align*}
Next, since $\delta = O_P(n^{-\eta})$, it is straightforward to show that
\begin{align*}
    \sup_{\bm{U}}\max_{g \in [p]}\bm{U}^{T} (\bm{P}_g^{\perp})^2 \bm{U} \leq c\{1 + o_P(1)\}
\end{align*}
for some constant $c>0$, which implies 
\begin{align*}
    &\norm*{ (\lambda p)^{-1}\sum_{g=1}^p \bm{P}_g^{\perp}\bm{U}(\bm{U}^{\top} \bm{P}_g^{\perp} \bm{U})^{-1} \bm{U}^{\top} \bm{P}_g^{\perp} \bm{e}_g  \bm{e}_g^{\top} \bm{P}_g^{\perp} \bm{U} (\bm{U}^{\top} \bm{P}_g^{\perp} \bm{U})^{-1} \bm{U}^{T} \bm{P}_g^{\perp}}_2 \\
    \leq & (\lambda p)^{-1}\sum_{g=1}^p \Tr \{ \bm{P}_g^{\perp}\bm{U}(\bm{U}^{\top} \bm{P}_g^{\perp} \bm{U})^{-1} \bm{U}^{\top} \bm{P}_g^{\perp} \bm{e}_g  \bm{e}_g^{\top} \bm{P}_g^{\perp} \bm{U} (\bm{U}^{\top} \bm{P}_g^{\perp} \bm{U})^{-1} \bm{U}^{T} \bm{P}_g^{\perp} \}\\
    =& (\lambda p)^{-1}\sum_{g=1}^p \Tr \{ \bm{U}^{\top} \bm{P}_g^{\perp} \bm{e}_g  \bm{e}_g^{\top} \bm{P}_g^{\perp} \bm{U} (\bm{U}^{\top} \bm{P}_g^{\perp} \bm{U})^{-1} \bm{U}^{T} (\bm{P}_g^{\perp})^2 \bm{U}(\bm{U}^{\top} \bm{P}_g^{\perp} \bm{U})^{-1} \} \\
    \leq & c\{1+o_P(1)\}(\lambda p)^{-1} \sum_{g=1}^p \Tr \{ \bm{U}^{\top} \bm{P}_g^{\perp} \bm{e}_g  \bm{e}_g^{\top} \bm{P}_g^{\perp} \bm{U}\}.
\end{align*}
where the $o_P(1)$ error term is uniform over all $\bm{U}\in \Omega_{\delta}$. The proof of Lemma~\ref{supp:lemma:EtE} shows that
\begin{align*}
    \sup_{\bm{U}\in \Omega_{\delta}} (\lambda p)^{-1} \sum_{g=1}^p \Tr \{ \bm{U}^{\top} \bm{P}_g^{\perp} \bm{e}_g  \bm{e}_g^{\top} \bm{P}_g^{\perp} \bm{U}\} = O_P(\lambda^{-1+\epsilon})
\end{align*}
and implies
\begin{align*}
    \sup_{\bm{U}\in \Omega_{\delta}} \norm*{ (\lambda p)^{-1} \sum_{g=1}^p \bm{A}_g(\bm{U})\bm{e}_g\bm{e}_g^{\top} \bm{A}_g(\bm{U}) }_2 = O_P(\lambda^{-1+\epsilon}).
\end{align*}
For the remaining term in \eqref{supp:equation:Agyg}, let $\bm{U} = \tilde{\bm{C}}\bm{v}_u + \tilde{\bm{z}}_u$, where $\tilde{\bm{z}}_u \in \ker([\bm{X}, \tilde{\bm{C}}]^{\top})$ and $\bm{v}_u^{\top} \bm{v}_u + \tilde{\bm{z}}_u^{\top}\tilde{\bm{z}}_u = I_K$. By Lemma~\ref{supp:lemma:Vz}, $\norm*{ \bm{v}_u^{\top}\bm{v}_u-I_K }_2 \leq c \delta^2$ and $\norm*{ \tilde{\bm{z}}_u }_2 \leq c \delta$ for some constant $c>0$. Therefore, since $\delta= O_P(n^{-\eta})$ and by the proof of \eqref{supp:equation:Bbound} of Lemma~\ref{supp:lemma:Bs},
\begin{align*}
    \sup_{\bm{U}\in \Omega_{\delta}} \max_{g \in [p]} \norm*{ \bm{U}^{\top}\bm{P}_g^{\perp} \bm{U} - \bm{v}_u^{\top} \tilde{\bm{C}}^{\top} \bm{P}_g^{\perp} \tilde{\bm{C}}\bm{v}_u }_2 \leq c \delta\{1+o_P(1)\}
\end{align*}
for some constant $c>0$. Therefore,
\begin{align*}
    &(\bm{U}^{\top}\bm{P}_g^{\perp} \bm{U})^{-1} \bm{U}^{\top}\bm{P}_g^{\perp} \tilde{\bm{C}} = (\bm{U}^{\top}\bm{P}_g^{\perp} \bm{U})^{-1} \bm{v}_u^{\top} \tilde{\bm{C}}^{\top}\bm{P}_g^{\perp} \tilde{\bm{C}} + (\bm{U}^{\top}\bm{P}_g^{\perp} \bm{U})^{-1} \tilde{\bm{z}}_u^{\top} \bm{P}_g^{\perp} \tilde{\bm{C}}\\
    &\Rightarrow \sup_{\bm{U} \in \Omega_{\delta}}\max_{g \in [p]}\norm*{ (\bm{U}^{\top}\bm{P}_g^{\perp} \bm{U})^{-1} \bm{U}^{\top}\bm{P}_g^{\perp} \tilde{\bm{C}} - \bm{v}_u^{-1} }_2 = O_P(n^{-\eta}).
\end{align*}
Consequently,
\begin{align*}
    \sup_{\bm{U} \in \Omega_{\delta}}\max_{g \in [p]}\norm*{ \bm{A}_g(\bm{U})\tilde{\bm{C}} }_2 \leq O_P(n^{-\eta}) + \sup_{\bm{U} \in \Omega_{\delta}}\max_{g \in [p]}\norm{ \bm{P}_g^{\perp} \tilde{\bm{z}}_u }_2 = O_P(n^{-\eta +\epsilon})
\end{align*}
for any $\epsilon \in (0,\eta)$, which completes the proof.
\end{proof}

\begin{lemma}[Third term in \eqref{supp:equation:Htilde}]
\label{supp:Lemma:Htilde3}
Suppose the assumptions of Lemma~\ref{supp:Lemma:Htilde1} hold, let $\eta \in (0,1/2)$, and let $\Omega_{\delta}$ be as defined in Lemma~\ref{supp:lemma:f2}. Then if $\delta=O_P(n^{-\eta})$, there exists a unitary matrix $\bm{v} = \bm{v}(\bm{U}) \in \mathbb{R}^{K \times K}$ that depends on $\bm{U} \in \Omega_{\delta}$ such that
\begin{align*}
    \sup_{\bm{U} \in \Omega_{\delta}}\norm*{ (\lambda p)^{-1}\sum_{g=1}^p \{ \bm{B}_g(\bm{U})^{\top}\bm{y}_g \bm{y}_g^{\top} \bm{B}_g(\bm{U}) \} \otimes \bm{A}_g(\bm{U}) - (\lambda^{-1}\bm{v}^{\top}\bm{\Lambda}\bm{v}) \otimes \bm{Q}\bm{Q}^{\top} }_2 = O_P(n^{-\eta+\epsilon} + \lambda^{-1/2+\epsilon})
\end{align*}
for any constant $\epsilon > 0$.
\end{lemma}
\begin{proof}
We see that
\begin{align}
\label{supp:equation:ByyB}
\begin{aligned}
    \bm{B}_g(\bm{U})^{\top}\bm{y}_g\bm{y}_g^{\top} \bm{B}_g(\bm{U}) =& \bm{B}_g(\bm{U})^{\top} \tilde{\bm{C}}\tilde{\bm{\ell}}_g \tilde{\bm{\ell}}_g^{\top}\tilde{\bm{C}}^{\top} \bm{B}_g(\bm{U}) + \bm{B}_g(\bm{U})^{\top} \bm{e}_g \bm{e}_g^{\top} \bm{B}_g(\bm{U})\\
    &+\bm{B}_g(\bm{U})^{\top}\bm{e}_g \tilde{\bm{\ell}}_g^{\top}\tilde{\bm{C}}^{\top} \bm{B}_g(\bm{U}) + \{ \bm{B}_g(\bm{U})^{\top}\bm{e}_g \tilde{\bm{\ell}}_g^{\top}\tilde{\bm{C}}^{\top} \bm{B}_g(\bm{U}) \}^{\top}.
\end{aligned}
\end{align}
First, since $\bm{A}_g(\bm{U}) \preceq \bm{W}_g$,
\begin{align*}
    \{ \bm{B}_g(\bm{U})^{\top} \bm{e}_g \bm{e}_g^{\top} \bm{B}_g(\bm{U}) \} \otimes \bm{A}_g(\bm{U}) \preceq \{ \bm{B}_g(\bm{U})^{\top} \bm{e}_g \bm{e}_g^{\top} \bm{B}_g(\bm{U}) \} \otimes \bm{W}_g,
\end{align*}
where
\begin{align*}
    &\norm*{ (\lambda p)^{-1}\sum_{g=1}^p \{ \bm{B}_g(\bm{U})^{\top} \bm{e}_g \bm{e}_g^{\top} \bm{B}_g(\bm{U}) \} \otimes \bm{W}_g }_2 \leq \max_{i \in [n]} (\lambda p)^{-1}\sum_{g=1}^p w_{gi}\Tr\{ \bm{B}_g(\bm{U})^{\top} \bm{e}_g \bm{e}_g^{\top} \bm{B}_g(\bm{U}) \}.
\end{align*}
Identical techniques to those used to prove Lemma~\ref{supp:lemma:EtE} can be used to show that
\begin{align*}
    \sup_{\bm{U} \in \Omega_{\delta}}\max_{i \in [n]} (\lambda p)^{-1}\sum_{g=1}^p w_{gi}\Tr\{ \bm{B}_g(\bm{U})^{\top} \bm{e}_g \bm{e}_g^{\top} \bm{B}_g(\bm{U}) \}= O_P(\lambda^{-1+\epsilon/2} \max_{(g,i) \in [p] \times [n]}w_{gi}) &= O_P(\lambda^{-1+\epsilon}).
\end{align*}
Next, it is straightforward to show that for any $\epsilon \in (0,\eta)$,
\begin{align*}
    \sup_{\bm{U} \in \Omega_{\delta}}\max_{g \in [p]} \norm*{ \bm{B}_g(\bm{U})^{\top} \tilde{\bm{C}}\tilde{\bm{\ell}}_g \tilde{\bm{\ell}}_g^{\top}\tilde{\bm{C}}^{\top} \bm{B}_g(\bm{U}) - \bm{v}_u^{-1} \tilde{\bm{\ell}}_g \tilde{\bm{\ell}}_g^{\top}\bm{v}_u^{-\top} }_2 = O_P(n^{-\eta + \epsilon}\lambda).
\end{align*}
Since $\norm*{ \bm{A}_g(\bm{U}) }_2 \leq \norm*{ \bm{W}_g }_2$, this implies
\begin{align*}
    &\sup_{\bm{U} \in \hat{\Omega}} \norm*{ (\lambda p)^{-1} \sum_{g=1}^p \{ \bm{B}_g(\bm{U})^{\top} \tilde{\bm{C}}\tilde{\bm{\ell}}_g \tilde{\bm{\ell}}_g^{\top}\tilde{\bm{C}}^{\top} \bm{B}_g(\bm{U}) \} \otimes \bm{A}_g(U) - (\lambda p)^{-1} \sum_{g=1}^p ( \bm{v}_u^{-1} \tilde{\bm{\ell}}_g \tilde{\bm{\ell}}_g^{\top}\bm{v}_u^{-\top} ) \otimes \bm{A}_g(U) }_2\\
    =& O_P(n^{-\eta + \epsilon}),
\end{align*}
where
\begin{align*}
    (\lambda p)^{-1} \sum_{g=1}^p ( \bm{v}_u^{-1} \tilde{\bm{\ell}}_g \tilde{\bm{\ell}}_g^{\top}\bm{v}_u^{-\top} ) \otimes \bm{A}_g(U) = (\bm{v}_u^{-1}\otimes I_n ) (\lambda p)^{-1}\sum_{g=1}^p  (\tilde{\bm{\ell}}_g \tilde{\bm{\ell}}_g^{\top}) \otimes \bm{A}_g(U) (\bm{v}_u^{-\top}\otimes I_n ),
\end{align*}
and for $\tilde{\bm{s}}_g = \tilde{\bm{\ell}}_g \tilde{\bm{\ell}}_g^{\top}$,
\begin{align*}
    (\lambda p)^{-1}\sum_{g=1}^p  \tilde{\bm{s}}_g \otimes \bm{A}_g(U) =& (\lambda p)^{-1}\sum_{g=1}^p  \tilde{\bm{s}}_g \otimes \bm{P}_g^{\perp} - (\lambda p)^{-1}\sum_{g=1}^p  \tilde{\bm{s}}_g \otimes \{\bm{P}_g^{\perp}\bm{U}( \bm{U}^{\top} \bm{P}_g^{\perp} \bm{U})^{-1} \bm{U}^{\top} \bm{P}_g^{\perp}\}\\
    =& (\lambda p)^{-1}\sum_{g=1}^p  \tilde{\bm{s}}_g \otimes \bm{W}_g - (\lambda p)^{-1}\sum_{g=1}^p  \tilde{\bm{s}}_g \otimes \{\bm{W}_g\bm{X}( \bm{X}^{\top} \bm{W}_g \bm{X})^{-1} \bm{X}^{\top} \bm{W}_g\}\\
    -& (\lambda p)^{-1}\sum_{g=1}^p  \tilde{\bm{s}}_g \otimes \{\bm{P}_g^{\perp}\bm{U}( \bm{U}^{\top} \bm{P}_g^{\perp} \bm{U})^{-1} \bm{U}^{\top} \bm{P}_g^{\perp}\}.
\end{align*}
Since $\delta = O_P(n^{-\eta})$ and by Lemma~\ref{supp:lemma:Bs},
\begin{align*}
    &\sup_{\bm{U} \in \Omega_{\delta}}\norm*{ (\lambda p)^{-1}\sum_{g=1}^p  \tilde{\bm{s}}_g \otimes \{\bm{P}_g^{\perp}\bm{U}( \bm{U}^{\top} \bm{P}_g^{\perp} \bm{U})^{-1} \bm{U}^{\top} \bm{P}_g^{\perp}\} - (\lambda p)^{-1}\sum_{g=1}^p  \tilde{\bm{s}}_g \otimes \{\bm{P}_g^{\perp}\tilde{\bm{C}} \tilde{\bm{C}}^{\top} \bm{P}_g^{\perp}\} }_2\\
    =& O_P(n^{-\eta +\epsilon})
\end{align*}
for any $\epsilon \in (0,\eta)$. Next, for $\bm{R} = n^{-1}\bm{C}^{\top}P_X^{\perp} \bm{C}$ and $\bm{\Lambda} = np^{-1}\bm{L}^{\top}\bm{L}$,
\begin{align*}
    (\lambda p)^{-1}\sum_{g=1}^p  \tilde{\bm{s}}_g \otimes \bm{W}_g = \bm{R}^{1/2} \otimes I_n \{p^{-1} \sum_{g=1}^p (\lambda^{-1}\bm{\ell}_g \bm{\ell}_g^{\top}) \otimes \bm{W}_g \} \bm{R}^{1/2} \otimes I_n
\end{align*}
such that $\norm*{p^{-1} \sum_{g=1}^p (\lambda^{-1}\bm{\ell}_{g_r}\bm{\ell}_{g_s}) (\bm{W}_g-I_n)}_2 = O_P(n^{-1/2+\epsilon})$ by Corollary~\ref{supp:corollary:MaximalIne}, which implies
\begin{align*}
    \norm*{ (\lambda p)^{-1}\sum_{g=1}^p  \tilde{\bm{s}}_g \otimes \bm{W}_g - (\lambda^{-1}\bm{\Lambda}) \otimes I_n }_2 = O_P(n^{-1/2+\epsilon}).
\end{align*}
We also have that since $\max_{g \in [p]}\norm*{ n^{-1/2}\bm{X}^{\top}\bm{W}_g \tilde{\bm{C}} }_2 = O_P(n^{-1/2+\epsilon})$ and $\max_{g \in [p]}\norm*{ n^{-1}\bm{X}^{\top}\bm{W}_g \bm{X} - n^{-1}\bm{X}^{\top}\bm{X} }_2 = O_P(n^{-1/2+\epsilon})$,
\begin{align*}
    \norm*{ (\lambda p)^{-1}\sum_{g=1}^p  \bm{s}_g \otimes (\bm{W}_g\tilde{\bm{C}} \tilde{\bm{C}}^{\top} \bm{W}_g) - (\lambda p)^{-1}\sum_{g=1}^p  \tilde{\bm{s}}_g \otimes (\bm{P}_g^{\perp}\tilde{\bm{C}} \tilde{\bm{C}}^{\top} \bm{P}_g^{\perp}) }_2 = O_P(n^{-1/2+\epsilon}), \quad \bm{s}_g = \bm{\ell}_g \bm{\ell}_g^{\top}.
\end{align*}
For any $r,s \in [K]$, define $\bm{M}^{(rs)} = (\lambda p)^{-1}\sum_{g=1}^p  \bm{s}_{g_{rs}}(\bm{W}_g\tilde{\bm{C}} \tilde{\bm{C}}^{\top} \bm{W}_g)$. Then
\begin{align*}
    \bm{M}^{(rs)}_{ij} =& \tilde{\bm{C}}_{i\bigcdot}^{\top}  \tilde{\bm{C}}_{j\bigcdot}p^{-1}\sum_{g=1}^p (\bm{s}_{g_{rs}}/\lambda)(w_{gi}-1)(w_{gj}-1) + \tilde{\bm{C}}_{i\bigcdot}^{\top}  \tilde{\bm{C}}_{j\bigcdot}p^{-1}\sum_{g=1}^p (\bm{s}_{g_{rs}}/\lambda)(w_{gi}-1)\\
    & + \tilde{\bm{C}}_{i\bigcdot}^{\top}  \tilde{\bm{C}}_{j\bigcdot}p^{-1}\sum_{g=1}^p (\bm{s}_{g_{rs}}/\lambda)(w_{gj}-1) + \tilde{\bm{C}}_{i\bigcdot}^{\top}  \tilde{\bm{C}}_{j\bigcdot}p^{-1}\sum_{g=1}^p (\bm{s}_{g_{rs}}/\lambda), \quad i,j \in [n].
\end{align*}
Therefore,
\begin{align*}
    \bm{M}^{(rs)} &= \sum_{k=1}^K \diag(\tilde{\bm{C}}_{\bigcdot k}) \{p^{-1}\bm{W}^{\top} \bm{S}^{(rs)} \bm{W}\} \diag(\tilde{\bm{C}}_{\bigcdot k}) + \bar{\bm{C}}^{(rs)} \tilde{\bm{C}} + \{ \bar{\bm{C}}^{(rs)} \tilde{\bm{C}} \}^{\top} + \bm{\Lambda}_{rs}\tilde{\bm{C}}\tilde{\bm{C}}^{\top}\\
    \bm{W}_{gi} &= w_{gi} - 1, \quad \bar{\bm{C}}^{(rs)}_{i \bigcdot } = p^{-1}\sum_{g=1}^p (\bm{s}_{g_{rs}}/\lambda)(w_{gi}-1)\tilde{\bm{C}}_{i \bigcdot}, \quad g \in [p]; i \in [n]\\
    \bm{S}^{(rs)} &= \diag( \bm{s}_{1_{rs}}/\lambda, \ldots, \bm{s}_{p_{rs}}/\lambda ).
\end{align*}
First, for some constant $c>0$,
\begin{align*}
    \norm*{ \diag(\tilde{\bm{C}}_{\bigcdot k}) \{p^{-1}\bm{W}^{\top} \bm{S}^{(rs)} \bm{W}\} \diag(\tilde{\bm{C}}_{\bigcdot k}) }_2 \leq c\norm*{p^{-1/2}\bm{W}}_2^2 \max_{i \in [n]} \tilde{\bm{C}}_{i k}^2 = O_P(n^{-1+\epsilon})
\end{align*}
by Lemma~\ref{supp:lemma:LatalaW}. Next, it is easy to see that $\norm*{\bar{\bm{C}}^{(rs)} \tilde{\bm{C}}}_2 \leq \norm*{\bar{\bm{C}}^{(rs)}}_2 = O_P(n^{-1/2+\epsilon})$. Lastly, since the entries of $\bm{X}$ are uniformly bounded, identical techniques can be used to show that
\begin{align*}
    \norm*{ (\lambda p)^{-1}\sum_{g=1}^p  \tilde{\bm{s}}_g \otimes \{\bm{W}_g\bm{X}( \bm{X}^{\top} \bm{W}_g \bm{X})^{-1} \bm{X}^{\top} \bm{W}_g\} - \bm{\Lambda} \otimes P_X }_2 = O_P(n^{-1/2+\epsilon}).
\end{align*}
Putting this all together implies
\begin{align*}
    \sup_{\bm{U} \in \Omega_{\delta}}\norm*{ (\lambda p)^{-1}\sum_{g=1}^p \{ \bm{B}_g(\bm{U})^{\top}\tilde{\bm{C}}\tilde{\bm{\ell}}_g \tilde{\bm{\ell}}_g^{\top} \tilde{\bm{C}}^{\top} \bm{B}_g(\bm{U}) \} \otimes \bm{A}_g(\bm{U}) - (\lambda^{-1}\bm{v}_u^{-1}\bm{\Lambda}\bm{v}_u^{-\top}) \otimes \bm{Q}\bm{Q}^{\top} }_2 = O_P(n^{-\eta+\epsilon}),
\end{align*}
where $\norm*{ \bm{v}_u - \bm{v} }_2 = O(\delta^2) = O_P(n^{-2\eta})$ for some unitary matrix $\bm{v} \in \mathbb{R}^{K \times K}$ by Lemma~\ref{supp:lemma:Vz}.

For the remaining two terms in \eqref{supp:equation:ByyB}, we note that for
\begin{align*}
    \bm{S}(\bm{U}) = [\bm{B}_1(\bm{U})^{\top}\tilde{\bm{C}}\tilde{\bm{\ell}}_1 \cdots \bm{B}_p(\bm{U})^{\top}\tilde{\bm{C}}\tilde{\bm{\ell}}_p],\, \bm{T}(\bm{U}) = [\bm{B}_1(\bm{U})^{\top}\bm{e}_1 \cdots \bm{B}_p(\bm{U})^{\top}\bm{e}_p] \in \mathbb{R}^{K \times p},
\end{align*}
\begin{align*}
    \norm*{ (\lambda p)^{-1}\sum_{g=1}^p \{ \bm{B}_g(\bm{U})^{\top}\tilde{\bm{C}}\tilde{\bm{\ell}}_g \bm{e}_g^{\top} \bm{B}_g(\bm{U}) \} \otimes \bm{A}_g(\bm{U}) }_2 \leq & \norm*{(\lambda p)^{-1}\bm{S}(\bm{U})\bm{S}(\bm{U})^{\top}}_2^{1/2}\\
    &\times \norm*{(\lambda p)^{-1}\bm{T}(\bm{U})\bm{T}(\bm{U})^{\top}}_2^{1/2}.
\end{align*}
Our above work shows that $\norm*{(\lambda p)^{-1}\bm{S}(\bm{U})\bm{S}(\bm{U})^{\top}}_2^{1/2} = O_P(1)$ and $\norm*{(\lambda p)^{-1}\bm{T}(\bm{U})\bm{T}(\bm{U})^{\top}}_2^{1/2} = O_P(n^{-\eta + \epsilon} + \lambda^{-1/2+\epsilon})$, which completes the proof.
\end{proof}

\begin{lemma}[Second and fourth terms of \eqref{supp:equation:Htilde}]
\label{supp:Lemma:Htilde24}
Suppose the assumptions of Lemma~\ref{supp:Lemma:Htilde1} hold, let $\eta \in (0,1/2)$, and let $\Omega_{\delta}$ be as defined in Lemma~\ref{supp:lemma:f2}. Then if $\delta=O_P(n^{-\eta})$,
\begin{align*}
    \sup_{\bm{U} \in \Omega_{\delta}}\norm*{ (\lambda p)^{-1}\sum_{g=1}^p\bm{B}_g(\bm{U})^{\top} \otimes \{ \bm{A}_g(\bm{U})\bm{y}_g\bm{y}_g^{\top} \bm{B}_g(\bm{U}) \} \bm{\Pi} }_2 = O_P(n^{-\eta + \epsilon} + \lambda^{-1/2+\epsilon})
\end{align*}
for any constant $\epsilon>0$.
\end{lemma}
\begin{proof}
Since $\bm{\Pi}$ is a permutation matrix,
\begin{align*}
    \norm*{ (\lambda p)^{-1}\sum_{g=1}^p\bm{B}_g(\bm{U})^{\top} \otimes \{ \bm{A}_g(\bm{U})\bm{y}_g\bm{y}_g^{\top} \bm{B}_g(\bm{U}) \} \bm{\Pi} }_2\leq \norm*{ (\lambda p)^{-1}\sum_{g=1}^p\bm{B}_g(\bm{U})^{\top} \otimes \{ \bm{A}_g(\bm{U})\bm{y}_g\bm{y}_g^{\top} \bm{B}_g(\bm{U}) \} }_2.
\end{align*}
By the definition of $\bm{B}_g(\bm{U})$,
\begin{align*}
    &\bm{B}_g(\bm{U})^{\top} \otimes \{ \bm{A}_g(\bm{U})\bm{y}_g\bm{y}_g^{\top} \bm{B}_g(\bm{U})\} \\
    =& [(\bm{U}^{\top}\bm{P}_g^{\perp} \bm{U})^{-1/2} \otimes \{\bm{A}_g(\bm{U})\bm{y}_g\}] [ \{\bm{P}_g^{\perp}\bm{U}(\bm{U}^{\top}\bm{P}_g^{\perp} \bm{U})^{-1/2}\} \otimes \{\bm{B}_g(\bm{U})^{\top}\bm{y}_g\} ]^{\top}.
\end{align*}
Define
\begin{align*}
    \bm{S}(\bm{U}) &= \begin{pmatrix} (\bm{U}^{\top}\bm{P}_1^{\perp} \bm{U})^{-1/2} \otimes \{\bm{A}_1(\bm{U})\bm{y}_1\} & \cdots & (\bm{U}^{\top}\bm{P}_p^{\perp} \bm{U})^{-1/2} \otimes \{\bm{A}_p(\bm{U})\bm{y}_p\} \end{pmatrix}\\
    \bm{T}(\bm{U}) &= \begin{pmatrix} \{\bm{P}_1^{\perp}\bm{U}(\bm{U}^{\top}\bm{P}_1^{\perp} \bm{U})^{-1/2}\} \otimes \{\bm{B}_1(\bm{U})^{\top}\bm{y}_1\} & \cdots & \{\bm{P}_p^{\perp}\bm{U}(\bm{U}^{\top}\bm{P}_p^{\perp} \bm{U})^{-1/2}\} \otimes \{\bm{B}_p(\bm{U})^{\top}\bm{y}_p\} \end{pmatrix},
\end{align*}
where $\bm{S}(\bm{U}), \bm{T}(\bm{U}) \in \mathbb{R}^{nK \times pK}$ and $\sum_{g=1}^p\bm{B}_g(\bm{U})^{\top} \otimes \{ \bm{A}_g(\bm{U})\bm{y}_g\bm{y}_g^{\top} \bm{B}_g(\bm{U})\} = \bm{S}(\bm{U}) \{\bm{T}(\bm{U})\}^{\top}$. Therefore,
\begin{align*}
   &\norm*{  (\lambda p)^{-1}\sum_{g=1}^p\bm{B}_g(\bm{U})^{\top} \otimes \{ \bm{A}_g(\bm{U})\bm{y}_g\bm{y}_g^{\top} \bm{B}_g(\bm{U}) \} }_2\\
   \leq & \norm*{ (\lambda p)^{-1}\bm{S}(\bm{U})\{\bm{S}(\bm{U})\}^{\top} }_2^{1/2}\norm*{ (\lambda p)^{-1}\bm{T}(\bm{U})\{\bm{T}(\bm{U})\}^{\top} }_2^{1/2},
\end{align*}
where $\sup_{\bm{U} \in \Omega_{\delta}}\norm*{ (\lambda p)^{-1}\bm{S}(\bm{U})\{\bm{S}(\bm{U})\}^{\top} }_2 = O_P(n^{-2\eta + \epsilon} + \lambda^{-1 + \epsilon})$ by Lemma~\ref{supp:Lemma:Htilde1}. We also see that
\begin{align*}
    (\lambda p)^{-1}\bm{T}(\bm{U})\{\bm{T}(\bm{U})\}^{\top} = (\lambda p)^{-1} \sum_{g=1}^p \{\bm{P}_g^{\perp}\bm{U}(\bm{U}^{\top} \bm{P}_g^{\perp}\bm{U})^{-\top}\bm{U}^{\top} \bm{P}_g^{\perp}\} \otimes \{\bm{B}_g(\bm{U})^{\top} \bm{y}_g \bm{y}_g^{\top} \bm{B}_g(\bm{U})\},
\end{align*}
where the same techniques used to prove Lemma~\ref{supp:Lemma:Htilde3} can be used to show that 
\begin{align*}
    \sup_{\bm{U} \in \Omega_{\delta}}\norm*{ (\lambda p)^{-1} \sum_{g=1}^p \{\bm{P}_g^{\perp}\bm{U}(\bm{U}^{\top} \bm{P}_g^{\perp}\bm{U})^{-1}\bm{U}^{\top} \bm{P}_g^{\perp}\} \otimes \{\bm{B}_g(\bm{U})^{\top} \bm{y}_g \bm{y}_g^{\top} \bm{B}_g(\bm{U})\} }_2 = O_P(1),
\end{align*}
which completes the proof.
\end{proof}

\begin{corollary}
\label{supp:corollary:Htilde}
Let $\tilde{\bm{H}}(\bm{U})$ be as defined in \eqref{supp:equation:Htilde}, $\bm{\Lambda} = np^{-1}\bm{L}^{\top}\bm{L}$, let $\eta \in (0,1/2)$, and let $\Omega_{\delta}$ be as defined in Lemma~\ref{supp:lemma:f2}. Then if the assumptions of Lemma~\ref{supp:Lemma:Htilde1} hold and $\delta=O_P(n^{-\eta})$, there exists a unitary matrix $\bm{v} = \bm{v}(\bm{U}) \in \mathbb{R}^{K \times K}$ for each $\bm{U} \in \Omega_{\delta}$ such that $\sup_{\bm{U} \in \Omega_{\delta}} \norm*{ \tilde{\bm{H}}(\bm{U}) + (\lambda^{-1}\bm{v}^{\top}\bm{\Lambda}\bm{v}) \otimes (\bm{0}_{K \times K}\oplus I_{n-d-K}) }_2 = O_P(n^{-\eta+\epsilon} + \lambda^{-1/2+\epsilon})$ for any constant $\epsilon>0$.
\end{corollary}
\begin{proof}
This follows directly from Lemmas~\ref{supp:Lemma:Htilde1}, \ref{supp:Lemma:Htilde3}, and \ref{supp:Lemma:Htilde24}.
\end{proof}
\begin{remark}
\label{supp:remark:ConstructV}
We can construct $\bm{v} = \bm{v}(\bm{U})$ using the following procedure. For $\bm{U} \in \Omega_{\delta}$ , let $\bm{v}_u$ be as defined in Lemma~\ref{supp:lemma:Vz}, and let $\bm{v}_u = \bm{A}_u \bm{\Sigma}_u \bm{B}_u^{\top}$ be its singular value decomposition. By the proof of Lemma~\ref{supp:Lemma:Htilde3}, Corollary~\ref{supp:corollary:Htilde} holds with $\bm{v}$ replaced with $\bm{v}_u^{-\top}$. Since $\norm*{ I_K - \bm{\Sigma}_u }_2 = O(\delta^2)$ by the proof of Lemma~\ref{supp:lemma:Vz}, Corollary~\ref{supp:corollary:Htilde} holds with $\bm{v} = \bm{A}_u\bm{B}_u^{\top}$.
\end{remark}

\begin{lemma}[First term in \eqref{supp:equation:stildeC}]
\label{supp:lemma:s11}
Under the assumptions of Lemma~\ref{supp:Lemma:Htilde1} and for any $\epsilon \in (0,1/2)$,
\begin{align*}
    \norm*{ (\lambda p)^{-1}\sum_{g=1}^p \bm{P}_{g}^{\perp}\tilde{\bm{C}}(\tilde{\bm{C}}^{\top}\bm{P}_g^{\perp} \tilde{\bm{C}})^{-1} \tilde{\bm{C}}^{\top} \bm{P}_g^{\perp}\bm{e}_g\tilde{\bm{\ell}}_g^{\top} }_2 = O_P(\lambda^{-1 + \epsilon}).
\end{align*}
\end{lemma}
\begin{proof}
Without loss of generality, we may assume $n^{-1}\bm{X}^{\top}\bm{X}=I_d$. Then
\begin{align*}
    \tilde{\bm{C}}^{\top} \bm{P}_g^{\perp}\bm{e}_g = \tilde{\bm{C}}^{\top} \bm{W}_g \bm{e}_g - \tilde{\bm{C}}^{\top}\bm{W}_g\bm{X}(\bm{X}^{\top}\bm{W}_g\bm{X})^{-1}\bm{X}^{\top}\bm{W}_g\bm{e}_g.
\end{align*}
For $\bm{R} = n^{-1}\bm{C}^{\top} P_{X}^{\perp} \bm{C}$,
\begin{align*}
    \tilde{\bm{C}}^{\top} \bm{W}_g \bm{e}_g = \bm{R}^{-1/2}\{n^{-1/2}\bm{C}^{\top} \bm{e}_g + n^{-1/2}\bm{C}^{\top}(\bm{W}_g - I_n) \bm{e}_g\} - \bm{R}^{-1/2} (n^{-1}\bm{C}^{\top}\bm{X})(n^{-1/2}\bm{X}^{\top} \bm{W}_g \bm{e}_g),
\end{align*}
where $\max_{g \in [p]}\norm*{ n^{-1/2}\bm{C}^{\top} \bm{e}_g }_2 = O_P(n^{\epsilon})$ by Lemma~\ref{supp:lemma:Cte} and $\max_{g \in [p]}\norm*{n^{-1/2}\bm{C}^{\top}(\bm{W}_g - I_n) \bm{e}_g}_2 = O_P(n^{\epsilon})$ by Lemma~\ref{supp:lemma:Powera}. An identical argument can be used to show that $\max_{g \in [p]}\norm*{n^{-1/2}\bm{X}^{\top}\allowbreak \bm{W}_g \allowbreak \bm{e}_g}_2 = O_P(n^{\epsilon})$, which implies $\max_{g \in [p]}\norm*{ \tilde{\bm{C}}^{\top} \bm{W}_g \bm{e}_g }_2 = O_P(n^{\epsilon})$. A similar argument can be used to show that 
\begin{align*}
    \max_{g \in [p]}\norm*{ \tilde{\bm{C}}^{\top}\bm{W}_g\bm{X}(\bm{X}^{\top}\bm{W}_g\bm{X})^{-1}\bm{X}^{\top}\bm{W}_g\bm{e}_g }_2 = O_P(n^{-1/2+\epsilon}).
\end{align*}
Putting all this together implies that $\max_{g \in [p]}\norm*{ \tilde{\bm{C}}^{\top} \bm{W}_g \bm{e}_g }_2 = O_P(n^{\epsilon})$, which by \eqref{supp:equation:Abound} in Lemma~\ref{supp:lemma:Bs} and the fact that $\max_{g \in [p]}\norm*{\tilde{\bm{C}}^{\top} (\bm{P}_g^{\perp})^2 \tilde{\bm{C}}}_2 \leq c\{1+o_P(1)\}$, further implies
\begin{align*}
    \norm*{ (\lambda p)^{-1}\sum_{g=1}^p\bm{P}_{g}^{\perp}\tilde{\bm{C}}(\tilde{\bm{C}}^{\top}\bm{P}_g^{\perp} \tilde{\bm{C}})^{-1} \tilde{\bm{C}}^{\top} \bm{P}_g^{\perp}\bm{e}_g\tilde{\bm{\ell}}_g^{\top} - (\lambda p)^{-1}\sum_{g=1}^p\bm{P}_{g}^{\perp}\tilde{\bm{C}} \tilde{\bm{C}}^{\top} \bm{P}_g^{\perp}\bm{e}_g\tilde{\bm{\ell}}_g^{\top} }_2 &= O_P(n^{-1/2+\epsilon}\lambda^{-1/2})\\
    &= O_P(\lambda^{-1+\epsilon}).
\end{align*}
Next,
\begin{align}
\label{supp:equation:Stilde1X}
\begin{aligned}
    \bm{P}_{g}^{\perp}\tilde{\bm{C}} \tilde{\bm{C}}^{\top} \bm{P}_g^{\perp}\bm{e}_g\tilde{\bm{\ell}}_g^{\top} =& \bm{W}_g\tilde{\bm{C}} \tilde{\bm{C}}^{\top} \bm{P}_g^{\perp}\bm{e}_g\tilde{\bm{\ell}}_g^{\top}\\
    & - (n^{-1/2}\bm{W}_g\bm{X})( n^{-1}\bm{X}^{\top}\bm{W}_g\bm{X} )^{-1}(n^{-1/2}\bm{X}^{\top} \bm{W}_g\tilde{\bm{C}}) \tilde{\bm{C}}^{\top} \bm{P}_g^{\perp}\bm{e}_g\tilde{\bm{\ell}}_g^{\top}.
\end{aligned}
\end{align}
Starting with the second term in \eqref{supp:equation:Stilde1X},
\begin{align*}
    &\norm*{ (n^{-1/2}\bm{W}_g\bm{X})( n^{-1}\bm{X}^{\top}\bm{W}_g\bm{X} )^{-1}(n^{-1/2}\bm{X}^{\top} \bm{W}_g\tilde{\bm{C}}) \tilde{\bm{C}}^{\top} \bm{P}_g^{\perp}\bm{e}_g\tilde{\bm{\ell}}_g^{\top} }_2\\
     \leq &\norm*{n^{-1/2}\bm{W}_g\bm{X}}_2 \norm*{ ( n^{-1}\bm{X}^{\top}\bm{W}_g\bm{X} )^{-1} }_2\norm*{ n^{-1/2}\bm{X}^{\top} \bm{W}_g\tilde{\bm{C}} }_2\norm*{\tilde{\bm{C}}^{\top} \bm{P}_g^{\perp}\bm{e}_g}_2 \norm*{ \tilde{\bm{\ell}}_g }_2,
\end{align*}
where Lemma~\ref{supp:lemma:Powera} and the above derivation of the behavior of $\norm*{\tilde{\bm{C}}^{\top} \bm{P}_g^{\perp}\bm{e}_g}_2$ implies that for some constant $c>0$,
\begin{align*}
    &\max_{g \in [p]}\norm*{n^{-1/2}\bm{W}_g\bm{X}}_2,\, \max_{g \in [p]} \norm*{ ( n^{-1}\bm{X}^{\top}\bm{W}_g\bm{X} )^{-1} }_2\leq c\{1+o_P(1)\}\\
    &\max_{g \in [p]}\norm*{ n^{-1/2}\bm{X}^{\top} \bm{W}_g\tilde{\bm{C}} }_2 = O_P(n^{-1/2+\epsilon}), \quad \max_{g \in [p]} \norm*{\tilde{\bm{C}}^{\top} \bm{P}_g^{\perp}\bm{e}_g}_2 = O_P(n^{\epsilon}).
\end{align*}
Therefore,
\begin{align*}
    \norm*{ (\lambda p)^{-1}\sum_{g=1}^p (n^{-1/2}\bm{W}_g\bm{X})( n^{-1}\bm{X}^{\top}\bm{W}_g\bm{X} )^{-1}(n^{-1/2}\bm{X}^{\top} \bm{W}_g\tilde{\bm{C}}) \tilde{\bm{C}}^{\top} \bm{P}_g^{\perp}\bm{e}_g \tilde{\bm{\ell}}_g^{\top} }_2 = O_P(\lambda^{-1+\epsilon}).
\end{align*}
The first term in \eqref{supp:equation:Stilde1X} can be expressed as
\begin{align*}
    \bm{W}_g\tilde{\bm{C}} \tilde{\bm{C}}^{\top} \bm{P}_g^{\perp}\bm{e}_g\tilde{\bm{\ell}}_g^{\top} =& \bm{W}_g\tilde{\bm{C}} \tilde{\bm{C}}^{\top} \bm{W}_g\bm{e}_g\tilde{\bm{\ell}}_g^{\top}\\
    &- \bm{W}_g\tilde{\bm{C}} (n^{-1/2}\tilde{\bm{C}}^{\top} \bm{W}_g\bm{X})(n^{-1}\bm{X}^{\top}\bm{W}_g\bm{X})^{-1}(n^{-1/2} \bm{X}^{\top}\bm{W}_g\bm{e}_g\tilde{\bm{\ell}}_g^{\top} )\\
    =&\sum_{k=1}^K \bm{W}_g\tilde{\bm{C}}_{\bigcdot k} \tilde{\bm{C}}_{\bigcdot k}^{\top} \bm{W}_g\bm{e}_g\tilde{\bm{\ell}}_g^{\top}\\
    &- \bm{W}_g\tilde{\bm{C}} (n^{-1/2}\tilde{\bm{C}}^{\top} \bm{W}_g\bm{X})(n^{-1}\bm{X}^{\top}\bm{W}_g\bm{X})^{-1}(n^{-1/2} \bm{X}^{\top}\bm{W}_g\bm{e}_g\tilde{\bm{\ell}}_g^{\top} )
\end{align*}
where an identical analysis to the one above can be used to show that
\begin{align*}
    \norm*{ (\lambda p)^{-1}\sum_{g=1}^p \bm{W}_g\tilde{\bm{C}} (n^{-1/2}\tilde{\bm{C}}^{\top} \bm{W}_g\bm{X})(n^{-1}\bm{X}^{\top}\bm{W}_g\bm{X})^{-1}(n^{-1/2} \bm{X}^{\top}\bm{W}_g\bm{e}_g\tilde{\bm{\ell}}_g^{\top} ) }_2 = O_P(\lambda^{-1+\epsilon}).
\end{align*}
Lastly, 
\begin{align}
\label{supp:equation:WCCW}
\begin{aligned}
    (\lambda p)^{-1} \sum_{g=1}^p \bm{W}_g\tilde{\bm{C}}_{\bigcdot k} \tilde{\bm{C}}_{\bigcdot k}^{\top} \bm{W}_g\bm{e}_g\tilde{\bm{\ell}}_g^{\top} =& (\lambda p)^{-1} \sum_{g=1}^p (\bm{W}_g-I_n)\tilde{\bm{C}}_{\bigcdot k} \tilde{\bm{C}}_{\bigcdot k}^{\top} (\bm{W}_g-I_n)\bm{e}_g\tilde{\bm{\ell}}_g^{\top}\\
    &+ (\lambda p)^{-1} \sum_{g=1}^p (\bm{W}_g-I_n)\tilde{\bm{C}}_{\bigcdot k} \tilde{\bm{C}}_{\bigcdot k}^{\top} \bm{e}_g\tilde{\bm{\ell}}_g^{\top}\\
    &+ \tilde{\bm{C}}_{\bigcdot k}(\lambda p)^{-1} \sum_{g=1}^p \tilde{\bm{C}}_{\bigcdot k}^{\top} (\bm{W}_g-I_n)\bm{e}_g\tilde{\bm{\ell}}_g^{\top}\\
    &- \tilde{\bm{C}}_{\bigcdot k}(\lambda p)^{-1} \sum_{g=1}^p \tilde{\bm{C}}_{\bigcdot k}^{\top} \bm{e}_g\tilde{\bm{\ell}}_g^{\top}.
\end{aligned}
\end{align}
Result~\eqref{supp:equation:CtE:2} in Lemma~\ref{supp:lemma:Cte} and Remark~\ref{supp:remark:CtE} imply $\norm*{ \tilde{\bm{C}}_{\bigcdot k}(\lambda p)^{-1} \sum_{g=1}^p \tilde{\bm{C}}_{\bigcdot k}^{\top} \bm{e}_g\tilde{\bm{\ell}}_g^{\top} }_2\allowbreak = \allowbreak O_P\{\allowbreak(\lambda p)^{-1/2}\}$. For the third term in \eqref{supp:equation:WCCW}, we see that
\begin{align*}
    &(\lambda p)^{-1} \sum_{g=1}^p \tilde{\bm{C}}_{\bigcdot k}^{\top} (\bm{W}_g-I_n)\bm{e}_g\tilde{\bm{\ell}}_g^{\top} = (\bm{R}^{-1/2})_{k \bigcdot}^{\top}(\lambda p)^{-1} \sum_{g=1}^p n^{-1/2}\bm{C}^{\top} (\bm{W}_g-I_n)\bm{e}_g(n^{1/2}\bm{\ell}_g)^{\top} \bm{R}^{1/2}\\
    &- (\bm{R}^{-1/2})_{k \bigcdot}^{\top}(n^{-1}\bm{C}_{\bigcdot k}^{\top}\bm{X})(\lambda p)^{-1} \sum_{g=1}^p (n^{-1/2}\bm{X})^{\top} (\bm{W}_g-I_n)\bm{e}_g(n^{1/2}\bm{\ell}_g)^{\top} \bm{R}^{1/2}.
\end{align*}
Since $\V\{(\bm{W}_g-I_n)\bm{e}_g\}$ is a diagonal matrix with uniformly bounded diagonal entries,
\begin{align*}
    \norm*{\bm{R}^{-1/2}(n^{-1}\bm{C}_{\bigcdot k}^{\top}\bm{X})(\lambda p)^{-1} \sum_{g=1}^p (n^{-1/2}\bm{X})^{\top} (\bm{W}_g-I_n)\bm{e}_g(n^{1/2}\bm{\ell}_g)^{\top}}_2 = O_P\{(\lambda p)^{-1/2}\}.
\end{align*}
Next, for $r,s \in [K]$ and some constants $c_1,c_2>0$,
\begin{align*}
    \V\{(\lambda p)^{-1} \sum_{g=1}^p n^{-1/2}\bm{C}_{\bigcdot r}^{\top} (\bm{W}_g-I_n)\bm{e}_g(n^{1/2}\bm{\ell}_{g_s})\} \leq & c_1\lambda^{-1}p^{-2} \sum_{g=1}^p [n^{-1}\sum_{i=1}^n \E\{\bm{C}_{i r}^2(w_{gi}-1)^2\bm{e}_{g_i}^2\} ]\\
    \leq & c_1c_2 (\lambda p)^{-1},
\end{align*}
which implies the third term in \eqref{supp:equation:WCCW} is $O_P\{(\lambda p)^{-1/2}\}$. The $i$th row of the second term in \eqref{supp:equation:WCCW} can be expressed as
\begin{align*}
    &\tilde{\bm{C}}_{i k}\bm{R}^{1/2}\left\{(\lambda p)^{-1}\sum_{g=1}^p (n^{1/2}\bm{\ell}_g)(w_{gi}-1)\bm{e}_g^{\top}(n^{-1/2}\bm{C})\right\} (\bm{R}^{-1/2})_{\bigcdot k}\\
    - &\tilde{\bm{C}}_{i k}\bm{R}^{1/2}\left\{(\lambda p)^{-1}\sum_{g=1}^p (n^{1/2}\bm{\ell}_g)(w_{gi}-1)\bm{e}_g^{\top}(n^{-1/2}\bm{X})\right\} (n^{-1}\bm{X}^{\top}\bm{C})(\bm{R}^{-1/2})_{\bigcdot k}\in \mathbb{R}^K,
\end{align*}
where $\E\{\bm{e}_g^{\top}(n^{-1/2}\bm{C}_{\bigcdot r})^{(2m)}\} \leq c_m$ for some constant $c_m > 0$ that only depends on the integer $m>0$ by Lemma~\ref{supp:lemma:Cte}. As a consequence, Corollary~\ref{supp:corollary:MaximalIne} implies
\begin{align*}
    &\max_{i \in [n]} \norm*{(\lambda p)^{-1}\sum_{g=1}^p (n^{1/2}\bm{\ell}_g)(w_{gi}-1)\bm{e}_g^{\top}(n^{-1/2}\bm{C})}_2= O_P(\lambda^{-1/2}p^{-1/2+\epsilon})= O_P(\lambda^{-1+\epsilon})\\
    &\max_{i \in [n]} \norm*{ (\lambda p)^{-1}\sum_{g=1}^p (n^{1/2}\bm{\ell}_g)(w_{gi}-1)\bm{e}_g^{\top}(n^{-1/2}\bm{X}) }_2 = O_P(\lambda^{-1/2}p^{-1/2+\epsilon})= O_P(\lambda^{-1+\epsilon}),
\end{align*}
which because $\sum_{i=1}^n \tilde{\bm{C}}_{i k}^2 = 1$, proves the second term in \eqref{supp:equation:WCCW} is $O_P(\lambda^{-1+\epsilon})$. We can then express the $i$th row of first term in \eqref{supp:equation:WCCW} as
\begin{align}
\label{supp:equation:Cik2}
\begin{aligned}
    &\tilde{\bm{C}}_{ik}^2 \bm{R}^{1/2} (\lambda p)^{-1}\sum_{g=1}^p (n^{1/2}\bm{\ell}_g) (w_{gi} - 1)^2 \bm{e}_{g_i}\\
    +& \tilde{\bm{C}}_{ik} \bm{R}^{1/2} (\lambda p)^{-1}\sum_{g=1}^p (n^{1/2}\bm{\ell}_g) (w_{gi}-1)\sum_{j \neq i}^n \tilde{\bm{C}}_{jk} \bm{e}_{g_j}(w_{gj}-1).
\end{aligned}
\end{align}
First,
\begin{align*}
    \max_{i \in [n]}\norm*{ (\lambda p)^{-1}\sum_{g=1}^p (n^{1/2}\bm{\ell}_g) (w_{gi} - 1)^2 \bm{e}_{g_i} }_2 = O_P(\lambda^{-1/2})
\end{align*}
and
\begin{align*}
    \tilde{\bm{C}}_{ik} = n^{-1/2}\bm{C}_{i \bigcdot}^{\top} (\bm{R}^{-1/2})_{\bigcdot k} - n^{-1/2} \bm{X}_{i \bigcdot} (n^{-1}\bm{X}^{\top}\bm{C}) (\bm{R}^{-1/2})_{\bigcdot k},
\end{align*}
which implies $\max_{i \in [\tilde{\bm{C}}_{ik}]}\tilde{\bm{C}}_{ik}^2 = O_P(n^{-1+\epsilon})$, and consequently that
\begin{align*}
    \norm*{ \tilde{\bm{C}}_{ik}^2 \bm{R}^{1/2} (\lambda p)^{-1}\sum_{g=1}^p (n^{1/2}\bm{\ell}_g) (w_{gi} - 1)^2 \bm{e}_{g_i} }_2 \leq \lambda^{-1/2}O_P\{(\max_{i \in [n]}n\tilde{\bm{C}}_{ik}^4)^{1/2}\} = O_P(\lambda^{-1 + \epsilon}).
\end{align*}
Finally, the second term in \eqref{supp:equation:Cik2} can be expressed as
\begin{align*}
    &\tilde{\bm{C}}_{ik} \bm{R}^{1/2}\left\{ (\lambda p)^{-1}\sum_{g=1}^p (n^{1/2}\bm{\ell}_g) (w_{gi}-1)\sum_{j \neq i}^n \bm{e}_{g_j}(w_{gj}-1)(n^{-1/2}\bm{C}_{j \bigcdot})^{\top}\right\} (\bm{R}^{-1/2})_{\bigcdot k}\\
    -& \tilde{\bm{C}}_{ik} \bm{R}^{1/2}\left\{ (\lambda p)^{-1}\sum_{g=1}^p (n^{1/2}\bm{\ell}_g) (w_{gi}-1)\sum_{j \neq i}^n \bm{e}_{g_j}(w_{gj}-1)(n^{-1/2}\bm{X}_{j \bigcdot})^{\top}\right\} (n^{-1}\bm{X}^{\top}\bm{C})(\bm{R}^{-1/2})_{\bigcdot k}.
\end{align*}
Since
\begin{align*}
    \E\left[\left\{ \sum_{j \neq i}^n \bm{e}_{g_j}(w_{gj}-1)(n^{-1/2}\bm{C}_{j r}) \right\}^{2m}\right], \, \E\left[\left\{ \sum_{j \neq i}^n \bm{e}_{g_j}(w_{gj}-1)(n^{-1/2}\bm{X}_{j s}) \right\}^{2m}\right] \leq c_m
\end{align*}
for $r \in [K]$, $s \in [d]$, and any integer $m>0$ and constant $c_m$ that only depends on $m$ by the proofs of Lemmas~\ref{supp:lemma:Cte} and \ref{supp:lemma:Powera}, Corollary~\ref{supp:corollary:MaximalIne} implies
\begin{align*}
    &\max_{i \in [n]}\norm*{ (\lambda p)^{-1}\sum_{g=1}^p (n^{1/2}\bm{\ell}_g) (w_{gi}-1)\sum_{j \neq i}^n \bm{e}_{g_j}(w_{gj}-1)(n^{-1/2}\bm{C}_{j \bigcdot})^{\top} }_2 = O_P(\lambda^{-1/2}p^{-1/2+\epsilon}) = O_P(\lambda^{-1+\epsilon})\\
    & \max_{i \in [n]}\norm*{ (\lambda p)^{-1}\sum_{g=1}^p (n^{1/2}\bm{\ell}_g) (w_{gi}-1)\sum_{j \neq i}^n \bm{e}_{g_j}(w_{gj}-1)(n^{-1/2}\bm{X}_{j \bigcdot})^{\top} }_2=O_P(\lambda^{-1/2}p^{-1/2+\epsilon}) = O_P(\lambda^{-1+\epsilon}).
\end{align*}
Since $\sum_{i=1}^n \tilde{\bm{C}}_{ik}^2 = 1$, this shows the first term in \eqref{supp:equation:WCCW} is $O_P(\lambda^{-1+\epsilon})$, and completes the proof.
\end{proof}

\begin{lemma}[First term in \eqref{supp:equation:stildeC}]
\label{supp:lemma:s12}
Suppose the assumptions of Lemma~\ref{supp:Lemma:Htilde1} hold and let $\bm{s}^{(1)} = (\lambda p)^{-1}\sum_{g=1}^p \bm{P}_g^{\perp} \bm{e}_g \tilde{\bm{\ell}}_g^{\top}$. Then for any constant $\epsilon \in (0,1/2)$, $\norm*{ \bm{s}^{(1)}}_2 = O_P(\lambda^{-1/2+\epsilon})$.
\end{lemma}
\begin{proof}
Without loss of generality, assume $n^{-1}\bm{X}^{\top}\bm{X} = I_K$. We can express $\bm{P}_g^{\perp} \bm{e}_g \tilde{\bm{\ell}}_g^{\top}$ as
\begin{align*}
    \bm{P}_g^{\perp} \bm{e}_g \tilde{\bm{\ell}}_g^{\top} = \bm{W}_g \bm{e}_g \tilde{\bm{\ell}}_g^{\top} - n^{-1/2}\bm{W}_g\bm{X}(n^{-1}\bm{X}^{\top}\bm{W}_g\bm{X})^{-1}(n^{-1/2}\bm{X}^{\top}\bm{W}_g\bm{e}_g) \tilde{\bm{\ell}}_g^{\top}.
\end{align*}
The same techniques used to prove Lemma~\ref{supp:lemma:s11} can be used to show
\begin{align*}
    \max_{g \in [p]}\norm*{ n^{-1/2}\bm{W}_g\bm{X}(n^{-1}\bm{X}^{\top}\bm{W}_g\bm{X})^{-1}(n^{-1/2}\bm{X}^{\top}\bm{W}_g\bm{e}_g) }_2 = n^{\epsilon}
\end{align*}
for any $\epsilon \in (0,1/2)$, which implies
\begin{align*}
    \norm*{ (\lambda p)^{-1}\sum_{g=1}^p n^{-1/2}\bm{W}_g\bm{X}(n^{-1}\bm{X}^{\top}\bm{W}_g\bm{X})^{-1}(n^{-1/2}\bm{X}^{\top}\bm{W}_g\bm{e}_g) \tilde{\bm{\ell}}_g^{\top} }_2 = O_P(\lambda^{-1/2 + \epsilon}).
\end{align*}
Next,
\begin{align*}
    \norm*{ (\lambda p)^{-1}\sum_{g=1}^p\bm{W}_g \bm{e}_g \tilde{\bm{\ell}}_g^{\top} }_2 = \norm*{ (\lambda p)^{-1}\sum_{g=1}^p\bm{W}_g \bm{e}_g (n^{1/2}\bm{\ell}_g)^{\top} }_2 O_P(1).
\end{align*}
To prove the results, we therefore only have to show that
\begin{align*}
    \norm*{ (\lambda p)^{-1}\sum_{g=1}^p\bm{W}_g \bm{e}_g (n^{1/2}\bm{\ell}_g)^{\top} }_2 = O_P(\lambda^{-1/2+\epsilon}),
\end{align*}
which follows because for any $k \in [K]$ and some constants $c_1,c_2>0$,
\begin{align*}
    \E\left[\sum_{i=1}^n \{ (\lambda p)^{-1}\sum_{g=1}^p\bm{W}_g \bm{e}_g (n^{1/2}\bm{\ell}_{g_k}) \}_i^2 \right] =& (\lambda p)^{-2}\Tr\{ \V( \sum_{g=1}^p\bm{W}_g \bm{e}_g n^{1/2}\bm{\ell}_{g_k} ) \}\\
    =& (\lambda p)^{-2} \sum_{g=1}^p\Tr[ \E\{ \V(\bm{W}_g \bm{e}_g n^{1/2}\bm{\ell}_{g_k} \mid \bm{C}) \} ]\\
    \leq & c_1 \lambda^{-1}p^{-2} \sum_{g=1}^p \Tr\{\V(\bm{W}_g \bm{e}_g)\}\leq c_1 c_2\lambda^{-1}.
\end{align*}
\end{proof}

\begin{lemma}[Second term in \eqref{supp:equation:stildeC}]
\label{supp:lemma:s2}
Define
\begin{align*}
    \bm{s}^{(2)} = (p\lambda)^{-1}\sum_{g=1}^p \{ \bm{P}_g^{\perp} - \bm{P}_g^{\perp} \tilde{\bm{C}}( \tilde{\bm{C}}^{\top} \bm{P}_g^{\perp}\tilde{\bm{C}} )^{-1}\tilde{\bm{C}}^{\top}  \bm{P}_g^{\perp} \}\bm{e}_g \bm{e}_g^{\top} \bm{P}_g^{\perp} \tilde{\bm{C}}( \tilde{\bm{C}}^{\top} \bm{P}_g^{\perp}\tilde{\bm{C}} )^{-1}.
\end{align*}
Then under the assumptions of Lemma~\ref{supp:Lemma:Htilde1} and for any $\epsilon \in (0,1/2)$, $\norm*{ \bm{s}^{(2)} }_2 = O_P(\lambda^{-1 + \epsilon})$
\end{lemma}
\begin{proof}
Define $\bm{S} = \begin{pmatrix} \bm{A}_1(\tilde{\bm{C}})\bm{e}_1 \cdots \bm{A}_p(\tilde{\bm{C}})\bm{e}_p \end{pmatrix} \in \mathbb{R}^{n \times p}$ and $\bm{T} = \begin{pmatrix} \bm{B}_1(\tilde{\bm{C}})^{\top}\bm{e}_1 \cdots \bm{B}_p(\tilde{\bm{C}})^{\top}\bm{e}_p \end{pmatrix} \in \mathbb{R}^{K \times p}$. Then $\bm{s}^{(2)} = (\lambda p)^{-1}\bm{S} \bm{T}^{\top}$, which implies
\begin{align*}
    \norm*{\bm{s}^{(2)}}_2 \leq \norm*{(\lambda p)^{-1} \bm{S} \bm{S}^{\top}}_2^{1/2} \norm*{(\lambda p)^{-1} \bm{T} \bm{T}^{\top}}_2^{1/2}.
\end{align*}
The proof of Lemma~\ref{supp:Lemma:Htilde1} shows that
\begin{align*}
    \norm*{(\lambda p)^{-1} \bm{S} \bm{S}^{\top}}_2 = \norm*{ (\lambda p)^{-1} \sum_{g=1}^p \bm{A}_g(\tilde{\bm{C}})\bm{e}_g \bm{e}_g^{\top} \bm{A}_g(\tilde{\bm{C}}) }_2 = O_P(\lambda^{-1 + \epsilon})
\end{align*}
and the proof of Lemma~\ref{supp:Lemma:Htilde3} shows that
\begin{align*}
    \norm*{(\lambda p)^{-1} \bm{T} \bm{T}^{\top}}_2 = \norm*{ (\lambda p)^{-1} \sum_{g=1}^p \bm{B}_g(\tilde{\bm{C}})^{\top}\bm{e}_g \bm{e}_g^{\top} \bm{B}_g(\tilde{\bm{C}}) }_2 = O_P(\lambda^{-1 + \epsilon}).
\end{align*}
\end{proof}

\begin{theorem}
\label{supp:theorem:ChatProp}
Suppose the assumptions of Lemma~\ref{supp:Lemma:Htilde1} hold and let $\bm{\Lambda} = np^{-1}\bm{L}^{\top}\bm{L}$, $f$ and $\Omega_{\delta}$ be as defined in \eqref{supp:equation:fMax} and Lemma~\ref{supp:lemma:f2}, respectively, and $\hat{\bm{C}} \in \argmax_{\bm{U} \in \Omega_{\delta}} f(\bm{U})$. Then there exist $\hat{\bm{v}} \in \mathbb{R}^{K \times K}$, $\hat{\bm{z}} \in \mathbb{R}^{(n-d-K) \times K}$, and a unitary matrix $\bm{v} \in \mathbb{R}^{K \times K}$ such that $\hat{\bm{v}}^{\top}\hat{\bm{v}} + \hat{\bm{z}}^{\top}\hat{\bm{z}} = I_K$ and the following hold for any constant $\epsilon \in (0,1/2)$:
\begin{align}
\label{supp:equation:CperpExp}
    &\hat{\bm{C}} = \tilde{\bm{C}}\hat{\bm{v}} + \bm{Q}\hat{\bm{z}}, \quad \norm*{\hat{\bm{v}} - \bm{v}}_2,\norm*{\hat{\bm{z}} - p^{-1}\sum_{g=1}^p \bm{Q}^{\top}\bm{P}_g^{\perp}\bm{e}_g (n^{1/2}\bm{\ell}_g)^{\top}\bm{\Lambda}^{-1}\bm{v}}_2 = O_P(\lambda^{-1+\epsilon}).
\end{align}
\end{theorem}

\begin{proof}
The expression for $\hat{\bm{C}}$ is a direct consequence of Lemma~\ref{supp:lemma:Vz}. By Lemma~\ref{supp:lemma:Vz}, $\norm*{ \hat{\bm{v}} - \bm{v} }_2 = O(\norm*{\hat{\bm{z}}}_2^2)$ for $\bm{v} = \hat{\bm{A}}\hat{\bm{B}}^{\top}$ and $\hat{\bm{A}},\hat{\bm{B}} \in \mathbb{R}^{K \times K}$ the left and right singular vectors of $\hat{\bm{v}}$. For $t\in[0,1]$, let $\hat{\bm{z}}(t) = t\hat{\bm{z}}$, $\hat{\bm{v}}(t)=\bm{v}\{ I_K - \hat{\bm{z}}(t)^{\top}\hat{\bm{z}}(t) \}^{1/2}$, and $\bm{\gamma}(t) = (\hat{\bm{v}}(t)^{\top},\, \hat{\bm{z}}(t)^{\top})^{\top} \in \mathbb{R}^{(n-d) \times K}$. Since $\hat{\bm{v}}^{\top}\hat{\bm{v}} + \hat{\bm{z}}^{\top}\hat{\bm{z}} = I_K$, $\hat{\bm{z}}$, $\hat{\bm{v}} = \bm{v}( I_K - \hat{\bm{z}}^{\top}\hat{\bm{z}} )^{1/2}$, meaning $\bm{\gamma}(0) = (\bm{v}^{\top},\, \bm{0})^{\top}$, $\bm{\gamma}(1) = (\hat{\bm{v}}^{\top},\, \hat{\bm{z}}^{\top})^{\top}$, and for $\hat{\bm{C}}(t) = \bm{C}\hat{\bm{v}}(t) + \bm{Q}\hat{\bm{z}}(t)$,
\begin{align*}
    \norm*{ \bm{\gamma}(t)-\bm{\gamma}(0) }_2 = \norm*{ \hat{\bm{C}}(t) - \bm{C}\bm{v} }_2 \leq c_1\norm*{ \hat{\bm{C}} - \bm{C}\bm{v} }_2 \leq c_2\norm*{ P_{\hat{\bm{C}}} - P_{\bm{C}} }_2, \quad t \in [0,1]
\end{align*}
for some constants $c_1,c_2>0$. By Taylor's Theorem
\begin{align}
    &\bm{0} = \myvec[\tilde{\bm{s}}\{ \bm{\gamma}(1) \}] = \myvec[\tilde{\bm{s}}\{ \bm{\gamma}(0) \}] + \smallint_0^1 \tilde{\bm{H}}\{ \bm{\gamma}(t) \} \nabla_t \myvec\{\bm{\gamma}(t)\} \text{d}t \nonumber\\
    \label{supp:equation:GradGamma}
    &\sum_{t \in [0,1]}\norm{ \nabla_t \myvec\{\bm{\gamma}(t)\} }_2 \leq \norm{ \hat{\bm{z}} }_F[ 1 + \norm{ \hat{\bm{z}} }_F^2\{1+o_P(1)\} ],
\end{align}
where $\norm{ \hat{\bm{z}} }_F = O_P(n^{-\eta})$ for any $\eta \in (0,1/4)$ by Corollary~\ref{supp:corollary:Consistency} and Lemma~\ref{supp:lemma:Vz}. Then by the expression for $\tilde{\bm{s}}$ in \eqref{supp:equation:stildeC},
\begin{align*}
   &\myvec\left\{ \begin{pmatrix} \bm{0}_{K \times K}\\ -\tilde{\bm{s}}_z \end{pmatrix} \right\} = \tilde{\bm{H}}^*\myvec(\hat{\bm{z}}) + \smallint_0^1 [\tilde{\bm{H}}\{ \bm{\gamma}(t) \}-\tilde{\bm{H}}^*] \nabla_t \myvec\{\bm{\gamma}(t)\} \text{d}t + \myvec\left\{ \begin{pmatrix} \bm{0}_{K \times K}\\ \hat{\bm{\xi}} \end{pmatrix} \right\}\\
   &\tilde{\bm{s}}_z = (\lambda p)^{-1}\sum_{g=1}^p \bm{Q}^{\top}\bm{P}_g^{\perp}\bm{e}_g(n^{1/2}\bm{\ell}_g)^{\top}\bm{v}, \quad \tilde{\bm{H}}^* = -(\lambda^{-1}\bm{v}^{\top}\bm{\Lambda}\bm{v}) \otimes (\bm{0}_{K \times K} \oplus I_{n-d-K})
\end{align*}
for any $\epsilon>0$, where $\norm*{\hat{\bm{\xi}}}_2 = O_P(\lambda^{-1+\epsilon})$ by Lemmas~\ref{supp:lemma:s11} and \ref{supp:lemma:s2}. Corollary~\ref{supp:corollary:Htilde} and \eqref{supp:equation:GradGamma} imply
\begin{align*}
    \norm{ \smallint_0^1 [\tilde{\bm{H}}\{ \bm{\gamma}(t) \}-\tilde{\bm{H}}^*] \nabla_t \myvec\{\bm{\gamma}(t)\} \text{d}t }_2 = o_P( \norm*{ \hat{\bm{z}} }_2 ),
\end{align*}
where an application of Lemma~\ref{supp:lemma:s12} then implies $\norm*{ \hat{\bm{z}} }_2= O_P(\lambda^{-1/2 + \epsilon})$ for any $\epsilon > 0$. An application Lemma~\ref{supp:lemma:Vz} and further applications of Corollary~\ref{supp:corollary:Htilde} and \eqref{supp:equation:GradGamma} complete the proof.
\end{proof}

\begin{corollary} 
\label{supp:corollary:Pchat}
Suppose the assumptions of Theorem~\ref{supp:theorem:ChatProp} hold. Then the conclusions of Theorem~\ref{theorem:PxC} hold.
\end{corollary}

\begin{proof}
This is a direct consequence of Theorem~\ref{supp:theorem:ChatProp} and Lemma~\ref{supp:lemma:Vz}.
\end{proof}

\begin{corollary}
\label{supp:corollary:InfNormC}
Suppose the assumptions of Theorem~\ref{supp:theorem:ChatProp} hold. Then $\norm*{ \hat{\bm{C}} - \tilde{\bm{C}}\bm{v} }_{\infty} = O_P(\lambda^{-1 + \epsilon})$ for any $\epsilon > 0$.
\end{corollary}

\begin{proof}
Let $\bm{Z} = n^{-1/2}[\bm{C},\bm{X}]$, $\bar{\bm{\ell}}_g = \bm{v}^{\top}(\lambda^{-1}\bm{\Lambda})^{-1}\bm{\ell}_g$, and $\bm{\Delta} = \hat{\bm{z}} - (\lambda p)^{-1}\sum_{g=1}^p \bm{Q}^{\top}\bm{P}_g^{\perp}\bm{e}_g \bar{\bm{\ell}}_g^{\top}$. Then \eqref{supp:equation:CperpExp} in Theorem~\ref{supp:theorem:ChatProp} implies
\begin{align}
\label{supp:equation:InftyExpand}
\begin{aligned}
    \norm*{ \hat{\bm{C}} - \tilde{\bm{C}}\bm{v} }_{\infty} \leq & \norm*{ \tilde{\bm{C}}(\hat{\bm{v}} - \bm{v}) }_{\infty} + \norm*{\bm{Q}\bm{\Delta}}_{\infty} + \sum_{k=1}^K\norm*{ (\lambda p)^{-1}\sum_{g=1}^p \bm{W}_g \bm{e}_g\bar{\bm{\ell}}_{gk} }_{\infty}\\
    &+ \norm*{\bm{Z}(\bm{Z}^{\top}\bm{Z})^{-1} (\lambda p)^{-1}\sum_{g=1}^p \bm{Z}^{\top}\bm{W}_g \bm{e}_g\bar{\bm{\ell}}_{g}^{\top}}_{\infty}\\
    &+ \norm*{(\lambda p)^{-1}\sum_{g=1}^p (\bm{W}_g-I_n)\bm{X}(\bm{X}^{\top}\bm{W}_g\bm{X})^{-1}\bm{X}^{\top}\bm{W}_g \bm{e}_g\bar{\bm{\ell}}_{g}^{\top}}_{\infty}\\
    &+ \norm*{\bm{Z}(\bm{Z}^{\top}\bm{Z})^{-1}(\lambda p)^{-1}\sum_{g=1}^p \bm{Z}^{\top}(\bm{W}_g-I_n)\bm{X}(\bm{X}^{\top}\bm{W}_g\bm{X})^{-1}\bm{X}^{\top}\bm{W}_g \bm{e}_g\bar{\bm{\ell}}_{g}^{\top}}_{\infty},
\end{aligned}
\end{align}
where since $\tilde{\bm{C}}(\hat{\bm{v}} - \bm{v})$ and $\bm{Q}\bm{\Delta}$ are at most rank $2K$,
\begin{align*}
    &\norm*{ \tilde{\bm{C}}(\hat{\bm{v}} - \bm{v}) }_{\infty} = O\{\norm*{ \tilde{\bm{C}}(\hat{\bm{v}} - \bm{v}) }_{2}\} = O_P(\lambda^{-1+\epsilon}), \quad \norm*{\bm{Q}\bm{\Delta}}_{\infty} = O(\norm*{\bm{Q}\bm{\Delta}}_{2})= O_P(\lambda^{-1+\epsilon})
\end{align*}
for any $\epsilon > 0$ by Theorem~\ref{supp:theorem:ChatProp}. Similarly, since the fourth and sixth matrices to the right of the inequality in \eqref{supp:equation:InftyExpand} are at most rank $K$,
\begin{align*}
    &\norm*{\bm{Z}(\bm{Z}^{\top}\bm{Z})^{-1} (\lambda p)^{-1}\sum_{g=1}^p \bm{Z}^{\top}\bm{W}_g \bm{e}_g\bar{\bm{\ell}}_{g}^{\top}}_{\infty} = O_P\{\norm*{(\lambda p)^{-1}\sum_{g=1}^p \bm{Z}^{\top}\bm{W}_g \bm{e}_g\bar{\bm{\ell}}_{g}^{\top}}_{2}\}\\
    &\norm*{\bm{Z}(\bm{Z}^{\top}\bm{Z})^{-1}(\lambda p)^{-1}\sum_{g=1}^p \bm{Z}^{\top}(\bm{W}_g-I_n)\bm{X}(\bm{X}^{\top}\bm{W}_g\bm{X})^{-1}\bm{X}^{\top}\bm{W}_g \bm{e}_g\bar{\bm{\ell}}_{g}^{\top}}_{\infty}\\
    =& O_P\{ \norm*{(\lambda p)^{-1}\sum_{g=1}^p \bm{Z}^{\top}(\bm{W}_g-I_n)\bm{X}(\bm{X}^{\top}\bm{W}_g\bm{X})^{-1}\bm{X}^{\top}\bm{W}_g \bm{e}_g\bar{\bm{\ell}}_{g}^{\top}}_{2} \}.
\end{align*}
To derive the asymptotic properties of these Euclidean norms, we first see that for some constants $c_1,c_2>0$ and $\tilde{\bm{c}}_i =  (\bm{C}_{i\bigcdot}^{\top}, \bm{X}_{i\bigcdot}^{\top})^{\top}$,
\begin{align*}
    &(\lambda p)^{-1/2}\sum_{g=1}^p \bm{Z}^{\top}\bm{W}_g \bm{e}_g\bar{\bm{\ell}}_{g}^{\top} = (\lambda p)^{-1/2}\sum_{g=1}^p \bm{Z}^{\top} \bm{e}_g\bar{\bm{\ell}}_{g}^{\top} + (\lambda p)^{-1/2}\sum_{g=1}^p \bm{Z}^{\top}(\bm{W}_g-I_n) \bm{e}_g\bar{\bm{\ell}}_{g}^{\top}\\
    &\V\{ (\lambda p)^{-1/2}\sum_{g=1}^p \bm{Z}^{\top}(\bm{W}_g-I_n) \bm{e}_g\bar{\bm{\ell}}_{gk} \} \leq c_1 p^{-1}\sum_{g=1}^p n^{-1}\sum_{i=1}^n \E\{(w_{gi}-1)^2 \bm{e}_{gi}^2\tilde{\bm{c}}_i \tilde{\bm{c}}_i^{\top}\} \preceq c_2 I_K, \quad k \in [K],
\end{align*}
where $\norm*{ (\lambda p)^{-1/2}\sum_{g=1}^p \bm{Z}^{\top} \bm{e}_g\bar{\bm{\ell}}_{g}^{\top} }_2 = O_P(1)$ by Lemma~\ref{supp:lemma:Cte}. This implies the fourth term in \eqref{supp:equation:InftyExpand} is $O_P(\lambda^{-1})$. Next, Lemma~\ref{supp:lemma:Powera} and Corollary~\ref{supp:corollary:MaximalIne} imply that for some constant $c>0$ and any $\epsilon>0$,
\begin{align*}
    \max_{g \in [p]}\norm*{ \bm{Z}^{\top}(\bm{W}_g - I_n)(n^{-1/2}\bm{X}) }_2 =& O_P(n^{-1/2+\epsilon}), \quad \max_{g \in [p]}\norm*\{(\bm{X}^{\top}\bm{W}_g\bm{X})^{-1}\}_2 \leq c\{1+o_P(1)\}\\
    \max_{g \in [p]}\norm*{ (\lambda n)^{-1/2}\bm{X}^{\top}\bm{W}_g \bm{e}_g \bar{\bm{\ell}}_g }_2 \leq & c\{ \max_{g \in [p]}\norm*{ n^{-1/2}\bm{X}^{\top}\bm{e}_g }_2 + \max_{g \in [p]}\norm*{ n^{-1/2}\bm{X}^{\top}(\bm{W}_g-I_n)\bm{e}_g }_2 \}\\
    =& O_P(n^{\epsilon}),
\end{align*}
which implies the sixth term in \eqref{supp:equation:InftyExpand} is $O_P(\lambda^{-1+\epsilon})$ for any $\epsilon>0$. We next consider the third term in \eqref{supp:equation:InftyExpand}. For $i \in [n]$ and $k \in [K]$, the $i$th element of the third vector can be expressed as
\begin{align*}
    &x_{ik} = p^{-1}\lambda^{-1/2}\sum_{g=1}^p w_{gi}\bm{e}_{gi}(\lambda^{-1/2}\bar{\bm{\ell}}_{gk}) = p^{-1/2}\lambda^{-1/2} (p^{-1/2}\sum_{g=1}^p a_{gi} b_{gk})\\
    &a_{gi} = w_{gi}\bm{e}_{gi}, \quad \abs*{b_{gk}} \leq c
\end{align*}
for some constant $c>0$. Since $a_{1i},\ldots,a_{pi}$ are independent and mean 0, Lemma~\ref{supp:lemma:Powera} and Corollary~\ref{supp:corollary:MaximalIne} imply $\max_{i \in [n]}\abs*{x_{ik}} = O_P(\lambda^{-1+\epsilon})$ for any $\epsilon>0$, which implies the third term in \eqref{supp:equation:InftyExpand} is $O_P(\lambda^{-1+\epsilon})$. For the fifth and final term in \eqref{supp:equation:InftyExpand}, assume without loss of generality that $n^{-1}\bm{X}^{\top}\bm{X}=I_d$. Then the fifth term in \eqref{supp:equation:InftyExpand} can be bounded above by
\begin{align}
\label{supp:equation:InfinityExpand:5}
\begin{aligned}
    &\norm*{ (\lambda p)^{-1}\sum_{g=1}^p (\bm{W}_g-I_n)(n^{-1/2}\bm{X})\{ (n^{-1}\bm{X}^{\top}\bm{W}_g\bm{X})^{-1} - I_d \}(n^{-1/2}\bm{X})^{\top}\bm{W}_g \bm{e}_g\bar{\bm{\ell}}_{g}^{\top} }_{\infty}\\
    &+ \sum_{k=1}^K \sum_{j=1}^d \norm*{ (\lambda p)^{-1}\sum_{g=1}^p (\bm{W}_g-I_n)(n^{-1/2}\bm{X}_{\bigcdot j})(n^{-1/2}\bm{X}_{\bigcdot j})^{\top}\bm{W}_g \bm{e}_g\bar{\bm{\ell}}_{gk}^{\top} }_{\infty}.
\end{aligned}
\end{align}
First, since the first matrix is at most rank $K$, and $\max_{g \in [p]}\norm*{ (n^{-1}\bm{X}^{\top}\bm{W}_g\bm{X})^{-1} - I_K }_2 = O_P(n^{-1/2+\epsilon})$ for any $\epsilon > 0$,
\begin{align*}
    &\norm*{ (\lambda p)^{-1}\sum_{g=1}^p (\bm{W}_g-I_n)(n^{-1/2}\bm{X})\{ (n^{-1}\bm{X}^{\top}\bm{W}_g\bm{X})^{-1} - I_K \}(n^{-1/2}\bm{X})^{\top}\bm{W}_g \bm{e}_g\bar{\bm{\ell}}_{g}^{\top} }_{\infty}\\
    \leq & c\norm*{ (\lambda p)^{-1}\sum_{g=1}^p (\bm{W}_g-I_n)(n^{-1/2}\bm{X})\{ (n^{-1}\bm{X}^{\top}\bm{W}_g\bm{X})^{-1} - I_K \}(n^{-1/2}\bm{X})^{\top}\bm{W}_g \bm{e}_g\bar{\bm{\ell}}_{g}^{\top} }_{2} = O_P(\lambda^{-1+\epsilon})
\end{align*}
for some constant $c>0$. Next, for fixed $j \in [d]$ and $k \in [K]$, the $i$th element of the second matrix in \eqref{supp:equation:InfinityExpand:5} can be expressed as
\begin{align*}
    (\lambda n)^{-1/2} p^{-1}\sum_{g=1}^p (w_{gi} - 1)\bm{X}_{ij} (n^{-1/2}\bm{X}_{\bigcdot j}^{\top} \bm{W}_g \bm{e}_g)(\lambda^{-1/2}\bar{\bm{\ell}}_{gk}) = (\lambda n)^{-1/2} a_{ijk}, \quad i \in [n],
\end{align*}
where for some constant $c>0$,
\begin{align*}
    \max_{i \in [n]}\abs*{a_{ijk}} \leq c(\max_{i \in [n], g \in [p]} \abs*{w_{gi} - 1} ) (\max_{g \in [p]} \abs*{n^{-1/2}\bm{X}_{\bigcdot j}^{\top} \bm{W}_g \bm{e}_g} ) = O_P(n^{\epsilon})
\end{align*}
for some any constant $\epsilon > 0$, which completes the proof.
\end{proof}

\begin{corollary}
\label{supp:corollary:CtC}
Suppose the assumptions of Theorem~\ref{supp:theorem:ChatProp} hold and let $\hat{\bm{C}}$ and $\bm{v}$ be as defined in the statement of Theorem~\ref{supp:theorem:ChatProp}. Then for any $\epsilon >0$,
\begin{subequations}
\label{supp:equation:CtPC}
\begin{align}
    \label{supp:equation:CtPC:Inv}
    &\max_{g \in [p]}\norm*{ (\hat{\bm{C}}^{\top} \bm{P}_g^{\perp} \hat{\bm{C}})(\bm{v}^{\top}\tilde{\bm{C}}^{\top} \bm{P}_g^{\perp} \tilde{\bm{C}}\bm{v})^{-1} - I_K }_2 = O_P(\lambda^{-1+\epsilon})\\
    \label{supp:equation:CtPC:Reg}
    &\max_{g \in [p]} \norm*{ \hat{\bm{C}}^{\top} \bm{P}_g^{\perp} \hat{\bm{C}} - I_K }_2 = O_P(n^{-1/2+\epsilon}).
\end{align}
\end{subequations}
\end{corollary}

\begin{proof}
Let $\bm{A}_g = \bm{v}^{\top}\tilde{\bm{C}}^{\top} \bm{P}_g^{\perp} \tilde{\bm{C}}\bm{v}$ and $\hat{\bm{v}}$, $\hat{\bm{z}}$ be as defined in Theorem~\ref{supp:theorem:ChatProp}. We can express $\hat{\bm{C}}^{\top} \bm{P}_g^{\perp} \hat{\bm{C}}$ as
\begin{align*}
    \hat{\bm{C}}^{\top} \bm{P}_g^{\perp} \hat{\bm{C}} =& \hat{\bm{v}}^{\top}\tilde{\bm{C}}^{\top}\bm{P}_g^{\perp}\tilde{\bm{C}} \hat{\bm{v}} + \hat{\bm{z}}^{\top}\bm{Q}^{\top}\bm{P}_g^{\perp}\tilde{\bm{C}} \hat{\bm{v}}+ ( \hat{\bm{z}}^{\top}\bm{Q}^{\top}\bm{P}_g^{\perp}\tilde{\bm{C}} \hat{\bm{v}} )^{\top} + \hat{\bm{z}}^{\top}\bm{Q}^{\top}\bm{P}_g^{\perp} \bm{Q}\hat{\bm{z}}.
\end{align*}
By \eqref{supp:equation:Abound} in Lemma~\ref{supp:lemma:Bs}, $\max_{g \in [p]} \norm*{ \tilde{\bm{C}}^{\top}\bm{P}_g^{\perp}\tilde{\bm{C}} }_2 = 1 + O_P(n^{-1/2 + \epsilon})$ for any $\epsilon > 0$. Therefore,
\begin{align*}
    \max_{g \in [p]}\norm{ \hat{\bm{z}}^{\top}\bm{Q}^{\top}\bm{P}_g^{\perp} \bm{Q}\hat{\bm{z}}\bm{A}_g^{-1} }_2 \leq \norm{ \hat{\bm{z}} }_2^2 \max_{g \in [p]}\norm*{ \bm{A}_g^{-1} }_2 \max_{g \in [p]}\norm{ \bm{W}_g }_2 = O_P(\lambda^{-1+\epsilon})
\end{align*}
for any $\epsilon>0$ by Theorem~\ref{supp:theorem:ChatProp}. Next,
\begin{align*}
    \norm*{ \hat{\bm{v}}^{\top}\tilde{\bm{C}}^{\top}\bm{P}_g^{\perp}\tilde{\bm{C}} \hat{\bm{v}}\bm{A}_g^{-1} - I_K }_2 \leq 2\norm{ \hat{\bm{v}}-\bm{v} }_2\norm*{\bm{A}_g}_2\norm*{\bm{A}_g^{-1}}_2 + \norm{ \hat{\bm{v}}-\bm{v} }_2^2\norm*{\bm{A}_g^{-1}}_2\norm*{\bm{A}_g}_2.
\end{align*}
Since $\norm{ \hat{\bm{v}}-\bm{v} }_2 = O_P(\lambda^{-1+\epsilon})$ for any $\epsilon>0$ by Theorem~\ref{supp:theorem:ChatProp}, $\max_{g \in [p]} \norm*{ \hat{\bm{v}}^{\top}\tilde{\bm{C}}^{\top}\bm{P}_g^{\perp}\tilde{\bm{C}} \hat{\bm{v}}\bm{A}_g^{-1} - I_K }_2 = O_P(\lambda^{-1+\epsilon})$. Next, let $\bm{s} = (\lambda p)^{-1}\sum_{g=1}^p \bm{P}_g^{\perp} \bm{e}_g (n^{1/2}\bm{\ell}_g)^{\top}(\lambda^{-1}\bm{\Lambda})^{-1}$. Then
\begin{align*}
    \norm*{ \hat{\bm{z}}^{\top}\bm{Q}^{\top}\bm{P}_g^{\perp}\tilde{\bm{C}} \hat{\bm{v}} }_2 \leq \norm*{ \bm{s}^{\top}\bm{Q}\bm{Q}^{\top}\bm{P}_g^{\perp}\tilde{\bm{C}} }_2 + \norm*{ \bm{v}^{\top}\bm{s}^{\top}\bm{Q} - \hat{\bm{z}} }_2 \norm*{ \bm{P}_g^{\perp}\tilde{\bm{C}} }_2,
\end{align*}
where $\max_{g \in [p]}\norm*{ \bm{P}_g^{\perp}\tilde{\bm{C}}}_2 \leq c\{1+o_P(1)\}$ for some constant $c>0$ by the proof of Lemma~\ref{supp:lemma:Bs}. Therefore,
\begin{align*}
    \max_{g \in [p]}( \norm*{ \bm{v}^{\top}\bm{s}^{\top}\bm{Q} - \hat{\bm{z}} }_2 \norm*{ \bm{P}_g^{\perp}\tilde{\bm{C}} }_2 ) \leq c\{1+o_P(1)\} \norm*{ \bm{v}^{\top}\bm{s}^{\top}\bm{Q} - \hat{\bm{z}} }_2 = O_P(\lambda^{-1+\epsilon})
\end{align*}
by Theorem~\ref{supp:theorem:ChatProp}. Since an application of Lemma~\ref{supp:lemma:Bs} and \eqref{supp:equation:CtPC:Inv} imply \eqref{supp:equation:CtPC:Reg}, we need only show that $\max_{g \in [p]}\norm*{ \bm{s}^{\top}\bm{Q}\bm{Q}^{\top}\bm{P}_g^{\perp}\tilde{\bm{C}} }_2 = O_P(\lambda^{-1+\epsilon})$ for any $\epsilon>0$ to complete the proof.

Define $\bm{H} = \bm{Q}\bm{Q}^{\top} = P_{[\bm{X},\tilde{\bm{C}}]}^{\perp}$. Then for any $r,t \in [K]$,
\begin{align}
\label{supp:equation:sHPC}
\begin{aligned}
    &\bm{s}_{\bigcdot r}^{\top}\bm{Q}\bm{Q}^{\top}\bm{P}_g^{\perp}\tilde{\bm{C}}_{\bigcdot t} = (\lambda p)^{-1} \bar{\bm{\ell}}_{gr}\bm{e}_g^{\top}\bm{P}_g^{\perp} \bm{H} \bm{P}_g^{\perp}\tilde{\bm{C}}_{\bigcdot t} + (\lambda p)^{-1}\sum_{h \neq g} \bar{\bm{\ell}}_{hr}\bm{e}_h^{\top}\bm{P}_h^{\perp} \bm{H} \bm{P}_g^{\perp}\tilde{\bm{C}}_{\bigcdot t}\\
    &\bar{\bm{\ell}}_{hr} = n^{1/2}\lambda^{-1}\bm{\Lambda}_{r \bigcdot}^{\top}\bm{\ell}_h, \quad h \in [p],    
\end{aligned}
\end{align}
where $\abs*{\bar{\bm{\ell}}_{hr}} \leq c\lambda^{1/2}$ for some constant $c>0$. Therefore,
\begin{align*}
    \max_{g \in [p]} \abs*{ (\lambda p)^{-1} \bar{\bm{\ell}}_{gr}\bm{e}_g^{\top}\bm{P}_g^{\perp} \bm{H} \bm{P}_g^{\perp}\tilde{\bm{C}}_{\bigcdot t} } \leq c (\lambda n)^{-1/2} (\max_{g \in [p]} \norm*{n^{-1/2} \bm{e}_g}_2) (\max_{g \in [p]}\norm*{ \bm{W}_g^2\tilde{\bm{C}}_{\bigcdot t}  }_2)
\end{align*}
for some constant $c>0$. It is easy to see that $\max_{g \in [p]} \norm*{n^{-1/2} \bm{e}_g}_2 = O_P(n^{\epsilon})$ and $\max_{g \in [p]}\allowbreak\norm*{ \bm{W}_g^2\allowbreak\tilde{\bm{C}}_{\bigcdot t}  }_2 = O_P(1)$, where the latter follows from Corollary~\ref{supp:corollary:MaximalIne} and implies 
\begin{align*}
    \max_{g \in [p]} \abs*{ (\lambda p)^{-1} \bar{\bm{\ell}}_{gr}\bm{e}_g^{\top}\bm{P}_g^{\perp} \bm{H} \bm{P}_g^{\perp}\tilde{\bm{C}}_{\bigcdot t} } = O_P(\lambda^{-1+\epsilon})
\end{align*}
for any $\epsilon>0$. For the second term in \eqref{supp:equation:sHPC},
\begin{align}
\label{supp:equation:ePHPC}
\begin{aligned}
    &(\lambda p)^{-1}\sum_{h \neq g} \bar{\bm{\ell}}_{hr}\bm{e}_h^{\top}\bm{P}_h^{\perp} \bm{H} \bm{P}_g^{\perp}\tilde{\bm{C}}_{\bigcdot t} = (\lambda p)^{-1}\sum_{h \neq g} \bar{\bm{\ell}}_{hr} \bm{e}_h^{\top}\bm{W}_h \bm{H} (\bm{W}_g - I_n)\tilde{\bm{C}}_{\bigcdot t}\\
    &- (\lambda p)^{-1}\sum_{h \neq g} \bar{\bm{\ell}}_{hr}\bm{e}_h^{\top}\bm{W}_h \bm{X}(\bm{X}^{\top}\bm{W}_h \bm{X})^{-1}\bm{X}^{\top}(\bm{W}_h-I_n)\bm{H} (\bm{W}_g - I_n)\tilde{\bm{C}}_{\bigcdot t}\\
    &- (\lambda p)^{-1}\sum_{h \neq g} \bar{\bm{\ell}}_{hr}\bm{e}_h^{\top}\bm{W}_h \bm{H}(\bm{W}_g-I_n)\bm{X}(\bm{X}^{\top}\bm{W}_g\bm{X})^{-1}\bm{X}^{\top}(\bm{W}_g-I_n) \tilde{\bm{C}}_{\bigcdot t}\\
    &+ (\lambda p)^{-1}\sum_{h \neq g} \{\bar{\bm{\ell}}_{hr}\bm{e}_h^{\top}\bm{W}_h \bm{X}(\bm{X}^{\top}\bm{W}_h \bm{X})^{-1}\bm{X}^{\top}(\bm{W}_h-I_n)\bm{H} (\bm{W}_g-I_n)\bm{X}(\bm{X}^{\top}\bm{W}_g\bm{X})^{-1}\\
    &\times \bm{X}^{\top}(\bm{W}_g-I_n) \tilde{\bm{C}}_{\bigcdot t}\}.
\end{aligned}
\end{align}
Define $\bm{R}=(n^{-1}\bm{C}^{\top}P_{\bm{X}}^{\perp}\bm{C})^{-1/2}$ and $\bm{Z} = [\bm{C}, \bm{X}]$, where we assume without loss of generality that $n^{-1}\bm{X}^{\top}\bm{X}=I_d$. Then the first term in \eqref{supp:equation:ePHPC} can be expressed as
\begin{align}
\label{supp:equation:ePHPC:1}
\begin{aligned}
    &(\lambda p)^{-1}\sum_{h \neq g} \bar{\bm{\ell}}_{hr} \bm{e}_h^{\top}\bm{W}_h \bm{H} (\bm{W}_g - I_n)\tilde{\bm{C}}_{\bigcdot t} = (\lambda p)^{-1}\sum_{h \neq g} \bar{\bm{\ell}}_{hr} \bm{e}_h^{\top} (\bm{W}_g - I_n)(n^{-1/2}\bm{C})\bm{R}_{\bigcdot t}\\
    &+(\lambda p)^{-1}\sum_{h \neq g} \bar{\bm{\ell}}_{hr} \bm{e}_h^{\top}(\bm{W}_h-I_n) (\bm{W}_g - I_n)(n^{-1/2}\bm{C})\bm{R}_{\bigcdot t}\\
    &-(\lambda p)^{-1}\sum_{h \neq g} \bar{\bm{\ell}}_{hr} \bm{e}_h^{\top} (\bm{W}_g - I_n)(n^{-1/2}\bm{X})(n^{-1}\bm{X}^{\top}\bm{C}\bm{R}_{\bigcdot t})\\
    &-(\lambda p)^{-1}\sum_{h \neq g} \bar{\bm{\ell}}_{hr} \bm{e}_h^{\top}(\bm{W}_h-I_n) (\bm{W}_g - I_n)(n^{-1/2}\bm{X})(n^{-1}\bm{X}^{\top}\bm{C}\bm{R}_{\bigcdot t})\\
    &-(\lambda p)^{-1}\sum_{h \neq g} \bar{\bm{\ell}}_{hr} \bm{e}_h^{\top}\bm{W}_h (n^{-1/2}\bm{Z})(n^{-1}\bm{Z}^{\top}\bm{Z})\{n^{-1} \bm{Z}^{\top}(\bm{W}_g - I_n)P_{\bm{X}}^{\perp}\bm{C} \}\bm{R}_{\bigcdot t}.
\end{aligned}
\end{align}
Define $\bm{x}_g = (\lambda p)^{-1/2}\sum_{h \neq g} \bar{\bm{\ell}}_{hr}\bm{e}_h$. Since $\abs*{\lambda^{-1/2}\bar{\bm{\ell}}_{hr}} \leq c$ for some constant $c>0$, $\E(\bm{x}_{g_i}^{2m}) \leq c_m$ for all $i \in [n]$ for some constant $c_m>0$ that only depends on the positive integer $m$ by Lemma~\ref{supp:lemma:Powera}. Since the rows of $(\bm{W}_g - I_n)(n^{-1/2}\bm{C})$ are mean $\bm{0}$ and independent conditional on $\{\bm{C},\bm{x}_g\}$, Corollary~\ref{supp:corollary:MaximalIne} implies
\begin{align*}
    \max_{g \in [p]} \abs*{ (\lambda p)^{-1/2}\sum_{h \neq g} \bar{\bm{\ell}}_{hr} \bm{e}_h^{\top} (\bm{W}_g - I_n)(n^{-1/2}\bm{C}_{\bigcdot k}) } = O_P(n^{\epsilon}), \quad k \in [K]
\end{align*}
for any $\epsilon > 0$, which then implies
\begin{align*}
    \max_{g \in [p]}\abs*{(\lambda p)^{-1}\sum_{h \neq g} \bar{\bm{\ell}}_{hr} \bm{e}_h^{\top} (\bm{W}_g - I_n)(n^{-1/2}\bm{C})\bm{R}_{\bigcdot t}} = O_P(\lambda^{-1 + \epsilon})
\end{align*}
for any $\epsilon > 0$. Identical analyses and repeated applications of Lemma~\ref{supp:lemma:Powera} and Corollary~\ref{supp:corollary:MaximalIne} can be used to show that the maximum, over $g \in [p]$, absolute value of the remaining four terms in \eqref{supp:equation:ePHPC:1} are all $O_P(\lambda^{-1 + \epsilon})$, which shows that the maximum, over $g \in [p]$, absolute value of the first term in \eqref{supp:equation:ePHPC} is $O_P(\lambda^{-1 + \epsilon})$ for any $\epsilon > 0$. For the second term in \eqref{supp:equation:ePHPC}, we see that for $h \neq g$,
\begin{align}
\label{supp:equation:ePHPC:2}
\begin{aligned}
    &n^{-1/2}\bm{X}^{\top}(\bm{W}_h - I_n)\bm{H}(\bm{W}_g - I_n)\tilde{\bm{C}}_{\bigcdot t} = n^{-1}\bm{X}^{\top}(\bm{W}_h - I_n)(\bm{W}_g - I_n)\bm{C}\bm{R}_{\bigcdot t}\\
    &-\{n^{-1}\bm{X}^{\top}(\bm{W}_h - I_n)(\bm{W}_g - I_n)\bm{X}\}( n^{-1}\bm{X}^{\top} \bm{C}\bm{R}_{\bigcdot t} )\\
    &- \{n^{-1}\bm{X}^{\top}(\bm{W}_h - I_n)\bm{Z}\}(n^{-1}\bm{Z}^{\top}\bm{Z})\{n^{-1}\bm{Z}^{\top}(\bm{W}_g - I_n)\bm{C}\bm{R}_{\bigcdot t}\}\\
    &+ \{n^{-1}\bm{X}^{\top}(\bm{W}_h - I_n)\bm{Z}\}(n^{-1}\bm{Z}^{\top}\bm{Z})\{n^{-1}\bm{Z}^{\top}(\bm{W}_g - I_n)\bm{X}\} (n^{-1}\bm{X}^{\top}\bm{C}\bm{R}_{\bigcdot t}).
\end{aligned}
\end{align}
Since the diagonal entries of $(\bm{W}_h - I_n)(\bm{W}_g - I_n)$ are independent and mean 0, Corollary~\ref{supp:corollary:MaximalIne} implies the terms in \eqref{supp:equation:ePHPC:2} satisfy
\begin{align*}
    &\max_{g \neq h \in [p] \times [p]} \norm*{ n^{-1}\bm{X}^{\top}(\bm{W}_h - I_n)(\bm{W}_g - I_n)\bm{C}\bm{R}_{\bigcdot t} }_2  = O_P(n^{-1/2+\epsilon})\\
    &\max_{g \neq h \in [p] \times [p]}\norm*{ \{n^{-1}\bm{X}^{\top}(\bm{W}_h - I_n)(\bm{W}_g - I_n)\bm{X}\}( n^{-1}\bm{X}^{\top} \bm{C}\bm{R}_{\bigcdot t} ) }_2  = O_P(n^{-1/2+\epsilon})\\
    &\max_{g \neq h \in [p] \times [p]}\norm*{ \{n^{-1}\bm{X}^{\top}(\bm{W}_h - I_n)\bm{Z}\}(n^{-1}\bm{Z}^{\top}\bm{Z})\{n^{-1}\bm{Z}^{\top}(\bm{W}_g - I_n)\bm{C}\bm{R}_{\bigcdot t}\} }_2  = O_P(n^{-1+\epsilon})\\
    &\max_{g \neq h \in [p] \times [p]}\norm*{ \{n^{-1}\bm{X}^{\top}(\bm{W}_h - I_n)\bm{Z}\}(n^{-1}\bm{Z}^{\top}\bm{Z})\{n^{-1}\bm{Z}^{\top}(\bm{W}_g - I_n)\bm{X}\} (n^{-1}\bm{X}^{\top}\bm{C}\bm{R}_{\bigcdot t}) }_2  = O_P(n^{-1+\epsilon}),
\end{align*}
which implies $\max_{g \neq h \in [p] \times [p]}\norm*{ n^{-1/2}\bm{X}^{\top}(\bm{W}_h - I_n)\bm{H}(\bm{W}_g - I_n)\tilde{\bm{C}}_{\bigcdot t} }_2  = O_P(n^{-1/2+\epsilon})$ for any $\epsilon>0$. Next, for some constant $c>0$,
\begin{align*}
    \norm*{ \lambda^{-1/2}\bar{\bm{\ell}}_{hr}n^{-1/2}\bm{e}_h^{\top}\bm{W}_h \bm{X} }_2 \leq c\{ \norm*{ n^{-1/2}\bm{e}_h^{\top}\bm{X} }_2 + \norm*{ n^{-1/2}\bm{e}_h^{\top}(\bm{W}_h - I_n)\bm{X} }_2\},
\end{align*}
where, for any $\epsilon>0$, 
\begin{align*}
\max_{h \in [p]}\norm*{ n^{-1/2}\bm{e}_h^{\top}(\bm{W}_h - I_n)\bm{X} }_2, \, \max_{h \in [p]} \norm*{ n^{-1/2}\bm{e}_h^{\top}\bm{X} }_2 = O_P(n^{\epsilon})   
\end{align*}
by Corollary~\ref{supp:corollary:MaximalIne} and because $\bm{e}_h$ is sub-Gaussian, respectively. Therefore, the second term in \eqref{supp:equation:ePHPC} satisfies
\begin{align*}
    \max_{g \in [p]}\abs*{ (\lambda p)^{-1}\sum_{h \neq g} \bar{\bm{\ell}}_{hr}\bm{e}_h^{\top}\bm{W}_h \bm{X}(\bm{X}^{\top}\bm{W}_h \bm{X})^{-1}\bm{X}^{\top}(\bm{W}_h-I_n)\bm{H} (\bm{W}_g - I_n)\tilde{\bm{C}}_{\bigcdot t} } &= O_P\{(\lambda n)^{-1/2 + \epsilon}\}\\
    &= O_P(\lambda^{-1+\epsilon})
\end{align*}
for any $\epsilon > 0$. Identical techniques to those used to derive the properties of the second term in \eqref{supp:equation:ePHPC} can also be used to show that the maximum, over $g \in [p]$, absolute values of the third and fourth terms in \eqref{supp:equation:ePHPC} are $O_P(\lambda^{-1+\epsilon})$ for any $\epsilon > 0$. The details have been omitted.
\end{proof}

\begin{corollary}
\label{supp:corollary:LtL}
Suppose the assumptions of Theorem~\ref{supp:theorem:ChatProp} hold, let $\hat{\bm{Z}} = [\hat{\bm{C}}, \bm{X}]$ for $\hat{\bm{C}}$ as defined in the statement of Theorem~\ref{supp:theorem:ChatProp}, and define $\hat{\bm{\ell}}_g$ to be the first $K$ elements of the $K+d$ vector $(\hat{\bm{Z}}^{\top}\bm{W}_g \hat{\bm{Z}})^{-1}\hat{\bm{Z}}^{\top}\bm{W}_g \bm{y}_g$. Then $\norm*{ (\lambda p)^{-1}\sum_{g=1}^p \hat{\bm{\ell}}_g \hat{\bm{\ell}}_g^{\top} - (\lambda p)^{-1}\bm{v}^{\top}\sum_{g=1}^p \tilde{\bm{\ell}}_g \tilde{\bm{\ell}}_g^{\top}\bm{v} }_2 = O_P(\lambda^{-1+\epsilon})$ and $\norm*{ (\lambda p)^{-1}\sum_{g=1}^p \hat{\bm{\ell}}_g \tilde{\bm{\ell}}_g^{\top} - (\lambda p)^{-1}\bm{v}^{\top}\sum_{g=1}^p \tilde{\bm{\ell}}_g \tilde{\bm{\ell}}_g^{\top} }_2 = O_P(\lambda^{-1+\epsilon})$ for any $\epsilon > 0$ and $\bm{v}$ as defined in the statement of Theorem~\ref{supp:theorem:ChatProp}.
\end{corollary}

\begin{proof}
Let $\hat{\bm{v}}$ and $\hat{\bm{z}}$ be as defined in Theorem~\ref{supp:theorem:ChatProp}, and let $\bm{s}$ and $\bm{H}$ be as defined in the proof of Corollary~\ref{supp:corollary:CtC}. We can express $\hat{\bm{\ell}}_g$ as
\begin{align}
\label{supp:equation:lhat}
\begin{aligned}
    \hat{\bm{\ell}}_g =& (\hat{\bm{C}}^{\top}\bm{P}_g^{\perp} \hat{\bm{C}})^{-1}\hat{\bm{C}}^{\top}\bm{P}_g^{\perp} \bm{y}_g = (\hat{\bm{C}}^{\top}\bm{P}_g^{\perp} \hat{\bm{C}})^{-1}\hat{\bm{v}}^{\top}\tilde{\bm{C}}^{\top}\bm{P}_g^{\perp} \tilde{\bm{C}}\tilde{\bm{\ell}}_g + (\hat{\bm{C}}^{\top}\bm{P}_g^{\perp} \hat{\bm{C}})^{-1}\hat{\bm{v}}^{\top}\tilde{\bm{C}}^{\top}\bm{P}_g^{\perp} \bm{e}_g\\
    &+(\hat{\bm{C}}^{\top}\bm{P}_g^{\perp} \hat{\bm{C}})^{-1}\hat{\bm{z}}^{\top}\bm{Q}^{\top}\bm{P}_g^{\perp} \tilde{\bm{C}}\tilde{\bm{\ell}}_g + (\hat{\bm{C}}^{\top}\bm{P}_g^{\perp} \hat{\bm{C}})^{-1}\hat{\bm{z}}^{\top}\bm{Q}^{\top}\bm{P}_g^{\perp}\bm{e}_g\\
    =& \bm{v}^{\top}\tilde{\bm{\ell}}_g + \{(\hat{\bm{C}}^{\top}\bm{P}_g^{\perp} \hat{\bm{C}})^{-1}(\hat{\bm{v}}^{\top}\bm{v})(\bm{v}^{\top}\tilde{\bm{C}}^{\top}\bm{P}_g^{\perp} \tilde{\bm{C}}\bm{v}) - I_K\} \bm{v}^{\top}\tilde{\bm{\ell}}_g\\
    &+ \bm{v}^{\top}\tilde{\bm{C}}^{\top}\bm{P}_g^{\perp} \bm{e}_g + \bm{v}^{\top}\{ (\tilde{\bm{C}}^{\top}\bm{P}_g^{\perp} \tilde{\bm{C}})^{-1} - I_K \}\tilde{\bm{C}}^{\top}\bm{P}_g^{\perp} \bm{e}_g\\
    &+ \{ (\hat{\bm{C}}^{\top}\bm{P}_g^{\perp} \hat{\bm{C}})^{-1}\hat{\bm{v}}^{\top} - (\bm{v}^{\top}\tilde{\bm{C}}^{\top}\bm{P}_g^{\perp} \tilde{\bm{C}}\bm{v})^{-1}\bm{v}^{\top} \}\tilde{\bm{C}}^{\top}\bm{P}_g^{\perp} \bm{e}_g\\
    &+(\hat{\bm{C}}^{\top}\bm{P}_g^{\perp} \hat{\bm{C}})^{-1}\bm{v}^{\top}\bm{s}^{\top}\bm{H}\bm{P}_g^{\perp} \tilde{\bm{C}}\tilde{\bm{\ell}}_g + (\hat{\bm{C}}^{\top}\bm{P}_g^{\perp} \hat{\bm{C}})^{-1}\bm{\Delta}^{\top}\bm{Q}^{\top}\bm{P}_g^{\perp} \tilde{\bm{C}}\tilde{\bm{\ell}}_g\\
    &+(\hat{\bm{C}}^{\top}\bm{P}_g^{\perp} \hat{\bm{C}})^{-1}\bm{v}^{\top}\bm{s}^{\top}\bm{H}\bm{P}_g^{\perp} \bm{e}_g + \bm{\Delta}^{\top}\bm{Q}^{\top}\bm{P}_g^{\perp} \bm{e}_g\\
    &+\{(\hat{\bm{C}}^{\top}\bm{P}_g^{\perp} \hat{\bm{C}})^{-1} - I_K\}\bm{\Delta}^{\top}\bm{Q}^{\top}\bm{P}_g^{\perp} \bm{e}_g
\end{aligned}
\end{align}
for $\bm{\Delta} = \hat{\bm{z}} - \bm{Q}^{\top}\bm{s}\bm{v}$. Lemmas~\ref{supp:lemma:Cte} and \ref{supp:lemma:Bs}, Theorem~\ref{supp:theorem:ChatProp}, and Corollary~\ref{supp:corollary:CtC} imply
\begin{align}
\label{supp:equation:ellhatTerms:1}
\begin{aligned}
    &\max_{g \in [p]}\norm*{ \{(\hat{\bm{C}}^{\top}\bm{P}_g^{\perp} \hat{\bm{C}})^{-1}(\hat{\bm{v}}^{\top}\bm{v})(\bm{v}^{\top}\tilde{\bm{C}}^{\top}\bm{P}_g^{\perp} \tilde{\bm{C}}\bm{v}) - I_K\} \bm{v}^{\top}\tilde{\bm{\ell}}_g }_2 = O_P(\lambda^{-1/2+\epsilon})\\
    &\max_{g \in [p]} \norm*{ \bm{v}^{\top}\tilde{\bm{C}}^{\top}\bm{P}_g^{\perp} \bm{e}_g }_2 = O_P(n^{\epsilon})\\
    &\max_{g \in [p]}\norm*{ \bm{v}^{\top}\{ (\tilde{\bm{C}}^{\top}\bm{P}_g^{\perp} \tilde{\bm{C}})^{-1} - I_K \}\tilde{\bm{C}}^{\top}\bm{P}_g^{\perp} \bm{e}_g }_2= O_P(n^{-1/2+\epsilon})\\
    &\max_{g \in [p]}\norm*{ \{ (\hat{\bm{C}}^{\top}\bm{P}_g^{\perp} \hat{\bm{C}})^{-1}\hat{\bm{v}}^{\top} - (\bm{v}^{\top}\tilde{\bm{C}}^{\top}\bm{P}_g^{\perp} \tilde{\bm{C}}\bm{v})^{-1}\bm{v}^{\top} \}\tilde{\bm{C}}^{\top}\bm{P}_g^{\perp} \bm{e}_g }_2 = O_P(\lambda^{-1+\epsilon})\\
    &\max_{g \in [p]}\norm*{ (\hat{\bm{C}}^{\top}\bm{P}_g^{\perp} \hat{\bm{C}})^{-1}\bm{\Delta}^{\top}\bm{Q}^{\top}\bm{P}_g^{\perp} \tilde{\bm{C}}\tilde{\bm{\ell}}_g }_2 = O_P(\lambda^{-1/2+\epsilon})\\
    &\max_{g \in [p]}\norm*{ \bm{\Delta}^{\top}\bm{Q}^{\top}\bm{P}_g^{\perp} \bm{e}_g }_2 = O_P(n^{-\delta})\\
    &\max_{g \in [p]} \norm*{ \{(\hat{\bm{C}}^{\top}\bm{P}_g^{\perp} \hat{\bm{C}})^{-1} - I_K\}\bm{\Delta}^{\top}\bm{Q}^{\top}\bm{P}_g^{\perp} \bm{e}_g }_2 = O_P(\lambda^{-1 + \epsilon})
\end{aligned}
\end{align}
for $\delta > 0$ sufficiently small and any $\epsilon > 0$. The second line follows from the fact that for $\bm{R}=(n^{-1}\bm{C}^{\top}P_{\bm{X}}^{\perp}\bm{C})^{-1/2}$,
\begin{align*}
    \tilde{\bm{C}}^{\top}\bm{P}_g^{\perp} \bm{e}_g =& \bm{R}(n^{-1/2}\bm{C}^{\top}P_{\bm{X}}^{\perp}\bm{e}_g) + \bm{R}\{n^{-1/2}\bm{C}^{\top}(\bm{W}_g - I_n)\bm{e}_g\}\\& - \bm{R} (n^{-1}\bm{C}^{\top}\bm{X})(n^{-1}\bm{X}^{\top}\bm{X})\{n^{-1/2}\bm{X}^{\top}(\bm{W}_g-I_n)\bm{e}_g\}\\
    &- \{n^{-1/2} \tilde{\bm{C}}^{\top}(\bm{W}_g-I_n)\bm{X} \}(n^{-1} \bm{X}^{\top} \bm{W}_g \bm{X} )^{-1}(n^{-1/2} \bm{X}^{\top}\bm{W}_g \bm{e}_g),
\end{align*}
where
\begin{align*}
    &\max_{g \in [p]}\norm*{ \bm{R}(n^{-1/2}\bm{C}^{\top}P_{\bm{X}}^{\perp}\bm{e}_g) }_2 = O_P(  \max_{g \in [p]}\norm*{ n^{-1/2}\bm{C}^{\top}P_{\bm{X}}^{\perp}\bm{e}_g }_2 ) = O_P(n^{\epsilon})\\
    &\max_{g \in [p]}\norm*{ n^{-1/2}\bm{C}^{\top}(\bm{W}_g - I_n)\bm{e}_g }_2, \, \max_{g \in [p]}\norm*{ n^{-1/2}\bm{X}^{\top}(\bm{W}_g - I_n)\bm{e}_g }_2 = O_P(n^{\epsilon})
\end{align*}
by Lemma~\ref{supp:lemma:Cte} and Corollary~\ref{supp:corollary:MaximalIne}, respectively, and because
\begin{align*}
    \max_{g \in [p]} \norm*{ \{n^{-1/2} \tilde{\bm{C}}^{\top}(\bm{W}_g-I_n)\bm{X} \}(n^{-1} \bm{X}^{\top} \bm{W}_g \bm{X} )^{-1}(n^{-1/2} \bm{X}^{\top}\bm{W}_g \bm{e}_g) }_2 = O_P(n^{-1/2 + \epsilon})
\end{align*}
for any $\epsilon>0$. Identical techniques used to prove Corollary~\ref{supp:corollary:CtC} can also be used to show
\begin{align}
\label{supp:equation:ellhatTerms:2}
    \max_{g \in [p]}\norm*{ (\hat{\bm{C}}^{\top}\bm{P}_g^{\perp} \hat{\bm{C}})^{-1}\bm{v}^{\top}\bm{s}^{\top}\bm{H}\bm{P}_g^{\perp} \tilde{\bm{C}}\tilde{\bm{\ell}}_g }_2 = O_P(\lambda^{-1/2+\epsilon})
\end{align}
for any $\epsilon > 0$. Lastly, $\max_{g \in [p]}\norm*{ (\hat{\bm{C}}^{\top}\bm{P}_g^{\perp} \hat{\bm{C}})^{-1}\bm{v}^{\top}\bm{s}^{\top}\bm{H}\bm{P}_g^{\perp} \bm{e}_g }_2 \leq \{1+o_P(1)\}\max_{g \in [p]}\norm*{ \bm{s}^{\top}\bm{H}\bm{P}_g^{\perp} \bm{e}_g }_2$, where for $\bar{\bm{\ell}}_{hr}$ as defined in \eqref{supp:equation:sHPC},
\begin{align}
\label{supp:equation:sHPe}
    \bm{s}_{\bigcdot r}^{\top}\bm{H}\bm{P}_g^{\perp} \bm{e}_g = (\lambda p)^{-1} \bar{\bm{\ell}}_{gr} \bm{e}_g^{\top}\bm{P}_g^{\perp} \bm{H}\bm{P}_g^{\perp}\bm{e}_g + (\lambda p)^{-1}\sum_{h \neq g}\bar{\bm{\ell}}_{hr} \bm{e}_h^{\top}\bm{P}_h^{\perp} \bm{H}\bm{P}_g^{\perp}\bm{e}_g, \quad r \in [K].
\end{align}
We first see that for some constant $c>0$, where
\begin{align*}
    \abs*{ (\lambda p)^{-1} \bar{\bm{\ell}}_{gr} \bm{e}_g^{\top}\bm{P}_g^{\perp} \bm{H}\bm{P}_g^{\perp}\bm{e}_g } \leq c \lambda^{-1/2} \{ \max_{g \in [p]}(p^{-1}\bm{e}_g^{\top}\bm{e}_g) \} (\max_{g \in [p], i \in [n]}w_{gi}^2) = O_P(\lambda^{-1/2 +\epsilon})
\end{align*}
for any $\epsilon > 0$. For $\bm{Z} = [\bm{C},\bm{X}]$, the second term in \eqref{supp:equation:sHPe} can be expressed as
\begin{align}
\label{supp:equation:etPHPe}
\begin{aligned}
    &(\lambda p)^{-1}\sum_{h \neq g}\bar{\bm{\ell}}_{hr} \bm{e}_h^{\top}\bm{P}_h^{\perp} \bm{H}\bm{P}_g^{\perp}\bm{e}_g = (\lambda p)^{-1}\sum_{h \neq g}\bar{\bm{\ell}}_{hr} \bm{e}_h^{\top}\bm{W}_h \bm{W}_g\bm{e}_g\\
    &- (\lambda p)^{-1}\sum_{h \neq g}\bar{\bm{\ell}}_{hr} (n^{-1/2}\bm{e}_h^{\top}\bm{W}_h\bm{Z})(n^{-1}\bm{Z}^{\top}\bm{Z}) (n^{-1/2}\bm{Z}^{\top}\bm{W}_g\bm{e}_g)\\
    &-(\lambda p)^{-1}\sum_{h \neq g}\bar{\bm{\ell}}_{hr} \bm{e}_h^{\top}\bm{W}_h \bm{H}(\bm{W}_g - I_n)\bm{X}(\bm{X}^{\top}\bm{W}_g \bm{X})^{-1}\bm{X}^{\top}\bm{W}_g \bm{e}_g\\
    &-(\lambda p)^{-1}\sum_{h \neq g}\bar{\bm{\ell}}_{hr} \bm{e}_h^{\top}\bm{W}_h \bm{X}(\bm{X}^{\top}\bm{W}_h \bm{X})^{-1}\bm{X}^{\top}(\bm{W}_h - I_n)\bm{H} \bm{W}_g\bm{e}_g\\
    &+(\lambda p)^{-1}\sum_{h \neq g}\bar{\bm{\ell}}_{hr} \bm{e}_h^{\top}\bm{W}_h \bm{X}(\bm{X}^{\top}\bm{W}_h \bm{X})^{-1}\bm{X}^{\top}(\bm{W}_h - I_n)\bm{H} (\bm{W}_g - I_n)\bm{X}(\bm{X}^{\top}\bm{W}_g \bm{X})^{-1}\\
    &\times \bm{X}^{\top}\bm{W}_g \bm{e}_g.
\end{aligned}
\end{align}
First, $\norm*{ n^{-1/2}\bm{e}_h^{\top}\bm{W}_h\bm{Z} }_2 \leq \norm*{ n^{-1/2}\bm{e}_h^{\top}\bm{Z} }_2 + \norm*{ n^{-1/2}\bm{e}_h^{\top}(\bm{W}_h-I_n)\bm{Z} }_2$, where $\max_{g \in [p]}\norm*{ n^{-1/2}\bm{e}_h^{\top}\bm{Z} }_2 = O_P(n^{\epsilon})$ by Lemma~\ref{supp:lemma:Cte} and Corollary~\ref{supp:corollary:MaximalIne} implies $\max_{g \in [p]}\norm*{ n^{-1/2}\bm{e}_h^{\top}(\bm{W}_h-I_n)\bm{Z} }_2 = O_P(n^{\epsilon})$ for any $\epsilon>0$. Therefore, the second term in \eqref{supp:equation:etPHPe} satisfies
\begin{align*}
    \max_{g \in [p]} \abs*{ (\lambda p)^{-1}\sum_{h \neq g}\bar{\bm{\ell}}_{hr} (n^{-1/2}\bm{e}_h^{\top}\bm{W}_h\bm{Z})(n^{-1}\bm{Z}^{\top}\bm{Z}) (n^{-1/2}\bm{Z}^{\top}\bm{W}_g\bm{e}_g) } = O_P(\lambda^{-1/2+\epsilon})
\end{align*}
for any $\epsilon>0$. Identical analyses can be used to show that the the maxima, over $g \in [p]$, absolute values of the third through fifth terms in \eqref{supp:equation:etPHPe} are all $O_P(\lambda^{-1/2+\epsilon})$. The first term in \eqref{supp:equation:etPHPe} can be expressed as
\begin{align*}
    &(\lambda p)^{-1}\sum_{h \neq g}\bar{\bm{\ell}}_{hr} \bm{e}_h^{\top}\bm{W}_h \bm{W}_g\bm{e}_g = (\lambda p)^{-1}\sum_{h \neq g}\bar{\bm{\ell}}_{hr} \bm{e}_h^{\top}\bm{e}_g + (\lambda p)^{-1}\sum_{h \neq g}\bar{\bm{\ell}}_{hr} \bm{e}_h^{\top}(\bm{W}_h-I_n) \bm{e}_g\\
    &+(\lambda p)^{-1}\sum_{h \neq g}\bar{\bm{\ell}}_{hr} \bm{e}_h^{\top}(\bm{W}_g-I_n) \bm{e}_g + (\lambda p)^{-1}\sum_{h \neq g}\bar{\bm{\ell}}_{hr} \bm{e}_h^{\top}(\bm{W}_h-I_n)(\bm{W}_g-I_n) \bm{e}_g.
\end{align*}
Since $\bm{e}_g$ is sub-Gaussian random vector with independent entries and uniformly sub-Gaussian norm, Corollary~\ref{supp:corollary:MaximalIne} implies $\max_{g \in [p]}\abs*{ (\lambda p)^{-1}\sum_{h \neq g}\bar{\bm{\ell}}_{hr} \bm{e}_h^{\top}\bm{e}_g } = O_P(\lambda^{-1/2+\epsilon})$ for any $\epsilon>0$. For the second term, we see that for $c_{hr} = \lambda^{-1/2}\bar{\bm{\ell}}_{hr}$,
\begin{align*}
    x_{gr} = p^{-1}\lambda^{-1/2}\sum_{h \neq g} \bar{\bm{\ell}}_{hr} \bm{e}_h^{\top}(\bm{W}_h-I_n) \bm{e}_g = p^{-1}\sum_{h \neq g}\sum_{i=1}^n c_{hr}\bm{e}_{hi}\bm{e}_{gi}(w_{hi}-1).
\end{align*}
Since $\max_{(h,g)\in[p]\times [p];i \in [n]}\E[\{ c_{hr}\bm{e}_{hi}\bm{e}_{gi}(w_{hi}-1) \}^{2m}]$ is bounded from above by a constant that only depends on $m>0$ and the elements of $\{w_{hi}-1\}_{h\in [p]\setminus \{g\}; i \in [n]}$ are mean 0 and independent conditional on $\{\bm{E},\bm{C}\}$, Lemma~\ref{supp:lemma:Powera} implies $\max_{g \in [p]}\E(x_{gr}^{2m})$ is bounded above by a constant that only depends on $m>0$. Corollary~\ref{supp:corollary:MaximalIne} therefore implies $\max_{g \in [p]}\abs*{ (\lambda p)^{-1}\sum_{h \neq g}\bar{\bm{\ell}}_{hr} \bm{e}_h^{\top}(\bm{W}_h-I_n) \bm{e}_g } = O_P(\lambda^{-1/2+\epsilon})$ for any $\epsilon>0$. Further applications of Lemma~\ref{supp:lemma:Powera} and Corollary~\ref{supp:corollary:MaximalIne} can be used to show
\begin{align*}
    \max_{g \in [p]}\abs*{ (\lambda p)^{-1}\sum_{h \neq g}\bar{\bm{\ell}}_{hr} \bm{e}_h^{\top}(\bm{W}_g-I_n) \bm{e}_g }, \, \max_{g \in [p]}\abs*{ (\lambda p)^{-1}\sum_{h \neq g}\bar{\bm{\ell}}_{hr} \bm{e}_h^{\top}(\bm{W}_h-I_n)(\bm{W}_g-I_n) \bm{e}_g } = O_P(\lambda^{-1/2 + \epsilon})
\end{align*}
for any $\epsilon>0$. This implies the first term in \eqref{supp:equation:etPHPe} satisfies
\begin{align*}
    \max_{g \in [p]}\abs*{ (\lambda p)^{-1}\sum_{h \neq g}\bar{\bm{\ell}}_{hr} \bm{e}_h^{\top}\bm{W}_h \bm{W}_g\bm{e}_g\ } = O_P(\lambda^{-1/2 + \epsilon})
\end{align*}
for any $\epsilon>0$, which gives us that \eqref{supp:equation:sHPe} satisfies
\begin{align}
\label{supp:equation:ellhatTerms:3}
    \max_{g \in [p]}\norm*{ \bm{s}_{\bigcdot r}^{\top}\bm{H}\bm{P}_g^{\perp} \bm{e}_g }_2 = O_P(\lambda^{-1/2+\epsilon})
\end{align}
for any $\epsilon > 0$. The expression for $\hat{\bm{\ell}}_g$ in \eqref{supp:equation:lhat} and the maximal inequalities in \eqref{supp:equation:ellhatTerms:1}, \eqref{supp:equation:ellhatTerms:2}, and \eqref{supp:equation:ellhatTerms:3} imply
\begin{align*}
    \max_{g \in [p]}\lVert &\hat{\bm{\ell}}_g\hat{\bm{\ell}}_g^{\top} - \{\bm{v}^{\top}\tilde{\bm{\ell}}_g\tilde{\bm{\ell}}_g^{\top}\bm{v} + \bm{v}^{\top}\tilde{\bm{C}}^{\top}\bm{P}_g^{\perp}\bm{e}_g \tilde{\bm{\ell}}_g^{\top}\bm{v} + ( \bm{v}^{\top}\tilde{\bm{C}}^{\top}\bm{P}_g^{\perp}\bm{e}_g \tilde{\bm{\ell}}_g^{\top}\bm{v} )^{\top}\\
    &+ \bm{\Delta}^{\top}\bm{Q}^{\top}\bm{P}_g^{\perp}\bm{e}_g\tilde{\bm{\ell}}_g^{\top}\bm{v} + ( \bm{\Delta}^{\top}\bm{Q}^{\top}\bm{P}_g^{\perp}\bm{e}_g \tilde{\bm{\ell}}_g^{\top}\bm{v})^{\top} \}  \rVert_2 = O_P(\lambda^{\epsilon})\\
    \max_{g \in [p]}\lVert &\hat{\bm{\ell}}_g\tilde{\bm{\ell}}_g^{\top} - \{\bm{v}^{\top}\tilde{\bm{\ell}}_g\tilde{\bm{\ell}}_g^{\top} + \bm{v}^{\top}\tilde{\bm{C}}^{\top}\bm{P}_g^{\perp}\bm{e}_g \tilde{\bm{\ell}}_g^{\top} + ( \bm{v}^{\top}\tilde{\bm{C}}^{\top}\bm{P}_g^{\perp}\bm{e}_g \tilde{\bm{\ell}}_g^{\top} )^{\top}\\
    &+ \bm{\Delta}^{\top}\bm{Q}^{\top}\bm{P}_g^{\perp}\bm{e}_g\tilde{\bm{\ell}}_g^{\top} + ( \bm{\Delta}^{\top}\bm{Q}^{\top}\bm{P}_g^{\perp}\bm{e}_g \tilde{\bm{\ell}}_g^{\top})^{\top} \}  \rVert_2 = O_P(\lambda^{\epsilon})
\end{align*}
for any $\epsilon > 0$. To complete the proof, we need only show that $\norm*{ (\lambda p)^{-1}\sum\limits_{g=1}^p \tilde{\bm{C}}^{\top}\bm{P}_g^{\perp}\bm{e}_g \tilde{\bm{\ell}}_g^{\top} }_2$ and $\norm*{ (\lambda p)^{-1}\sum\limits_{g=1}^p \bm{\Delta}^{\top}\bm{Q}^{\top}\bm{P}_g^{\perp}\bm{e}_g \tilde{\bm{\ell}}_g^{\top} }_2$ are both $O_P(\lambda^{-1 + \epsilon})$ for any $\epsilon > 0$. Let $\bm{R} = (n^{-1}\bm{C}P_{\bm{X}}^{\perp}\bm{C})^{1/2}$ and assume $n^{-1}\bm{X}^{\top}\bm{X}=I_d$ without loss of generality. Then for $\bar{\bm{\ell}}_g = \lambda^{-1/2}n^{1/2}\bm{\ell}_g$ and $k \in [K]$,
\begin{align*}
     &\norm*{ (\lambda p)^{-1}\sum\limits_{g=1}^p \tilde{\bm{C}}^{\top}\bm{P}_g^{\perp}\bm{e}_g \tilde{\bm{\ell}}_g^{\top} }_2 \leq \underbrace{\norm*{\bm{R}}_2 \norm*{\bm{R}^{-1}}_2}_{=O_P(1)} \norm*{ (\lambda p)^{-1}\sum\limits_{g=1}^p (n^{-1/2}\bm{C})^{\top}P_{\bm{X}}^{\perp}\bm{P}_g^{\perp}\bm{e}_g (n^{1/2}\bm{\ell}_g)^{\top} }_2
\end{align*}
\begin{align}
\label{supp:equation:CtPPe}
\begin{aligned}
    &(\lambda p)^{-1}\sum\limits_{g=1}^p (n^{-1/2}\bm{C})^{\top}P_{\bm{X}}^{\perp}\bm{P}_g^{\perp}\bm{e}_g (n^{1/2}\bm{\ell}_{gk}) = \lambda^{-1/2}p^{-1}\sum\limits_{g=1}^p (n^{-1/2}\bm{C})^{\top}P_{\bm{X}}^{\perp}\bm{W}_g\bm{e}_g \bar{\bm{\ell}}_{gk}\\
     &- \lambda^{-1/2}p^{-1}\sum\limits_{g=1}^p \{n^{-1}\bm{C}^{\top}P_{\bm{X}}^{\perp}(\bm{W}_g-I_n)\bm{X}\}(n^{-1}\bm{X}^{\top}\bm{W}_g\bm{X})^{-1}(n^{-1/2}\bm{X}^{\top}\bm{W}_g\bm{e}_g) \bar{\bm{\ell}}_{gk}
\end{aligned}
\end{align}
where
\begin{align}
\label{supp:equation:CtPWe}
\begin{aligned}
    &\lambda^{-1/2}p^{-1}\sum\limits_{g=1}^p (n^{-1/2}\bm{C})^{\top}P_{\bm{X}}^{\perp}\bm{W}_g\bm{e}_g \bar{\bm{\ell}}_{gk} = \lambda^{-1/2}p^{-1}\sum\limits_{g=1}^p (n^{-1/2}\bm{C})^{\top}P_{\bm{X}}^{\perp}\bm{e}_g \bar{\bm{\ell}}_{gk}\\
    &+ \lambda^{-1/2}p^{-1}\sum\limits_{g=1}^p (n^{-1/2}\bm{C})^{\top}(\bm{W}_g - I_n)\bm{e}_g \bar{\bm{\ell}}_{gk}\\
    &-  (n^{-1}\bm{C}^{\top}\bm{X})\lambda^{-1/2}p^{-1}\sum\limits_{g=1}^p (n^{-1/2}\bm{X})^{\top}(\bm{W}_g - I_n)\bm{e}_g \bar{\bm{\ell}}_{gk}.
\end{aligned}
\end{align}
The first term is $O_P(\lambda^{-1/2}p^{-1/2})= O_P(\lambda^{-1})$ by Lemma~\ref{supp:lemma:Cte} and Remark~\ref{supp:remark:CtE}. For the second term,
\begin{align*}
    \V\{ p^{-1/2}\sum\limits_{g=1}^p (n^{-1/2}\bm{C})^{\top}(\bm{W}_g - I_n)\bm{e}_g \bar{\bm{\ell}}_{gk} \} = p^{-1}\sum\limits_{g=1}^p \bar{\bm{\ell}}_{gk}^2 n^{-1}\sum_{i=1}^n \E\{ (w_{gi}-1)^2 \bm{e}_{gi}^2\bm{C}_{i \bigcdot}\bm{C}_{i \bigcdot}^{\top} \} \preceq c I_K
\end{align*}
for some $c > 0$, meaning the second term in \eqref{supp:equation:CtPWe} is $O_P(\lambda^{-1})$. An identical analysis shows that the third term in \eqref{supp:equation:CtPWe} is also $O_P(\lambda^{-1})$, which proves the first term in \eqref{supp:equation:CtPPe} is $O_P(\lambda^{-1})$. For the second term in \eqref{supp:equation:CtPPe}, we first see that
\begin{align*}
    &n^{-1}\bm{C}^{\top}P_{\bm{X}}^{\perp}(\bm{W}_g-I_n)\bm{X} = n^{-1}\bm{C}^{\top}(\bm{W}_g-I_n)\bm{X} - (n^{-1}\bm{C}^{\top}\bm{X})\{ n^{-1}\bm{X}^{\top}(\bm{W}_g-I_n)\bm{X} \}\\
    &\max_{g \in [p]}\norm*{ n^{-1}\bm{C}^{\top}(\bm{W}_g-I_n)\bm{X} }_2,\, \max_{g \in [p]}\norm*{ n^{-1}\bm{X}^{\top}(\bm{W}_g-I_n)\bm{X} }_2 = O_P(n^{-1/2+\epsilon})
\end{align*}
for any $\epsilon > 0$. And since $\max_{g \in [p]}\norm*{ n^{-1/2}\bm{X}^{\top}\bm{W}_g\bm{e}_g }_2= O_P(n^{\epsilon})$ for any $\epsilon > 0$, the second term in \eqref{supp:equation:CtPPe} is $O_P(\lambda^{-1+\epsilon})$ for any $\epsilon > 0$, which proves $\norm*{ (\lambda p)^{-1}\sum_{g=1}^p \tilde{\bm{C}}^{\top}\bm{P}_g^{\perp}\bm{e}_g \tilde{\bm{\ell}}_g^{\top} }_2 = O_P(\lambda^{-1+\epsilon})$ for any $\epsilon > 0$. Lastly,
\begin{align*}
    \norm*{ (\lambda p)^{-1}\sum\limits_{g=1}^p \bm{\Delta}^{\top}\bm{Q}^{\top}\bm{P}_g^{\perp}\bm{e}_g \tilde{\bm{\ell}}_g^{\top} }_2 \leq \underbrace{\norm*{ \bm{R} }_2}_{=O_P(1)}\norm*{ (\lambda p)^{-1}\sum\limits_{g=1}^p \bm{\Delta}^{\top}\bm{Q}^{\top}\bm{P}_g^{\perp}\bm{e}_g (n^{1/2}\bm{\ell}_g)^{\top} }_2,
\end{align*}
where for $k \in [K]$,
\begin{align*}
    &(\lambda p)^{-1}\bm{\Delta}^{\top}\bm{Q}^{\top}\sum\limits_{g=1}^p \bm{P}_g^{\perp}\bm{e}_g (n^{1/2}\bm{\ell}_{gk}) = \lambda^{-1/2}\bm{\Delta}^{\top}\bm{Q}^{\top}p^{-1}\sum\limits_{g=1}^p \bm{e}_g \bar{\bm{\ell}}_{gk} + \lambda^{-1/2}\bm{\Delta}^{\top}\bm{Q}^{\top}p^{-1}\sum\limits_{g=1}^p (\bm{W}_g-I_n)\bm{e}_g \bar{\bm{\ell}}_{gk}\\
     &- \lambda^{-1/2}\bm{\Delta}^{\top}\bm{Q}^{\top}p^{-1}\sum\limits_{g=1}^p \{n^{-1/2}(\bm{W}_g-I_n)\bm{X}\}(n^{-1}\bm{X}^{\top}\bm{W}_g\bm{X})^{-1}(n^{-1/2}\bm{X}^{\top}\bm{W}_g\bm{e}_g) \bar{\bm{\ell}}_{gk}.
\end{align*}
We first see that
\begin{align*}
    \E\left(\norm*{ p^{-1}\sum\limits_{g=1}^p \bm{e}_g \bar{\bm{\ell}}_{gk} }_2^2\right) = p^{-1}\sum_{g=1}^p \bar{\bm{\ell}}_{gk}^2 p^{-1}\Tr\{ \V(\bm{e}_g) \} \leq c
\end{align*}
for some constant $c>0$. Next,
\begin{align*}
    \E\left\{\norm*{ p^{-1}\sum\limits_{g=1}^p (\bm{W}_g-I_n)\bm{e}_g \bar{\bm{\ell}}_{gk} }_2^2\right\} = p^{-1}\sum_{g=1}^p p^{-1}\sum_{i=1}^n \E\{(w_{gi}-1)^2\bm{e}_{gi}^2\} \leq c
\end{align*}
for some constant $c>0$. Since
\begin{align*}
    \max_{g \in [p]}\norm*{ \{n^{-1/2}(\bm{W}_g-I_n)\bm{X}\}(n^{-1}\bm{X}^{\top}\bm{W}_g\bm{X})^{-1}(n^{-1/2}\bm{X}^{\top}\bm{W}_g\bm{e}_g) \bar{\bm{\ell}}_{gk} }_2 = O_P(n^{\epsilon}),
\end{align*}
this implies
\begin{align*}
    \norm*{ p^{-1}\sum\limits_{g=1}^p \{n^{-1/2}(\bm{W}_g-I_n)\bm{X}\}(n^{-1}\bm{X}^{\top}\bm{W}_g\bm{X})^{-1}(n^{-1/2}\bm{X}^{\top}\bm{W}_g\bm{e}_g) \bar{\bm{\ell}}_{gk} }_2 = O_P(n^{\epsilon})
\end{align*}
for any $\epsilon > 0$. Since $\norm*{ \bm{\Delta} }_2 = O_P(\lambda^{-1+\epsilon})$ for any $\epsilon > 0$ by Theorem~\ref{supp:theorem:ChatProp},
\begin{align*}
    \norm*{ (\lambda p)^{-1}\sum\limits_{g=1}^p \bm{\Delta}^{\top}\bm{Q}^{\top}\bm{P}_g^{\perp}\bm{e}_g \tilde{\bm{\ell}}_g^{\top} }_2 = O_P(\lambda^{-1+\epsilon})
\end{align*}
for any $\epsilon > 0$, which completes the proof.
\end{proof}

\subsection{Properties of our estimate for $\bm{\Omega}$}
\label{supp:subsection:Omega}

\begin{theorem}
\label{supp:theorem:Omega}
Suppose the Assumptions of Theorem~\ref{supp:theorem:ChatProp} hold, let $\bm{\Omega} = (n^{-1}\bm{C}^{\top}P_{\bm{X}}^{\perp}\bm{C})^{-1/2}\bm{C}^{\top} \allowbreak \bm{X}(\bm{X}^{\top}\bm{X})^{-1}$, and for $\hat{\bm{\beta}}_g^{\naive} = (\bm{X}^{\top}\bm{W}_g \bm{X})^{-1}\bm{X}^{\top}\bm{W}_g \bm{y}_g$ and $\hat{\bm{\ell}}_g$ as defined in the statement of Corollary~\ref{supp:corollary:LtL}, define $\hat{\bm{\Omega}} = n^{1/2}\left(\sum_{g=1}^p \hat{\bm{\ell}}_g\hat{\bm{\ell}}_g^{\top}\right)^{-1}\left[ \sum_{g=1}^p\hat{\bm{\ell}}_g \{\hat{\bm{\beta}}_g^{\naive}\}^{\top}\right]$. Then $\norm*{ \bm{v}^{\top}\bm{\Omega}_{\bigcdot j} - \hat{\bm{\Omega}}_{\bigcdot j} }_2 = o_P(n^{-1/2})$ for $\bm{v}$ as defined in the statement of Theorem~\ref{supp:theorem:ChatProp} and all $j \in [d_1]$.
\end{theorem}

\begin{proof}
By definition,
\begin{align*}
    \hat{\bm{\beta}}_g^{\naive} = \bm{\beta}_g + n^{-1/2}\bm{\Omega}^{\top}\tilde{\bm{\ell}}_g + (\bm{X}^{\top}\bm{W}_g \bm{X})^{-1}\bm{X}^{\top}(\bm{W}_g-I_n)\tilde{\bm{C}}\tilde{\bm{\ell}}_g +  (\bm{X}^{\top}\bm{W}_g \bm{X})^{-1}\bm{X}^{\top}\bm{W}_g \bm{e}_g.
\end{align*}
Define $\hat{\bm{S}} = (\lambda p)^{-1}\sum_{g=1}^p \hat{\bm{\ell}}_g\hat{\bm{\ell}}_g^{\top}$, where $\norm*{\hat{\bm{S}}}_2 = O_P(1)$ by Corollary~\ref{supp:corollary:LtL}. Then for $j \in [d_1]$ and $\bm{a}_j$ the $j$th standard basis vector in $\mathbb{R}^d$,
\begin{align}
\label{supp:equation:OmegaHatExpand}
\begin{aligned}
    \hat{\bm{\Omega}}_{\bigcdot j} =& \hat{\bm{S}}^{-1}\{ n^{1/2}(\lambda p)^{-1}\sum_{g=1}^p \hat{\bm{\ell}}_g \bm{\beta}_{gj}  \} + \hat{\bm{S}}^{-1}\{ (\lambda p)^{-1}\sum_{g=1}^p \hat{\bm{\ell}}_g\tilde{\bm{\ell}}_g^{\top} \bm{\Omega}_{\bigcdot j}  \}\\
    &+\hat{\bm{S}}^{-1} [ (\lambda p)^{-1}\sum_{g=1}^p \hat{\bm{\ell}}_g\tilde{\bm{\ell}}_g^{\top} \{\tilde{\bm{C}}^{\top} (\bm{W}_g-I_n)(n^{-1/2}\bm{X})\}(n^{-1}\bm{X}^{\top}\bm{W}_g \bm{X})^{-1}\bm{a}_j  ]\\
    &+ \hat{\bm{S}}^{-1}\{ (\lambda p)^{-1}\sum_{g=1}^p \hat{\bm{\ell}}_g(n^{-1/2}\bm{e}_g^{\top}\bm{W}_g \bm{X})(n^{-1}\bm{X}^{\top}\bm{W}_g \bm{X})^{-1}\bm{a}_j  \}.
\end{aligned}
\end{align}
By Corollary~\ref{supp:corollary:LtL},
\begin{align*}
    \norm*{ \hat{\bm{S}}^{-1}\{ (\lambda p)^{-1}\sum_{g=1}^p \hat{\bm{\ell}}_g\tilde{\bm{\ell}}_g^{\top} \bm{\Omega}_{\bigcdot j}  \} - \bm{v}^{\top}\bm{\Omega}_{\bigcdot j} }_2 = o_P(n^{-1/2}).
\end{align*}
For the fourth term in \eqref{supp:equation:OmegaHatExpand},
\begin{align*}
    &\norm*{ (\lambda p)^{-1}\sum_{g=1}^p \hat{\bm{\ell}}_g(n^{-1/2}\bm{e}_g^{\top}\bm{W}_g \bm{X})(n^{-1}\bm{X}^{\top}\bm{W}_g \bm{X})^{-1} }_2\\ \leq& \norm*{ (\lambda p)^{-1}\sum_{g=1}^p \tilde{\bm{\ell}}_g(n^{-1/2}\bm{e}_g^{\top}\bm{W}_g \bm{X}) }_2\underbrace{\norm*{ (n^{-1}\bm{X}^{\top}\bm{X})^{-1} }_2}_{=O(1)}\\
    &+ \underbrace{\norm*{ (\lambda p)^{-1}\sum_{g=1}^p \tilde{\bm{\ell}}_g(n^{-1/2}\bm{e}_g^{\top}\bm{W}_g \bm{X})\{ (n^{-1}\bm{X}^{\top}\bm{W}_g \bm{X})^{-1} - (n^{-1}\bm{X}^{\top}\bm{X})^{-1} \} }_2}_{=O_P(\lambda^{-1+\epsilon})} \\
    &+\underbrace{\norm*{ (\lambda p)^{-1}\sum_{g=1}^p (\hat{\bm{\ell}}_g - \bm{v}^{\top}\tilde{\bm{\ell}}_g)(n^{-1/2}\bm{e}_g^{\top}\bm{W}_g \bm{X})(n^{-1}\bm{X}^{\top}\bm{W}_g \bm{X})^{-1}}_2}_{=O_P(\lambda^{-1+\epsilon})},
\end{align*}
where the second and third lines follow because $\norm*{ \hat{\bm{\ell}}_g - \bm{v}^{\top}\tilde{\bm{\ell}}_g }_2 = O_P(n^{\epsilon})$ by Corollary~\ref{supp:corollary:LtL} and $\max_{g \in [p]}\norm*{ (n^{-1}\bm{X}^{\top}\bm{W}_g \bm{X})^{-1} - (n^{-1}\bm{X}^{\top}\bm{X})^{-1} }_2 = O_P(n^{-1/2+\epsilon})$ for any $\epsilon > 0$. Next, 
\begin{align*}
    (\lambda p)^{-1}\sum_{g=1}^p \tilde{\bm{\ell}}_g\bm{e}_g^{\top}\bm{W}_g (n^{-1/2}\bm{X}) =& (\lambda p)^{-1}\sum_{g=1}^p \tilde{\bm{\ell}}_g\bm{e}_g^{\top} (n^{-1/2}\bm{X})\\
    &+ (\lambda p)^{-1}\sum_{g=1}^p \tilde{\bm{\ell}}_g\bm{e}_g^{\top}(\bm{W}_g-I_n) (n^{-1/2}\bm{X})\\
    \V\{ (\lambda p)^{-1/2}\sum_{g=1}^p \tilde{\bm{\ell}}_g\bm{e}_g^{\top} (n^{-1/2}\bm{X}_{\bigcdot j})  \} =& (\lambda p)^{-1}\sum_{g=1}^p \tilde{\bm{\ell}}_g\tilde{\bm{\ell}}_g^{\top}\{ n^{-1}\bm{X}_{\bigcdot j}^{\top} \V(\bm{e}_g) \bm{X}_{\bigcdot j}\} \preceq c I_K\\
    \V\{ (\lambda p)^{-1/2}\sum_{g=1}^p \tilde{\bm{\ell}}_g\bm{e}_g^{\top}(\bm{W}_g - I_n) (n^{-1/2}\bm{X}_{\bigcdot j})  \} =& (\lambda p)^{-1}\sum_{g=1}^p \tilde{\bm{\ell}}_g\tilde{\bm{\ell}}_g^{\top}[ n^{-1}\sum_{i=1}^n \E\{ (w_{gi}- 1)^2 \bm{e}_{gi}^2 \bm{X}_{i j}^2 \}]\\
    \preceq & cI_K
\end{align*}
for some constant $c>0$ and all $j \in [d]$, which implies
\begin{align*}
    \norm*{ (\lambda p)^{-1}\sum_{g=1}^p \hat{\bm{\ell}}_g(n^{-1/2}\bm{e}_g^{\top}\bm{W}_g \bm{X})(n^{-1}\bm{X}^{\top}\bm{W}_g \bm{X})^{-1} }_2 = o_P(n^{-1/2}).
\end{align*}
For the third term, we first see that
\begin{align*}
    &\lambda^{-1}\hat{\bm{\ell}}_g\tilde{\bm{\ell}}_g^{\top}\tilde{\bm{C}}^{\top} (\bm{W}_g-I_n)(n^{-1/2}\bm{X}) = \lambda^{-1}\hat{\bm{\ell}}_g(n^{1/2}\bm{\ell}_g)^{\top} \{ n^{-1}\bm{C}^{\top}(\bm{W}_g-I_n)\bm{X} \}\\
    &- \lambda^{-1}\hat{\bm{\ell}}_g(n^{1/2}\bm{\ell}_g)^{\top} (n^{-1}\bm{C}^{\top}\bm{X})(n^{-1}\bm{X}^{\top}\bm{X})^{-1} \{ n^{-1}\bm{X}^{\top}(\bm{W}_g-I_n)\bm{X} \},
\end{align*}
which implies
\begin{align*}
    \max_{g \in [p]}\norm*{ \lambda^{-1}\hat{\bm{\ell}}_g\tilde{\bm{\ell}}_g^{\top}\tilde{\bm{C}}^{\top} (\bm{W}_g-I_n)(n^{-1/2}\bm{X}) }_2 = O_P(n^{-1/2+\epsilon})
\end{align*}
for any $\epsilon>0$. Consequently, for $\bm{R} = (n^{-1}\bm{C}^{\top}P_{\bm{X}}^{\perp}\bm{C})^{1/2}$ and $j \in [d]$,
\begin{align*}
    &\norm*{ (\lambda p)^{-1}\sum_{g=1}^p \hat{\bm{\ell}}_g\tilde{\bm{\ell}}_g^{\top} \{\tilde{\bm{C}}^{\top} (\bm{W}_g-I_n)(n^{-1/2}\bm{X})\}(n^{-1}\bm{X}^{\top}\bm{W}_g \bm{X})^{-1} }_2\\
    \leq & \underbrace{\norm*{\bm{R}}_2}_{=O_P(1)}\underbrace{\norm*{ (n^{-1}\bm{X}^{\top}\bm{X})^{-1} }_2}_{=O(1)}\norm*{ (\lambda p)^{-1}\sum_{g=1}^p (n\bm{\ell}_g\bm{\ell}_g^{\top}) \{n^{-1}\bm{C}^{\top}P_{\bm{X}}^{\perp} (\bm{W}_g-I_n)\bm{X}\} }_2 + o_P(n^{-1/2})\\
    &(\lambda p)^{-1}\sum_{g=1}^p (n\bm{\ell}_g\bm{\ell}_g^{\top}) \{n^{-1}\bm{C}^{\top}P_{\bm{X}}^{\perp} (\bm{W}_g-I_n)\bm{X}_{\bigcdot j}\} = (\lambda p)^{-1}\sum_{g=1}^p (n\bm{\ell}_g\bm{\ell}_g^{\top}) \{n^{-1}\bm{C}^{\top} (\bm{W}_g-I_n)\bm{X}_{\bigcdot j}\}\\
    &- (\lambda p)^{-1}\sum_{g=1}^p (n\bm{\ell}_g\bm{\ell}_g^{\top}) (n^{-1}\bm{C}^{\top}\bm{X})(n^{-1}\bm{X}^{\top}\bm{X})\{ n^{-1}\bm{X}^{\top}(\bm{W}_g-I_n)\bm{X}_{\bigcdot j}\}.
\end{align*}
Define $\bm{S}_g = n\lambda^{-1}\bm{\ell}_g\bm{\ell}_g^{\top}$, which has uniformly bounded entries. Then
\begin{align*}
    &(\lambda p)^{-1}\sum_{g=1}^p (n\bm{\ell}_g\bm{\ell}_g^{\top}) \{n^{-1}\bm{C}^{\top} (\bm{W}_g-I_n)\bm{X}_{\bigcdot j}\} = (np)^{-1}\sum_{i=1}^n \sum_{g=1}^p (w_{gi} - 1)\bm{S}_g \bm{C}_{i \bigcdot}\bm{X}_{ij}\\
    &\V\{ (np)^{-1/2}\sum_{i=1}^n \sum_{g=1}^p (w_{gi} - 1)\bm{S}_g \bm{C}_{i \bigcdot}\bm{X}_{ij} \} = p^{-1} \sum_{i=1}^n \underbrace{[n^{-1}\sum_{g=1}^p \bm{X}_{ij}^2\bm{S}_g \E\{ (w_{gi} - 1)^2\bm{C}_{i \bigcdot}\bm{C}_{i \bigcdot}^{\top} \}\bm{S}_g]}_{\preceq c I_K}
\end{align*}
for some constant $c>0$, which implies
\begin{align*}
    \norm*{(\lambda p)^{-1}\sum_{g=1}^p (n\bm{\ell}_g\bm{\ell}_g^{\top}) \{n^{-1}\bm{C}^{\top} (\bm{W}_g-I_n)\bm{X}_{\bigcdot j}\}}_2 = o_P(n^{-1/2}).
\end{align*}
As an identical analysis can be used to show that 
\begin{align*}
    \norm*{ (\lambda p)^{-1}\sum_{g=1}^p (n\bm{\ell}_g\bm{\ell}_g^{\top}) (n^{-1}\bm{C}^{\top}\bm{X})(n^{-1}\bm{X}^{\top}\bm{X})\{ n^{-1}\bm{X}^{\top}(\bm{W}_g-I_n)\bm{X}_{\bigcdot j}\} }_2 = o_P(n^{-1/2}),
\end{align*}
the third term in \eqref{supp:equation:OmegaHatExpand} satisfies
\begin{align*}
    \norm*{ (\lambda p)^{-1}\sum_{g=1}^p \hat{\bm{\ell}}_g\tilde{\bm{\ell}}_g^{\top} \{\tilde{\bm{C}}^{\top} (\bm{W}_g-I_n)(n^{-1/2}\bm{X})\}(n^{-1}\bm{X}^{\top}\bm{W}_g \bm{X})^{-1} }_2 = o_P(n^{-1/2}).
\end{align*}
For the first term in \eqref{supp:equation:OmegaHatExpand}, we note that $\max_{g \in [p]}\norm*{ \hat{\bm{\ell}}_g }_2 \leq \lambda^{1/2} c\{1+o_P(1)\}$ for some constant $c>0$. Therefore, for some constant $c>0$,
\begin{align*}
    \norm*{ n^{1/2}(\lambda p)^{-1}\sum_{g=1}^p \hat{\bm{\ell}}_g \bm{\beta}_{gj} }_2 \leq c\{1+o_P(1)\}(n/\lambda)^{1/2} \{p^{-1}\sum_{g=1}^p I(\bm{\beta}_{gj} \neq 0)\} = o_P(n^{1/2}), \quad j \in [d_1]
\end{align*}
by Assumption~\ref{supp:assumptions:FA}, which completes the proof.
\end{proof}

\begin{corollary}
\label{supp:corollary:OmegaInference}
In addition to the assumptions of Theorem~\ref{supp:theorem:Omega}, suppose $\E(\bm{C}_{i \bigcdot}) = \sum_{j=1}^d \bm{X}_{ij}\bm{\omega}_j$ for $\bm{\omega}_j \in \mathbb{R}^K$. Then for a fixed $j \in [d_1]$ and $Z \sim \chi_K^2$, $[\{(\bm{X}^{\top}\bm{X})^{-1}\}_{jj}]^{-1}\hat{\bm{\Omega}}_{\bigcdot j}^{\top}\hat{\bm{\Omega}}_{\bigcdot j} \edist Z + o_P(1)$ if $\bm{\omega}_j=\bm{0}$.
\end{corollary}

\begin{proof}
This follows directly from Theorem~\ref{supp:theorem:Omega} and the proof of Theorem 3 in \citet{BCconf}. The details have been omitted.
\end{proof}

\subsection{Estimating coefficients in differential abundance analyses}
\label{supp:subsection:DiffAbund}
For notational convenience, we let $\hat{\bm{C}}_{\perp} = P_{\bm{X}}^{\perp}\hat{\bm{C}}$ be the estimator obtained from \eqref{supp:equation:fMax} for the remainder of the supplement. Note that by construction, $\hat{\bm{C}}_{\perp}^{\top}\hat{\bm{C}}_{\perp} = I_K$.

\begin{lemma}
\label{supp:lemma:IPWest}
Let $\hat{\bm{C}} = n^{1/2}\hat{\bm{C}}_{\perp} + P_{\bm{X}_2}^{\perp}\bm{X}_1 \hat{\bm{\Omega}}_1^{\top}$ and $\hat{\bm{Z}} = [P_{\bm{X}_2}^{\perp}\bm{X}_1, \hat{\bm{C}}, \bm{X}_2]$ for $\bm{X}_j \in \mathbb{R}^{n \times d_j}$, $j=1,2$, given in Assumption~\ref{supp:assumptions:FA} and $\hat{\bm{\Omega}}_1 \in \mathbb{R}^{K \times d_1}$ the first $d_1$ columns $\hat{\bm{\Omega}}$ defined in the statement of Theorem~\ref{supp:theorem:Omega}. Define the inverse probability weighted (IPW) estimator
\begin{align*}
    &\hat{\bm{\theta}}_g^{\ipw} = (\hat{\bm{Z}}^{\top}\bm{W}_g\hat{\bm{Z}})^{-1}\hat{\bm{Z}}^{\top}\bm{W}_g \bm{y}_g
\end{align*}
and the parameter vector
\begin{align*}
    \bm{\theta}_g^* = (\bm{\beta}_{g1}^{\top}, \{ \bm{v}^{\top}(n^{-1}\bm{C}^{\top}P_{\bm{X}}^{\perp}\bm{C})^{1/2}\bm{\ell}_g \}^{\top}, \{\bm{\beta}_{g2} + (\bm{X}_2^{\top}\bm{X}_2)^{-1}\bm{X}_2^{\top}(\bm{X}_1\bm{\beta}_{g1} + \bm{C}\bm{\ell}_g)\}^{\top})^{\top} \in \mathbb{R}^{d+K},
\end{align*}
where $\bm{\beta}_{g1} \in \mathbb{R}^{d_1}$ and $\bm{\beta}_{g2} \in \mathbb{R}^{d_2}$ are the first $d_1$ and last $d_2$ elements of $\bm{\beta}_g \in \mathbb{R}^{d_1+d_2}$. Then under the assumptions of Theorem~\ref{supp:theorem:Omega}, $\norm*{\hat{\bm{\theta}}_g^{\ipw} - \bm{\theta}_g^*}_2 = O_P(n^{-1/2})$.
\end{lemma}


\begin{proof}
Note that $\E(\bm{y}_g) = \bm{X}_1\bm{\beta}_{g1} + \bm{C}\bm{\ell}_g + \bm{X}_2\bm{\beta}_{g2} = \bm{Z}\bm{\theta}_{g}^*$ for $\bm{Z} = [\tilde{\bm{X}}_1, \tilde{\bm{C}}\bm{v} + \tilde{\bm{X}}_1\bm{\Omega}_1^{\top}\bm{v}, \bm{X_2}]$, $\tilde{\bm{X}}_1=P_{\bm{X}_2}^{\perp}\bm{X}_1$, and $\bm{\theta}_{g}^*$ as defined in the statement of Lemma~\ref{supp:lemma:IPWest}. Therefore,
\begin{align*}
    \hat{\bm{\theta}}_g^{\ipw} - \bm{\theta}_{g}^* =& (n^{-1}\hat{\bm{Z}}^{\top}\bm{W}_g\hat{\bm{Z}})^{-1}\bm{\delta}^{\top}\bm{W}_g(n^{-1/2}\bm{Z})\bm{\theta}_{g}^* + (n^{-1}\hat{\bm{Z}}^{\top}\bm{W}_g\hat{\bm{Z}})^{-1}(n^{-1}\bm{Z}^{\top}\bm{W}_g \bm{e}_g)\\
    &+   (n^{-1}\hat{\bm{Z}}^{\top}\bm{W}_g\hat{\bm{Z}})^{-1}(n^{-1/2}\bm{\delta}^{\top}\bm{W}_g \bm{e}_g)\\
    \bm{\delta} =& [\bm{0}_{n \times d_1}, \tilde{\bm{C}}(\hat{\bm{v}} - \bm{v}) + \bm{Q}\hat{\bm{z}} + (n^{-1/2}\tilde{\bm{X}}_1)(\hat{\bm{\Omega}}_1^{\top} - \bm{\Omega}_1^{\top}\bm{v}), \bm{0}_{n \times d_2}].
\end{align*}
First, identical techniques used to prove Corollary~\ref{supp:corollary:CtC} can be used to show $\norm*{ n^{-1}\hat{\bm{Z}}^{\top}\bm{W}_g\hat{\bm{Z}} }_2 = O_P(1)$. Next,
\begin{align*}
    \norm*{ n^{-1}\bm{Z}^{\top}\bm{W}_g \bm{e}_g }_2 = O_P(\norm*{ n^{-1}\bm{X}^{\top} \bm{e}_g }_2) + O_P(\norm*{ n^{-1}\bm{C}^{\top} \bm{e}_g }_2) + O_P(\norm*{ n^{-1}\bm{Z}^{\top}(\bm{W}_g-I_n) \bm{e}_g }_2),
\end{align*}
where the first term is trivially $O_P(n^{-1/2})$. For the second,
\begin{align*}
    \norm*{ n^{-1}\bm{C}^{\top} \bm{e}_g - (n^{-1}\bm{G}_{s_g *}^{\top} \bm{G}_{s_g *})\bm{\gamma}^{(e)}_{s_g g}\bm{\gamma}^{(c)}_{s_g *} + \bm{\gamma}^{(e)}_{s_g g}  }_2 = O_P(n^{-1/2}),
\end{align*}
where $\norm*{ \bm{\gamma}^{(e)}_{s_g g}\bm{\gamma}^{(c)}_{s_g *} }_2 = o_P(n^{-1/2})$ by Assumption~\ref{supp:assumptions:FA}. Therefore, $\norm*{ n^{-1}\bm{C}^{\top} \bm{e}_g }_2=O_P(n^{-1/2})$. For the third term, there exists an $\bm{R}$ such that $\bm{Z}=[\bm{X},\bm{C}]\bm{R}$ and $\norm*{\bm{R}}_2=O_P(1)$, meaning
\begin{align*}
    \norm*{ n^{-1}\bm{Z}^{\top}(\bm{W}_g-I_n) \bm{e}_g }_2 \leq O_P\{ \norm*{ n^{-1}[\bm{X},\bm{C}]^{\top}(\bm{W}_g-I_n) \bm{e}_g }_2 \},
\end{align*}
where for $\bar{\bm{c}}_{i} = [\bm{X},\bm{C}]_{i \bigcdot}$,
\begin{align*}
    \V\{ n^{-1/2}[\bm{X},\bm{C}]^{\top}(\bm{W}_g-I_n) \bm{e}_g \} = n^{-1}\sum_{i=1}^n \E\{(w_{gi}-1)^2 \bm{e}_{gi}^2 \bar{\bm{c}}_{i}\bar{\bm{c}}_{i}^{\top} \} \preceq c I_{d+K}
\end{align*}
for some constant $c>0$. This proves $\norm*{ n^{-1}\bm{Z}^{\top}\bm{W}_g \bm{e}_g }_2 = O_P(n^{-1/2})$. We next see that by Theorems~\ref{supp:theorem:ChatProp} and \ref{supp:theorem:Omega},
\begin{align*}
    \norm*{ \bm{\delta}^{\top}\bm{W}_g(n^{-1/2}\bm{Z}) }_2 = \norm*{ \hat{\bm{z}}^{\top}\bm{Q}^{\top}\bm{W}_g(n^{-1/2}\bm{Z}) }_2 + o_P(n^{-1/2}),
\end{align*}
where identical techniques used to prove Corollary~\ref{supp:corollary:CtC} can be used to show $\norm*{ \hat{\bm{z}}^{\top}\bm{Q}^{\top}\bm{W}_g(n^{-1/2}\bm{Z}) }_2 = o_P(n^{-1/2})$. A second application of Theorems~\ref{supp:theorem:ChatProp} and \ref{supp:theorem:Omega} imply
\begin{align*}
     \norm*{ n^{-1/2}\bm{\delta}^{\top}\bm{W}_g\bm{e}_g }_2 = \norm*{ n^{-1/2}\hat{\bm{z}}^{\top}\bm{Q}^{\top}\bm{W}_g\bm{e}_g }_2 + o_P(n^{-1/2}),
\end{align*}
where techniques used to prove Corollary~\ref{supp:corollary:LtL} can be used to show $\norm*{ n^{-1/2}\hat{\bm{z}}^{\top}\bm{Q}^{\top}\bm{W}_g\bm{e}_g }_2 = o_P(n^{-1/2})$.
\end{proof}

\begin{lemma}
\label{supp:lemma:IPWestVar}
Fix a $g \in [p]$ and suppose the assumptions of Lemma~\ref{supp:lemma:IPWest} hold, let $\hat{\bm{\theta}}_g^{\IPW}$ and $\hat{\bm{Z}}$ be as defined in the statement of Lemma~\ref{supp:lemma:IPWest}, and let $\{ \hat{\sigma}_g^{\IPW} \}^2 = (\sum_{i=1}^n w_{gi})^{-1} \allowbreak \sum_{i=1}^n w_{gi} \allowbreak \{y_{gi} - \hat{\bm{Z}}_{i*}^{\top} \hat{\bm{\theta}}^{\IPW}\}^2$. Then $\abs*{\hat{\sigma}_g^{\IPW} - \sigma_g} = O_P(n^{-1/2})$.
\end{lemma}

\begin{proof}
Since $\{w_{gi}\}_{i \in [n]}$ are independent with uniformly bounded $m$ moments for all non-negative $m$, $n^{-1} \sum_{i=1}^n w_{gi} = 1 + O_P(n^{-1/2})$. Next,
\begin{align*}
    n^{-1}\sum_{i=1}^n w_{gi}\{y_{gi} - \hat{\bm{Z}}_{i*}^{\top} \hat{\bm{\theta}}_g^{\IPW}\}^2 =& n^{-1}\{ \bm{e}_g + \bm{\delta}\bm{\theta}_g^* - \hat{\bm{Z}}\bm{\epsilon}_g \}^{\top}\bm{W}_g \{ \bm{e}_g + \bm{\delta}\bm{\theta}_g^* - \hat{\bm{Z}}\bm{\epsilon}_g\}\\
    =& n^{-1}\bm{e}_g^{\top}\bm{e}_g + n^{-1}\bm{e}_g^{\top}(\bm{W}_g-I_n) \bm{e}_g + 2n^{-1}\bm{e}_g^{\top}\bm{W}_g \bm{\delta}\bm{\theta}_g^*\\
    &- 2n^{-1} \bm{e}_g^{\top}\bm{W}_g \hat{\bm{Z}}\bm{\epsilon}_g + n^{-1}\{ \bm{\theta}_g^* \}^{\top}\bm{\delta}^{\top}\bm{W}_g \bm{\delta}\bm{\theta}_g^*\\
    &- 2n^{-1}\{ \bm{\theta}_g^* \}^{\top}\bm{\delta}^{\top}\bm{W}_g \hat{\bm{Z}}\bm{\epsilon}_g + n^{-1} \bm{\epsilon}_g^{\top} \hat{\bm{Z}}^{\top}\bm{W}_g \hat{\bm{Z}} \bm{\epsilon}_g
\end{align*}
for $\bm{\delta} = \bm{Z} - \hat{\bm{Z}}$ and $\bm{\epsilon}_g = \hat{\bm{\theta}}_g^{\IPW} - \bm{\theta}_g^*$. Note that $\norm*{ \bm{\epsilon}_g }_2 = O_P(n^{-1/2})$ by Lemma~\ref{supp:lemma:IPWest}. Going through each of the above seven terms, it is easy to see that for any $\epsilon>0$,
\begin{align*}
    &\abs*{n^{-1}\bm{e}_g^{\top}\bm{e}_g - \sigma_g^2} = O_P(n^{-1/2}), \quad \abs*{ n^{-1}\bm{e}_g^{\top}(\bm{W}_g-I_n) \bm{e}_g } = O_{P}(n^{-1/2})\\
    &\abs*{ n^{-1} \bm{e}_g^{\top}\bm{W}_g \hat{\bm{Z}}\bm{\epsilon}_g } \leq \underbrace{\norm*{ n^{-1/2}\bm{W}_g\bm{e}_g }_2}_{O_P(1)} \underbrace{\norm*{n^{-1/2} \hat{\bm{Z}}}_2}_{O_P(1)} \underbrace{\norm*{ \bm{\epsilon}_g }_2}_{O_P(n^{-1/2})} = O_P(n^{-1/2})\\
    &\abs*{ n^{-1}\{ \bm{\theta}_g^* \}^{\top}\bm{\delta}^{\top}\bm{W}_g \bm{\delta}\bm{\theta}_g^* } \leq \underbrace{\norm*{ \bm{\theta}_g^* }_2^2}_{O_P(1)} \underbrace{\norm*{ n^{-1/2} \bm{\delta} }_2^2}_{O_P(\lambda^{-1+\epsilon/2})} \underbrace{\norm*{ \bm{W}_g }_2}_{O_P(n^{\epsilon/2})} = O_P(\lambda^{-1+\epsilon}) = o_P(n^{-1/2})\\
    &\norm*{ n^{-1}\{ \bm{\theta}_g^* \}^{\top}\bm{\delta}^{\top}\bm{W}_g \hat{\bm{Z}}\bm{\epsilon}_g }_2 \leq \underbrace{\norm*{\bm{\theta}_g^*}_2}_{O_P(1)} \underbrace{\norm*{ \bm{\delta} }_2}_{O_P(\lambda^{-1/2+\epsilon})} \underbrace{\norm*{ n^{-1/2}\bm{W}_g }_2}_{O_P(1)} \underbrace{\norm*{ n^{-1/2} }_2}_{O_P(1)} \underbrace{\norm*{\bm{\epsilon}_g}_2}_{O_P(n^{-1/2})} = O_P(n^{-1/2})\\
    &\norm*{ n^{-1} \bm{\epsilon}_g^{\top} \hat{\bm{Z}}^{\top}\bm{W}_g \hat{\bm{Z}} \bm{\epsilon}_g }_2 \leq \underbrace{\norm*{\bm{\epsilon}_g}_2^2}_{O_P(n^{-1})}\underbrace{\norm*{n^{-1/2} \hat{\bm{Z}}}_2^2}_{O_P(1)} \underbrace{\norm*{\bm{W}_g}_2}_{O_P(n^{\epsilon})} = o_P(n^{-1/2}). 
\end{align*}
The proof will be complete if we can show $\norm*{ n^{-1}\bm{e}_g^{\top}\bm{W}_g \bm{\delta} }_2 = O_P(n^{-1/2})$, where
\begin{align*}
    \norm*{ n^{-1}\bm{e}_g^{\top}\bm{W}_g \bm{\delta} }_2 \leq \norm*{ n^{-1}\bm{e}_g^{\top}\bm{W}_g (P_{\bm{X}_2}^{\perp}\bm{X}_1) }_2 \underbrace{\norm*{ \hat{\bm{\Omega}}_1 }_2}_{O_P(1)} + \norm*{ n^{-1/2}\bm{e}_g^{\top}\bm{W}_g \hat{\bm{C}}_{\perp} }_2
\end{align*}
for $\hat{\bm{\Omega}}_1$ and $\hat{\bm{C}}_{\perp}$ as defined in the statement of Lemma~\ref{supp:lemma:IPWest}. It is easy to see $\norm*{ n^{-1}\bm{e}_g^{\top}\bm{W}_g (P_{\bm{X}_2}^{\perp}\bm{X}_1) }_2 \allowbreak = O_P(n^{-1/2})$. And since we showed $n^{-1/2}\bm{e}_g^{\top}\bm{W}_g \hat{\bm{C}}_{\perp} = O_P(n^{-1/2})$ in the proof of Lemma~\ref{supp:lemma:IPWest}, the proof is complete.
\end{proof}

\begin{lemma}
\label{supp:lemma:Cknown}
Let $\tilde{\bm{C}}$ and $\bm{v}$ be as defined in Theorem~\ref{supp:theorem:ChatProp} and $\bm{\Omega}_1 \in \mathbb{R}^{K \times d_1}$ be the first $d_1$ columns of $\bm{\Omega}$ defined in Theorem~\ref{supp:theorem:Omega}. Suppose Assumption~\ref{supp:assumptions:FA} holds, fix a $g \in [p]$, let $\bm{Z}=[P_{\bm{X}_2}\bm{X}_1,n^{1/2}\tilde{\bm{C}}\bm{v} + P_{\bm{X}_2}\bm{X}_1\bm{\Omega}_1^{\top}\bm{v}, \bm{X}_2]$ and $\bm{\theta}_g^*$ be as defined in the statement of Lemma~\ref{supp:lemma:IPWest}, let $\bm{\eta}_g^* = ( \{\bm{\theta}_g^*\}^{\top},\sigma_g^2 )^{\top}$, and for some constant $\delta>0$ small enough, define the maximum likelihood estimator
\begin{align*}
    \{\hat{\bm{\theta}}_g^{\known}, \hat{\sigma}_g^{\known}\} =& \argmax_{\{\bm{\theta},\sigma\} \in \mathcal{H} \times \mathcal{S}} f_g^{\known}(\bm{\theta},\sigma)\\
    \mathcal{H} =& \{\bm{\theta} \in \mathbb{R}^{K + d}: \norm*{\bm{\eta}-\bm{\eta}_g^*}_2 \leq \delta\}, \quad \mathcal{S}= \{\sigma>0: \abs*{\sigma - \sigma_g} \leq \delta\}\\
    f_g^{\known}(\bm{\theta},\sigma) =& n^{-1}\sum_{i=1}^n -r_{gi}\{\bm{y}_{gi} - \mu_{i}(\bm{\theta})\}^2/(2\sigma^2)\\
    &+ (1-r_{gi})\log\left(\smallint \phi(\epsilon)\Psi[ -\alpha_g\{ \mu_{i}(\bm{\theta}) + \sigma \epsilon - \delta_g \} ]\text{d}\epsilon\right), \quad \mu_{i}(\bm{\theta}) = \bm{Z}_{i \bigcdot}^{\top}\bm{\theta}.
\end{align*}
Then for $\bm{\eta}_g^*=(\bm{\theta}_g^*,\sigma_g)^{\top}$ and $\hat{\bm{\eta}}_g^{\known} = (\hat{\bm{\theta}}^{\known}_g,\hat{\sigma}^{\known}_g)^{\top}$, $\norm*{ \hat{\bm{\eta}}_g^{\known} - \bm{\eta}_g^* }_2 = O_P(n^{-1/2})$, $\{n \bm{H}_g^{\known} \}^{1/2}\{\hat{\bm{\eta}}_g^{\known} - \bm{\eta}_g^*\} \edist \bm{V}_g + o_P(1)$, and $\norm*{ \bm{H}_g^{\known} - \E\{ \bm{H}_g^{\known} \mid \bm{C},\bm{G} \} }_2 = o_P(1)$, where
\begin{align*}
    \bm{V}_g \sim N_{K+d+1}(\bm{0},I_{K+d+1}), \quad \bm{H}_g^{\known} = -\nabla^2 f_g^{\known}(\bm{\theta}_g^*,\sigma_g).
\end{align*}
\end{lemma}
\begin{remark}
\label{supp:remark:Cknown}
Proposition~\ref{proposition:Cknown} follows from Lemma~\ref{supp:lemma:Cknown} by letting $\bm{X}=\bm{X}_1$ and $\bm{X}_2=0$.
\end{remark}

\begin{proof}
It is easy to see that there exists a change-of-basis matrix $\hat{\bm{R}}$ that depends on $\bm{C}$ such that $\bm{Z}=[\bm{X},\bm{C}]\hat{\bm{R}}$, $\bm{\theta}_g^* = \hat{\bm{R}}^{-1}( \bm{\beta}_g^{\top},\bm{\ell}_g^{\top} )^{\top}$, and $\norm*{\hat{\bm{R}} - \bm{R}}_2 = o_P(1)$ for some non-random $\bm{R}$ with $\norm*{\bm{R}}_2\leq c$ and $\norm*{\bm{R}^{-1}}_2\leq c$ for some constant $c>0$. Therefore, it suffices to prove the theorem assuming $\bm{Z} = [\bm{X},\bm{C}]$ and $\bm{\theta}_g^* = ( \bm{\beta}_g^{\top},\bm{\ell}_g^{\top} )^{\top}$. Let $\gamma = \bm{\gamma}_{s_g g}^{(e)}$ and define
\begin{align*}
    \tilde{f}^{\known}(\bm{\theta},\sigma) =& n^{-1}\sum_{i=1}^n -r_{gi}\{\bm{y}_{gi} - \mu_{i}(\bm{\theta}) - \bm{G}_{s_g i}\gamma\}^2/(2\sigma^2)\\
    &+ (1-r_{gi})\log\left(\smallint \phi(\epsilon)\Psi[ -\alpha_g\{ \mu_{i}(\bm{\theta}) + \bm{G}_{s_g i}\gamma + \sigma \epsilon - \delta_g \} ]\text{d}\epsilon\right).
\end{align*}
Then for $h(\mu,\sigma) = \log\left(\smallint \phi(\epsilon)\Psi[ -\alpha_g\{ \mu + \sigma \epsilon - \delta_g \} ]\text{d}\epsilon\right)$,
\begin{align*}
    \tilde{f}^{\known}(\bm{\theta},\sigma) - f^{\known}(\bm{\theta},\sigma) =& 2\gamma n^{-1}\sum_{i=1}^n r_{gi}\{\bm{y}_{gi} - \mu_{i}(\bm{\theta})\}\bm{G}_{s_g i}/(2\sigma^2)\\
    & - \gamma^2n^{-1} \sum_{i=1}^n r_{gi}\bm{G}_{s_g i}^2/(2\sigma^2)\\
    &+ n^{-1}\gamma\sum_{i=1}^n (1-r_{gi}) \bm{G}_{s_g i} \frac{\partial}{\partial \mu} h( \tilde{\mu}_i,\sigma ),
\end{align*}
where $\tilde{\mu}_i = \alpha_i \mu_i(\bm{\theta}) + (1-\alpha_i)\{ \mu_i(\bm{\theta}) + \gamma \bm{G}_{s_g i} \}$ for some $\alpha_i \in [0,1]$. Since $\sup_{i \in [n]}\abs*{\gamma \bm{G}_{s_g i}}=o(n^{-1/4})$, $\sup_{i \in [n], \{\bm{\theta},\sigma\} \in \mathcal{H} \times \mathcal{S}}\abs*{ \frac{\partial}{\partial \mu} h( \tilde{\mu}_i,\sigma ) } \leq c$ for some constant $c>0$ by Lemma~\ref{supp:lemma:PsiBounds}. Therefore,
\begin{align*}
    \sup_{ \{\bm{\theta},\sigma\} \in \mathcal{H} \times \mathcal{S} }\abs*{ \tilde{f}^{\known}(\bm{\theta},\sigma) - f^{\known}(\bm{\theta},\sigma) } = o_P(n^{-1/4}).
\end{align*}

Next, define $F_g(\bm{\theta}) = f_g^{\known}(\bm{\theta}) - \E\{\tilde{f}_g^{\known}(\bm{\theta}) \mid \bm{C},\bm{G}\}$. Then $\hat{\bm{\theta}}_g^{\known}$ and $\hat{\sigma}_g^{\known}$ are consistent if $\sup_{ \{\bm{\theta},\sigma\} \in \mathcal{H} \times \mathcal{S} } \abs*{F_g(\bm{\theta})} = o_P(1)$. We see that
\begin{align}
\label{supp:equation:ftildefknown}
\begin{aligned}
    \sup_{ \{\bm{\theta},\sigma\} \in \mathcal{H} \times \mathcal{S} } \abs*{F_g(\bm{\theta},\sigma)} \leq & \underbrace{\sup_{ \{\bm{\theta},\sigma\} \in \mathcal{H} \times \mathcal{S} } \abs*{ \tilde{f}^{\known}(\bm{\theta},\sigma) - f^{\known}(\bm{\theta},\sigma) }}_{o_P(n^{-1/4})}\\
    &+ \sup_{ \{\bm{\theta},\sigma\} \in \mathcal{H} \times \mathcal{S} } \abs*{ \underbrace{\tilde{f}_g^{\known}(\bm{\theta},\sigma) - \E\{\tilde{f}_g^{\known}(\bm{\theta},\sigma) \mid \bm{C},\bm{G}\}}_{=\tilde{F}_g(\bm{\theta},\sigma)} }.
\end{aligned}
\end{align}
Since $\E [\abs*{ \sup_{ \{\bm{\theta},\sigma\} \in \mathcal{H} \times \mathcal{S} } h\{\mu_i(\bm{\theta}),\sigma\} }^m] \leq c_m$ for some constant $c_m$ that only depends on $m>0$, $\abs*{\tilde{F}_g(\bm{\theta},\sigma)} = o_P(1)$ for all $\{\bm{\theta},\sigma\} \in \mathcal{H} \times \mathcal{S}$. Next, let $\bm{R}_{gi} = \diag(r_{g1},\ldots,r_{gn})$ and $B(\bm{\theta},\sigma; \epsilon) = \{\{ \bm{x},v \} \in \mathcal{H} \times \mathcal{S}: \norm*{ \bm{x} - \bm{\theta} }_2,\abs*{v-\sigma}\leq \epsilon\}$. Then because $\sup{\mu \in \mathbb{R}}\abs*{ \frac{\partial}{\partial \mu}h(\mu,\sigma) } \leq c$ and $\sup{\sigma \in \mathcal{S}}\abs*{ \frac{\partial}{\partial \mu}h(\mu,\sigma) } \leq c$ for some constant $c>0$ by Lemma~\ref{supp:lemma:PsiBounds}, it is straightforward to show that
\begin{align*}
    \sup_{ \{\bm{\theta}_1,\sigma_1\}, \{\bm{\theta}_2,\sigma_2\} \in B(\bm{\theta},\sigma; \epsilon) }\abs*{ \tilde{F}_g(\bm{\theta}_1,\sigma_1) - \tilde{F}_g(\bm{\theta}_2,\sigma_2) } = O(\epsilon).
\end{align*}
Therefore, $\tilde{F}_g(\bm{\theta},\sigma)$ is stochastically equicontinuous on a compact set, which means $\sup_{ \{\bm{\theta},\sigma\} \in \mathcal{H} \times \mathcal{S} } \allowbreak \abs*{F_g(\bm{\theta})} = o_P(1)$, and therefore implies $\hat{\bm{\theta}}_g^{\known}$ and $\hat{\sigma}_g^{\known}$ are consistent.

Next, define
\begin{align*}
    &\bm{s}^{\known}(\bm{\theta},\sigma) = \nabla f^{\known}(\bm{\theta},\sigma) = ( \bm{s}_1^{\known}(\bm{\theta},\sigma)^{\top},s_2^{\known}(\bm{\theta},\sigma) )^{\top}\\
    &\tilde{\bm{s}}^{\known}(\bm{\theta},\sigma) = \nabla \tilde{f}^{\known}(\bm{\theta},\sigma) = ( \tilde{\bm{s}}_1^{\known}(\bm{\theta},\sigma)^{\top},\tilde{s}_2^{\known}(\bm{\theta},\sigma) )^{\top},
\end{align*}
where $s_2^{\known}(\bm{\theta},\sigma), \tilde{s}_2^{\known}(\bm{\theta},\sigma) \in \mathbb{R}$ are partial derivatives with respect to $\sigma$. Then
\begin{align*}
    \bm{s}_1^{\known}(\bm{\theta},\sigma) = &n^{-1} \sum_{i=1}^n [r_{gi}\{y_{gi}-\mu_i(\bm{\theta})\}/\sigma^2 + (1-r_{gi})\alpha_g \frac{\partial}{\partial \mu} h\{\mu_i(\bm{\theta}),\sigma\}]\bm{Z}_{i*}\\
    \tilde{\bm{s}}_1^{\known}(\bm{\theta},\sigma) = &n^{-1} \sum_{i=1}^n [r_{gi}\{y_{gi}-\mu_i(\bm{\theta}) - \gamma\bm{G}_{s_g i}\}/\sigma^2 + (1-r_{gi}) \frac{\partial}{\partial \mu} h\{\mu_i(\bm{\theta})+\gamma\bm{G}_{s_g i},\sigma\}]\bm{Z}_{i*}\\
    s_2^{\known}(\bm{\theta},\sigma) = &n^{-1} \sum_{i=1}^n r_{gi}\{y_{gi}-\mu_i(\bm{\theta})\}^2/\sigma^3 + (1-r_{gi}) \frac{\partial}{\partial \sigma} h\{\mu_i(\bm{\theta}),\sigma\}\\
    \tilde{s}_2^{\known}(\bm{\theta},\sigma) = &n^{-1} \sum_{i=1}^n r_{gi}\{y_{gi}-\mu_i(\bm{\theta}) - \gamma\bm{G}_{s_g i}\}^2/\sigma^3 + (1-r_{gi}) \frac{\partial}{\partial \sigma} h\{\mu_i(\bm{\theta})+\gamma\bm{G}_{s_g i},\sigma\}.
\end{align*}
Since $\bm{G}_{s_g i}$ is a bounded random variable and $\gamma = o(n^{-1/4})$,
\begin{align}
\label{supp:equation:s1Diff}
\begin{aligned}
    \tilde{\bm{s}}_1^{\known}(\bm{\theta}_g^*,\sigma_g) - \bm{s}_1^{\known}(\bm{\theta}_g^*,\sigma_g) =& (n\sigma_g^2)^{-1}\gamma\sum_{i=1}^n r_{gi}\bm{G}_{s_g i}\bm{Z}_{i*}\\
    &+ n^{-1} \gamma\sum_{i=1}^n(1-r_{gi})\bm{G}_{s_g i} \frac{\partial^2}{\partial^2 \mu} h\{\mu_i(\bm{\theta}_g),\sigma_g\}\bm{Z}_{i*}\\
    &+ n^{-1} \gamma^2\sum_{i=1}^n(1-r_{gi})\bm{G}_{s_g i}^2 \frac{\partial^3}{\partial^3 \mu} h(\tilde{\mu}_i,\sigma_g)\bm{Z}_{i*},
\end{aligned}
\end{align}
where $\tilde{\mu}_i = \alpha_i \mu_i(\bm{\theta}_g) + (1-\alpha_i)\{ \mu_i(\bm{\theta}_g) + \gamma \bm{G}_{s_g i} \}$ for $\alpha_i \in [0,1]$. To show this difference is $o(n^{-1/2})$, we first note that because the terms inside each of the three summands in \eqref{supp:equation:s1Diff} are independent with uniformly bounded moments and $\gamma = o(n^{-1/4})$, all three sums in \eqref{supp:equation:s1Diff} have variance equal to $o(n^{-1/2})$. We therefore need only show that the expectation of the above three sums is $o(n^{-1/2})$. Since $\gamma^2 = o(n^{-1/2})$ and $\frac{\partial^3}{\partial^3 \mu} h(\tilde{\mu}_i,\sigma_g)$ is uniformly bounded from above and below by Lemma~\ref{supp:lemma:PsiBounds}, the expectation of the third sum is $o(n^{-1/2})$. Next, for the first sum in \eqref{supp:equation:s1Diff},
\begin{align*}
    \E\left( n^{-1}\gamma\sum_{i=1}^n r_{gi}\bm{G}_{s_g i}\bm{Z}_{i*} \right) =&  n^{-1}\gamma\sum_{i=1}^n \E(\Psi[ \alpha_g\{ \mu_i(\bm{\theta}_g^*) + \gamma \bm{G}_{s_g i} + \bm{\Delta}^{(e)}_{gi} - \delta_g  \} ]\bm{G}_{s_g i}\bm{Z}_{i*} )\\
    =& n^{-1}\gamma\sum_{i=1}^n \E(\Psi[ \alpha_g\{ \mu_i(\bm{\theta}_g^*) + \gamma \bm{G}_{s_g i} + \bm{\Delta}^{(e)}_{gi} - \delta_g  \} ]\bm{G}_{s_g i}\tilde{\bm{Z}}_{i*} ) + o(n^{-1/2}),
\end{align*}
where $\tilde{\bm{Z}}_{i*}$ is independent of $\bm{G}_{s_g*}$ by Assumption~\ref{supp:assumptions:FA}. Further,
\begin{align*}
    &\sup_{i \in [n]}\abs*{ \{\mu_i(\bm{\theta}_g^*) + \gamma \bm{G}_{s_g i}\} - \tilde{\mu}_i(\bm{\theta}_g^*) } = o(n^{-1/4}), \quad \tilde{\mu}_i(\bm{\theta}_g^*)=\bm{X}_{i *}^{\top}\bm{\beta}_g + \{\bm{\Delta}^{(c)}_{i *} + \textstyle\sum_{s \neq s_g}\bm{G}_{si} \bm{\gamma}^{(c)}_{s*}\}^{\top} \bm{\ell}_g,
\end{align*}
where $\tilde{\mu}_i(\bm{\theta}_g^*)$ is independent of $\bm{G}_{s_g *}$ by Assumption~\ref{supp:assumptions:FA}. Therefore, since $\frac{d}{dx}\Psi(x)$ is bounded,
\begin{align*}
    &\sup_{i \in [n]}\abs*{ \E(\Psi[ \alpha_g\{ \mu_i(\bm{\theta}_g^*) + \gamma \bm{G}_{s_g i} + \bm{\Delta}^{(e)}_{gi} - \delta_g  \} ]\bm{G}_{s_g i}\tilde{\bm{Z}}_{i*} ) - \E(\Psi[ \alpha_g\{ \tilde{\mu}_i(\bm{\theta}_g^*) + \bm{\Delta}^{(e)}_{gi} - \delta_g  \} ]\bm{G}_{s_g i}\tilde{\bm{Z}}_{i*} ) }\\
    &= o\left\{ n^{-1/4}\sup_{i \in [n]} \E( \norm*{ \tilde{\bm{Z}}_{i*} }_2 ) \right\} = o(n^{-1/4})
\end{align*}
Putting this all together gives us
\begin{align*}
    \E\left( n^{-1}\gamma\sum_{i=1}^n r_{gi}\bm{G}_{s_g i}\bm{Z}_{i*} \right) =& n^{-1}\gamma \sum_{i=1}^n \E(\Psi[ \alpha_g\{ \tilde{\mu}_i(\bm{\theta}_g^*) + \bm{\Delta}^{(e)}_{gi} - \delta_g  \} ]\bm{G}_{s_g i}\tilde{\bm{Z}}_{i*} ) + o(n^{-1/2})\\
    =&n^{-1}\gamma \sum_{i=1}^n \E(\bm{G}_{s_g i})\E(\Psi[ \alpha_g\{ \tilde{\mu}_i(\bm{\theta}_g^*) + \bm{\Delta}^{(e)}_{gi} - \delta_g  \} ]\tilde{\bm{Z}}_{i*} ) + o(n^{-1/2})\\
    =& o(n^{-1/2}),
\end{align*}
where the second equality follows because $\bm{G}_{s_g*}$ is independent of $\{ \tilde{\mu}(\bm{\theta}_g^*),\bm{\Delta}^{(e)},\tilde{\bm{Z}} \}$ and the third because $\E(\bm{G}_{s_g i}) = 0$. This implies the first sum in \eqref{supp:equation:s1Diff} is $o_P(n^{-1/2})$. For the second and final term in \eqref{supp:equation:s1Diff}, we note that for $\tilde{\mu}_i(\bm{\theta}_g^*)$ as defined above,
\begin{align*}
    \sup_{i \in [n]}\abs*{ \frac{\partial^2}{\partial^2 \mu} h\{\mu_i(\bm{\theta}_g^*),\sigma_g\} - \frac{\partial^2}{\partial^2 \mu} h\{\tilde{\mu}_i(\bm{\theta}_g^*),\sigma_g\} } = o(n^{-1/4})
\end{align*}
by Assumption~\ref{supp:assumptions:FA} and Lemma~\ref{supp:lemma:PsiBounds}. An identical analysis to the one applied to the first term of \eqref{supp:equation:s1Diff} can then be used to show the second term of \eqref{supp:equation:s1Diff} is $o_P(n^{-1/2})$.

We next consider the difference
\begin{align*}
    \tilde{s}_2^{\known}(\bm{\theta}_g^*,\sigma_g) - s_2^{\known}(\bm{\theta}_g^*,\sigma_g) =& -\frac{1}{n\sigma_g^3} \gamma \sum_{i=1}^n r_{gi}\{y_{gi} - \mu_i(\bm{\theta}_g^*)\}\bm{G}_{s_g i}\\
    &+ n^{-1} \gamma\sum_{i=1}^n (1-r_{gi}) \bm{G}_{s_g i}\frac{\partial^2}{\partial \sigma \partial \mu} h\{ \mu_i(\theta),\sigma \}\\
    &+ n^{-1} \gamma^2\sum_{i=1}^n (1-r_{gi})\bm{G}_{s_g i}^2\frac{\partial^2}{\partial \sigma \partial^2 \mu} h( \tilde{\mu}_i,\sigma ) + o_P(n^{-1/2}),
\end{align*}
where $\tilde{\mu}_i = \alpha_i\mu(\bm{\theta}_g^*) + (1-\alpha_i)\{ \mu(\bm{\theta}_g^*) + \gamma \bm{G}_{s_g i} \}$ for some $\alpha_i \in [0,1]$. Identical techniques to those used to show $\norm*{ \tilde{\bm{s}}_1^{\known}(\bm{\theta}_g^*,\sigma_g) - \bm{s}_1^{\known}(\bm{\theta}_g^*,\sigma_g) }_2 = o_P(n^{-1/2})$ can be used to show $\tilde{s}_2^{\known}(\bm{\theta}_g^*,\sigma_g) - s_2^{\known}(\bm{\theta}_g^*,\sigma_g) = o_P(n^{-1/2})$. The details have been omitted. Putting all this together implies $\norm*{ \bm{s}^{\known}(\bm{\theta}_g^*,\sigma_g) - \tilde{\bm{s}}^{\known}(\bm{\theta}_g^*,\sigma_g) }_2 = o_P(n^{-1/2})$.

We next consider the Hessians:
\begin{align*}
    \bm{H}_{11}(\bm{\theta},\sigma) = -\nabla_{\bm{\theta}} \bm{s}^{\known}_1(\bm{\theta},\sigma) =& n^{-1}\sum_{i=1}^n [ r_{gi}/\sigma^2 - (1-r_{gi}) \frac{\partial^2}{\partial \mu^2} h\{ \mu_i(\bm{\theta}),\sigma \} ]\bm{Z}_{i*}\bm{Z}_{i*}^{\top}\\
    \tilde{\bm{H}}_{11}(\bm{\theta},\sigma) = -\nabla_{\bm{\theta}} \tilde{\bm{s}}^{\known}_1(\bm{\theta},\sigma) =& n^{-1}\sum_{i=1}^n [ r_{gi}/\sigma^2 - (1-r_{gi}) \frac{\partial^2}{\partial \mu^2} h\{ \mu_i(\bm{\theta})+\gamma\bm{G}_{s_g i},\sigma \} ]\bm{Z}_{i*}\bm{Z}_{i*}^{\top}\\
    \bm{H}_{22}(\bm{\theta},\sigma) = -\frac{\partial}{\partial \sigma} s^{\known}_2(\bm{\theta},\sigma) =& n^{-1}\sum_{i=1}^n [ 3r_{gi}\{ y_{gi} - \mu_i(\bm{\theta}) \}^2/\sigma^4 - (1-r_{gi}) \frac{\partial^2}{\partial \sigma^2} h\{ \mu_i(\bm{\theta}),\sigma \} ]\\
    \tilde{\bm{H}}_{22}(\bm{\theta},\sigma) = -\frac{\partial}{\partial \sigma} \tilde{s}^{\known}_2(\bm{\theta},\sigma) =& n^{-1}\sum_{i=1}^n [ 3r_{gi}\{ y_{gi} - \mu_i(\bm{\theta}) - \gamma\bm{G}_{s_g i} \}^2/\sigma^4 \\& - (1-r_{gi}) \frac{\partial^2}{\partial \sigma^2} h\{ \mu_i(\bm{\theta})+\gamma\bm{G}_{s_g i},\sigma \} ]\\
    \bm{H}_{12}(\bm{\theta},\sigma) = -\nabla_{\bm{\theta}} s^{\known}_2(\bm{\theta},\sigma) =& n^{-1}\sum_{i=1}^n [ 2r_{gi}\{y_{gi} - \mu_i(\bm{\theta})\}/\sigma^3\\& - (1-r_{gi}) \frac{\partial^2}{\partial\sigma \partial\mu}h\{\mu_i(\bm{\theta}),\sigma\} ]\bm{Z}_{i*}\\
    \tilde{\bm{H}}_{12}(\bm{\theta},\sigma) = -\nabla_{\bm{\theta}} \tilde{s}^{\known}_2(\bm{\theta},\sigma) =& n^{-1}\sum_{i=1}^n [ 2r_{gi}\{y_{gi} - \mu_i(\bm{\theta})-\gamma\bm{G}_{s_g i}\}/\sigma^3\\& - (1-r_{gi}) \frac{\partial^2}{\partial\sigma \partial\mu}h\{\mu_i(\bm{\theta})+\gamma\bm{G}_{s_g i},\sigma\} ]\bm{Z}_{i*}.
\end{align*}
Let $B(\bm{\theta},\sigma; \epsilon)$ be the Euclidean ball with radius $\epsilon$ and centered at $\{\bm{\theta},\sigma\}$ defined above. Using Lemma~\ref{supp:lemma:PsiBounds}, it is straightforward to show that
\begin{align*}
    &\sup_{\{\bm{\theta},\sigma\} \in B(\bm{\theta}_g,\sigma_g; \epsilon)}\norm*{ \bm{H}_{ij}(\bm{\theta},\sigma) - \bm{H}_{ij}(\bm{\theta}_g,\sigma_g) }_2 = O_P(\epsilon), \quad i,j \in [2]\\
    &\sup_{\{\bm{\theta},\sigma\} \in B(\bm{\theta}_g,\sigma_g; \epsilon)}\norm*{ \tilde{\bm{H}}_{ij}(\bm{\theta}_g,\sigma_g) - \bm{H}_{ij}(\bm{\theta}_g,\sigma_g) }_2 = O_P(\epsilon), \quad i,j \in [2]\\
    &\norm*{ \tilde{\bm{H}}_{ij}(\bm{\theta}_g,\sigma_g) - \E\{ \tilde{\bm{H}}_{ij}(\bm{\theta}_g,\sigma_g) \} }_2 = O_P(n^{-1/2}), \quad i,j \in [2].
\end{align*}
Let $\bm{H}_g^{\known}$ and $\tilde{\bm{H}}_g^{\known}$ be the $(d+K+1)\times (d+K+1)$ be the minus Hessians of $f^{\known}(\bm{\theta},\sigma)$ and $\tilde{f}^{\known}(\bm{\theta},\sigma)$ evaluated at $(\bm{\theta}_g^*,\sigma_g)$. Note that the first $(d+K)\times (d+K)$ and last diagonal elements of $\bm{H}_g^{\known}$ and $\tilde{\bm{H}}_g^{\known}$ are given by $\bm{H}_{11}(\bm{\theta}_g^*,\sigma_g),\tilde{\bm{H}}_{11}(\bm{\theta}_g^*,\sigma_g)$ and $H_{22}(\bm{\theta}_g^*,\sigma_g),\tilde{H}_{22}(\bm{\theta}_g^*,\sigma_g)$, and the off diagonal is given by $\bm{H}_{12}(\bm{\theta}_g^*,\sigma_g),\tilde{\bm{H}}_{12}(\bm{\theta}_g^*,\sigma_g)$. It is straightforward to show that $\E( \tilde{\bm{H}}_g ) \succeq c I_{d+K+1}$ for some constant $c>0$ and all $n$ large enough. Putting all this together implies
\begin{align*}
    \{n\bm{H}_g^{\known}\}^{1/2} \{\hat{\bm{\eta}}_g^{\known} - \bm{\eta}_g^*\} = \{ \E( \tilde{\bm{H}}_g ) \}^{-1/2}\{n^{1/2}\tilde{\bm{s}}^{\known}(\bm{\theta}_g^*,\sigma_g)\} + o_P(1).
\end{align*}
The result then follows by an application of the Lindeberg central limit theorem.
\end{proof}

\begin{theorem}
\label{supp:theorem:InferenceBeta}
Suppose the assumptions of Theorem~\ref{supp:theorem:Omega} and Lemma~\ref{supp:lemma:Cknown} hold, let $\hat{\bm{C}}$, $\hat{\bm{Z}}$, $\bm{\theta}_g^*$, $\hat{\bm{\theta}}_g^{\ipw}$, and $\hat{\sigma}_g^{\ipw}$ be as defined in the statements of Lemma~\ref{supp:lemma:IPWest} and \ref{supp:lemma:IPWestVar}, let $\bm{H}_g^{\known}$ be as defined in the statement of Lemma~\ref{supp:lemma:Cknown}, and define the log-likelihood function
\begin{align*}
    f_g(\bm{\theta},\sigma) =& n^{-1}\sum_{i=1}^n \left[-r_{gi}\{\bm{y}_{gi} - \hat{\mu}_{i}(\bm{\theta})\}^2/(2\sigma^2)\right.\\
    &+ \left.(1-r_{gi})\log( 1-\smallint \phi(\epsilon)\Psi[ \alpha_g\{ \hat{\mu}_{i}(\bm{\theta}) + \sigma \epsilon - \delta_g \} ]\text{d}\epsilon )\right], \quad \hat{\mu}_{i}(\bm{\theta}) = \hat{\bm{Z}}_{i \bigcdot}^{\top}\bm{\theta},
\end{align*}
where $\phi$ is the probability density function of the standard normal. Let $m \geq 1$ be a constant integer, and define $\hat{\bm{\eta}}_g^{\fs} = ( \{ \hat{\bm{\theta}}_g^{\fs} \}^{\top}, \hat{\sigma}_g^{\fs} )^{\top}$ to be the estimator for $\bm{\eta}_g^*=(\{\bm{\theta}_g^*\}^{\top},\sigma_g)^{\top}$ that, for starting point $\hat{\bm{\eta}}_g^{\ipw} = ( \{ \hat{\bm{\theta}}_g^{\ipw} \}^{\top}, \hat{\sigma}_g^{\ipw} )^{\top}$, uses $J \in [m]$ iterations of Fisher scoring to maximize $f_g(\bm{\theta},\sigma)$. Then as $n \to \infty$ and for $\hat{\bm{H}}_g^{\fs}$ the plug-in estimator for $\bm{H}_g^{\known}$ that plugs in $\hat{\bm{C}}$ for $\bm{C}$ and $\hat{\bm{\eta}}_g^{\fs}$ for $\bm{\eta}_g^*$,
\begin{align}
    \label{supp:equation:asyequ}
    &n^{1/2}\abs*{ \hat{\bm{\theta}}_{g_j}^{\fs} - \hat{\bm{\theta}}_{g_j}^{\known} } = o_P(1), \quad j \in [d_1]\\
    \label{supp:equation:InfoMat}
    &\abs*{ [\{\bm{H}_g^{\fs}\}^{-1}]_{rs} - [\{\bm{H}_g^{\known}\}^{-1}]_{rs} } = o_P(1), \quad r,s \in [d_1].
\end{align}
\end{theorem}

\begin{remark}
\label{supp:remark:betag1}
The first $d_1$ entries of $\hat{\bm{\theta}}_{g}^{\fs}$ and $\hat{\bm{\theta}}_{g}^{\known}$ are estimates for the first $d_1$ entries of $\bm{\beta}_g$, which are exactly the coefficients of interest. Theorem~\ref{supp:theorem:InferenceBeta} therefore implies estimation for and inference on the coefficients of interest with our estimated $\bm{C}$ is asymptotically equivalent to that when $\bm{C}$ is known. A trivial corollary of Lemma~\ref{supp:lemma:Cknown} and Theorem~\ref{supp:theorem:InferenceBeta} is that our Fisher scoring estimator for the coefficients of interest is asymptotically normal, where the first $d_1 \times d_1$ block of $\{\bm{H}_g^{\fs}\}^{-1}$ is an estimator for its asymptotic variance.
\end{remark}

\begin{proof}
Let $\Theta,\mathcal{S}$ be as defined in the statement of Lemma~\ref{supp:lemma:Cknown} and define
\begin{align*}
    h(\mu,\sigma) = \log[ \smallint \Psi\{-\alpha_g(\mu + \sigma e - \delta_g) \} \phi(e)de ].
\end{align*}
We prove Theorem~\ref{supp:theorem:InferenceBeta} by showing that (a) $\{\hat{\bm{\theta}}_g,\hat{\sigma}_g\} = \argmax_{\bm{\theta} \in \Theta,\sigma\in\mathcal{S}} f_{g}(\bm{\theta},\sigma)$ are asymptotically equivalent to $\{\hat{\bm{\theta}}_g^{\known},\hat{\sigma}_g^{\known}\}$ defined in the statement of Lemma~\ref{supp:lemma:Cknown}, and (b) that $\{\hat{\bm{\theta}}_g^{\fs},\sigma_g^{\fs}\}$ are asymptotically equivalent to $\{\hat{\bm{\theta}}_g,\hat{\sigma}_g\}$.

For (a), let $\bm{R}_g = \diag(r_{g1},\ldots,r_{gn})$ and $\tilde{\bm{X}}_1=P_{\bm{X}_2}^{\perp}\bm{X}_1$. Then for $\bm{\theta} \in \Theta$ and $\bar{\bm{\ell}} \in \mathbb{R}^K$ entries $d_1+1,\ldots,d_1+K$ of $\bm{\theta}$,
\begin{align*}
    f_g(\bm{\theta},\sigma) - f_g^{\known}(\bm{\theta},\sigma) =& -\sigma^{-2}(n^{-1/2}\bm{y_g}^{\top})^{\top}\bm{R}_g\bm{\Delta}_C\bar{\bm{\ell}}\\& + \sigma^{-2}\bar{\bm{\ell}}^{\top}\bm{\Delta}_C^{\top}\bm{R}_g \{\tilde{\bm{C}}\bm{v} + (n^{-1/2}\tilde{\bm{X}}_1)(\bm{v}^{\top}\bm{\Omega}_1)^{\top}\}\bar{\bm{\ell}}\\
    &+ (2\sigma^2)^{-1}\bar{\bm{\ell}}^{\top}\bm{\Delta}_C^{\top}\bm{R}_g \bm{\Delta}_C\bar{\bm{\ell}} + n^{-1}\sum_{i=1}^n (1-r_{gi})\bm{\delta}_i^{\top}\bar{\bm{\ell}} a_{gi}(\bm{\theta},\sigma)\\
    \bm{\Delta}_C =& (\hat{\bm{C}}_{\perp} - \tilde{\bm{C}}\bm{v})  + (n^{-1/2}\tilde{\bm{X}}_1)( \hat{\bm{\Omega}}_1 - \bm{v}^{\top}\bm{\Omega}_1 )^{\top}\\
    \bm{\delta}_i =& n^{1/2}(\hat{\bm{C}}_{\perp_{i \bigcdot}} - \bm{v}^{\top}\tilde{\bm{C}}_{i \bigcdot}) + (\hat{\bm{\Omega}}_1 - \bm{v}^{\top}\bm{\Omega}_1)\tilde{\bm{X}}_{1_{i \bigcdot}}, \quad i \in [n]\\
    a_{gi}(\bm{\theta},\sigma) =& -\alpha_g \frac{ \smallint \dot{\Psi}[-\alpha_g\{ \mu_i(\bm{\theta}) + \zeta_{gi} \bm{\delta}_i^{\top}\bar{\bm{\ell}} + \sigma_g \epsilon - \delta_g \}] \phi(\epsilon)d\epsilon }{ \smallint \Psi[-\alpha_g\{ \mu_i(\bm{\theta}) + \zeta_{gi} \bm{\delta}_i^{\top}\bar{\bm{\ell}} + \sigma \epsilon - \delta_g \}] \phi(\epsilon)d\epsilon }, \quad \zeta_{gi} \in [0,1],\quad i \in [n],
\end{align*}
where Theorems~\ref{supp:theorem:ChatProp} and \ref{supp:theorem:Omega} imply $\norm*{ \bm{\Delta}_C }_2 = O_P(\lambda^{-1/2+\epsilon})$ for any $\epsilon>0$, and Corollary~\ref{supp:corollary:InfNormC} and Theorem~\ref{supp:theorem:Omega} imply $\sup_{\bm{\theta}_g \in \Theta, i \in [n]}\abs*{\bm{\delta}_i^{\top}\bar{\bm{\ell}}_g} = O_P(n^{-\eta})$ for some sufficiently small $\eta>0$. Since $\abs*{a_{gi}(\bm{\theta},\sigma)}\leq c$ for some constant $c>0$ by Lemma~\ref{supp:lemma:PsiBounds}, this implies $\sup_{\bm{\theta} \in \Theta,\sigma\in\mathcal{S}} \abs*{ f_g(\bm{\theta},\sigma) - f_g^{\known}(\bm{\theta},\sigma) }= O_P(n^{-\eta})$ for some sufficiently small $\eta>0$. Therefore, for $\tilde{f}^{\known}$ as defined in Lemma~\ref{supp:lemma:Cknown},
\begin{align*}
    &\sup_{\bm{\theta} \in \Theta,\sigma \in \mathcal{S}} \abs*{ f_g(\bm{\theta},\sigma) - \E\{ \tilde{f}_g^{\known}(\bm{\theta},\sigma) \mid \bm{C},\bm{G} \} } \leq  \underbrace{\sup_{\bm{\theta} \in \Theta,\sigma \in \mathcal{S}} \abs*{ f_g(\bm{\theta},\sigma) -  f_g^{\known}(\bm{\theta},\sigma) }}_{O_P(n^{-\eta})}\\
    &+ \underbrace{\sup_{\bm{\theta} \in \Theta,\sigma \in \mathcal{S}} \abs*{ f_g^{\known}(\bm{\theta},\sigma) -  \E\{ \tilde{f}_g^{\known}(\bm{\theta},\sigma) \mid \bm{C},\bm{G} \} }}_{\text{$o_P(1)$ by properties of \eqref{supp:equation:ftildefknown} in Lemma~\ref{supp:lemma:Cknown}}},
\end{align*}
meaning $\norm*{ \hat{\bm{\theta}}_g - \bm{\theta}_g^* }_2 = o_P(1)$ and $\abs*{\hat{\sigma}_g - \sigma_g}=o_P(1)$. Next, define
\begin{align*}
    &\bm{Z} =  [(n^{-1/2}\tilde{\bm{X}}_1), \{ \tilde{\bm{C}}\bm{v} + (n^{-1/2}\tilde{\bm{X}}_1)(\bm{v}^{\top}\bm{\Omega}_1)^{\top} \}, (n^{-1/2}\bm{X}_2)], \quad n^{1/2}\bm{Z}\bm{\theta}=(\mu_1(\bm{\theta}),\ldots,\mu_n(\bm{\theta}))^{\top}\\
    &\bm{s}_{g1}^{\known}(\bm{\theta},\sigma) = \nabla_{\bm{\theta}} f_g^{\known}(\bm{\theta},\sigma) = \sigma^{-2}\bm{Z}^{\top}\bm{R}_g(n^{-1/2}\bm{y}_g - \bm{Z}\bm{\theta}) + n^{-1/2}\bm{Z}^{\top}(I_n - \bm{R}_g) \dot{\bm{h}}_1(\bm{\theta},\sigma)\\
    &s_{g2}^{\known}(\bm{\theta},\sigma) = \frac{\partial}{\partial \sigma} f_g^{\known}(\bm{\theta},\sigma) = -\frac{\bm{1}_n^{\top}\bm{R}_g \bm{1}_n}{n\sigma} + \sigma^{-3}(n^{-1/2}\bm{y}_g - \bm{Z}\bm{\theta})^{\top}\bm{R}_g (n^{-1/2}\bm{y}_g - \bm{Z}\bm{\theta})\\& + n^{-1}\bm{1}_n^{\top}(I_n-\bm{R}_g) \dot{\bm{h}}_2(\bm{\theta},\sigma)\\
    &\dot{\bm{h}}_1(\bm{\theta},\sigma) = \left( \frac{\partial}{\partial \mu}h\{\mu_1(\bm{\theta}),\sigma\},\ldots,\frac{\partial}{\partial \mu}h\{\mu_n(\bm{\theta}),\sigma\} \right)^{\top}\\& \dot{\bm{h}}_2(\bm{\theta},\sigma) = \left( \frac{\partial}{\partial \sigma}h\{\mu_1(\bm{\theta}),\sigma\},\ldots,\frac{\partial}{\partial \sigma}h\{\mu_n(\bm{\theta}),\sigma\} \right)^{\top}\\
    &\bm{H}_{11}^{\known}(\bm{\theta},\sigma) = -\nabla_{\bm{\theta}}^2 f_g^{\known}(\bm{\theta},\sigma)=\bm{Z}^{\top}\{\sigma^{-2}\bm{R}_g - (I_n - \bm{R}_g)\ddot{\bm{H}}_{11}(\bm{\theta},\sigma)\}\bm{Z}\\
    &H_{22}^{\known}(\bm{\theta},\sigma) = -\frac{\partial^2}{\partial \sigma^2} f_g^{\known}(\bm{\theta},\sigma) = 3\sigma^{-4}  (n^{-1/2}\bm{y}_g - \bm{Z}\bm{\theta})^{\top}\bm{R}_g (n^{-1/2}\bm{y}_g - \bm{Z}\bm{\theta}) - \frac{\bm{1}_n^{\top}\bm{R}_g \bm{1}_n}{n\sigma^2}\\& - n^{-1}\bm{1}_n^{\top}(I_n-\bm{R}_g)\ddot{\bm{H}}_{22}(\bm{\theta},\sigma)\bm{1}_n \\
    &\bm{H}_{12}^{\known}(\bm{\theta},\sigma) = -\nabla_{\bm{\theta}} s_{g2}^{\known}(\bm{\theta},\sigma)= 2\sigma^{-3} \bm{Z}^{\top}\bm{R}_g(n^{-1/2}\bm{y}_g- \bm{Z}\bm{\theta})\\& - n^{-1/2}\bm{Z}^{\top} (I_n - \bm{R}_g)\ddot{\bm{H}}_{12}(\bm{\theta},\sigma)\bm{1}_n\\
    &\ddot{\bm{H}}_{11}(\bm{\theta},\sigma) = \diag\left[ \frac{\partial^2}{\partial \mu^2} h\{\mu_1(\bm{\theta}),\sigma\},\ldots, \frac{\partial^2}{\partial \mu^2} h\{\mu_n(\bm{\theta}),\sigma\}\right]\\ &\ddot{\bm{H}}_{22}(\bm{\theta},\sigma) = \diag\left[ \frac{\partial^2}{\partial \sigma^2} h\{\mu_1(\bm{\theta}),\sigma\},\ldots, \frac{\partial^2}{\partial \sigma^2} h\{\mu_n(\bm{\theta}),\sigma\}\right] \\ &\ddot{\bm{H}}_{12}(\bm{\theta},\sigma) = \diag\left[ \frac{\partial^2}{\partial\mu \partial \sigma} h\{\mu_1(\bm{\theta}),\sigma\},\ldots, \frac{\partial^2}{\partial\mu \partial \sigma} h\{\mu_n(\bm{\theta}),\sigma\}\right]
\end{align*}
and
\begin{align}
    \bm{s}_{g1}(\bm{\theta},\sigma) =& \nabla_{\bm{\theta}} f_g(\bm{\theta},\sigma) = \bm{s}_{g1}^{\known}(\bm{\theta},\sigma) + \sigma^{-2}\bm{\delta}^{\top}\bm{R}_g (n^{-1/2}\bm{y}_g - \bm{Z}\bm{\theta}) + \sigma^{-2}\bm{Z}^{\top}\bm{R}_g\bm{\delta}\bm{\theta}\nonumber\\
    &+ n^{-1/2}\bm{Z}^{\top}(I_n - \bm{R}_g)\bm{\varepsilon}_1(\bm{\theta},\sigma) + n^{-1/2}\bm{\delta}^{\top}(I_n - \bm{R}_g) \dot{\bm{h}}(\bm{\theta},\sigma) \nonumber \\
    &+ n^{-1/2}\bm{\delta}^{\top}(I_n - \bm{R}_g)\bm{\varepsilon}_1(\bm{\theta},\sigma) + o_P(n^{-1/2})\nonumber\\
    s_{g2}(\bm{\theta},\sigma) =& \frac{\partial}{\partial \sigma} f_g(\bm{\theta},\sigma) = s_{g2}^{\known}(\bm{\theta},\sigma) - 2\sigma^{-3}(n^{-1/2}\bm{y}_g-\bm{Z}\bm{\theta})^{\top}\bm{R}_g \bm{\delta}\bm{\theta}\nonumber\\& + n^{-1}\bm{1}^{\top}(I_n-\bm{R}_g)\bm{\varepsilon}_2(\bm{\theta},\sigma) +  o_P(n^{-1/2}) \nonumber\\
    \bm{H}_{11}(\bm{\theta},\sigma) =& \nabla^2_{\bm{\theta}}f_g(\bm{\theta},\sigma), \quad H_{22}(\bm{\theta},\sigma) = \frac{\partial^2}{\partial \sigma^2} f_g(\bm{\theta},\sigma), \quad \bm{H}_{12}(\bm{\theta},\sigma) = \nabla_{\bm{\theta}} \frac{\partial}{\partial \sigma} f_g(\bm{\theta},\sigma)\nonumber\\
    \bm{\delta} =& [\bm{0}_{n \times d_1}, \bm{\Delta}_C,\bm{0}_{n \times d_2}]\\
    \bm{\varepsilon}_1(\bm{\theta},\sigma)=& \left( \frac{\partial}{\partial \mu}h\{\mu_1(\bm{\theta}) + n^{1/2}\bm{\delta}_{1*}^{\top}\bm{\theta},\sigma\},\ldots,\frac{\partial}{\partial \mu}h\{\mu_n(\bm{\theta}) + n^{1/2}\bm{\delta}_{n*}^{\top}\bm{\theta},\sigma\} \right)^{\top} - \dot{\bm{h}}_1(\bm{\theta},\sigma) \nonumber\\
    \bm{\varepsilon}_2(\bm{\theta},\sigma)=& \left( \frac{\partial}{\partial \sigma}h\{\mu_1(\bm{\theta}) + n^{1/2}\bm{\delta}_{1*}^{\top}\bm{\theta},\sigma\},\ldots,\frac{\partial}{\partial \sigma}h\{\mu_n(\bm{\theta}) + n^{1/2}\bm{\delta}_{n*}^{\top}\bm{\theta},\sigma\} \right)^{\top} - \dot{\bm{h}}_2(\bm{\theta},\sigma) \nonumber,
\end{align}
where the $o_P(n^{-1/2})$ term is uniform over all $\{\bm{\theta},\sigma\} \in \Theta \times \mathcal{S}$. We prove two critical lemmas regarding the behavior $\bm{s}_{g1},s_{g2}$ and $\bm{H}_{11},H_{22},\bm{H}_{12}$.
\begin{lemma}
\label{supp:lemma:sg1sg2}
Suppose the assumptions of Theorem~\ref{supp:theorem:InferenceBeta} hold and let $\tilde{\bm{\theta}} \in \mathbb{R}^{d+K}$ and $\tilde{\sigma}>0$ be such that $\norm*{ \tilde{\bm{\theta}}-\bm{\theta}_g^* }_2,\abs*{\tilde{\sigma} - \sigma_g} = O_P(n^{-1/2})$. Then
\begin{align*}
    \norm*{ \bm{s}_{g1}(\tilde{\bm{\theta}},\tilde{\sigma}) - \bm{s}_{g1}^{\known}(\tilde{\bm{\theta}},\tilde{\sigma}) }_2, \, \norm*{ s_{g2}(\tilde{\bm{\theta}},\tilde{\sigma}) - s_{g2}^{\known}(\tilde{\bm{\theta}},\tilde{\sigma}) }_2 = o_P(n^{-1/2}).
\end{align*}
\end{lemma}
\begin{proof}
We prove the result for $\bm{s}_{g1}$. The proof for $s_{g2}$ uses identical arguments, and has been omitted. The proof of Corollary~\ref{supp:corollary:CtC} can be used to show that
\begin{align*}
    &\norm*{ \bm{\delta}^{\top}\bm{R}_g(n^{-1/2}\bm{y}_g - \bm{Z}\bm{\theta}) }_2 \leq \norm*{ n^{-1/2}\bm{\delta}^{\top}\bm{R}_g \bm{y}_g }_2 + \norm*{ \bm{\delta}^{\top}\bm{R}_g \bm{Z} }_2\norm*{\bm{\theta}}_2 \leq  (1 + \norm*{\bm{\theta}}_2) o_P(n^{-1/2})\\
    &\norm*{ \bm{Z}^{\top}\bm{R}_g\bm{\delta}\bm{\theta} }_2 \leq \norm*{ \bm{Z}^{\top}\bm{R}_g\bm{\delta} }_2\norm*{\bm{\theta}}_2 = \norm*{\bm{\theta}}_2 o_P(n^{-1/2})
\end{align*}
for any $\bm{\theta} \in \Theta$. Next, for any $\{\bm{\theta},\sigma\} \in \Theta \times \mathcal{S}$,
\begin{align*}
    n^{-1/2}\bm{\delta}^{\top}(I_n - \bm{R}_g)\bm{\varepsilon}_1(\bm{\theta},\tilde{\sigma}) =& \bm{\delta}^{\top}\bm{V}(\bm{\theta},\sigma)\bm{\delta}\bm{\theta}\\
    \bm{V}(\bm{\theta},\sigma) =& \diag\left[ (1-r_{g1})\frac{\partial^2}{\partial \mu^2} h\{ \mu_1(\bm{\theta}_g) + \zeta_{1 }(n^{1/2}\bm{\delta}_{1 \bigcdot}^{\top}\bm{\theta}_g) ,\sigma\},\right.\\
    &\left.\ldots,(1-r_{gn})\frac{\partial^2}{\partial \mu^2} h\{ \mu_n(\bm{\theta}) + \zeta_{n }(n^{1/2}\bm{\delta}_{n \bigcdot}^{\top}\bm{\theta}),\sigma \} \right].
\end{align*}
for some $\zeta_1,\ldots,\zeta_n \in [0,1]$ that depend on $\bm{\theta}$ and $\sigma$. Since $\norm*{ \bm{V}(\bm{\theta},\sigma)}_2 \leq c$ for some constant $c>0$ that does not depend on $\bm{\theta}$ or $\sigma$ by Lemma~\ref{supp:lemma:PsiBounds}, $\sup_{\{\bm{\theta},\sigma\} \in \Theta \times \mathcal{S}} \norm*{n^{-1/2} \bm{\delta}^{\top}(I_n - \bm{R}_g)\bm{\varepsilon}_1(\bm{\theta},\tilde{\sigma}) }_2 = o_P(n^{-1/2})$ by Theorems~\ref{supp:theorem:ChatProp} and \ref{supp:theorem:Omega}. Next, since the entries of $\dot{\bm{h}}_1(\bm{\theta},\sigma)$ have uniformly bounded gradient (and entries) by Lemma~\ref{supp:lemma:PsiBounds} and $\norm*{\bm{\delta}}_2 = o_P(1)$,
\begin{align*}
    \norm*{ n^{-1/2}\bm{\delta}^{\top}(I_n-\bm{R}_g)\dot{\bm{h}}_1(\tilde{\bm{\theta}},\sigma) }_2 = \norm*{ n^{-1/2}\hat{\bm{z}}^{\top}\bm{Q}^{\top}(I_n-\bm{R}_g)\dot{\bm{h}}_1(\bm{\theta}_g^*,\sigma) }_2 + o_P(n^{-1/2})
\end{align*}
by Theorem~\ref{supp:theorem:ChatProp} for $\hat{\bm{z}}$ and $\bm{Q}$ as defined in \eqref{supp:equation:CperpExp}. Lemma~\ref{supp:lemma:Cte}, along with the same techniques used to prove Corollary~\ref{supp:corollary:CtC}, can be used to prove $\norm*{ n^{-1/2}\hat{\bm{z}}^{\top}\bm{Q}^{\top}(I_n-\bm{R}_g)\dot{\bm{h}}(\bm{\theta}_g^*) }_2 = o_P(n^{-1/2})$. The details have been omitted. Lastly,
\begin{align*}
    &n^{-1/2}\bm{Z}^{\top}(I_n-\bm{R}_g)\bm{\varepsilon}_1(\tilde{\bm{\theta}},\tilde{\sigma}) = \bm{Z}^{\top}(I_n-\bm{R}_g)\ddot{\bm{H}}_{11}(\tilde{\bm{\theta}},\tilde{\sigma})\bm{\delta}\tilde{\bm{\theta}} + \bm{Z}^{\top}(I_n-\bm{R}_g)\bm{r}(\tilde{\bm{\theta}},\tilde{\sigma})\\
    &\bm{r}(\bm{\theta},\sigma) = \frac{1}{2n^{1/2}}\left( ( n^{1/2}\bm{\delta}_{1*}^{\top} \bm{\theta} )^2\frac{\partial^3}{\partial \mu^3}h\{\mu_1(\bm{\theta})+\zeta_1,\sigma\} ,\ldots, (n^{1/2}\bm{\delta}_{n*}^{\top} \bm{\theta} )^2\frac{\partial^3}{\partial \mu^3}h\{\mu_n(\bm{\theta})+\zeta_n,\sigma\}\right)^{\top}
\end{align*}
for $\zeta_i = \alpha_i n^{1/2}\bm{\delta}_{i*}^{\top} \bm{\theta}$ and some $\alpha_i \in [0,1]$. Since the $d+1,\ldots,d+K$ entries of $n^{1/2}\bm{\theta}_g^*$ are $O(\lambda^{1/2})$ by Assumption~\ref{supp:assumptions:FA}, Corollary~\ref{supp:corollary:InfNormC} implies $\sup_{i \in [n]} (n^{1/2}\bm{\delta}_{i*}^{\top} \tilde{\bm{\theta}})^2 = o_P(n^{-1/2})$. Therefore, $\norm*{\bm{Z}^{\top}(I_n-\bm{R}_g)\bm{r}(\tilde{\bm{\theta}},\tilde{\sigma})}_2 = o_P(n^{1/2})$ by Lemma~\ref{supp:lemma:PsiBounds}. Finally, since $\norm*{ \ddot{\bm{H}}_{11}(\tilde{\bm{\theta}},\tilde{\sigma}) - \ddot{\bm{H}}_{11}(\bm{\theta}_g^*, \sigma_g) }_2 = O_P(n^{-1/2})$ by Lemma~\ref{supp:lemma:PsiBounds} and $\norm*{\bm{\delta}}_2 = o_P(1)$,
\begin{align*}
    \norm*{ \bm{Z}^{\top}(I_n-\bm{R}_g)\ddot{\bm{H}}_{11}(\tilde{\bm{\theta}},\tilde{\sigma})\bm{\delta}\tilde{\bm{\theta}} }_2 = O_P\{ \norm*{ \bm{Z}^{\top}(I_n-\bm{R}_g)\ddot{\bm{H}}_{11}(\bm{\theta}_g^*, \sigma_g) \bm{\delta} }_2 \} + o_P(n^{-1/2}).
\end{align*}
An application of Lemma~\ref{supp:lemma:Cte} and identical techniques used to prove Corollary~\ref{supp:corollary:CtC} can be used to show $\norm*{ \bm{Z}^{\top}(I_n-\bm{R}_g)\ddot{\bm{H}}_{11}(\bm{\theta}_g^*, \sigma_g) \bm{\delta} }_2 = o_P(n^{-1/2})$. This completes the proof.
\end{proof}

\begin{lemma}
\label{supp:lemma:Hessian}
Suppose the assumptions in the statement of Theorem~\ref{supp:theorem:InferenceBeta} hold, let $\bm{H}_g^{\known}$ be as defined in Lemma~\ref{supp:lemma:Cknown}, and let $B(\bm{\theta},\sigma;\epsilon)$ be the Euclidean ball centered at $(\bm{\theta}^{\top},\sigma)^{\top}$ with radius $\epsilon>0$. Then for all $\epsilon$ sufficiently small,
\begin{subequations}
\begin{align}
\label{supp:equation:Hessian:1}
    &\sup_{\{\bm{\theta},\sigma\} \in B(\bm{\theta}_g^*,\sigma_g;\epsilon)} \norm*{ \bm{H}(\bm{\theta},\sigma) - \bm{H}_g^{\known} }_2 = O_P(\epsilon) + o_P(1)\\
\label{supp:equation:Hessian:2}
    &\sup_{\{\bm{\theta},\sigma\} \in B(\bm{\theta}_g^*,\sigma_g;\epsilon)} \norm*{ \bm{F}(\bm{\theta},\sigma) - \bm{H}_g^{\known} }_2 = O_P(\epsilon) + o_P(1)
\end{align}
\end{subequations}
where $\bm{H}(\bm{\theta},\sigma) = -\nabla^2 f_{g}(\bm{\theta},\sigma)$ and $\bm{F}(\bm{\theta},\sigma) = \begin{pmatrix} \bm{F}_{11}(\bm{\theta},\sigma) & \bm{F}_{12}(\bm{\theta},\sigma)\\ \bm{F}_{12}(\bm{\theta},\sigma)^{\top} & F_{22}(\bm{\theta},\sigma) \end{pmatrix}$ is the Fisher information matrix given by
\begin{align*}
    \bm{F}_{11}(\bm{\theta},\sigma) =& n^{-1}\sum_{i=1}^n \left( \frac{q\{\mu_i(\bm{\theta}),\sigma\}}{\sigma^2} - [1-q\{\mu_i(\bm{\theta}),\sigma\}]\frac{\partial^2}{\partial \mu^2} h\{\mu_i(\bm{\theta}),\sigma\}  \right)\hat{\bm{Z}}_{i*}\hat{\bm{Z}}_{i*}^{\top}\\
    \bm{F}_{12}(\bm{\theta},\sigma) =& n^{-1}\sum_{i=1}^n \left(2\sigma^{-2} \smallint e \Psi[\alpha_g\{\mu_i(\bm{\theta}) + \sigma e - \delta_g\}] \phi(e)de\right.\\&\left. - [1-q\{\mu_i(\bm{\theta}),\sigma\}]\frac{\partial^2}{\partial \mu \partial \sigma} h\{\mu_i(\bm{\theta}),\sigma\}\right)\hat{\bm{Z}}_{i*}\\
    F_{22}(\bm{\theta},\sigma) =& n^{-1}\sum_{i=1}^n \left( \sigma^{-2}\smallint (3e^2 - 1) \Psi[\alpha_g\{\mu_i(\bm{\theta}) + \sigma e - \delta_g\}] \phi(e)de \right. \\& \left. - [1-q\{\mu_i(\bm{\theta}),\sigma\}]\frac{\partial^2}{\partial \sigma^2} h\{\mu_i(\bm{\theta}),\sigma\}\right)\\
    q(\mu,\sigma) =& \smallint \Psi\{\alpha_g (\mu + \sigma e - \delta_g)\} \phi(e)de
\end{align*}
\end{lemma}
\begin{remark}
\label{supp:remark:InferenceBeta:Done}
Result \eqref{supp:equation:InfoMat} in the statement of Theorem~\ref{supp:theorem:InferenceBeta} follows from \eqref{supp:equation:asyequ} and \eqref{supp:equation:Hessian:1}. We therefore need only prove \eqref{supp:equation:asyequ} to complete the proof of Theorem~\ref{supp:theorem:InferenceBeta}.
\end{remark}
\begin{remark}
Since $\norm*{ \bm{H}_g^{\known} - \E\{ \bm{H}_g^{\known} \mid \bm{G},\bm{C} \} }_2 = o_P(1)$, Lemma~\ref{supp:lemma:Hessian} implies 
\begin{align*}
    \sup_{\{\bm{\theta},\sigma\} \in B(\bm{\theta}_g^*,\sigma_g;\epsilon)} \norm*{ \bm{H}(\bm{\theta},\sigma) - \E\{ \bm{H}_g^{\known} \mid \bm{G},\bm{C} \} }_2 = O_P(\epsilon) + o_P(1).
\end{align*}
\end{remark}
\begin{proof}
We first note that
\begin{align*}
    &\bm{H}(\bm{\theta},\sigma) = \begin{pmatrix} \bm{H}_{11}(\bm{\theta},\sigma) & \bm{H}_{12}(\bm{\theta},\sigma)\\ \bm{H}_{12}(\bm{\theta},\sigma)^{\top} & H_{22} (\bm{\theta},\sigma) \end{pmatrix}\\
    &\bm{H}_g^{\known} = \bm{H}^{\known}(\bm{\theta}_g^*,\sigma_g) = \begin{pmatrix} \bm{H}_{11}^{\known}(\bm{\theta},\sigma) & \bm{H}_{12}^{\known}(\bm{\theta},\sigma)\\ \bm{H}_{12}^{\known}(\bm{\theta},\sigma)^{\top} & H^{\known}_{22} (\bm{\theta},\sigma) \end{pmatrix}
\end{align*}
and
\begin{align*}
    \sup_{\{\bm{\theta},\sigma\} \in B(\bm{\theta}_g^*,\sigma_g;\epsilon)} \norm*{ \bm{H}(\bm{\theta},\sigma) - \bm{H}_g^{\known} }_2 \leq & \sup_{\{\bm{\theta},\sigma\} \in B(\bm{\theta}_g^*,\sigma_g;\epsilon)} \norm*{ \bm{H}^{\known}(\bm{\theta},\sigma) - \bm{H}_g^{\known} }_2\\
    &+ \sup_{\{\bm{\theta},\sigma\} \in B(\bm{\theta}_g^*,\sigma_g;\epsilon)} \norm*{ \bm{H}(\bm{\theta},\sigma) - \bm{H}^{\known}(\bm{\theta},\sigma) }_2,
\end{align*}
where the first term after the $\leq$ is $O_P(\epsilon) + o_P(1)$ by the proof of Lemma~\ref{supp:lemma:Cknown}. Straightforward applications of Theorem~\ref{supp:theorem:ChatProp} and Lemma~\ref{supp:lemma:PsiBounds} can be used to show
\begin{align*}
    \sup_{\{\bm{\theta},\sigma\} \in B(\bm{\theta}_g^*,\sigma_g;\epsilon)} \norm*{ \bm{H}(\bm{\theta},\sigma) - \bm{H}^{\known}(\bm{\theta},\sigma) }_2 = o_P(1),
\end{align*}
which proves \eqref{supp:equation:Hessian:1}. For \eqref{supp:equation:Hessian:2}, we see that
\begin{align*}
    \sup_{\{\bm{\theta},\sigma\} \in B(\bm{\theta}_g^*,\sigma_g;\epsilon)} \norm*{ \bm{F}(\bm{\theta},\sigma) - \bm{H}_g^{\known} }_2 \leq & \sup_{\{\bm{\theta},\sigma\} \in B(\bm{\theta}_g^*,\sigma_g;\epsilon)} \norm*{ \bm{F}(\bm{\theta},\sigma) - \bm{F}(\bm{\theta}_g^*,\sigma_g) }_2\\ &+ \norm*{ \bm{F}(\bm{\theta}_g^*,\sigma_g) - \bm{H}_g^{\known} }_2.
\end{align*}
The same techniques used to prove $\norm*{ \bm{F}(\bm{\theta}_g^*,\sigma_g) - \bm{H}_g^{\known} }_2 = o_P(1)$. Lastly, An application of Lemma~\ref{supp:lemma:PsiBounds} can be used to prove
\begin{align*}
    \sup_{\{\bm{\theta},\sigma\} \in B(\bm{\theta}_g^*,\sigma_g;\epsilon)} \norm*{ \bm{F}(\bm{\theta},\sigma) - \bm{F}(\bm{\theta}_g^*,\sigma_g) }_2 = O_P(\epsilon) + o_P(1),
\end{align*}
which completes the proof.
\end{proof}

Returning to the proof of Theorem~\ref{supp:theorem:InferenceBeta}, the observation that $( \hat{\bm{\theta}}_g,\hat{\sigma}_g )^{\top}$ and $( \hat{\bm{\theta}}_g^{\known},\hat{\sigma}_g^{\known} )^{\top}$ are consistent for $(\{\bm{\theta}_g^*\}^{\top},\sigma_g)^{\top}$, as well as Lemma~\ref{supp:lemma:Hessian} imply
\begin{align*}
    0 = \begin{pmatrix} \bm{s}_{g1}(\hat{\bm{\theta}}_g,\hat{\sigma}_g)\\ s_{g2}(\hat{\bm{\theta}}_g,\hat{\sigma}_g) \end{pmatrix} =& \begin{pmatrix} \bm{s}_{g1}\{\hat{\bm{\theta}}_g^{\known},\hat{\sigma}_g^{\known}\}\\ s_{g2}\{\hat{\bm{\theta}}_g^{\known},\hat{\sigma}_g^{\known}\} \end{pmatrix} - \begin{pmatrix} \bm{s}_{g1}^{\known}\{\hat{\bm{\theta}}_g^{\known},\hat{\sigma}_g^{\known}\}\\ s_{g2}^{\known}\{\hat{\bm{\theta}}_g^{\known},\hat{\sigma}_g^{\known}\} \end{pmatrix}\\& + \underbrace{\begin{pmatrix} \bm{s}_{g1}^{\known}\{\hat{\bm{\theta}}_g^{\known},\hat{\sigma}_g^{\known}\}\\ s_{g2}^{\known}\{\hat{\bm{\theta}}_g^{\known},\hat{\sigma}_g^{\known}\} \end{pmatrix}}_{0} - \bm{H}_g^{\known}\bm{\Delta} + o_P(\norm*{ \bm{\Delta} }_2)
\end{align*}
for $\bm{\Delta} = (\hat{\bm{\theta}}_g^{\top},\hat{\sigma}_g)^{\top} -  (\{\hat{\bm{\theta}}_g^{\known}\}^{\top},\hat{\sigma}_g^{\known})^{\top}$. Since $\norm*{(\{\hat{\bm{\theta}}_g^{\known}\}^{\top},\hat{\sigma}_g^{\known})^{\top}}_2=O_P(n^{-1/2})$, Lemma~\ref{supp:lemma:sg1sg2} implies $\norm*{ \bm{\Delta} }_2 = o_P(n^{-1/2})$, which completes part (a) (the first part of the proof; see above).

For part (b) (the second part of the proof; see above), let $\bm{F}(\bm{\theta},\sigma) \in \mathbb{R}^{(d+K+1) \times (d+K+1)}$ be the Fisher scoring matrix defined in the statement of Lemma~\ref{supp:lemma:Hessian}. Then for $\hat{\bm{\eta}}_g^{(j)} = ( \{ \hat{\bm{\theta}}^{(j)}_g \}^{\top},\hat{\sigma}_g^{(j)} )^{\top}$ the $j$th Fisher scoring updates, $\hat{\bm{\eta}}_g^{\ipw} = ( \{ \hat{\bm{\theta}}^{\ipw}_g \}^{\top},\hat{\sigma}_g^{\ipw} )^{\top}$, and $\hat{\bm{\eta}}_g = (  \hat{\bm{\theta}}_g^{\top},\hat{\sigma}_g )^{\top}$,
\begin{align*}
    &\hat{\bm{\eta}}_g^{(1)} = \hat{\bm{\eta}}_g^{\ipw} + [ \bm{F}\{ \hat{\bm{\eta}}_g^{\ipw} \} ]^{-1}\begin{pmatrix}\bm{s}_{g1}\{ \hat{\bm{\eta}}_g^{\ipw} \}\\ s_{g2}\{ \hat{\bm{\eta}}_g^{\ipw} \} \end{pmatrix}\\
    &\hat{\bm{\eta}}_g^{(j+1)} = \hat{\bm{\eta}}_g^{(j)} + [ \bm{F}\{ \hat{\bm{\eta}}_g^{(j)} \} ]^{-1}\begin{pmatrix}\bm{s}_{g1}\{ \hat{\bm{\eta}}_g^{(j)} \}\\ s_{g2}\{ \hat{\bm{\eta}}_g^{(j)}\} \end{pmatrix}, \quad j=1,\ldots,m-1.
\end{align*}
We study the behavior of $\hat{\bm{\eta}}_g^{(1)}$ for simplicity, and note the extension to finite $j>1$ is trivial. Since $\norm*{ \hat{\bm{\eta}}_g^{\ipw} - \hat{\bm{\eta}}_g^* }_2 = O_P(n^{-1/2})$ by Lemmas~\ref{supp:lemma:IPWest} and \ref{supp:lemma:IPWestVar}, $\norm*{ \hat{\bm{\eta}}_g^{\ipw} - \hat{\bm{\eta}}_g }_2 = O_P(n^{-1/2})$. Lemma~\ref{supp:lemma:Hessian} then implies
\begin{align*}
    \hat{\bm{\eta}}_g^{(1)} =& \hat{\bm{\eta}}_g^{\ipw} + [ \bm{F}\{ \hat{\bm{\eta}}_g^{\ipw} \} ]^{-1} \left[ \underbrace{\begin{pmatrix}\bm{s}_{g1}( \hat{\bm{\eta}}_g )\\ s_{g2}( \hat{\bm{\eta}}_g ) \end{pmatrix}}_{0} - \smallint_{0}^1 \bm{H}\{ t\hat{\bm{\eta}}_g^{\ipw} + (1-t)\hat{\bm{\eta}}_g \}dt \{ \hat{\bm{\eta}}_g^{\ipw} - \hat{\bm{\eta}}_g \} \right]\\
    =& \hat{\bm{\eta}}_g + o_P(n^{-1/2}),
\end{align*}
which completes the proof.
\end{proof}

\begin{lemma}
\label{supp:lemma:PsiBounds}
Let $c,m,M>0$ be constants and suppose $\Psi(x)$ is a six times continuously differentiable cumulative distribution function, where
\begin{enumerate}[label=(\roman*)]
\item $\Psi(-x)=1-\Psi(x)$ and $\abs*{\Psi^{(j)}(x)} \leq c$ for all $j \in [6]$.
\item $\abs*{x}^m\Psi(x) \geq c$ for all $x < -M$
\item $\abs*{x}^m\abs*{\Psi^{(j)}(x)}\leq c$ for $j\in [6]$ and all $\abs*{x} > M$.
\end{enumerate}
Define $\mu(x,\sigma) = \log\{\smallint\Psi(x + \sigma e) \phi(e)\text{d}e\}$ for all $x \in \mathbb{R}$ and $\sigma \in (s^{-1},s)$ for some constant $s>1$. Then for some constant $\tilde{c}$ and $i,j\geq 0$ such that $i+j \in [3]$,
\begin{align*}
    \abs{ \frac{\partial^{(i+j)}}{\partial x^i \partial \sigma^j }  \mu(x,\sigma)} \leq \tilde{c}.
\end{align*}
\end{lemma}
\begin{proof}
\begin{align*}
    \frac{\partial^{1}}{\partial x^1 }  \mu(x,\sigma) =& \frac{ \smallint \dot{\Psi}(x + \sigma e) \phi(e)de }{ \smallint \Psi(x + \sigma e) \phi(e)de }\\
    \frac{\partial^{2}}{\partial x^2 }  \mu(x,\sigma) =&  \frac{ \smallint \ddot{\Psi}(x + \sigma e) \phi(e)de }{ \smallint \Psi(x + \sigma e) \phi(e)de } - \left\{ \frac{\partial^{1}}{\partial x^1 }  \mu(x,\sigma) \right\}^2\\
    \frac{\partial^{3}}{\partial x^3 }  \mu(x,\sigma) =& \frac{ \smallint \dddot{\Psi}(x + \sigma e) \phi(e)de }{ \smallint \Psi(x + \sigma e) \phi(e)de } - \frac{ \smallint \ddot{\Psi}(x + \sigma e) \phi(e)de }{ \smallint \Psi(x + \sigma e) \phi(e)de } \frac{ \smallint \dot{\Psi}(x + \sigma e) \phi(e)de }{ \smallint \Psi(x + \sigma e) \phi(e)de } \\&- 2 \frac{\partial^{1}}{\partial x^1 }  \mu(x,\sigma) \frac{\partial^{2}}{\partial x^2 }  \mu(x,\sigma)\\
    \frac{\partial^{1}}{\partial \sigma^1 }  \mu(x,\sigma) = &\sigma\frac{ \smallint \ddot{\Psi}(x + \sigma e) \phi(e)de }{ \smallint \Psi(x + \sigma e) \phi(e)de }\\
    \frac{\partial^{2}}{\partial \sigma^2 }  \mu(x,\sigma) =& \sigma^2\frac{ \smallint \Psi^{(4)}(x + \sigma e) \phi(e)de }{ \smallint \Psi(x + \sigma e) \phi(e)de } - \left\{ \frac{\partial^{1}}{\partial \sigma^1 }\mu(x,\sigma) \right\}^2 + \frac{ \smallint \ddot{\Psi}(x + \sigma e) \phi(e)de }{ \smallint \Psi(x + \sigma e) \phi(e)de } \\
    \frac{\partial^{3}}{\partial \sigma^3 }  \mu(x,\sigma) =& \sigma^3\frac{ \smallint \Psi^{(6)}(x + \sigma e) \phi(e)de }{ \smallint \Psi(x + \sigma e) \phi(e)de } - \sigma^3 \frac{ \smallint \Psi^{(4)}(x + \sigma e) \phi(e)de }{ \smallint \Psi(x + \sigma e) \phi(e)de } \frac{ \smallint \ddot{\Psi}(x + \sigma e) \phi(e)de }{ \smallint \Psi(x + \sigma e) \phi(e)de }\\& - 2 \frac{\partial^{1}}{\partial \sigma^1 }\mu(x,\sigma) \frac{\partial^{2}}{\partial \sigma^2 }\mu(x,\sigma) + 3\sigma \frac{ \smallint \Psi^{(4)}(x + \sigma e) \phi(e)de }{ \smallint \Psi(x + \sigma e) \phi(e)de } - \sigma^{-1}\left\{ \frac{\partial^{1}}{\partial \sigma^1 }  \mu(x,\sigma) \right\}^2 \\
    \frac{\partial^{2}}{\partial x^1 \partial \sigma^1}  \mu(x,\sigma) =&  \sigma\frac{ \smallint \dddot{\Psi}(x + \sigma e) \phi(e)de }{ \smallint \Psi(x + \sigma e) \phi(e)de } - \frac{\partial^{1}}{\partial x^1 }  \mu(x,\sigma) \frac{\partial^{1}}{\partial \sigma^1 }  \mu(x,\sigma) \\
    \frac{\partial^{3}}{\partial x^2 \partial \sigma^1}  \mu(x,\sigma) =& \sigma\frac{ \smallint \Psi^{(4)}(x + \sigma e) \phi(e)de }{ \smallint \Psi(x + \sigma e) \phi(e)de } - \sigma\frac{ \smallint \dddot{\Psi}(x + \sigma e) \phi(e)de }{ \smallint \Psi(x + \sigma e) \phi(e)de } \frac{\partial^{1}}{\partial x^1 }  \mu(x,\sigma)\\ &  - \frac{\partial^{2}}{\partial x^2 }  \mu(x,\sigma) \frac{\partial^{1}}{\partial \sigma^1 }  \mu(x,\sigma)- \frac{\partial^{1}}{\partial x^1 }  \mu(x,\sigma) \frac{\partial^{2}}{\partial x^1\partial \sigma^1 }  \mu(x,\sigma)\\
    \frac{\partial^{3}}{\partial x^1 \partial \sigma^2}  \mu(x,\sigma) =& \sigma^2\frac{ \smallint \Psi^{(5)}(x + \sigma e) \phi(e)de }{ \smallint \Psi(x + \sigma e) \phi(e)de } - \frac{ \smallint \dddot{\Psi}(x + \sigma e) \phi(e)de }{ \smallint \Psi(x + \sigma e) \phi(e)de } \left\{ \sigma \frac{\partial^{1}}{\partial \sigma^1 }  \mu(x,\sigma) - 1 \right\}\\ &  - \frac{\partial^{1}}{\partial x^1 }  \mu(x,\sigma) \frac{\partial^{2}}{\partial \sigma^2 }  \mu(x,\sigma)- \frac{\partial^{1}}{\partial \sigma^1 }  \mu(x,\sigma) \frac{\partial^{2}}{\partial x^1\partial \sigma^1 }  \mu(x,\sigma).
\end{align*}
Since $\sigma$ is bounded above 0 and below $\infty$, we therefore only have to show that $\abs{ \frac{ \smallint \Psi^{(j)}(x + \sigma e) \phi(e)de }{ \smallint \Psi(x + \sigma e) \phi(e)de } }$ is bounded from above for all $j \in [6]$ to prove the lemma. First, $\abs*{ \Psi^{(j)}(x + \sigma e) }$ is bounded from above. Second, $\smallint\Psi(x + \sigma e)\phi(e)de>0$ is increasing in $\sigma$ for all $\abs*{x}$ suitably large. Third, $\smallint\Psi(x + \sigma e)\phi(e)de$ is increasing in $x$ for all fixed $\sigma$. The latter two imply $\smallint\Psi(x + \sigma e)\phi(e)de>a_k$ for all $x > -k$ and $\sigma \in (s^{-1},s)$, where $k>0$ and $a_k>0$ is a constant that only depends on $k$. These three imply we need only consider the behavior of $\frac{ \smallint \Psi^{(j)}(x + \sigma e) \phi(e)de }{ \smallint \Psi(x + \sigma e) \phi(e)de }$ when $(-x)$ is large to prove the lemma. Let $M$, $c$, and $m$ be as defined in the statement of Lemma~\ref{supp:lemma:PsiBounds}. Then for $Z \sim N(0,1)$,
\begin{align*}
    \int \Psi(x + \sigma e)\phi(e)de \geq & \int_{-\infty}^{\frac{-M-x}{\sigma}} \Psi( x + \sigma e ) \phi(e)de \geq c \int_{-\infty}^{\frac{-M-x}{\sigma}} \abs*{ x+\sigma e }^{-m} \phi(e)de\\
    = & \E\{ (-x+\sigma Z)^{-m} 1\{ -x+\sigma Z \geq M \} \}\\
    \geq & [\E\{ (-x+\sigma Z) 1\{ -x+\sigma Z \geq M \} \}]^{-m} \geq ( -x/2 )^{-m},
\end{align*}
where the last inequality holds for all $(-x)>0$ sufficiently large. Further, for all $(-x)>0$ sufficiently large and some constant $\epsilon>0$
\begin{align}
\label{supp:equation:PsiBound1}
\begin{aligned}
    \int \Psi^{(j)}(x + \sigma e)\phi(e)de =& \int_{\frac{-M-x}{\sigma}}^{\frac{M-x}{\sigma}} \Psi^{(j)}(x + \sigma e)\phi(e)de + \int_{\frac{M-x}{\sigma}}^{\infty} \Psi^{(j)}(x + \sigma e)\phi(e)de\\& + \int_{-\infty}^{\frac{-M-x}{\sigma}} \Psi^{(j)}(x + \sigma e)\phi(e)de \\
    \leq& \frac{2c M}{\sigma}\phi\left(\frac{-M-x}{\sigma}\right) + \epsilon \frac{ \phi\left(\frac{M-x}{\sigma}\right) }{(M-x)/\sigma}\\
    &+ c \int_{-\infty}^{\frac{-M-x}{\sigma}} \abs*{ x+\sigma e }^{-m}\phi(e)de.
\end{aligned}
\end{align}
First, $\abs*{x}^m\phi\left(\frac{-M-x}{\sigma}\right)$ is bounded from above as a function of $x \in \mathbb{R}$ and $\sigma \in (s^{-1},s)$. Second,
\begin{align}
\label{supp:equation:PsiBound2}
\begin{aligned}
    \int_{-\infty}^{\frac{-M-x}{\sigma}} \abs*{ x+\sigma e }^{-m}\phi(e)de =& \int_{-\infty}^{\frac{-M-x}{2\sigma}} \abs*{ x+\sigma e }^{-m}\phi(e)de + \int_{\frac{-M-x}{2\sigma}}^{\frac{-M-x}{\sigma}} \abs*{ x+\sigma e }^{-m}\phi(e)de\\
    \leq & \abs{ \frac{(-x)+M}{2} }^{-m} + M^{-m}\phi\left\{ \frac{(-x)-M}{2\sigma} \right\} 
\end{aligned}
\end{align}
for all $(-x)>0$ sufficiently large, which completes the proof.
\end{proof}

\begin{remark}
\label{supp:remark:PsiBound}
If we replace condition (iii) in the statement of Lemma~\ref{supp:lemma:PsiBounds} with $\abs*{x}^{m+\delta} \allowbreak\abs*{\Psi^{(j)}(x)}\allowbreak \leq c$ for any $\delta>0$, we would replace $-m$ with $-(m+\delta)$ in \eqref{supp:equation:PsiBound1} and \eqref{supp:equation:PsiBound2}, which would prove that $\abs{ \frac{\partial^{i+j}}{\partial x^i \partial \sigma^j} \mu(x,\sigma) } \to 0$ as $\abs*{x} \to \infty$. This shows that outlying missing data points have a trivial contribution to the gradient of the log-likelihood in \eqref{equation:thetagLikelihood}, suggesting that letting $\Psi$ be the CDF of a t-distribution makes estimation robust to outliers.
\end{remark}

\section{Theoretical guarantees for mtGWAS}
\label{supp:section:TheorymtGWAS}

\subsection{A restatement of Theorem~\ref{theorem:mtGWAS}}
\label{supp:subsection:TheorymtGWAS}
Before proving our results for mtGWAS, we first redefine $\eta_{gs}^{(e)}$, $\eta_{gs}^{(c)}$, and $\eta_{gs}^{(c,e)}$ to all observed nuisance covariates $\bm{x}_i$. First,
\begin{align}
\label{supp:equation:ScoreTest}
\begin{aligned}
\eta_{gs}^{(e)} =& \{ \textstyle\sum_{i=1}^n \frac{\partial}{\partial \gamma} h_{gsi}( \gamma,\hat{\bm{\theta}}_g,\hat{\sigma}_g )\mid_{\gamma=0} \}^2 [ \{ -\mathcal{I}_{gs}(\hat{\bm{\theta}}_g,\hat{\sigma}_g) \}^{-1} ]_{11}\\
\{\hat{\bm{\theta}}_g,\hat{\sigma}_g\} =& \argmax_{\bm{\theta}\in \mathbb{R}^{d+K}, \sigma \in \mathbb{R}_{+}} \sum_{i=1}^n h_{gsi}( 0,\bm{\theta},\sigma )\\
h_{gsi}( \gamma,\bm{\theta},\sigma ) =& -r_{gi}\{ y_{gi}-(\bm{\theta}^{\top}\hat{\bm{z}}_i + \gamma G_{si})\}^2/(2\sigma^2)\\ &+ (1-r_{gi})\log[1- \smallint \Psi\{ \hat{\alpha}_g( \bm{\theta}^{\top}\hat{\bm{z}}_i + \gamma G_{si} + \sigma e - \hat{\delta}_g ) \}\phi(e)\text{d}e ]\\ \hat{\bm{z}}_{i} =& (\bm{x}_i^{\top},\hat{\bm{c}}_i)^{\top},
\end{aligned}
\end{align}
where we solve the optimization problem in the second line using the one-step Fisher scoring algorithm detailed in the statement of Theorem~\ref{theorem:betag}. The matrix $\mathcal{I}_{gs}(\bm{\theta},\sigma)$ is the standard $(K+d+1) \times (K+d+1)$ Fisher information matrix evaluated at $\{\bm{\theta},\sigma\}$ and using covariates $\hat{\bm{z}}_i$. We next define $\eta_{gs}^{(c)}$ to be
\begin{align}
\label{supp:equation:ScoreTest:c}
    \eta_{gs}^{(c)} = \frac{ \{\hat{\bm{\ell}}_g^{\top} \hat{\bm{\gamma}}_s^{(c)}\}^2 }{ \hat{\bm{\ell}}_g^{\top}\hat{\V}\{ \hat{\bm{\gamma}}_{s}^{(c)} \}\hat{\bm{\ell}}_g + \{\hat{\bm{\gamma}}_s^{(c)}\}^{\top}\hat{\V}(\hat{\bm{\ell}}_g)\hat{\bm{\gamma}}_s^{(c)} }, \quad  \hat{\bm{\gamma}}_s^{(c)} = (\bm{G}_s^{\top} P_{\bm{X}}^{\perp} \bm{G}_s)^{-1}P_{\bm{X}}^{\perp} \hat{\bm{C}},
\end{align}
where $\bm{G}_s = (G_{s1},\ldots,G_{sn})^{\top}$. The estimate $\hat{\bm{\ell}}_g$ is the appropriate sub-vector of $\hat{\bm{\theta}}_g$ defined in \eqref{supp:equation:ScoreTest}, $\hat{\V}(\hat{\bm{\ell}}_g)$ is the appropriate $K \times K$ sub-matrix $\mathcal{I}_{gs}(\hat{\bm{\theta}}_g,\hat{\sigma}_g)$ defined in \eqref{supp:equation:ScoreTest}, and $\hat{\V}\{ \hat{\bm{\gamma}}_{s}^{(c)} \}$ is the usual ordinary least squares estimate for the variance of $\hat{\bm{\gamma}}_{s}^{(c)}$ from the regression of $\hat{\bm{C}}$ onto $\bm{G}_s$ and $\bm{X}$. We can now re-state Theorem~\ref{theorem:mtGWAS}.

\begin{theorem}
\label{supp:theorem:mtGWAS}
Suppose Assumption~\ref{supp:assumptions:FA} holds, fix a $g \in [p]$, and let $\eta_{gs}^{(e)}$ and $\eta_{gs}^{(c)}$ be as defined in \eqref{supp:equation:ScoreTest} and \eqref{supp:equation:ScoreTest:c}. Then $\eta_{gs}^{(e)} \tdist \chi^2_1$ if $H_{0,gs}^{(e)}: \gamma_{gs}^{(e)} = 0$ is true. If (i) $n^{1/2}\norm*{\bm{\ell}_g}_2 \to \infty$ and (ii) $\E(\bm{c}_i \mid G_{si}) = \bm{A}^{\top}\bm{x}_i + \bm{\gamma}_s^{(c)} G_{si}$ for some non-random $\bm{A} \in \mathbb{R}^{d \times K}$, then $\eta_{gs}^{(c)} \tdist \chi^2_1$ if $H_{0,gs}^{(c)}: \bm{\ell}_g^{\top}\gamma_{s}^{(c)} = 0$ is true and $\eta_{gs}^{(c,e)} = \eta_{gs}^{(c)} + \eta_{gs}^{(e)} \tdist \chi^2_2$ if $H_{0,gs}^{(c,e)}: \bm{\ell}_g^{\top}\gamma_{s}^{(c)} = \gamma_{gs}^{(e)} = 0$ is true.
\end{theorem}

\subsection{Proof of Theorem~\ref{supp:theorem:mtGWAS}}
\label{supp:subsection:mtGWASproof}
We prove Theorem~\ref{supp:theorem:mtGWAS} by first showing that $\eta_{gs}^{(e)}$ and $\eta_{gs}^{(c)}$ are asymptotically equivalent to the corresponding quantities when $\bm{C}$ is known and when we account for all genetic effects on $e_{gi}$.

\begin{lemma}
\label{supp:lemma:Score:etae:AE}
Fix a $g \in [p]$ and $s \in [S]$, suppose Assumption~\ref{supp:assumptions:FA} holds, and let $\bm{z}_i = (\bm{x}_i^{\top},\bm{c}_i^{\top})^{\top}$. Define $\mathcal{H}_g = \{r \in [S]: \gamma^{(e)}_{gr} \neq 0\}$ and $a_{gs}, \bm{A}_{gs}, a_{gs}^{\known},\bm{A}_{gs}^{\known}$ to be
\begin{align*}
    &a_{gs} = n^{-1}\sum_{i=1}^n \frac{\partial}{\partial \gamma} h_{gsi}( \gamma,\hat{\bm{\theta}}_g,\hat{\sigma}_g) \mid_{\gamma=0}, \quad a_{gs}^{\known} = n^{-1}\sum_{i=1}^n \frac{\partial}{\partial \gamma} h_{gsi}^{\known}\{\gamma,\hat{\bm{\theta}}_g^{\known},\hat{\sigma}_g^{\known}\}\mid_{\gamma=0}\\
    &\bm{A}_{gs} = n^{-1}\mathcal{I}_{gs}( \hat{\bm{\theta}}_g,\hat{\sigma}_g), \quad n^{-1}\bm{A}_{gs}^{\known} = \mathcal{I}_{gs}^{\known}\{ \hat{\bm{\theta}}_g^{\known},\hat{\sigma}_g^{\known} \}\\
    &h_{gsi}^{\known}(\gamma,\theta,\sigma) =  -r_{gi}\log (\sigma) -r_{gi}\left[ y_{gi} - \left\{\bm{z}_i^{\top}\bm{\theta} + G_{si}\gamma + \sum_{r \in \mathcal{H}_g \setminus \{ s \}}\gamma^{(e)}_{ri} G_{ri}\right\} \right]^2/(2\sigma^2)\\ &+ (1-r_{gi})\log\left( \int \Psi\left[ \alpha_g\left\{\bm{z}_i^{\top}\bm{\theta} + G_{si}\gamma + \sum_{r \in \mathcal{H}_g \setminus \{ s \}}\gamma^{(e)}_{ri} G_{ri} + \sigma e\right\} \right]\phi(e)\text{d}e \right)\\
    &\{\hat{\bm{\theta}}_g^{\known},\hat{\sigma}_g^{\known}\} = \argmax_{\bm{\theta} \in \Theta, \sigma \in \mathcal{S}} \sum_{i=1}^n h_{gsi}(0,\bm{\theta},\sigma).
\end{align*}
where $h_{gsi}$, $\{\hat{\bm{\theta}}_g,\hat{\sigma}\}$, and $\mathcal{I}_{gs}( \hat{\bm{\theta}}_g,\hat{\sigma}_g)$ are defined in \eqref{supp:equation:ScoreTest}, $\Theta,\mathcal{S}$ are as defined in the statement of Lemma~\ref{supp:lemma:Cknown}, and $\mathcal{I}_{gs}^{\known}\{ \bm{\theta},\sigma \}$ is the corresponding $(d+K+1) \times (d+K+1)$ minus Fisher information matrix evaluated at $\{\gamma=0,\bm{\theta},\sigma\}$. Then if the null hypothesis $H_{0,gs}^{(e)}:\gamma_{gs}^{(e)}=0 $ is true, then $n^{1/2}\abs*{ a_{gs} - a_{gs}^{\known} } = o_P(1)$ and $\norm*{ \bm{A}_{gs} - \bm{A}_{gs}^{\known} }_2 = o_P(1)$.
\end{lemma}

\begin{proof}
Note that $\frac{\partial}{\partial \gamma}  h_{gsi}^{\known}(\gamma,\bm{\theta},\sigma)\mid_{\gamma=0}$ and $\frac{\partial}{\partial \gamma}  h_{gsi}(\gamma,\bm{\theta},\sigma)\mid_{\gamma=0}$ are exactly the score functions from Lemma~\ref{supp:lemma:Cknown} and Theorem~\ref{supp:theorem:InferenceBeta}. Therefore, the results are a simple consequence of the proofs and results of Lemma~\ref{supp:lemma:Cknown} and Theorem~\ref{supp:theorem:InferenceBeta}.
\end{proof}

\begin{lemma}
\label{supp:lemma:Score:etac:AE}
Fix an $s \in [S]$, suppose Assumption~\ref{supp:assumptions:FA} holds, and let $\bm{b}_{gs}$, $\bm{B}_{gs}$, $\bm{b}_{gs}^{\known}$, and $\bm{B}_{gs}^{\known}$ be
\begin{align*}
    &\hat{\bm{\gamma}}_{s}^{(c)} = \{ (\bm{G}_s^{\top}P_{\bm{X}}^{\perp}\bm{G}_s)^{-1}\bm{G}_s^{\top}P_{\bm{X}}^{\perp} (n^{1/2}\hat{\bm{C}}_{\perp}) \}^{\top}, \quad \hat{\bm{\gamma}}_{s}^{(c),\known} = \{ (\bm{G}_s^{\top}P_{\bm{X}}^{\perp}\bm{G}_s)^{-1}\bm{G}_s^{\top}P_{\bm{X}}^{\perp} (n^{1/2}\tilde{\bm{C}}) \}^{\top}\\
    &\hat{\V}\{\hat{\bm{\gamma}}_{s}^{(c)}\} = n^{-1}(n^{1/2}\hat{\bm{C}}_{\perp})^{\top}P_{[\bm{X},\bm{G}_s]}^{\perp}(n^{1/2}\hat{\bm{C}}_{\perp}), \quad \hat{\V}\{\hat{\bm{\gamma}}_{s}^{(c),\known}\} = n^{-1}(n^{1/2}\tilde{\bm{C}})^{\top}P_{[\bm{X},\bm{G}_s]}^{\perp}(n^{1/2}\tilde{\bm{C}})
\end{align*}
for $\bm{G}_s=(G_{s1},\ldots,G_{sn})^{\top}$ and $\tilde{\bm{C}}$ defined in \eqref{supp:equation:Cobj}. Then for unitary matrix $\bm{v} \in \mathbb{R}^{K \times K}$ as defined in the statement of Theorem~\ref{supp:theorem:ChatProp}, $n^{1/2}\norm*{\hat{\bm{\gamma}}_{s}^{(c)} - \bm{v}^{\top}\hat{\bm{\gamma}}_{s}^{(c),\known}}_2 = o_P(1)$ and $\norm*{ \hat{\V}\{\hat{\bm{\gamma}}_{s}^{(c)}\} - \bm{v}^{\top}\hat{\V}\{\hat{\bm{\gamma}}_{s}^{(c),\known}\} \bm{v}}_2 = o_P(1)$.
\end{lemma}

\begin{proof}
The vector $\hat{\bm{\gamma}}_{s}^{(c)} \in \mathbb{R}^{K}$ is exactly the first column of
\begin{align*}
    \begin{pmatrix} n^{-1/2}\hat{\bm{C}}_{\perp}^{\top}\bm{G}_s & n^{-1/2}\hat{\bm{C}}_{\perp}^{\top}\bm{X} \end{pmatrix} \begin{pmatrix} n^{-1}\bm{G}_s^{\top}\bm{G}_s & n^{-1}\bm{G}_s^{\top}\bm{X}\\
    n^{-1}\bm{X}^{\top}\bm{G}_s & n^{-1}\bm{X}^{\top}\bm{X}\end{pmatrix}^{-1}.
\end{align*}
Therefore, to prove $n^{1/2}\norm*{ \hat{\bm{\gamma}}_{s}^{(c)} - \bm{v}^{\top}\hat{\bm{\gamma}}_{s}^{(c),\known} }_2=o_P(1)$, we need only show that $\norm*{ \hat{\bm{C}}_{\perp}^{\top} (n^{-1/2}\bm{G}_s) - \bm{v}^{\top}\tilde{\bm{C}}^{\top} (n^{-1/2}\bm{G}_s) }_2 = o_P(n^{-1/2})$ and $\norm*{ \hat{\bm{C}}_{\perp}^{\top} (n^{-1/2}\bm{X}) - \bm{v}^{\top}\tilde{\bm{C}}^{\top} (n^{-1/2}\bm{X}) }_2 = o_P(n^{-1/2})$. However, this can easily be shown using the exact same techniques used to prove Corollary~\ref{supp:corollary:CtC}. The same goes for showing that $\norm*{ \hat{\V}\{\hat{\bm{\gamma}}_{s}^{(c)}\} - \bm{v}^{\top}\hat{\V}\{\hat{\bm{\gamma}}_{s}^{(c),\known}\}\bm{v} }_2 = o_P(1)$. The details have been omitted.
\end{proof}

\begin{lemma}
\label{supp:lemma:ell:AE}
Fix a $g \in [p]$, suppose Assumption~\ref{supp:assumptions:FA} holds, and let $\mathcal{H}_g = \{s \in [S]: \gamma^{(e)}_{gs} \neq 0\}$. Let $\hat{\bm{z}}_i = ( \bm{x}_i^{\top},\hat{\bm{c}}_i )^{\top}$ and $\tilde{\bm{z}}_i = (\bm{x}_i^{\top},n^{1/2}\tilde{\bm{C}}_{i*}^{\top})$ for $\tilde{\bm{C}}$ given in \eqref{supp:equation:Cobj}. Let $\Theta$ and $\mathcal{S}$ be as given in the statement of Lemma~\ref{supp:lemma:Cknown}, and define $\hat{\bm{\theta}}_g$ and $\hat{\bm{\theta}}_g^{\known}$ to be
\begin{align*}
    &\{\hat{\bm{\theta}}_g,\hat{\sigma}_g\} = \argmax_{\bm{\theta},\sigma} \sum_{i=1}^n f_{gi}(\bm{\theta},\sigma), \quad \{\hat{\bm{\theta}}_g^{\known},\hat{\sigma}_g\} = \argmax_{\bm{\theta} \in \Theta,\sigma \in \mathcal{S}} \sum_{i=1}^n \tilde{f}_{gi}^{\known}(\bm{\theta},\sigma)\\
    &f_{gi}(\bm{\theta},\sigma) = -r_{gi}\log(\sigma) - r_{gi}\left( y_{gi} - \hat{\bm{z}}_i^{\top}\bm{\theta} \right)^2/(2\sigma^2)+ (1-r_{gi})\log[ \smallint \Psi\{ \alpha_g(\bm{z}_i^{\top}\bm{\theta}  + \sigma e) \}\phi(e)\text{d}e ]\\
    &\tilde{f}_{gi}^{\known}(\bm{\theta},\sigma) = -r_{gi}\log (\sigma) -r_{gi}\left[ y_{gi} - \left\{\tilde{\bm{z}}_i^{\top}\bm{\theta} + \sum_{s \in \mathcal{H}_g }\gamma^{(e)}_{si} G_{si}\right\} \right]^2/(2\sigma^2)\\ &+ (1-r_{gi})\log\left( \int \Psi\left[ \alpha_g\left\{\tilde{\bm{z}}_i^{\top}\bm{\theta} + \sum_{s \in \mathcal{H}_g}\gamma^{(e)}_{si} G_{si} + \sigma e\right\} \right]\phi(e)\text{d}e \right),
\end{align*}
where the first optimization is solved using the one step Fisher scoring method outlined in the statement of Theorem~\ref{theorem:betag}. Define $\hat{\bm{\ell}}_g$ and $\hat{\bm{\ell}}_g^{\known}$ to be the last $K$ elements of $\hat{\bm{\theta}}_g$ and $\hat{\bm{\theta}}_g^{\known}$, and let $\hat{\V}(\hat{\bm{\ell}}_g)$ and $\hat{\V}\{\hat{\bm{\ell}}_g^{\known}\}$ be the standard minus Fisher information-derived estimates for the variances. Then for $\bm{v}$ given in the statement of Theorem~\ref{supp:theorem:ChatProp}, $n^{1/2}\norm*{ \hat{\bm{\ell}}_g - \bm{v}^{\top}\hat{\bm{\ell}}_g^{\known} } = o_P(1)$ and $n\norm*{ \hat{\V}(\hat{\bm{\ell}}_g) - \bm{v}^{\top}\hat{\V}\{\hat{\bm{\ell}}_g^{\known}\}\bm{v} }_2 = o_P(1)$.
\end{lemma}
\begin{proof}
This is a direct consequence of the proofs of Lemma~\ref{supp:lemma:Cknown} and Theorem~\ref{supp:theorem:InferenceBeta}.
\end{proof}

We use these three lemmas to prove Theorem~\ref{supp:theorem:mtGWAS}.

\begin{proof}[Proof of Theorem~\ref{supp:theorem:mtGWAS}]
We first prove the properties of $\eta_{gs}^{(e)}$. If $H_{0,gs}^{(e)}$ is true, Lemma~\ref{supp:lemma:Score:etae:AE} implies it suffices to study the properties of $\eta_{gs}^{(e),\known} = \{a_{gs}^{\known}\}^2/[\{\bm{A}_{gs}^{\known}\}^{-1}]_{11}$. However, standard techniques can be used to show that this satisfies $\eta_{gs}^{(e),\known} \tdist \chi^2_1$.

We next consider $\eta_{gs}^{(c)}$. Let $\hat{\bm{\gamma}}_s^{(c)},\hat{\bm{\gamma}}_s^{(c),\known}$ and $\hat{\bm{\ell}}_g, \hat{\bm{\ell}}_g^{\known}$ be as defined in Lemmas~\ref{supp:lemma:Score:etac:AE} and \ref{supp:lemma:ell:AE}. To simplify notation, let $\hat{\bm{\ell}} = \hat{\bm{\ell}}_g$, $\bar{\bm{\ell}} = \hat{\bm{\ell}}_g^{\known}$, $\hat{\bm{\gamma}} = \hat{\bm{\gamma}}_s^{(c)}$, and $\bar{\bm{\gamma}} = \hat{\bm{\gamma}}_s^{(c),\known}$. First, since $\bar{\bm{\ell}}_g$ is estimated conditional on $\bm{C}$, it is straightforward to show that for $\tilde{\bm{\ell}}$ as defined in \eqref{supp:equation:Cobj},
\begin{align*}
    n^{1/2}\begin{pmatrix} \bar{\bm{\ell}} - \tilde{\bm{\ell}}\\ \bar{\bm{\gamma}} - \bm{\gamma}_s^{(c)} \end{pmatrix} = \begin{pmatrix} \bm{W}_{\ell}\\ \bm{W}_{\gamma} \end{pmatrix} + o_P(1),
\end{align*}
where $\bm{W}_{\ell} \sim N_K(0,n\hat{\V}(\bar{\bm{\ell}}))$ and $\bm{W}_{\gamma} \sim N_K(0,n\hat{\V}(\bar{\bm{\gamma}}))$ are independent. Since $\Prob\{ n\hat{\V}(\bar{\bm{\ell}}) \succeq c I_K \}$ and $\Prob\{ n\hat{\V}(\bar{\bm{\gamma}}) \succeq c I_K \}$ go to 1 as $n \to \infty$ for some constant $c>0$ small enough,
\begin{align*}
    b_{gs}=\frac{n^{1/2}\bar{\bm{\ell}}^{\top}\bar{\bm{\gamma}}}{ \{ n\bar{\bm{\ell}}^{\top}\hat{\V}(\bar{\bm{\gamma}})\bar{\bm{\ell}} + n\bar{\bm{\gamma}}^{\top}\hat{\V}(\bar{\bm{\ell}})\bar{\bm{\gamma}} \}^{1/2} } \tdist N(0,1)
\end{align*}
when $H_{0,gs}^{(c)}$ is true by the assumption that $n^{1/2}\norm*{ \bm{\ell}_g }_2 \to \infty$. Next, Lemmas~\ref{supp:lemma:Score:etac:AE} and \ref{supp:lemma:ell:AE} imply
\begin{align*}
    &n^{1/2}\abs*{ \hat{\bm{\ell}}^{\top}\hat{\bm{\gamma}} - \bar{\bm{\ell}}^{\top}\bar{\bm{\gamma}} } = o_P( \norm*{\bm{\gamma}_s^{(c)}}_2 + \norm*{\bm{\ell}_g}_2 + n^{-1/2} ) = o_P( \norm*{\bm{\gamma}_s^{(c)}}_2 + \norm*{\bm{\ell}_g}_2 )\\
    &\norm*{ n\hat{\V}(\hat{\bm{\gamma}})-n\hat{\V}(\bar{\bm{\gamma}}) }_2, \, \norm*{ n\hat{\V}(\hat{\bm{\ell}})-n\hat{\V}(\bar{\bm{\ell}}) }_2 = o_P(1),
\end{align*}
where the second equality in the first line follows from the fact that $n^{1/2}\norm*{\bm{\ell}_g}_2 \to \infty$. This proves $\eta_{gs}^{(c)} \tdist \chi^2_1$ when $H_{0,gs}^{(c)}$ is true.

We lastly prove $\eta_{gs}^{(c,e)} = \eta_{gs}^{(c)} + \eta_{gs}^{(e)} \tdist \chi^2_2$ when $H_{0,gs}^{(c,e)}$ is true, which by Lemmas~\ref{supp:lemma:Score:etae:AE}, \ref{supp:lemma:Score:etac:AE}, and \ref{supp:lemma:ell:AE}, holds if $\eta_{gs}^{(e),\known}$ is asymptotically independent of $\eta_{gs}^{(c),\known} = b_{gs}^2$ under $H_{0,gs}^{(c,e)}$. Asymptotic independence holds because $a_{gs}^{\known}$, $\bar{\bm{\ell}}$, and $\bar{\bm{\gamma}}$ are asymptotically independent, which completes the proof.
\end{proof}

\subsection{The computational efficiency of mtGWAS test statistics}
\label{supp:subsection:mtGWASComp}
The test statistic $\eta_{sg}^{(c)}$ involves simply regressing estimated latent factors onto genotype via ordinary least squares, and is therefore easy to compute at the genome-wide scale. For $\eta_{sg}^{(e)}$, the partial derivative in \eqref{supp:equation:ScoreTest} can be expressed as
\begin{align*}
    &\sum_{i=1}^n\frac{\partial}{\partial \gamma} h_{gsi}(\gamma, \hat{\bm{\theta}}_g,\hat{\sigma}_g) \mid_{\gamma=0} = \sum_{i=1}^n G_{si}s_{gi}(\hat{\bm{\theta}}_g,\hat{\sigma}_g)\\
    &s_{gi}(\bm{\theta},\sigma) = r_{gi}(y_{gi} - \bm{\theta}^{\top}\hat{\bm{z}}_i)/\sigma^2 - (1-r_{gi})\frac{\hat{\alpha}_g \smallint \dot{\Psi}\{ \hat{\alpha}_g(\bm{\theta}^{\top}\hat{\bm{z}}_i + \sigma e - \hat{\delta}_g) \}\phi(e)\text{d}e}{1 - \smallint \Psi \{ \hat{\alpha}_g(\bm{\theta}^{\top}\hat{\bm{z}}_i + \sigma e - \hat{\delta}_g) \}\phi(e)\text{d}e}
\end{align*}
for $\dot{\Psi}(x) = \frac{d}{dx} \Psi(x)$. Since $s_{gi}(\hat{\bm{\theta}}_g,\hat{\sigma}_g)$ does not depend on genotype, it can pre-computed. For the minus inverse Fisher information, let $\hat{\bm{D}}_g^{(11)} = \diag\{d_{g1}^{(11)}(\hat{\bm{\theta}}_g,\hat{\sigma}_g),\ldots,d_{gn}^{(11)}(\hat{\bm{\theta}}_g,\hat{\sigma}_g)\}$, $\hat{\bm{D}}_g^{(12)} = \diag\{d_{g1}^{(12)}(\hat{\bm{\theta}}_g,\hat{\sigma}_g),\ldots,d_{gn}^{(12)}(\hat{\bm{\theta}}_g,\hat{\sigma}_g)\}$, and $\hat{\bm{D}}_{g}^{(22)} = \diag\{ d_{g1}^{(22)}(\hat{\bm{\theta}}_g,\hat{\sigma}_g), \ldots, d_{gn}^{(22)}(\hat{\bm{\theta}}_g,\hat{\sigma}_g) \}$, where
\begin{align*}
    d_{gi}^{(11)}(\bm{\theta},\sigma) =& \sigma^{-2}\smallint\Psi\{ \hat{\alpha}_g(\bm{\theta}^{\top}\hat{\bm{z}}_i + \sigma e - \hat{\delta}_g) \} \phi(e)de - \hat{\alpha}_g^2 \smallint\ddot{\Psi}\{ -\hat{\alpha}_g(\bm{\theta}^{\top}\hat{\bm{z}}_i + \sigma e - \hat{\delta}_g) \} \phi(e)de\\& + \hat{\alpha}_g^2 \frac{ [ \smallint\dot{\Psi}\{ -\hat{\alpha}_g(\bm{\theta}^{\top}\hat{\bm{z}}_i + \sigma e - \hat{\delta}_g) \} \phi(e)de ]^2 }{ \smallint \Psi\{ -\hat{\alpha}_g(\bm{\theta}^{\top}\hat{\bm{z}}_i + \sigma e - \hat{\delta}_g) \} \phi(e)de }\\
    d_{gi}^{(12)}(\bm{\theta},\sigma) =& 2\sigma^{-2} \hat{\alpha}_g \smallint\dot{\Psi}\{ -\hat{\alpha}_g(\bm{\theta}^{\top}\hat{\bm{z}}_i + \sigma e - \hat{\delta}_g) \} \phi(e)de\\& + \hat{\alpha}_g^3\sigma \smallint \dddot{\Psi}\{ -\hat{\alpha}_g(\bm{\theta}^{\top}\hat{\bm{z}}_i + \sigma e - \hat{\delta}_g) \} \phi(e)de\\ &- \hat{\alpha}_g^3 \sigma\frac{ \smallint\dot{\Psi}\{ -\hat{\alpha}_g(\bm{\theta}^{\top}\hat{\bm{z}}_i + \sigma e - \hat{\delta}_g) \} \phi(e)de \smallint\ddot{\Psi}\{ -\hat{\alpha}_g(\bm{\theta}^{\top}\hat{\bm{z}}_i + \sigma e - \hat{\delta}_g) \} \phi(e)de }{ \smallint \Psi\{ -\hat{\alpha}_g(\bm{\theta}^{\top}\hat{\bm{z}}_i + \sigma e - \hat{\delta}_g) \} \phi(e)de }\\
    d_{gi}^{(22)}(\bm{\theta},\sigma) =& 2 \sigma^{-2}\smallint\Psi\{ \hat{\alpha}_g(\bm{\theta}^{\top}\hat{\bm{z}}_i + \sigma e - \hat{\delta}_g) \} \phi(e)de\\& - 4\hat{\alpha}_g^2 \smallint\ddot{\Psi}\{ -\hat{\alpha}_g(\bm{\theta}^{\top}\hat{\bm{z}}_i + \sigma e - \hat{\delta}_g) \} \phi(e)de\\ &+  \hat{\alpha}_g^4 \sigma^2\frac{ [\smallint\ddot{\Psi}\{ -\hat{\alpha}_g(\bm{\theta}^{\top}\hat{\bm{z}}_i + \sigma e - \hat{\delta}_g) \} \phi(e)de]^2 }{ \smallint\Psi\{ -\hat{\alpha}_g(\bm{\theta}^{\top}\hat{\bm{z}}_i + \sigma e - \hat{\delta}_g) \} \phi(e)de }\\& - \hat{\alpha}_g^4 \sigma^2 \smallint\ddddot{\Psi}\{ -\hat{\alpha}_g(\bm{\theta}^{\top}\hat{\bm{z}}_i + \sigma e - \hat{\delta}_g) \} \phi(e)de
\end{align*}
for $\ddot{\Psi}(x)$, $\dddot{\Psi}(x)$, and $\ddddot{\Psi}(x)$ the second, third, and fourth derivatives of $\Psi$. Then $[\{-\mathcal{I}_{gs}(\hat{\bm{\theta}}_g,\hat{\sigma}_g)\}^{-1}]_{11}$ is exactly the first diagonal element of
\begin{align*}
    \begin{pmatrix}
\bm{G}_s^{\top} \hat{\bm{D}}_g^{(11)}\bm{G}_s & \bm{G}_s^{\top} \hat{\bm{D}}_g^{(11)}\hat{\bm{Z}} & \bm{G}_s^{\top} \hat{\bm{D}}_g^{(12)}\bm{1}_n\\
\hat{\bm{Z}}^{\top}\hat{\bm{D}}_g^{(11)}\bm{G}_s & \hat{\bm{Z}}^{\top}\hat{\bm{D}}_g^{(11)}\hat{\bm{Z}} & \hat{\bm{Z}}_s^{\top} \hat{\bm{D}}_g^{(12)}\bm{1}_n\\
\bm{1}_n^{\top} \hat{\bm{D}}_g^{(12)}\bm{G}_s & \bm{1}_n^{\top} \hat{\bm{D}}_g^{(12)}\hat{\bm{Z}} & \bm{1}_n^{\top} \hat{\bm{D}}_g^{(22)}\bm{1}_n
\end{pmatrix}^{-1},
\end{align*}
where $\bm{G}_s = (G_{s1},\ldots,G_{sn})^{\top}$ and $\hat{\bm{Z}} = (\hat{\bm{z}}_1 \cdots \hat{\bm{z}}_n)^{\top}$. Since $\hat{\bm{D}}_g^{(11)}$, $\hat{\bm{D}}_g^{(12)}$, and $\hat{\bm{D}}_g^{(22)}$ do not depend on genotype, they can be pre-computed.
\newpage

\printbibliography


\end{document}